\documentclass[10pt,a4paper]{article}
\usepackage[T1]{fontenc}
\usepackage[utf8]{inputenc}
\usepackage[english]{babel}
\usepackage{lmodern,microtype,amsmath,amssymb,amsthm,mathtools}
\usepackage[margin=27mm]{geometry}
\usepackage{graphicx,tikz}
\usetikzlibrary{arrows.meta}
\usepackage[hidelinks]{hyperref}
\numberwithin{equation}{section}
\numberwithin{figure}{section}
\newtheorem{theorem}{Theorem}[section]
\newtheorem{proposition}[theorem]{Proposition}
\newtheorem{lemma}[theorem]{Lemma}
\newtheorem{corollary}[theorem]{Corollary}
\newcommand{\be}{\begin{equation}}
\newcommand{\ee}{\end{equation}}
\newcommand{\bse}{\begin{subequations}}
\newcommand{\ese}{\end{subequations}}
\newcommand{\del}{\partial}
\newcommand{\Ric}{\operatorname{Ric}}
\newcommand{\Hess}{\operatorname{Hess}}
\newcommand{\Scal}{\operatorname{Scal}}
\makeatletter
\renewcommand{\@seccntformat}[1]{\@ifundefined{#1numberformat}{\csname the#1\endcsname\quad}{\csname #1numberformat\endcsname}}

\renewcommand{\@pnumwidth}{4.5em}
\renewcommand{\@tocrmarg}{5.5em}
\makeatother
\newcommand{\Ocal}{\mathcal O}
\makeatletter
\renewcommand\paragraph{\@startsection{paragraph}{4}{\z@}{1ex \@plus .3ex \@minus .2ex}{-1em}{\normalfont\normalsize\bfseries}}
\makeatother
\title{\resizebox{\textwidth}{!}{\bf Existence of static axisymmetric solutions to the Einstein equations}}
\author{Philippe G. LeFloch and Luc Rozoy}
\date{}
\begin{document}
\maketitle

\footnotetext[1]{Laboratoire Jacques-Louis Lions,
Centre National de la Recherche Scientifique, Sorbonne Universit\'e,
4 Place Jussieu, 75252 Paris, France.
Email: {\tt contact@philippelefloch.org}.}

\footnotetext[2]{Université Grenoble Alpes, 38058 Grenoble, France.
Email: {\tt rozoy.luc@gmail.com}.}

\begin{abstract}
We construct smooth Riemannian metrics and scalar potentials that are invariant under a circle action and satisfy the static Einstein equations with positive cosmological constant. The underlying compact three-manifold is a \emph{ball with two disjoint open balls removed}. The potential (or lapse function) is positive in the interior and vanishes at the boundary, consisting of three non-degenerate horizons. The areas of the two inner boundary components are controlled by two arbitrary parameters, namely the masses of the black holes, which are taken to be sufficiently small and comparable. Each inner horizon radius scales linearly with the individual mass parameters, while the separation of their reference centers scales as the cube root of the total mass. These two scales are assumed to be sufficiently small at \emph{different rates}, whereas the cosmological length remains fixed. The main analytical challenge is to control the linearized equation and the nonlinear remainder simultaneously near the shrinking horizons, between the centers and in the outer region. Weighted H\"older estimates quantify the growth of the inverse and show that the error of an approximation joining \emph{Kottler collars} to \emph{de Sitter space} is small enough to close a contraction argument. Namely, we solve a \emph{reduced elliptic Einstein equation} that is endowed with \emph{two obstruction tensors}; varying the two axial positions removes their coefficients. The contracted Bianchi identity and an energy estimate for an auxiliary vector field (Killing field) introduced by the projection to vanish. Finally, for a fixed cosmological constant, the intrinsic metric of one of the two boundary components is sufficient to determine a static solution with simply connected interior and non-vanishing boundary gradient, \emph{up to} isometry and multiplication of the potential by a positive constant.
\end{abstract}


\

{\small 

\setcounter{tocdepth}{1}
\begingroup
\normalsize
\setlength{\parskip}{0pt}
\setlength{\baselineskip}{12pt}
\makeatletter
\renewcommand{\l@section}[2]{\vskip 0pt plus 1fill\@dottedtocline{1}{0pt}{1.8em}{\bfseries #1}{\bfseries #2}}
\par\addvspace{.4\baselineskip}
\noindent\textbf{\contentsname}\par\nobreak\vspace{.3\baselineskip}
\@starttoc{toc}
\makeatother
\endgroup
}

\

\clearpage

\section{Introduction}
\label{LS-section-introduction}

\subsection{The static Einstein equations}
\label{LS-section-equations}

We construct axisymmetric solutions of the {static Einstein equations} with positive cosmological constant.
This article is motivated by the work of Dias, Gibbons, Santos and Way~\cite{P2-Dias}, who \emph{numerically}constructed static solutions that are axisymmetric but not spherically symmetric, in contrast with the standard spherical symmetry property imposed by classical black-hole uniqueness theorems (in the asymptotically flat setting, say). Here, we establish that such solutions do exist, by formulating and solving a nonlinear elliptic problem together with two scalar compatibility conditions and introducing a suitable scalar of parameters involving the mass and the location of the black holes.

We consider a connected, compact, orientable three-dimensional manifold $M$ with nonempty smooth boundary, a smooth Riemannian metric $h$ on $\overline M$, and a smooth function $N$ satisfying $N>0$ in its interior $M^\circ$ and $N=0$ on $\del M$, whose connected components are the \emph{horizons}. We call $N$ the \emph{lapse function}. We denote the \emph{cosmological constant} by $\Lambda$. The four-dimensional \emph{static metric} must satisfy
\be
\label{LS-eq-spacetime}
\mathbf g=-N^2dt^2+h,\qquad \Ric_{\mathbf g}=\Lambda\mathbf g,\qquad \Lambda>0, 
\ee
so that\footnote{On $N>0$, the nonzero time connection coefficients are $\Gamma^t_{it}=N^{-1}\del_iN$ and $\Gamma^i_{tt}=Nh^{ij}\del_jN$. Their contraction gives $\Ric_{\mathbf g,ij}=\Ric_{h,ij}-N^{-1}\Hess_hN_{ij}$, $\Ric_{\mathbf g,tt}=N\Delta_hN$, and $\Ric_{\mathbf g,it}=0$. In the static region,~\eqref{LS-eq-spacetime} and~\eqref{LS-eq-static} are therefore equivalent.} the spatial equations are\footnote{We use the Laplacian $\Delta_h\coloneqq\operatorname{div}_h\nabla^h$ and the convention $R(X,Y)Z\coloneqq\nabla_X\nabla_YZ-\nabla_Y\nabla_XZ-\nabla_{[X,Y]}Z$, for which round spheres have positive sectional curvature.}
\be
\label{LS-eq-static}
\Hess_hN=N(\Ric_h-\Lambda h),\qquad \Delta_hN=-\Lambda N.
\ee 
Taking the trace of the tensor equation and using the scalar equation gives $N(\Scal_h-2\Lambda)=0$. Positivity in the interior and continuity at the boundary give $\Scal_h=2\Lambda$ on $\overline M$.

A \emph{positive cosmological constant} is necessary in the compact setting under consideration in this paper. Multiplying the lapse equation by $N$ and integrating by parts, with $N=0$ on the boundary, gives
\be
\label{AX-eq-positive-cosmological-constant}
0<\int_M|\nabla^hN|_h^2\,d\mu_h =\Lambda\int_MN^2\,d\mu_h.
\ee
Thus $\Lambda>0$ is the only sign compatible with the compact boundary problem considered here. For $\Lambda=0$, the theorem of Bunting--Masood-ul-Alam~\cite{AX-BuntingMasood} shows that a regular asymptotically Euclidean static vacuum exterior with nonempty smooth compact horizon boundary is Schwarzschild, hence spherically symmetric with a single horizon component. The static axisymmetric Myers--Korotkin--Nicolai solutions at $\Lambda=0$, described by Reiris--Peraza~\cite{AX-ReirisPeraza}, have a \emph{noncompact spatial region} and an asymptotically Kasner end, although their horizon boundary is compact. Thus spherical symmetry in the noncompact problem depends on the asymptotic conditions.


\subsection{Existence and uniqueness}
\label{LS-section-strategy}

\paragraph{Main result}

Our prove in this paper a small-mass existence theorem and a uniqueness statement from data prescribed on one boundary component. The two reference masses of the black holes provide natural parameters in our construction scheme, while the intrinsic horizon metric is a natural object to establish an injective parametrization of the solutions. .

\begin{theorem}[Existence of static axisymmetric solutions to the Einstein equations]
\label{AX-thm-intro-existence}
Fix $\Lambda>0$ and $0<\theta_0\leq1/2$. For any two positive reference masses $m_+,m_-$ with sufficiently small sum $m_-+m_+$ and with controlled ratio $m_+/(m_++m_-)\in[\theta_0,1-\theta_0]$, there exist a compact three-dimensional manifold $M$, a Riemannian metric $h$ and a function $N$ satisfying the static vacuum equations~\eqref{LS-eq-static}, with the following properties.
\begin{itemize}
\item The manifold $M$ is diffeomorphic to a closed ball with two disjoint open balls removed. An effective circle action preserves $h$ and $N$ and has a \emph{regular axis of symmetry}.
\item The fields $h,N$ are smooth up to the boundary and analytic in suitable interior coordinates. The lapse is positive in the interior and vanishes exactly on the three boundary spheres; each sphere is a \emph{nondegenerate horizon} (namely, $\nabla N\neq0$) with a regular analytic local spacetime extension.
\item The two inner horizon areas satisfy $A_j=16\pi m_j^2(1+o(1))$ as $(m_++m_-)\to0$. The two masses may be unequal; for equal masses, an additional reflection is available and interchanges the inner horizons.
\end{itemize}
\end{theorem}

Theorem~\ref{AX-thm-binary-existence}, below, gives the quantitative existence statement, Proposition~\ref{AX-lem-local-horizon-extension} the local horizon extension, and Proposition~\ref{AX-prop-prescribed-Kottler-masses} the formulation with prescribed Kottler reference masses. The lapse attains its maximum on a \emph{smooth embedded circle}, represented by one interior point of the orbit domain, with nearby levels diffeomorphic to tori (Proposition~\ref{AX-prop-lapse-maximum-circle}). Near the maximal circle, we expand the lapse in the two \emph{normal directions}. Axial symmetry makes the coefficients constant along the circle, but not independent of normal direction (Section~\ref{AX-section-local-critical-geometry}). The remaining \emph{critical set} (where $\nabla N=0$) consists of three \emph{nondegenerate axial saddles} (Propositions~\ref{AX-prop-three-axial-saddles} and~\ref{AX-prop-complete-critical-set}). Furthermore, for the special class of \emph{reflection-symmetric equal-mass solutions}, the regular levels pass from three spheres to two spheres and then to a torus as the lapse value increases (Proposition~\ref{AX-prop-equal-mass-level-topology}).

At fixed $\Lambda$, Theorem~\ref{AX-thm-horizon-global-uniqueness}, stated below, determines any compact static region with simply connected interior and regular nondegenerate horizons by the \emph{intrinsic metric of one marked horizon}, up to isometry and lapse rescaling. In the axial class, these data are a \emph{meridian profile} and its \emph{interval length}. This injective parametrization applies without a small-mass restriction. Its image and uniqueness from the two masses alone remain open.


\paragraph{The Kottler collars}
\label{AX-section-intro-construction}

In the proposed scheme,  
we construct the approximate solutions from Kottler and de Sitter metrics and solve a \emph{reduced nonlinear elliptic Einstein equation}. Two scalar compatibility conditions select the axial centers; the {contracted Bianchi identity} and a {coercive estimate for the symmetric gradient of the auxiliary vector field} then recover the static equations. \emph{Kottler collars} solve the static equations with the same positive cosmological constant as the outer de Sitter region. At fixed $\Lambda$, their local parameters are the \emph{horizon radius} and \emph{surface gravity}, the latter including the lapse normalization. Rescaling both lapse functions by the same positive constant leaves their relative normalization unchanged; global matching prevents treating all these local parameters as independent solution parameters. The Schwarzschild metric arises when the cosmological coefficient of the rescaled Kottler family vanishes and is used to calculate the \emph{kernel of the limiting Schwarzschild linearized operator}.

Dias, Gibbons, Santos and Way~\cite{P2-Dias} also solve \emph{numerically} a gauge-fixed elliptic problem with horizon and axis regularity. Their numerical reference modifies Bach--Weyl and joins it to de Sitter; our construction is analogous but different, and joins two exact Kottler collars to de Sitter through explicitly constructed intermediate fields.

\paragraph{Parameter ranges and uniform estimates}

Multiplying $N$ by a positive constant preserves the static equations and rescales the time coordinate. We define
\be
\label{AX-eq-lapse-normalization}
N_{\max}\coloneqq\max_{\overline M}N>0,\qquad \ell\coloneqq\sqrt{3/\Lambda}.
\ee
On fixed compact parameter ranges away from horizon degeneracy, Kottler collars admit uniform local estimates, which may help relax the mass restrictions. Larger masses also require control of the global inverse, the nonlinear correction and the two parameter equations. Theorem~\ref{AX-thm-binary-existence}, below, allows different reference masses with sufficiently small sum and ratio in a fixed compact interval. Proposition~\ref{AX-prop-prescribed-Kottler-masses}, below, prescribes the two Kottler reference parameters.

We denote by $m_+,m_->0$ the auxiliary mass scales of the two inner reference regions and by $d>0$ the \emph{separation} of their centers in the reference coordinates. With the cosmological length $\ell$ defined in~\eqref{AX-eq-lapse-normalization}, we introduce the \emph{dimensionless total-mass parameter} $\varepsilon_{\mathrm{tot}}\coloneqq((m_++m_-)/\ell)^{1/3}$. We write $a\asymp b$ for positive quantities whose ratio is bounded above and below by constants independent of sufficiently small $\varepsilon_{\mathrm{tot}}$, with the mass ratio in a fixed compact subset of $(0,\infty)$. In other words, our scales are
\be
\label{AX-eq-intro-scales}
m_\pm\asymp\varepsilon_{\mathrm{tot}}^3\ell, \qquad d\asymp\varepsilon_{\mathrm{tot}}\ell, \qquad \frac{m_\pm}{d}\asymp\varepsilon_{\mathrm{tot}}^2, \qquad \Lambda d^2\asymp\varepsilon_{\mathrm{tot}}^2.
\ee
Figure~\ref{AX-fig-intro-scales} illustrates these scales.

\begin{figure}[!htbp]
\centering
\begin{tikzpicture}[x=1cm,y=1cm,font=\small]
\draw[black] (0,0)--(12,0);
\draw[thick,black] (2,0) circle (1.35);
\fill[white] (1.86,0) circle (.045);
\draw[black] (1.86,0) circle (.045);
\fill[white] (2.14,0) circle (.06);
\draw[black] (2.14,0) circle (.06);
\draw[<->] (2,0.25)--(3.32,0.25) node[midway,above] {$\asymp\ell$};
\draw[dashed] (2,0) ellipse (.4 and .23);
\node at (2,-1.75) {Cosmological scale};
\draw[->,dashed] (3.55,0.65)--(5.1,0.65) node[midway,above] {enlargement};
\draw[thick,black,fill=white] (6.3,0) circle (.16);
\draw[thick,black,fill=white] (9.7,0) circle (.23);
\fill (6.3,0) circle (.025); \fill (9.7,0) circle (.025);
\draw[dashed] (6.3,-.15)--(6.3,-.75);
\draw[dashed] (9.7,-.23)--(9.7,-.75);
\draw[<->] (6.3,-.65)--(9.7,-.65) node[midway,below] {$d\asymp\varepsilon_{\mathrm{tot}}\ell$};
\draw[->] (5.4,1)--(6.2,.13);
\node[anchor=south] at (5.4,1) {radius $\asymp m_-$};
\draw[->] (10.6,1)--(9.85,.17);
\node[anchor=south] at (10.6,1) {radius $\asymp m_+$};
\node at (8,-1.75) {$m_\pm\asymp\varepsilon_{\mathrm{tot}}^3\ell\ll d\ll\ell$};
\end{tikzpicture}
\caption{The three length scales in the reference construction. The left panel shows a section through the symmetry axis at the cosmological scale; the right panel enlarges the neighborhood of the two inner centers. The inner radii are of order $m_\pm$, and their separation is of order $\varepsilon_{\mathrm{tot}}\ell$. The two masses may differ. The panels use different magnifications and the radii are exaggerated for visibility; the displayed asymptotic relations specify the relative sizes.}
\label{AX-fig-intro-scales}
\end{figure}
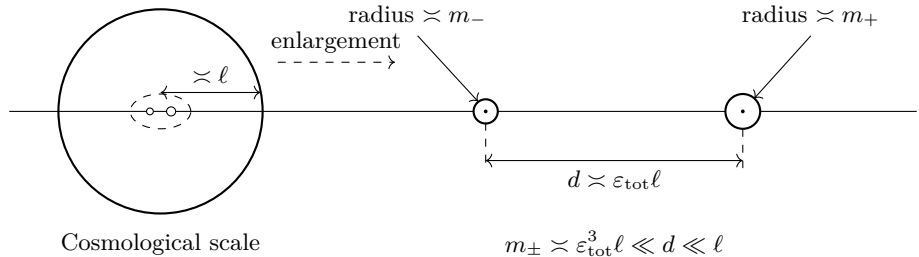

\subsection{Related nonlinear elliptic constructions}
\label{AX-section-related-works}

\paragraph{Axis regularity and asymptotically flat equilibrium}
For $\Lambda=0$, Weinstein~\cite{AX-WeinsteinEquilibrium,AX-WeinsteinTwoBody} develops the construction of stationary axisymmetric vacuum metrics through singular harmonic maps into the hyperbolic plane. His first existence and uniqueness result treats zero total angular momentum; the subsequent two-horizon analysis allows all five parameters of the reduced problem and establishes regularity at the horizon poles. The reconstructed metrics may carry a conical defect on the intervening axis, whose angle determines the interaction force. These works identify axis regularity as an additional equilibrium condition on solutions of the reduced equations.

The asymptotically flat problem provides several obstructions to static equilibrium. Beig and Schoen~\cite[Theorems~2.2 and~2.3]{AX-BeigSchoenStatic} prove rigidity for asymptotically flat static configurations with nonnegative scalar curvature and a closed, noncompact, totally geodesic surface in the vacuum region. Their result excludes configurations with a noncompact reflection surface disjoint from the bodies. For asymptotically flat, axisymmetric stationary vacuum configurations, Li and Tian~\cite[Section~3]{AX-LiTianHorizons} show that an involution interchanging the horizon components forces a positive interaction force between at least two of them. Their argument relates the force to the regularity condition on the intervening axis. Our construction places the equilibrium problem in a compact static region with positive cosmological constant and a cosmological horizon; the two scalar compatibility equations select the centers while preserving a regular axis.

In their recent preprint, Han--Khuri--Weinstein--Xiong~\cite[Theorem~1.1]{AX-HKWXKerrUniqueness} prove a strict sign result for asymptotically flat stationary axisymmetric vacuum metrics reconstructed from singular harmonic maps with the prescribed axis, pole and infinity asymptotics. Assuming nonzero angular momentum on every degenerate horizon component, they obtain a nonzero attractive interaction force on every bounded axis interval between horizons. This strengthens the asymptotically flat equilibrium obstruction in the class covered by their theorem. The balance equations in our construction include the positive cosmological term.

A different equilibrium construction for $\Lambda=0$ is given by Khuri--Weinstein--Yamada~\cite[Theorem~1(ii) and Section~3]{AX-KWYSignedMasses}. In four spacetime dimensions, they construct balanced asymptotically flat static configurations with three constituents of both mass signs. The axis is regular, while negative-mass singularities remain. Their formulas describe how angle defects change across horizons and negative-mass singularities. Our construction achieves a regular axis with two positive-mass inner reference regions and a positive cosmological constant.

\paragraph{Periodic solutions with a Kasner end}
Reiris and Peraza~\cite[Introduction and Theorems~2--3]{AX-ReirisPeraza} describe the Myers--Korotkin--Nicolai solutions and classify complete circle-symmetric static vacuum black holes for $\Lambda=0$. The Myers--Korotkin--Nicolai construction starts from an infinite periodic array of spherical horizons with a regular axis and takes a quotient along the periodic direction, possibly combined with an axial rotation. The quotient has the topology of an open solid torus $\mathbb R^2\times\mathbb S^1$ with finitely many open balls removed. Its horizon boundary is compact, but the radial direction extends to infinity through an \emph{asymptotically Kasner end} with toroidal cross-sections. Our solutions have $\Lambda>0$ and a \emph{compact, simply connected spatial region}, diffeomorphic to a ball with two open balls removed; its boundary consists of two inner spherical horizons and one outer \emph{cosmological horizon}. The lapse equation expresses this distinction analytically: the harmonic lapse at $\Lambda=0$ is defined on a noncompact exterior, whereas on our compact domain the condition $N=0$ on every boundary component gives the positive Dirichlet eigenvalue identity~\eqref{AX-eq-positive-cosmological-constant}.

\paragraph{Gluing for Einstein and other nonlinear elliptic equations}
For positive cosmological constant, Hintz~\cite[Theorem~1.1 and Sections~3.1--3.3]{AX-HintzDeSitterGluing} constructs dynamical vacuum spacetimes by gluing Schwarzschild--de Sitter metrics near prescribed points of the future conformal boundary. A balance condition on the masses and points removes the obstruction to a linear divergence equation for trace-free symmetric tensors. Both constructions combine exact Kottler regions with de Sitter geometry and impose compatibility conditions on the gluing parameters. Hintz solves a spacetime problem near conformal infinity, whereas we solve a static elliptic problem on a compact region bounded by three horizons.

Our construction shares several analytical features with gluing and concentration methods for other nonlinear elliptic equations. Mazzeo--Pacard--Pollack~\cite{AX-MPPConnectedSums} construct connected sums of constant mean curvature surfaces using weighted estimates on the constituent pieces and matching their Cauchy data. Their work illustrates how estimates adapted to a degenerating geometry can be combined with finite-dimensional matching conditions. Breiner--Kapouleas~\cite{AX-BKEmbeddedCMC} solve the linearized equation modulo localized obstruction functions, estimate the nonlinear correction in weighted H\"older spaces, and vary geometric parameters to remove the remaining obstructions through integral identities. This last step is closely related to our selection of the two axial centers through scalar compatibility equations. In the construction of Riemannian Einstein metrics with negative Ricci curvature, Bamler~\cite[Sections~4--6]{AX-BamlerDehn} incorporates cusp deformation modes into the norms to obtain a uniform inverse estimate and carry out a nonlinear fixed-point argument; the related issue here is to identify the kernels of the limiting linearized operators and quantify the inverse bounds as the inner radii and their separation shrink at different rates. Esposito--Pistoia--V\'etois~\cite[Section~2, Propositions~2.1 and~2.3]{AX-EPVYamabe} likewise separate an infinite-dimensional equation from a finite-dimensional parameter problem: they solve the projected Yamabe equation orthogonally to dilation and translation modes, then select the concentration parameters through a reduced energy. In our setting, the reduced Einstein equation and the two scalar compatibility equations play these respective roles, while the contracted Bianchi identity and an energy estimate for the auxiliary vector field recover the full static Einstein equations.


\subsection{Outline of the proof}
\label{AX-section-intro-outline}

Section~\ref{AX-section-problem}, below, sets out the axial formulation, including the tensor constraint required for equivalence with the static equations, and the regularity conditions at the horizons, the axis and their intersections. The first four points explain how the analytical construction gives an exact solution; the last two describe its critical geometry and its determination by horizon data.
\begin{enumerate}
\item \emph{An approximation adapted to the different scales.} Section~\ref{AX-section-existence} constructs the Kottler--de Sitter reference. Near each inner horizon, we match the constant parts of the companion and cosmological fields and use regular Cartesian lifts, with the local angular coordinate normalized by that horizon's surface gravity. In units where $\ell=1$, the transition radii $\varepsilon_{\mathrm{tot}}^{5/3}$ and $\varepsilon_{\mathrm{tot}}^{3/5}$ balance the competing matching errors. Lemma~\ref{AX-lem-improved-reference-error} makes the Einstein error small enough to compensate for the growth of the linear inverse while preserving the fields on the annuli used for the compatibility integrals.
\item \emph{Cancellation that preserves the leading compatibility equations.} Proposition~\ref{AX-prop-flux-stability} shows that linearization at the flat product metric gives zero contribution to each closed-sphere compatibility integral for any $C^2$ correction independent of the auxiliary angular coordinate and defined on an annular neighborhood of the sphere. An explicit quadratic calculation for the reference then yields the leading terms, of order $\varepsilon_{\mathrm{tot}}^4$. A correction of size $o(\varepsilon_{\mathrm{tot}}^2)$ in the rescaled annular $C^2$ norm preserves these leading terms. This cancellation makes the accuracy of the nonlinear correction sufficient for selecting the centers.
\item \emph{Control of the two displacement modes.} Sections~\ref{AX-section-two-masses} and~\ref{AX-section-force-surjectivity} classify the kernels of the linearized operators in the rescaled Kottler, flat and de Sitter limits and derive weighted estimates. The inner orthogonality conditions remove the two axial displacement modes, while Proposition~\ref{AX-prop-centered-index} identifies the two scalar obstructions to solvability. Testing against localized displacement tensors controls the two obstruction coefficients. Theorem~\ref{AX-thm-physical-force-inverse} then solves the linear equation with these two additional scalar unknowns and quantifies the growth of its inverse; Corollary~\ref{AX-cor-improved-projection-quadratic} controls the nonlinear remainder, including the mass dependence introduced by the tensor projection.
\item \emph{Selection of the centers and recovery of the Einstein equation.} In Section~\ref{AX-section-comparable-mass-existence}, a contraction solves the reduced equation. Varying the axial barycenter and separation gives a zero of the two compatibility integrals by Lemma~\ref{AX-prop-two-mass-balance}, also for different comparable masses. The contracted Bianchi identity and a coercive estimate control the auxiliary vector contribution by the two obstruction coefficients. The matrix of compatibility integrals, together with this estimate, forces both coefficients to vanish at the selected zero. The same estimate eliminates the auxiliary contribution, giving the full Einstein equation and hence Theorem~\ref{AX-thm-binary-existence}. The resulting reference separation satisfies $d^3\sim(m_++m_-)\ell^2$. Proposition~\ref{AX-prop-prescribed-Kottler-masses} also gives existence with prescribed Kottler reference masses.
\item \emph{The maximal circle and the remaining critical points.} Section~\ref{AX-section-critical-geometry} uses an explicit positive factorization of the leading potential to locate its circle of minima. Global estimates and local $C^2$ control transfer this description to a circle of maxima of the exact lapse, with nearby torus levels. The same estimates identify three nondegenerate axial saddles, describe their conical critical levels within the analytic spatial geometry and determine the topology of the regular levels. These conclusions describe the solutions produced by the existence argument.
\item \emph{Determination by one intrinsic horizon metric.} Section~\ref{AX-section-uniqueness} distinguishes uniqueness in the fixed contraction ball from global determination by one horizon. The static equations recursively determine all normal expansion coefficients from the intrinsic horizon metric and the lapse normalization. Analyticity gives local uniqueness, and continuation across the spatial doubles gives Theorem~\ref{AX-thm-horizon-global-uniqueness} under its regularity, nondegeneracy and simple connectivity assumptions. This yields an injective parametrization by the horizon metrics realized by regular solutions.
\end{enumerate}
The appendices supply the detailed local calculations used above: Appendix~\ref{AX-appendix-horizon} proves the horizon extension, Appendix~\ref{AX-appendix-Schwarzschild} gives the weighted inner estimates and classifies the kernel of the Schwarzschild linearized operator, Appendix~\ref{AX-appendix-quadratic-flux} computes the quadratic compatibility integrals, and Appendix~\ref{AX-appendix-weighted-loss} establishes the estimates for the vector field exhibiting growth of the inverse constant.

The main analytical contribution of this paper is the quantitative compatibility analysis made between the approximation, the growth of the linear inverse bound and the integral identities selecting the geometric parameters. We expect this approach to be useful for other nonlinear geometric elliptic constructions involving several small domains, provided sufficient structure is shared by the leading equations and the small domain geometry. Applying it in another setting would require estimates for the (co-)kernels of the limiting linearized operators and the nonlinear error, together with control of the resulting finite-dimensional parameter equations.  


\section{Reduction of the equations and boundary regularity}
\label{AX-section-problem}


\subsection{The metric on the orbit domain}
\label{AX-section-equations}

\paragraph{The three-horizon spatial region}
We express the static equations in axial variables and specify regularity at the horizons, the axis and their intersections. These conditions and the small-hole scaling determine the class of approximate metrics used in the construction. We specialize the static equations~\eqref{LS-eq-static} to an axial region with three horizon components. The reduced equations govern the metric on the orbit domain, the circular orbit radius and the lapse. The connected spatial manifold $M$ is diffeomorphic to a closed ball with the interiors of two disjoint balls removed. Its boundary consists of three spheres, denoted by $H_+$, $H_-$ and $H_c$: the first two represent black-hole horizons and the last represents the cosmological horizon. Thus the construction concerns a compact static region with three horizon components. We require an effective action of the circle $S^1$, with angular coordinate $\varphi\in\mathbb R/(2\pi\mathbb Z)$, preserving $h$ and $N$. The circle action is free away from its fixed axis. We also impose integrability of the planes orthogonal to the circular orbits. In adapted orthogonal coordinates, the reflection $\varphi\mapsto-\varphi$ across a plane containing the symmetry axis preserves the data and each horizon component. For equal masses, the construction can additionally impose a reflection interchanging $H_+$ and $H_-$. These two reflections have distinct geometric roles.

The spatial region is a ball with two holes. Its orbit domain $\Omega\coloneqq M/S^1$ is a disk with boundary arcs and corners, illustrated in Figure~\ref{AX-fig-orbit-domain}. On its interior the integrability assumption gives a metric on the orbit domain $q$ and a positive orbit radius $R$, with
\be
\label{AX-eq-metric}
h=q+R^2d\varphi^2,\qquad \mathbf g=-N^2dt^2+q+R^2d\varphi^2.
\ee
The length of a circular orbit is $2\pi R$. The fields $q$, $R$ and $N$ are defined on $\Omega$, with $R,N>0$ in its interior. In local isothermal coordinates $(x,y)$, the metric takes the form $q=e^{2\sigma}(dx^2+dy^2)$. Its covariant derivative, Laplacian and Gauss curvature are denoted by $D$, $\Delta_q$ and $K_q$, respectively. The isothermal chart is a convenient description of the interior equations; regularity at a horizon or an axis must be expressed in regular spatial coordinates.

Figure~\ref{AX-fig-orbit-domain} illustrates the orbit domain with auxiliary level curves. In dimensionless drawing coordinates $(z,r)$, we use fixed parameters $a_{\rm fig},b_{\rm fig},L_{\rm fig}>0$ satisfying $b_{\rm fig}<a_{\rm fig}$, $a_{\rm fig}+b_{\rm fig}<L_{\rm fig}$ and $b_{\rm fig}L_{\rm fig}^2>a_{\rm fig}^2(a_{\rm fig}-b_{\rm fig})$. We put $r_\pm^{\rm fig}\coloneqq((z\mp a_{\rm fig})^2+r^2)^{1/2}$ and define
\be
\label{AX-eq-figure-lapse}
N_{\rm fig}(z,r)\coloneqq \Bigl(1-\frac{z^2+r^2}{L_{\rm fig}^2}\Bigr) \Bigl(1-\frac{b_{\rm fig}}{r_+^{\rm fig}}\Bigr) \Bigl(1-\frac{b_{\rm fig}}{r_-^{\rm fig}}\Bigr).
\ee
This function is positive in the drawn region, vanishes on its three horizon arcs and is even across each open axis interval. Figure~\ref{AX-fig-spatial-3D} shows the corresponding spatial picture under rotation about the axis.

\begin{figure}[!htbp]
\centering
\resizebox{0.74\textwidth}{!}{\input{FIGURES/AX-orbit-levels-picture.tex}}
\caption{Illustration of the axial orbit domain and lapse levels. Red arcs are $H_\pm$, the blue arc is $H_c$, and the dashed axis intervals are $\mathsf A_-,\mathsf A_0,\mathsf A_+$. The black points $\mathsf B_-,\mathsf B_0,\mathsf B_+$ are the axial critical points of $N_{\rm fig}$. The curves and critical points belong to the auxiliary function $N_{\rm fig}$. Table~\ref{AX-tab-level-topologies} gives the level topologies of the constructed solutions for both equal and unequal comparable masses. Closed curves away from the axis lift to tori under rotation.}
\label{AX-fig-orbit-domain}
\par\medskip
\includegraphics[width=0.64\textwidth]{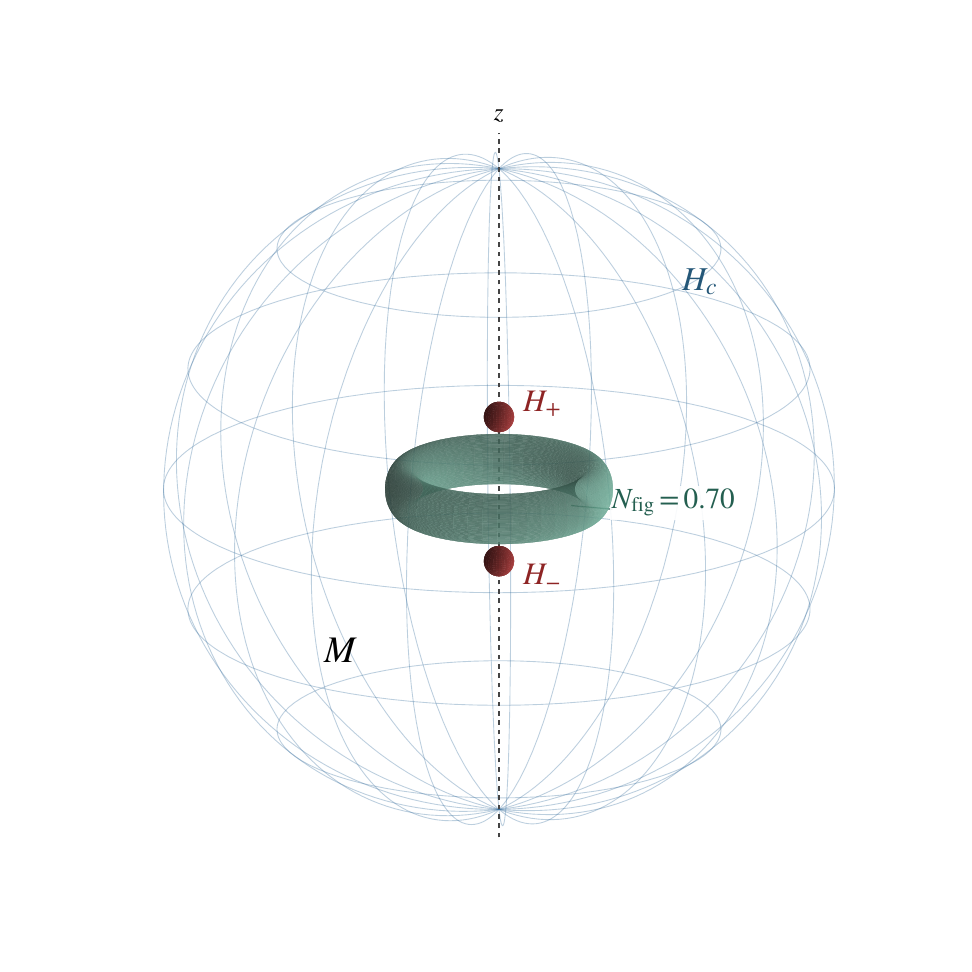}
\caption{A three-dimensional representation of the compact spatial region $M$. The blue mesh represents the outer boundary $H_c$ and the two red spheres represent $H_\pm$. The green torus is obtained by rotating a closed auxiliary level disjoint from the axis in Figure~\ref{AX-fig-orbit-domain}. This is a topological illustration, not an isometric embedding of the static metric.}
\label{AX-fig-spatial-3D}
\end{figure}

\paragraph{Local geometry of the illustrative function}
The symmetric auxiliary function~\eqref{AX-eq-figure-lapse} illustrates regular axes and the local geometry of conical critical levels. The ordering of the critical values of the constructed Einstein solutions follows from the estimates below. For equal masses, the lapse value at the exact middle saddle is smaller than at the two outer saddles, which have the same value (Corollary~\ref{AX-cor-equal-mass-critical-order}). For unequal comparable masses, we order the three exact critical values as $c_{(1)}\leq c_{(2)}\leq c_{(3)}$. We assume no reflection symmetry and make no claim about their order according to the saddle labels (Proposition~\ref{AX-prop-general-level-topology}). Table~\ref{AX-tab-level-topologies} gives the corresponding level topologies.

The function $N_{\rm fig}$ is analytic at the three marked points. Its Hessian in Cartesian coordinates has two positive transverse eigenvalues and one negative axial eigenvalue. We can verify these signs directly. The second axial derivative of $\log N_{\rm fig}(z,0)$ is strictly negative on each axis interval. Since $N_{\rm fig}$ is positive inside that interval and vanishes at its endpoints, it has exactly one axial critical point. At the central point, positivity of the transverse eigenvalues is exactly the condition $b_{\rm fig}L_{\rm fig}^2>a_{\rm fig}^2(a_{\rm fig}-b_{\rm fig})$. At the positive outer point with coordinate $z_*$, the distances are $r_\pm^{\rm fig}=z_*\mp a_{\rm fig}$. The critical-point equation gives the transverse eigenvalue $\lambda_\perp$ as
\be
\label{AX-eq-figure-critical-hessian}
\frac{\lambda_\perp}{N_{\rm fig}(z_*,0)} =\frac{a_{\rm fig}b_{\rm fig}}{z_*} \Bigl(\frac{1}{(r_+^{\rm fig})^2(r_+^{\rm fig}-b_{\rm fig})} -\frac{1}{(r_-^{\rm fig})^2(r_-^{\rm fig}-b_{\rm fig})}\Bigr)>0.
\ee
The reflected point has the same signs. Write $N_{\rm fig}=F(z,x^2+y^2)$ near a marked axial point with coordinate $z_B$. The implicit function theorem gives the critical level in the form $x^2+y^2=(z-z_B)^2g(z)$, where $g$ is positive and analytic. With the coordinate $w=(z-z_B)\sqrt{g(z)}$, this level becomes the double cone $x^2+y^2=w^2$. The level surface is singular at its vertex, while the function remains analytic there. For the Einstein solutions, both $h$ and $N$ are regular at every interior critical point of $N$. Proposition~\ref{AX-prop-three-axial-saddles} supplies three axial saddle points of the exact solutions and their local critical cones. The marked positions and critical values in this drawing still belong to the auxiliary function.

\paragraph{The full reduced system}
The reduction retains both the scalar equations and the two components of a trace-free tensor equation. The product $NR$ is useful, but the cosmological term changes its equation. The following proposition expresses the full static system in the three fields $q$, $R$ and $N$.
\begin{proposition}[Axial reduction of the static equations]
\label{AX-prop-reduction}
Fix an open orbit domain, smooth positive functions $N,R$ and a smooth Riemannian metric $q$ on that domain. The metric~\eqref{AX-eq-metric} satisfies the static vacuum equations if and only if
\be
\label{AX-eq-reduced}
\begin{aligned}
K_q q-R^{-1}\Hess_qR&=N^{-1}\Hess_qN+\Lambda q,\\
\Delta_qR+N^{-1}\langle DR,DN\rangle_q+\Lambda R&=0,\\
\Delta_qN+R^{-1}\langle DR,DN\rangle_q+\Lambda N&=0.
\end{aligned}
\ee
For the product $\rho\coloneqq NR$, one has
\be
\label{AX-eq-area}
\Delta_q\rho=-2\Lambda\rho.
\ee
In an isothermal chart, the system consists of the three scalar equations~\eqref{AX-eq-isothermal} and the constraint~\eqref{AX-eq-constraint} below. Equivalence with the static vacuum equations requires the constraint as well as the three scalar equations.
\end{proposition}
The statement applies where both circular directions have positive length. The boundary conditions in Section~\ref{AX-section-boundaries} describe the regular and conical extensions of this system.
\begin{proof}
\bse
We use indices $a,b\in\{1,2\}$ for the two orbit-domain coordinates. Besides the Christoffel coefficients of $q$, the nonzero spatial coefficients involving the angular direction are $\Gamma^a_{\varphi\varphi}=-R D^aR$ and $\Gamma^\varphi_{a\varphi}=R^{-1}\del_aR$. Substitution in the Ricci and Hessian formulas gives
\be
\label{AX-eq-warped-components}
\begin{aligned}
(\Ric_h)_{ab}&=K_q q_{ab}-R^{-1}D_aD_bR,&
(\Ric_h)_{\varphi\varphi}&=-R\Delta_qR,\\
(\Hess_hN)_{ab}&=D_aD_bN,&
(\Hess_hN)_{\varphi\varphi}&=R\langle DR,DN\rangle_q.
\end{aligned}
\ee
For example, the additional terms in the Ricci component with orbit-domain indices are $-\del_b(R^{-1}\del_aR)+R^{-1}\Gamma^c_{ab}\del_cR-R^{-2}\del_aR\del_bR$, which combine to $-R^{-1}D_aD_bR$. In the angular component the terms containing $|DR|_q^2$ cancel, leaving $-R\Delta_qR$. The mixed components vanish because the fields are independent of $\varphi$ and the metric has no mixed term between the orbit-domain and angular directions. Taking the spatial trace of the Hessian yields $\Delta_hN=\Delta_qN+R^{-1}\langle DR,DN\rangle_q$. The components in the orbit-domain and angular directions of the tensor equation, together with the lapse equation, now give~\eqref{AX-eq-reduced}; the same substitutions prove the converse.

The product rule and the two scalar equations give
\be
\label{AX-eq-product-calculation}
\Delta_q(NR)=N\Delta_qR+R\Delta_qN+2\langle DR,DN\rangle_q=-2\Lambda NR,
\ee
which proves~\eqref{AX-eq-area}. To express the remaining equations, we denote the Euclidean derivative and Laplacian in $(x,y)$ by $\nabla_0$ and $\Delta_0$. The trace of the tensor equation is $2K_q=R^{-1}\Delta_qR+N^{-1}\Delta_qN+2\Lambda$. Substitution of the scalar equations gives $K_q=-(NR)^{-1}\langle DR,DN\rangle_q$. Since $K_q=-e^{-2\sigma}\Delta_0\sigma$ and $\Delta_q=e^{-2\sigma}\Delta_0$ on functions, the three scalar equations are
\be
\label{AX-eq-isothermal}
\begin{aligned}
\Delta_0R+N^{-1}\nabla_0R\cdot\nabla_0N+\Lambda e^{2\sigma}R&=0,\\
\Delta_0N+R^{-1}\nabla_0R\cdot\nabla_0N+\Lambda e^{2\sigma}N&=0,\\
\Delta_0\sigma&=(NR)^{-1}\nabla_0R\cdot\nabla_0N.
\end{aligned}
\ee
For a symmetric tensor $T$ on the two-dimensional orbit domain, we define its trace-free part by $T^\circ\coloneqq T-\tfrac12(\operatorname{tr}_qT)q$. The trace-free part of the equation with orbit-domain indices is
\be
\label{AX-eq-constraint}
\bigl(R^{-1}\Hess_qR+N^{-1}\Hess_qN\bigr)^\circ=0.
\ee
A symmetric trace-free two-tensor has two independent components, so this is an additional pair of conditions. Conversely, the trace equation supplied by~\eqref{AX-eq-isothermal} and the constraint~\eqref{AX-eq-constraint} recover the tensor equation on the orbit domain. This proves the final assertion.
\ese
\end{proof}

The product $\rho=NR$ vanishes on every horizon and axis arc. With $\Lambda>0$, its equation~\eqref{AX-eq-area} contains a strictly negative right-hand side in the interior, whereas it is harmonic for $\Lambda=0$. Thus $\rho=NR$ cannot serve as a harmonic Weyl coordinate when $\Lambda>0$. The constraint also remains necessary after solving the three scalar equations. We therefore solve a tensor equation and use the contracted Bianchi identity to obtain the full static system. Near the critical set of the lapse, we use regular coordinates on the orbit domain.


\subsection{Horizon, axis and corner conditions}
\label{AX-section-boundaries}

\paragraph{Regularity at nondegenerate horizons}
We express regularity at the horizons and the axis through expansions in smooth local coordinates. These conditions must hold jointly where a horizon meets the axis, so that every pole is regular in the spatial metric. The boundary arcs of the orbit disk occur in the cyclic order $\mathsf A_+,H_+,\mathsf A_0,H_-,\mathsf A_-,H_c$, with $\mathsf A_0$ the axis segment between the two black holes. When the additional reflection interchanging the inner horizons is imposed, it exchanges the two black-hole arcs and the two exterior axis arcs. Each horizon arc meets the axis at its two poles. If the spatial data are smooth, the quotient has a right-angle corner at such a meeting point; the point is a regular boundary point of the three-dimensional spatial manifold. This corner geometry follows from the linear rotation action on the tangent plane to the horizon and from the inward spatial normal, which is fixed by the rotations. The quotient of the tangent half-space is therefore a quadrant.

Khuri and Woolgar~\cite[Theorem~1(i)]{LS-KhuriWoolgar} exclude degenerate horizon components in the static vacuum black-hole class with positive cosmological constant. The construction below supplies positive surface gravities and analytic collars at all three horizons.

On a connected nondegenerate horizon, the surface gravity $\kappa_j\coloneqq|\nabla^hN|_h$ is constant and positive, as proved in Proposition~\ref{AX-lem-local-horizon-extension}. For the regular angular lifts used in the construction, the same lemma gives analytic parity. With inward proper distance $s$, the data have the local form
\be
\label{AX-eq-horizon-data}
h=ds^2+\gamma_j(s^2,x),\qquad N=\kappa_j s\,b_j^{\rm hor}(s^2,x),\qquad b_j^{\rm hor}(0,x)=1,\qquad j\in\{+,-,c\}.
\ee
Here $x$ denotes horizon coordinates, and $\gamma_j,b_j^{\rm hor}$ are analytic for an exact regular static solution. The additional reflection interchanging the inner horizons, when imposed, gives $\kappa_+=\kappa_-$ and $\operatorname{Area}_h(H_+)=\operatorname{Area}_h(H_-)$. The general two-mass class permits different inner surface gravities and areas. The lapse scaling changes all three constants by the same positive factor. Integration of the lapse equation gives
\be
\label{AX-eq-horizon-flux}
\kappa_+\operatorname{Area}_h(H_+)+\kappa_-\operatorname{Area}_h(H_-)+\kappa_c\operatorname{Area}_h(H_c)=\Lambda\int_MN\,d\mu_h.
\ee
Indeed, the outward normal derivative is $-\kappa_j$, and the divergence theorem applied to $\Delta_hN=-\Lambda N$ proves this relation. Thus the surface gravities and the horizon geometries cannot be prescribed as independent boundary functions.

\paragraph{Regularity on the axis and at horizon poles}
On an open axis arc, we use proper distance $r_\perp$ from that arc and define its angular ratio by
\be
\label{AX-eq-angle}
\beta_i^{\rm ax}\coloneqq\lim_{r_\perp\downarrow0}\frac{R}{r_\perp},\qquad \delta_i^{\rm ax}\coloneqq2\pi(1-\beta_i^{\rm ax}),\qquad i\in\{-,0,+\}.
\ee
The circumference-to-radius ratio tends to $2\pi\beta_i^{\rm ax}$. In smooth normal Cartesian coordinates, rotation invariance makes $N$ and the tangential metric coefficients smooth even functions of $r_\perp$, while the orbit radius is $R=r_\perp+\Ocal(r_\perp^3)$. This proves $\beta_i^{\rm ax}=1$ for a regular axis. Conversely, the scalar condition $\beta_i^{\rm ax}=1$ must be accompanied by the higher parity conditions expressing smoothness of the tensor in Cartesian coordinates. The same conditions must hold jointly with the horizon parity at every pole. After filling the local auxiliary circle, a pole is described by two orthogonal polar planes, with radii $s$ and $r_\perp$; regular invariant tensors in their Cartesian coordinates provide the simultaneous conditions. Separate conditions on the two open faces do not ensure compatibility at their intersection.

For $\Lambda=0$, the interpretation of axis angle defects as interaction forces is developed by Weinstein~\cite{AX-WeinsteinTwoBody} and Khuri--Weinstein--Yamada~\cite[Section~3]{AX-KWYSignedMasses}. The auxiliary angular coordinate used to fill $H_j$ is $\theta_{E,j}\coloneqq\kappa_j t_E$, where $t_E$ is the auxiliary coordinate in the Riemannian metric. Different surface gravities therefore require different local angular normalizations. We formulate regularity separately at each horizon, using the local period of $t_E$ $2\pi/\kappa_j$. Proposition~\ref{AX-lem-local-horizon-extension} supplies the local Lorentzian extensions from these Cartesian regularity conditions. The number of regular metric on the orbit domain coefficients is three before fixing spatial coordinates; together with $N$ and $R$, there are five coefficient functions. We impose a coordinate condition on the metric on the orbit domain while preserving these regularity conditions.

The following proposition converts the Cartesian regularity conditions into an analytic local extension. The proof is given in Appendix~\ref{AX-appendix-horizon}.

\begin{proposition}[Analytic extension across a nondegenerate horizon]
\label{AX-lem-local-horizon-extension}
Fix smooth static data satisfying~\eqref{LS-eq-static} and a connected nondegenerate horizon $H$. Then $H$ is totally geodesic and its surface gravity $\kappa$ is constant. If the local metric $G_H\coloneqq h+(N/\kappa)^2d\theta^2$, with $\theta$ of period $2\pi$, extends as a positive $C^{2,\alpha}$ metric in normal Cartesian coordinates at its fixed surface $H$, where $0<\alpha<1$, then $G_H$ is analytic in suitable coordinates, the spatial data have the parity~\eqref{AX-eq-horizon-data}, and the Lorentzian metric has an analytic local extension across $H$. This Cartesian regularity is imposed on the reference metrics and corrections; analyticity follows from the Einstein equation.
\end{proposition}

We seek smooth static data on the compact ball with two holes, with positive interior lapse, exactly three nondegenerate horizon components, a regular axis with $\beta_-^{\rm ax}=\beta_0^{\rm ax}=\beta_+^{\rm ax}=1$, and axial symmetry with the reflection across a plane containing the symmetry axis. Each horizon has the local Lorentzian extension proved below. The maximum set of $N$ is an interior set and is not an additional boundary.


\subsection{Topology of the domain and level sets}
\label{AX-section-model-topologies}

We compare the topology of a static region with three boundary components with those of de Sitter, Kottler and Nariai. The topology of the spatial manifold and that of a regular lapse level are different objects, summarized in Figure~\ref{AX-fig-model-topologies} and Table~\ref{AX-tab-level-topologies}.

We normalize each lapse by $\max N=1$ and write $I$ for a compact interval. In de Sitter, the domain is a ball, each regular level is a sphere and the maximum is a point. In subextremal Kottler of mass $\mu>0$, the domain is $I\times S^2$. In the areal radius $s$, the squared lapse before normalization is $F(s)\coloneqq1-2\mu/s-\Lambda s^2/3$. The derivative $F'(s)=2\mu/s^2-2\Lambda s/3$ has a unique zero, at $s=(3\mu/\Lambda)^{1/3}$, and the lapse reaches its maximum there. Each regular level therefore consists of two spheres, while the maximum set is a sphere. The Nariai solution is
\be
\label{AX-eq-Nariai-comparison}
h=\Lambda^{-1}(d\chi^2+d\Omega^2),\qquad N=\sin\chi,\qquad 0\leq\chi\leq\pi,
\ee
where $d\Omega^2$ is the unit round metric. Direct substitution verifies~\eqref{LS-eq-static}. Its geometry is a product, with the same domain and level topologies as Kottler.

For the constructed solutions, nondegeneracy gives a collar of each horizon on which $N$ is a regular coordinate. The complement of these collars is compact and has a positive minimum of $N$. More explicitly, choose disjoint horizon collars $U_j$ on which $dN\ne0$ and which are foliated by the levels of $N$. Their inner boundary levels have positive values $c_j$. On the compact complement, the minimum $c_K$ of $N$ is positive. For $c_0\coloneqq\tfrac12\min(c_K,c_+,c_-,c_c)$, each level $\{N=c\}$ satisfying $0<c<c_0$ is the union of three spheres, one in each horizon collar. The solution determines the threshold $c_0$. Here $c$ is a lapse value, not a mass or a separation parameter. Proposition~\ref{AX-prop-lapse-maximum-circle} shows that the maximum set of each constructed solution is one smooth embedded circle, represented by one interior point of $\Omega$, and that regular levels sufficiently close to its maximum are tori. Proposition~\ref{AX-prop-three-axial-saddles} gives three axial saddles and Section~\ref{AX-section-local-critical-geometry} describes their local levels. Proposition~\ref{AX-prop-complete-critical-set} exhausts the critical set, and Proposition~\ref{AX-prop-equal-mass-level-topology} gives the complete level topologies for equal masses. For arbitrary comparable masses, sorting the three critical values gives the full regular-level description in Proposition~\ref{AX-prop-general-level-topology}. Table~\ref{AX-tab-level-topologies} records the topologies established for the constructed solutions.

\begin{figure}[!htbp]
\centering
\begingroup%
\colorlet{axtopouter}{blue!65!black}%
\colorlet{axtopinner}{red!70!black}%
\colorlet{axtoplevel}{teal!70!black}%
\colorlet{axtopink}{black!75}%
\newcommand{\AxTopSphere}[4]{%
  \begin{scope}[shift={(#1,#2)}]
  \fill[#4!6] (0,0) circle[radius=#3];
  \draw[#4,line width=.65pt] (0,0) circle[radius=#3];
  \draw[#4,densely dashed,line width=.35pt]
    (#3,0) arc[start angle=0,end angle=180,x radius=#3,y radius={.30*#3}];
  \draw[#4,line width=.45pt]
    (-#3,0) arc[start angle=180,end angle=360,x radius=#3,y radius={.30*#3}];
  \end{scope}}%
\newcommand{\AxTopRegion}[3]{%
  \begin{scope}[shift={(#1,#2)}]
  \fill[axtopouter!6] (0,0) circle[radius=.85];
  \draw[axtopouter,line width=.8pt] (0,0) circle[radius=.85];
  \draw[axtopouter!65,densely dashed,line width=.35pt]
    (.85,0) arc[start angle=0,end angle=180,x radius=.85,y radius=.27];
  \draw[axtopouter!75,line width=.45pt]
    (-.85,0) arc[start angle=180,end angle=360,x radius=.85,y radius=.27];
  \draw[axtopouter!35,line width=.35pt] (0,0) ellipse[x radius=.27,y radius=.85];
  #3
  \end{scope}}%
\newcommand{\AxTopHole}[3]{%
  \begin{scope}[shift={(#1,#2)}]
  \fill[white] (0,0) circle[radius=#3];
  \fill[axtopinner!5] (0,0) circle[radius=#3];
  \draw[axtopinner,line width=.75pt] (0,0) circle[radius=#3];
  \draw[axtopinner,densely dashed,line width=.35pt]
    (#3,0) arc[start angle=0,end angle=180,x radius=#3,y radius={.30*#3}];
  \draw[axtopinner,line width=.45pt]
    (-#3,0) arc[start angle=180,end angle=360,x radius=#3,y radius={.30*#3}];
  \end{scope}}%
\noindent\begin{tikzpicture}[x=1cm,y=1cm,line cap=round,line join=round,
    font=\small,text=axtopink]
\path[use as bounding box] (0,-11.92) rectangle (15.4,.70);
\fill[black!2] (0,-.12) rectangle (15.4,-2.92);
\fill[black!2] (0,-5.72) rectangle (15.4,-8.52);
\draw[black!32,line width=.55pt] (0,-.12)--(15.4,-.12);
\foreach \y in {-2.92,-5.72,-8.52,-11.32}
  \draw[black!17,line width=.35pt] (0,\y)--(15.4,\y);
\node[anchor=west,font=\small\bfseries] at (.15,.30) {Static model};
\node[font=\small\bfseries] at (4.65,.30) {Spatial domain};
\node[font=\small\bfseries] at (9.15,.30) {Regular level $\{N=c\}$};
\node[font=\small\bfseries] at (13.35,.30) {Maximum set $\{N=1\}$};

\node[anchor=west,font=\small\bfseries] at (.15,-1.42) {de Sitter};
\AxTopRegion{4.65}{-1.32}{}
\node at (4.65,-2.46) {$\overline B^{\,3}$};
\AxTopSphere{9.15}{-1.32}{.59}{axtoplevel}
\node at (9.15,-2.30) {$S^2$};
\fill[axtopink] (13.35,-1.32) circle[radius=.065];
\node at (13.35,-2.30) {point};

\node[anchor=west,font=\small\bfseries] at (.15,-4.22) {Kottler};
\AxTopRegion{4.65}{-4.12}{\AxTopHole{0}{0}{.40}}
\node at (4.65,-5.26) {$I\times S^2$};
\AxTopSphere{8.48}{-4.12}{.52}{axtoplevel}
\AxTopSphere{9.82}{-4.12}{.52}{axtoplevel}
\node at (9.15,-5.10) {$S^2\sqcup S^2$};
\AxTopSphere{13.35}{-4.12}{.59}{axtopink}
\node at (13.35,-5.10) {$S^2$};

\node[anchor=west,font=\small\bfseries] at (.15,-7.02) {Nariai};
\AxTopRegion{4.65}{-6.92}{\AxTopHole{0}{0}{.40}}
\node at (4.65,-7.98) {$I\times S^2$};
\node[font=\footnotesize,text=black!60] at (4.65,-8.32) {product geometry};
\AxTopSphere{8.48}{-6.92}{.52}{axtoplevel}
\AxTopSphere{9.82}{-6.92}{.52}{axtoplevel}
\node at (9.15,-7.90) {$S^2\sqcup S^2$};
\AxTopSphere{13.35}{-6.92}{.59}{axtopink}
\node at (13.35,-7.90) {$S^2$};

\node[anchor=west,font=\small\bfseries,align=left]
  at (.15,-9.82) {Axial\\solutions};
\AxTopRegion{4.65}{-9.72}{%
  \AxTopHole{-.38}{.06}{.27}\AxTopHole{.38}{-.06}{.27}}
\node[font=\footnotesize] at (4.65,-10.86)
  {$\overline B^{\,3}\setminus(B_+^3\sqcup B_-^3)$};
\AxTopSphere{7.93}{-9.72}{.43}{axtoplevel}
\AxTopSphere{9.15}{-9.72}{.43}{axtoplevel}
\AxTopSphere{10.37}{-9.72}{.43}{axtoplevel}
\node at (9.15,-10.66) {$S^2\sqcup S^2\sqcup S^2$};
\node[font=\footnotesize,text=black] at (9.15,-11.05)
  {$0<c<c_{(1)}$ (before the first saddle)};
\draw[axtopink,line width=.9pt] (13.35,-9.72) ellipse[x radius=.66,y radius=.24];
\node at (13.35,-10.55) {$S^1$};

\draw[axtopouter,line width=1pt] (3.0,-11.69)--(3.47,-11.69);
\node[anchor=west,font=\footnotesize] at (3.60,-11.69) {outer boundary};
\draw[axtopinner,line width=1pt] (7.10,-11.69)--(7.57,-11.69);
\node[anchor=west,font=\footnotesize] at (7.70,-11.69) {inner boundaries};
\end{tikzpicture}
\endgroup
\caption{Topological representations of compact static regions and selected lapse levels, with $\max N=1$. The sets $B_+^3$ and $B_-^3$ are the two removed open balls. Blue and red distinguish the boundary components. For the constructed axial solutions, the maximum set is one circle (Proposition~\ref{AX-prop-lapse-maximum-circle}). Table~\ref{AX-tab-level-topologies} describes the level topology as the lapse value varies.}
\label{AX-fig-model-topologies}
\end{figure}

\begin{table}[htbp]
\centering
\small
\begin{tabular}{@{}lll@{}}
\hline
Model & Lapse value & Topology of $\{N=c\}$\\
\hline
All four models & $c=0$ & boundary components in Figure~\ref{AX-fig-model-topologies}\\
de Sitter & $0<c<1$; $c=1$ & $S^2$; one point\\
Kottler and Nariai & $0<c<1$; $c=1$ & $S^2\sqcup S^2$; $S^2$\\
Axial solutions (comparable masses) & $0<c<c_{(1)}$ & $S^2\sqcup S^2\sqcup S^2$\\
 & $c_{(1)}<c<c_{(2)}$ & $S^2\sqcup S^2$\\
 & $c_{(2)}<c<c_{(3)}$ & $S^2$\\
 & $c_{(3)}<c<1$ & $S^1\times S^1$\\
 & $c=1$ & $S^1$\\
Equal-mass solutions & $0<c<c_{\rm mid}^{\rm ex}$ & $S^2\sqcup S^2\sqcup S^2$\\
 & $c=c_{\rm mid}^{\rm ex}$ & $(S^2\vee S^2)\sqcup S^2$\\
 & $c_{\rm mid}^{\rm ex}<c<c_{\rm out}^{\rm ex}$ & $S^2\sqcup S^2$\\
 & $c=c_{\rm out}^{\rm ex}$ & two spheres meeting at two points\\
 & $c_{\rm out}^{\rm ex}<c<1$ & $S^1\times S^1$\\
\hline
\end{tabular}
\caption{Topologies of the lapse levels, with $\max N=1$. For de Sitter, Kottler and Nariai, the indicated regular levels occur at every $0<c<1$. For the axial solutions, $c_{(1)}\leq c_{(2)}\leq c_{(3)}$ are the three axial critical values, counted with multiplicity; intervals with equal endpoints are empty (Proposition~\ref{AX-prop-general-level-topology}). For equal masses the final rows give the complete level topology; $\vee$ denotes identification of one point on each sphere. The exact critical values satisfy $c_{\rm mid}^{\rm ex}<c_{\rm out}^{\rm ex}$ (Proposition~\ref{AX-prop-equal-mass-level-topology}).}
\label{AX-tab-level-topologies}
\end{table}


\subsection{Length scales and the leading equation for the horizon separation}
\label{AX-section-scales}

We introduce the length scales associated with two small reference horizons and their separation. The smooth part of the leading lapse at either center determines an explicit scalar equilibrium equation.

We keep $\Lambda>0$ fixed and recall the cosmological length $\ell=\sqrt{3/\Lambda}$. The positive parameter $m$ is an auxiliary inner mass scale. It specifies the reference horizon radius and lapse normalization; the associated Kottler mass is defined in~\eqref{AX-eq-kottler-fields}. It is not an asymptotic mass of the compact spatial region. We introduce $\varepsilon\coloneqq(m/\ell)^{1/3}$; for equal masses, the total-mass parameter in the introduction is $\varepsilon_{\mathrm{tot}}=2^{1/3}\varepsilon$. We use a dimensionless separation $D$ in a fixed compact interval $I\coloneqq[D_-,D_+]\Subset(0,\infty)$. If $e_z$ denotes the unit vector along the reference symmetry axis, the centers are $p_\pm\coloneqq\pm(d/2)e_z$. The scales satisfy
\be
\label{AX-eq-scales}
m=\varepsilon^3\ell,\qquad d=\varepsilon\ell D,\qquad \frac md=\frac{\varepsilon^2}{D},\qquad \Lambda d^2=3\varepsilon^2D^2.
\ee
The separation $d$ is measured in the reference coordinates; the proper distance between the two horizon boundaries is a different quantity. The three geometric scales are the hole radius, of order $m$, the separation $d$, and the cosmological radius $\ell$.

In the intermediate region, with Euclidean coordinates $X$, radii $r_\pm\coloneqq|X-p_\pm|$ and Euclidean metric $\delta_X$, the leading fields are
\be
\label{AX-eq-intermediate}
U_m\coloneqq\frac m{r_+}+\frac m{r_-},\qquad Q_\ell\coloneqq\frac{|X|^2}{2\ell^2},\qquad N_I\coloneqq1-U_m-Q_\ell,\qquad h_I\coloneqq(1+2U_m-Q_\ell)\delta_X.
\ee
These expressions are approximate fields, used on regions separated from the centers. The operators $\Delta_0$ and $\Hess_0$ are Euclidean in the coordinates $X$. The fields satisfy the linearized static equations in these coordinates: $\Delta_0U_m=0$ and $\Delta_0Q_\ell=3/\ell^2=\Lambda$ imply $\Delta_0(N_I-1)=-\Lambda$. For the conformal perturbation $h=\delta_X+2a\delta_X$, the three-dimensional linearized Ricci tensor has the expression $-(\Hess_0a+\Delta_0a\,\delta_X)$. Indeed, the connection variation is $\dot\Gamma^k_{ij}=\delta^k_j\del_i a+\delta^k_i\del_j a-\delta_{ij}\del^k a$, and its Ricci contraction $\del_k\dot\Gamma^k_{ij}-\del_j\dot\Gamma^k_{ik}$ gives that expression. Taking $a=U_m-Q_\ell/2$ gives $\Ric^{\rm lin}=-\Hess_0U_m+2\ell^{-2}\delta_X$, which equals $\Hess_0(N_I-1)+\Lambda\delta_X$. The omitted terms are at least quadratic in $U_m,Q_\ell$ and their derivatives, or are products of $\Lambda$ with the metric perturbation.

At the upper center, the part of $N_I$ that is smooth at $p_+$ is $1-m/r_- -Q_\ell$. Its negative axial derivative is $d/(2\ell^2)-m/d^2$. We define the dimensionless equilibrium function $B(D)\coloneqq D/2-D^{-2}$ and obtain
\be
\label{AX-eq-balance}
\frac d{2\ell^2}-\frac m{d^2}=\frac{\varepsilon}{\ell}B(D),\qquad B(D_0)=0,\qquad D_0\coloneqq2^{1/3},\qquad B'(D)=\frac12+\frac2{D^3}>0.
\ee
The expression on the left is the negative axial derivative of $-m/|X-p_-|-|X|^2/(2\ell^2)$ at $p_+$. Its zero gives $d^3=2m\ell^2$, the scaling also used in~\cite[equation (3)]{P2-Dias}. This calculation determines a candidate separation. The nonlinear proof requires a quantitative comparison between $B$ and the obstruction of the corrected global metric.


\section{The reduced Einstein equation and its approximation}
\label{AX-section-existence}


\subsection{Weighted spaces and the reduced Einstein equation}
\label{AX-section-analytic}

\paragraph{Weighted tensor norms and the projection of the Einstein equation}
We construct the Kottler--de Sitter approximation and estimate its Einstein tensor on the different length scales. The axial flux calculation and an energy estimate for the auxiliary vector field identify the conditions under which a projected solution satisfies the Einstein equation. We introduce weighted H\"older spaces adapted to the horizon, separation and cosmological scales. The inverse of the operator on vector fields defines the tensor projection, and the admissible exponents control the dependence of the linear and nonlinear bounds on the masses. We fix an integer $k\geq3$, $0<\alpha<1$, and a weight $-1<\gamma<0$. A positive length function $\omega$ is comparable to $m$ near each inner horizon, to the distance from the nearer center away from those horizons, and to $\ell$ near the cosmological horizon. For each fixed pair of positive masses, $\omega$ is positive on the closed reference domain and measures the local chart radius. At the inner horizons one has $\omega\asymp m_\pm$, and at the outer horizon $\omega\asymp\ell$. The lapse instead has the boundary expansion $N=\kappa_j s+\Ocal(s^3)$ in inward normal distance $s$ at each horizon. On charts of radius $c\omega(x)$, where $c>0$ is fixed and small, we use the scaled $C^{r,\alpha}$ norm: derivatives of order $j$ are multiplied by $\omega(x)^j$, and the order-$r$ H\"older seminorm by $\omega(x)^{r+\alpha}$. A dimensionless weighted tensor norm is defined by
\be
\label{AX-eq-weighted-norm}
\|u\|_{r,\alpha,\gamma;\omega}\coloneqq \sup_x\Bigl(\frac{\omega(x)}\ell\Bigr)^{-\gamma} \|u\|_{C^{r,\alpha}_{\omega(x)}(B_{c\omega(x)}(x))}.
\ee

For the reference coordinates used below, a convenient choice is
\be
\label{AX-eq-explicit-length-weight}
\omega(X)\coloneqq\bigl(\ell^{-2}+r_+^{-2}+r_-^{-2}\bigr)^{-1/2}, \qquad r_\pm\coloneqq|X-p_\pm|.
\ee
The weight specifies the local length on which the geometry is examined. Near an inner horizon, the scaled norm measures a correction on the horizon scale; farther from the centers, it measures the same correction on the larger local scale. Multiplying each derivative by the corresponding power of $\omega$ allows these measurements to be compared after rescaling the coordinate neighborhoods to fixed size. The extra factor $(\omega/\ell)^{-\gamma}$ in~\eqref{AX-eq-weighted-norm} prescribes the allowed change in the amplitude of the tensor between these scales. In particular, $\omega$ remains positive at the horizons, whereas $N$ vanishes there. Figure~\ref{AX-fig-omega-domain} displays the levels of $\omega$ and its graph on a planar slice of the reference domain.

\begin{figure}[!htbp]
\centering
\includegraphics[width=\textwidth]{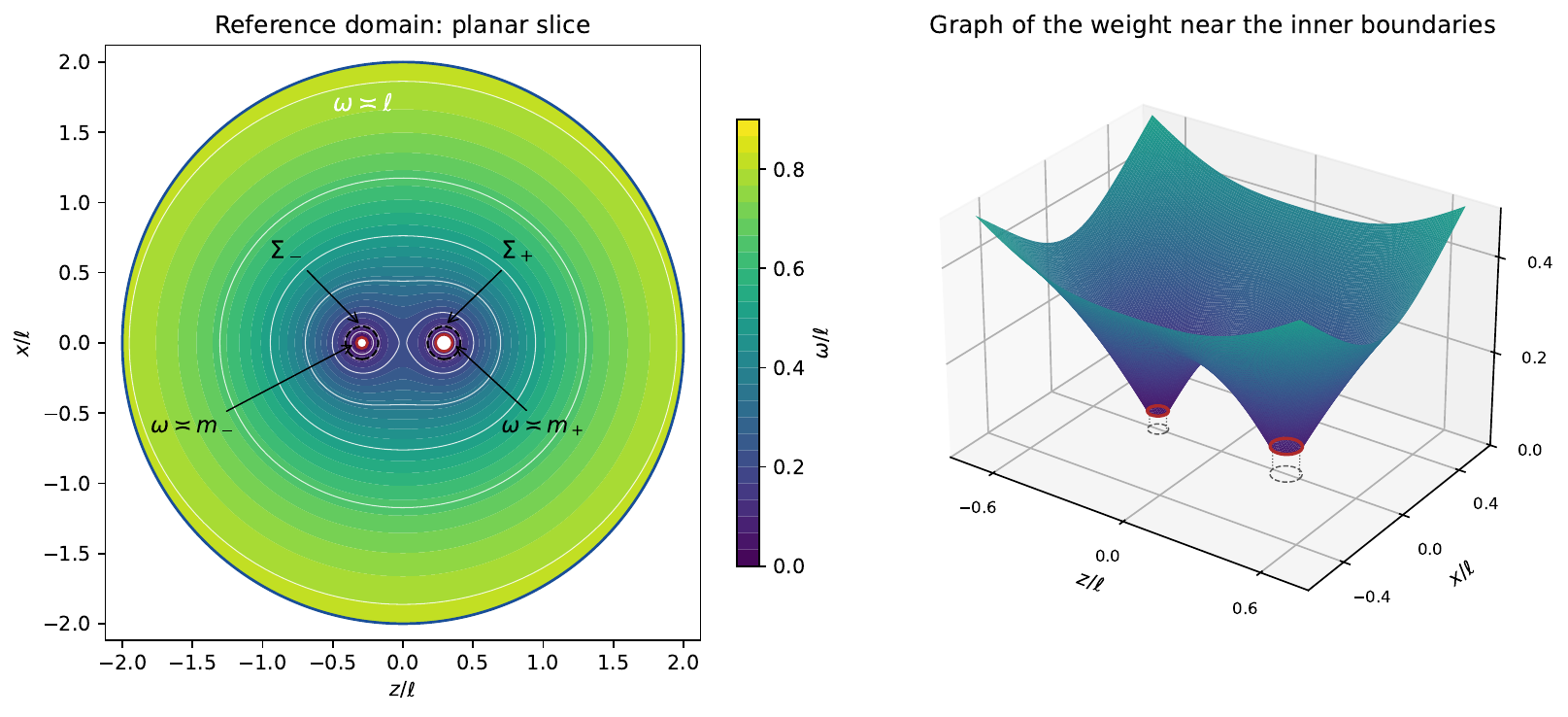}
\caption{The reference domain and the positive length weight $\omega$. The left panel shows a planar slice through the two centers, with level curves of $\omega$ and the inner boundaries in red. The outer boundary is the reference sphere $|X|=2\ell$. The dashed black circles are planar sections of the integration spheres $\Sigma_\pm=\{|X-p_\pm|=\vartheta d\}$ in~\eqref{AX-eq-spatial-flux}. The right panel shows the graph of $\omega$ over a neighborhood of the inner boundaries; the dashed circles lie in the plane $\omega=0$, below the positive red boundary values. The boundary sizes are $\omega\asymp m_\pm$ at the inner horizons and $\omega\asymp\ell$ at the outer horizon. The illustrative unequal masses are chosen to make the inner boundaries visible.}
\label{AX-fig-omega-domain}
\end{figure}

The components are measured relative to the reference metric in regular Cartesian lifts. At an axis or horizon pole the same definition uses the corresponding lifted chart, including the simultaneous parity conditions of Section~\ref{AX-section-boundaries}. We abbreviate the correction norm of order $k+2$ and weight $\gamma$ by $\|\cdot\|_X$, and the source norm of order $k$ and weight $\gamma-2$ by $\|\cdot\|_Y$. These are norm conventions; the coordinate condition and the orthogonality conditions on the inner Kottler regions will be imposed when the global tensor domain is defined in Section~\ref{AX-section-projected-tensor-estimates}.

The interval $-1<\gamma<0$ gives decay of the flat metric corrections and local integrability of the corresponding sources of order $r^{\gamma-2}$. We choose a vector weight parameter $0<\sigma_{\rm v}<1/2$ and a cutoff exponent $0<\zeta<4/9$. For $j\in\{+,-\}$, we write $m_j$ for the corresponding reference mass. The exponent $\zeta$ specifies growing inner cutoffs $R_j=m_j^{-\zeta}$ after normalization to $\ell=1$. The cutoff restriction $\zeta<4/9$ ensures that $m_jR_j=m_j^{1-\zeta}=o(\varepsilon^{5/3})$, so the growing cutoffs remain inside the Kottler collars of~\eqref{AX-eq-improved-reference-constants}. We define the exponent measuring the approximation error, the two exponents measuring growth in the inverse and quadratic estimates, and the resulting correction exponents by
\be
\label{AX-eq-exponent-relations}
\begin{aligned}
p&\coloneqq(12-3\gamma)/5,&
s_{\rm inv}&\coloneqq3(\sigma_{\rm v}-\gamma),&
s_{\rm quad}&\coloneqq3(\sigma_{\rm v}-2\gamma),\\
p_{\rm lin}&\coloneqq p+6\gamma-\tfrac32,&
p_{\rm cor}&\coloneqq p-s_{\rm inv},&
p_{\rm bal}&\coloneqq p_{\rm cor}+\gamma-2.
\end{aligned}
\ee
The exponent $p=3(4-\gamma)/5$ comes from the outer matching radius $r_{\rm out}=\varepsilon^{3/5}$ in~\eqref{AX-eq-improved-reference-constants}; the radii $\varepsilon^{5/3}$ and $\varepsilon^{3/5}$ balance the competing matching errors. The factor $3$ in the exponents for these bounds converts powers of the inner mass scale $m_j\asymp\varepsilon^3$ into powers of $\varepsilon$. The exponent $p_{\rm lin}$ records the projected linearization error in~\eqref{AX-eq-small-projection-on-linear-image}, while $p_{\rm cor}$ subtracts the exponent for growth of the inverse bound from the approximation order and $p_{\rm bal}$ measures the remaining margin in the compatibility equations.

We call a triple admissible when it lies in the ranges above and satisfies the inequalities in Table~\ref{AX-tab-admissible-exponents}. These conditions suffice for the estimates used below. The constants in subsequent estimates may depend on the fixed triple, $k,\alpha$, and the fixed compact ranges of mass fractions and geometric parameters; they are independent of the masses tending to zero.

\begin{table}[htbp]
\centering\small
\begin{tabular}{p{.43\textwidth}p{.49\textwidth}}
\hline
Condition on the exponents & Estimate controlled\\
\hline
$\zeta(1-\gamma)>\sigma_{\rm v}-\gamma$ & Absorption of the commutator with an inner cutoff\\
$2-\zeta(2+\gamma)>\sigma_{\rm v}-\gamma$ & Absorption of the cosmological remainder of the Kottler displacement tensors\\
$p_{\rm lin}>s_{\rm inv}$ & Absorption of the error in the linearized Bianchi identity\\
$p>2s_{\rm inv}+s_{\rm quad}$ & Contraction of the nonlinear correction map\\
$p_{\rm cor}+3\gamma>2$ & Relative metric error, horizon areas and estimation of the obstruction coefficients from the axial fluxes\\
\hline
\end{tabular}
\caption{Sufficient inequalities for the weighted estimates and the nonlinear argument.}
\label{AX-tab-admissible-exponents}
\end{table}

\paragraph{An equivalent description of the admissible exponents}
For the parameter ranges above, Table~\ref{AX-tab-admissible-exponents} imposes precisely the following inequalities\footnote{The fractions in~\eqref{AX-eq-admissible-parameter-family} follow by substituting~\eqref{AX-eq-exponent-relations} into the sufficient inequalities of Table~\ref{AX-tab-admissible-exponents}. The coefficients in~\eqref{AX-eq-direct-Kottler-absorption-margin} follow by using the resulting upper bounds for $\sigma_{\rm v}$ and $\zeta$. These numbers reflect the chosen matching scales and the ensuing algebraic conditions. By contrast, the constants specifying annulus widths and cutoff profiles may be chosen freely within the stated ranges.}:
\be
\label{AX-eq-admissible-parameter-family}
-\frac{2}{27}<\gamma<0,\qquad 0<\sigma_{\rm v}<\frac{2}{15}+\frac95\gamma,\qquad \frac{\sigma_{\rm v}-\gamma}{1-\gamma}<\zeta<\frac49.
\ee
The interval for $\sigma_{\rm v}$ is nonempty under $\gamma>-2/27$. Its upper bound gives $(\sigma_{\rm v}-\gamma)/(1-\gamma)<4/9$. The upper bound for $\sigma_{\rm v}$ is equivalent to $p_{\rm cor}+3\gamma>2$. Substitution of~\eqref{AX-eq-exponent-relations} gives $3\sigma_{\rm v}<2/5+(27/5)\gamma$, hence the fractions $2/15$ and $9/5$; requiring a positive interval for $\sigma_{\rm v}$ gives $\gamma>-2/27$. The inner cutoff commutator condition follows from the lower bound for $\zeta$; substitution of the upper bounds for $\sigma_{\rm v},\zeta$ gives the other inequalities in Table~\ref{AX-tab-admissible-exponents}. In particular, we have
\be
\label{AX-eq-direct-Kottler-absorption-margin}
2-\zeta(2+\gamma)-(\sigma_{\rm v}-\gamma) >\frac{44}{45}-\frac{56}{45}\gamma>0.
\ee
\begin{table}[htbp]
\centering\small
\begin{tabular}{p{.54\textwidth}p{.38\textwidth}}
\hline
Quantity or operator & Order as $\varepsilon\to0$\\
\hline
Einstein error of the improved reference & $\Ocal(\varepsilon^p)$\\
Tensor component of the inverse with two additional unknowns & $\Ocal(\varepsilon^{-s_{\rm inv}})$\\
Lipschitz factor of the projected quadratic term & $\Ocal(\varepsilon^{-s_{\rm quad}})(\|v\|_X+\|w\|_X)$\\
Correcting tensor & $\Ocal(\varepsilon^{p_{\rm cor}})$\\
Relative correction on a separation annulus & $\Ocal(\varepsilon^{p_{\rm cor}+\gamma})$\\
Error in the normalized axial fluxes & $\Ocal(\varepsilon^{p_{\rm bal}})$\\
\hline
\end{tabular}
\caption{The exponents proved in the approximation, inverse and nonlinear estimates.}
\label{AX-tab-estimate-exponents}
\end{table}

We use these weighted H\"older norms for the global elliptic correction. For reference spatial data $(h_0,N_0)$, we define the auxiliary Riemannian metric $G\coloneqq h_0+N_0^2dt_E^2$. For a current auxiliary metric $g$, we define
\be
\label{AX-eq-projection}
\begin{gathered}
S_\Lambda(g)\coloneqq\Ric_g-\tfrac12\Scal_g\,g+\Lambda g,\qquad
B_GV\coloneqq\tfrac12\mathcal L_VG,\\
(\delta_GT)_j\coloneqq-\nabla_G^iT_{ij},\qquad
A_G\coloneqq\delta_GB_G,\qquad
\Pi_G\coloneqq I-B_GA_G^{-1}\delta_G.
\end{gathered}
\ee
Here $V$ is a regular spatial vector field orthogonal to the axial symmetry orbits, invariant under rotation and with the reflection symmetry. Thus it has components only in the two directions represented in the orbit domain. The estimates for this inverse of the operator on vector fields are proved below, starting with Proposition~\ref{AX-prop-uniform-vector}. The formulas are understood on the static interior and in the local angular lifts at the horizons. They do not assert smoothness of a globally filled metric with a single period of $t_E$. When the inverse exists, $\delta_G\Pi_GT=0$ is obtained by applying $\delta_G$ to the defining formula. Also, $\Pi_GT=T$ whenever $\delta_GT=0$. These two identities imply $\Pi_G^2=\Pi_G$.

We make the estimates explicit. For a regular tensor correction $s$ and its auxiliary metric $g=G+s$, we set $\eta\coloneqq\|g-G\|_{C^{2,\alpha}_{0;\omega}}$, using the weight-zero version of~\eqref{AX-eq-weighted-norm}, and let $\eta_{\rm sep}$ be its scaled $C^{2,\alpha}_d$ norm on an enlarged sphere neighborhood of radius comparable to $d$ about $p_+$. The required comparisons are
\be
\label{AX-eq-correction-comparison}
\eta\leq C\varepsilon^{3\gamma}\|s\|_X,\qquad \eta_{\rm sep}\leq C\varepsilon^\gamma\|s\|_X.
\ee
These powers follow from $\omega/\ell\geq c\varepsilon^3$ globally and $\omega/\ell\asymp\varepsilon$ on the separation neighborhood: $|s|$ in the scaled norm is bounded by $(\omega/\ell)^\gamma\|s\|_X$. The same comparisons hold in each fixed regular angular lift.

\paragraph{Integration over spheres to measure the axial obstructions}
The fluxes below detect the finite-dimensional obstructions to the corrected Einstein equation. Their interpretation as forces is only explanatory: the proof uses the flux integrals and the coefficients of the obstruction tensors.
We fix annulus constants $0<a_{\rm ann}<1<b_{\rm ann}$ and choose $\vartheta>0$ satisfying $2b_{\rm ann}\vartheta<1$. The annuli with radii between $a_{\rm ann}\vartheta d$ and $b_{\rm ann}\vartheta d$ contain their flux spheres and have disjoint closures. These constants specify the widths of the neighborhoods; all subsequent estimates allow any fixed choice satisfying these inequalities. For sufficiently small masses, the full two-dimensional Euclidean sphere $\Sigma_a\coloneqq\{|X-p_a|=\vartheta d\}$ lies in the spatial interior, surrounds the inner horizon $H_a$, and remains separated from the other horizon. Its planar section is shown as a dashed black circle in Figure~\ref{AX-fig-omega-domain}. For a tensor $T$, we define its axial spatial flux by integration over this sphere:
\be
\label{AX-eq-spatial-flux}
\operatorname{Flux}_a(T)\coloneqq \int_{|X-p_a|=\vartheta d}T_{ij}n_\delta^i(e_z)^j\,dA_\delta,\qquad a\in\{+,-\}.
\ee
Here $i,j$ are spatial indices, $n_\delta$ is the outward Euclidean unit normal to the sphere, pointing away from $p_a$, and $dA_\delta$ is its two-dimensional Euclidean area measure. This specifies the sign and normalization; the flux is a reference-coordinate functional. We use $P_E>0$ for an integration length in $t_E$ on the static interior, independently of the local horizon angles.


\subsection{Kottler and de Sitter reference metrics and the vector energy estimate}
\label{AX-section-uniform-vector}

\paragraph{The reference fields and their horizon parameters}
The de Sitter and Kottler models are spherically symmetric static solutions with positive cosmological constant. The de Sitter spatial region has a regular center and one cosmological horizon. For sufficiently small positive mass, the Kottler spatial region lies between a black-hole horizon and a cosmological horizon. In the reference construction, de Sitter supplies the outer region, while a Kottler collar supplies a neighborhood of each inner horizon. The formulas below specify the coordinates, mass parameters and lapse normalizations.  We join two exact Kottler collars to an outer de Sitter region through the intermediate fields. We estimate the Einstein tensor and the operator on spatial vector fields orthogonal to the axial symmetry orbits directly on this reference. We then improve the matching at the inner horizons, keeping the intermediate fields on the annuli used for the axial fluxes. All projections and coordinate conditions below use the resulting Kottler--de Sitter reference itself.

We set $\ell=1$, so $\Lambda=3$, and first use the equal scales $m_j=m=\varepsilon^3$ and $d=\varepsilon D$, with $D\in I\Subset(0,\infty)$. In the closed Euclidean ball $\overline B_2$, we remove the open balls $B_{m_j/2}(p_j)$ and denote the remaining compact spatial region by $M_\varepsilon$. The coordinate radii are $r=|X|$ and $r_j=|X-p_j|$. The outer de Sitter fields are
\be
\label{AX-eq-vector-reference-models}
h_D\coloneqq(1+r^2/4)^{-2}\delta_X,\qquad N_D\coloneqq\frac{1-r^2/4}{1+r^2/4}.
\ee
For each inner horizon, we choose constants $c_{h,j},c_{N,j}>0$. For the first reference they equal one; their matched values will be specified in~\eqref{AX-eq-improved-reference-constants}. They determine the horizon radius and surface gravity
\be
\label{AX-eq-kottler-horizon-parameters}
r_{H,j}\coloneqq2m_j\sqrt{c_{h,j}},\qquad \kappa_j\coloneqq\frac{c_{N,j}}{4m_j\sqrt{c_{h,j}}},\qquad Q_{H,j}\coloneqq\Lambda r_{H,j}^2.
\ee
The corresponding Kottler parameters and fields are
\be
\label{AX-eq-kottler-fields}
\begin{gathered}
\mu_j\coloneqq\frac{r_{H,j}}2(1-Q_{H,j}/3),\qquad
\alpha_{K,j}\coloneqq\frac{2\kappa_jr_{H,j}}{1-Q_{H,j}}
=\frac{c_{N,j}}{1-Q_{H,j}},\\
F_j(R)\coloneqq1-\frac{2\mu_j}{R}-\frac{\Lambda R^2}{3},\qquad
h_{K,j}\coloneqq F_j(R)^{-1}dR^2+R^2d\Omega^2,\qquad
N_{K,j}\coloneqq\alpha_{K,j}\sqrt{F_j(R)}.
\end{gathered}
\ee
Here $R$ is the areal radius on the spherical collar: its circular orbit radius, denoted by $R$ in~\eqref{AX-eq-metric}, is $R\sin\theta$ in spherical polar coordinates. The metric $d\Omega^2$ is the unit round metric, and $G_{K,j}\coloneqq h_{K,j}+N_{K,j}^2dt_E^2$ denotes the auxiliary metric on the collar. The identities $RF_j'+F_j-1+\Lambda R^2=0$ and $F_j''/2+F_j'/R=-\Lambda$ verify~\eqref{LS-eq-static}. Under the condition $Q_{H,j}<1$, the root $R=r_{H,j}$ defines a nondegenerate black-hole horizon. Its area is $4\pi r_{H,j}^2$ and its surface gravity is $\alpha_{K,j}F_j'(r_{H,j})/2=\kappa_j$. The Kottler mass $\mu_j=m_j\sqrt{c_{h,j}}(1-Q_{H,j}/3)$ differs from the auxiliary scale $m_j$. Figure~\ref{AX-fig-reference-profiles} displays the lapse profiles of the two models and locates their horizons. In these coordinates $F^{-1}$ is the radial metric coefficient; its divergence at a simple zero of $F$ is resolved by using normal distance to the horizon.

\begin{figure}[!htbp]
\centering
\includegraphics[width=\textwidth]{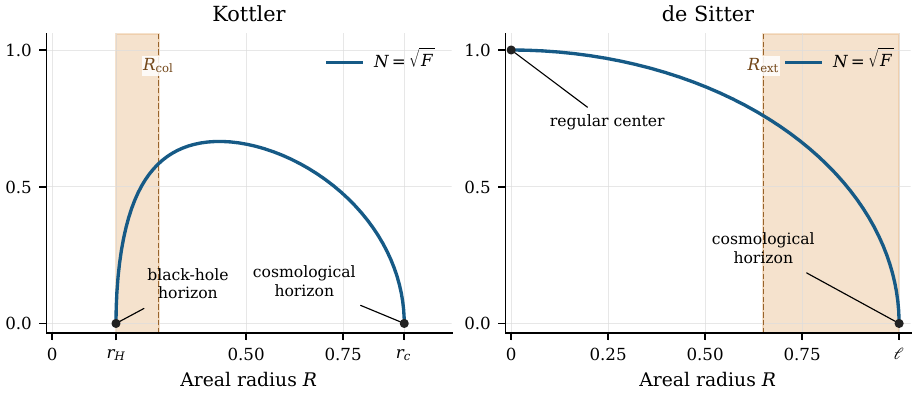}
\caption{Lapse profiles $N=\sqrt F$ in areal radius $R$, with $\ell=1$, $\Lambda=3$ and lapse factor one. On the left, for the illustrative Kottler mass $\mu=0.08$, $F=1-2\mu/R-R^2$ is positive between the horizon roots $r_H<r_c$. On the right, de Sitter has $F=1-R^2$, a regular center at $R=0$ and a horizon at $R=\ell$. The brown intervals schematically mark the reference pieces: the Kottler collar $r_H\leq R\leq R_{\mathrm{col}}$ on the left and the outer de Sitter region $R\geq R_{\mathrm{ext}}$ on the right. The areal cutoff locations are schematic; Figure~\ref{AX-fig-gluing-scales} specifies the reference-coordinate radii $r_{\rm in},r_{\rm out}$ and the transition annuli. At each simple horizon, normal distance gives $N=\kappa s+\Ocal(s^3)$, as in~\eqref{AX-eq-horizon-data}.}
\label{AX-fig-reference-profiles}
\end{figure}

At fixed $\Lambda$, each Kottler reference collar has two parameters, the mass $\mu_j$ and the positive lapse factor $\alpha_{K,j}$, or equivalently $r_{H,j}$ and $\kappa_j$. The triple $m_j,c_{h,j},c_{N,j}$ is redundant up to radial reparametrization. Before matching, two collars carry four local parameters; rescaling both lapse functions by the same positive constant leaves two masses and one relative surface gravity. In this construction, however, matching determines $c_{h,j},c_{N,j}$. The construction does not show that $\kappa_+/\kappa_-$ can be chosen independently, nor does it give a three-parameter classification of global solutions.

We use the radial coordinate and local angular coordinate
\be
\label{AX-eq-kottler-collar-identification}
R=\sqrt{c_{h,j}}\,r_j\Bigl(1+\frac{m_j}{2r_j}\Bigr)^2,\qquad y_j\coloneqq\sqrt{R/r_{H,j}-1} =\frac{r_j-m_j/2}{\sqrt{2m_jr_j}},\qquad \vartheta_j\coloneqq\kappa_jt_E.
\ee
Thus $r_j=m_j/2$ is the horizon. These are reference coordinates, not isotropic coordinates for $h_{K,j}$. To give the fields directly in these coordinates, we put $u_j\coloneqq m_j/(2r_j)$ and
$b_j\coloneqq1-\Lambda(R^2+r_{H,j}R+r_{H,j}^2)/3$. The fields are
\be
\label{AX-eq-kottler-radial-fields}
h_{K,j}=c_{h,j}(1+u_j)^4\bigl(b_j^{-1}dr_j^2+r_j^2d\Omega^2\bigr),\qquad N_{K,j}=\frac{c_{N,j}}{1-Q_{H,j}}\frac{1-u_j}{1+u_j}\sqrt{b_j}.
\ee
Indeed, $F_j(R)=((1-u_j)/(1+u_j))^2b_j$ and $dR/dr_j=\sqrt{c_{h,j}}(1-u_j^2)$. The regular angular form of these fields is proved in Lemma~\ref{AX-lem-kottler-comparison}.

We choose a fixed smooth function $\chi:[0,\infty)\to[0,1]$ equal to one on $[0,1]$ and zero on $[2,\infty)$. For the first reference, we put $a\coloneqq\varepsilon^2$ and $b\coloneqq\varepsilon^{\beta_0}$, where the fixed exponent $\beta_0$ satisfies $(2-2\gamma)/(4-\gamma)\leq\beta_0\leq(1+2\gamma)/(1+\gamma)$. We define $\chi_j\coloneqq\chi(r_j/a)$ and $\chi_o\coloneqq\chi(r/b)$. The choice of inner exponent $2$ balances the two unmatched error orders $2-a_0\gamma$ and $6-a_0(2+\gamma)$. Equating these orders gives $a_0=2$ and maximizes their minimum. Using the intermediate fields~\eqref{AX-eq-intermediate}, we define the reference by
\be
\label{AX-eq-vector-reference-patching}
\begin{aligned}
h_B&\coloneqq\chi_o h_I+(1-\chi_o)h_D,&
N_B&\coloneqq\chi_o N_I+(1-\chi_o)N_D,\\
h_\varepsilon&\coloneqq h_B+\sum_j\chi_j(h_{K,j}-h_B),&
N_\varepsilon&\coloneqq N_B+\sum_j\chi_j(N_{K,j}-N_B),\\
G_\varepsilon&\coloneqq h_\varepsilon+N_\varepsilon^2dt_E^2,&
\omega&=(1+r_+^{-2}+r_-^{-2})^{-1/2}.
\end{aligned}
\ee
Each inner term is supported in its collar. We fix a collar margin $\beta_{\rm col}\in(0,1)$. A sufficient condition for positivity of its Kottler factor is
\be
\label{AX-eq-kottler-collar-positivity}
\frac{\Lambda}{3}\max_j\bigl(L_j^2+r_{H,j}L_j+r_{H,j}^2\bigr)\leq1-\beta_{\rm col}, \qquad L_j\coloneqq2a\sqrt{c_{h,j}}\Bigl(1+\frac{m_j}{4a}\Bigr)^2.
\ee
It gives $b_j\geq\beta_{\rm col}$ on each inserted collar and $Q_{H,j}<1$. For sufficiently small $\varepsilon$, this inequality holds, the inner supports are disjoint, and the intermediate lapse is positive wherever a cutoff acts, since $U_m+Q_1\leq C\varepsilon$. The spatial metric and lapse are convex combinations of positive fields there. Near the inner boundaries the fields are exactly Kottler, and near $r=2$ they are exactly de Sitter. They have three nondegenerate regular horizons and regular axes. The isolated cosmological roots of the Kottler fields are outside the inserted collars and introduce no boundary. The fields are exactly $(h_I,N_I)$ on the enlarged flux annuli. The smooth cutoffs are constant on each side of a transition annulus, so the fields and all their derivatives agree across its boundary spheres. The final global correction is not required to vanish on the exact model regions of this reference.

\begin{lemma}[A bound for the Einstein tensor of the reference metric]
\label{AX-lem-reference-error}
For the reference~\eqref{AX-eq-vector-reference-patching} with $a=\varepsilon^2$ and $b=\varepsilon^{\beta_0}$ as above, every fixed exponent $\nu$ satisfying $2-\gamma<\nu\leq2-2\gamma$ gives, for each fixed $k\geq0$ and $0<\alpha<1$,
\be
\label{AX-eq-reference-error-proved}
\|S_3(G_\varepsilon)\|_{k,\alpha,\gamma-2;\omega}\leq C_{k,\alpha,\gamma}\varepsilon^\nu.
\ee
Here $\ell=1$, and the norms are those of~\eqref{AX-eq-weighted-norm}, using the reference metric and regular local angular lifts. The constants are uniform in $D\in I$ and for the comparable unequal masses specified below.
\end{lemma}
\begin{proof}
\bse
On a chart of radius comparable to $\omega$, the weighted source norm is the scaled norm of $\omega^2S_3$ multiplied by $\omega^{-\gamma}$. We estimate this quantity in the four regions of the construction. Every estimate below holds with any fixed number of scaled derivatives; an additional derivative also gives its H\"older version by the mean value theorem.

The Einstein tensor vanishes on the exact Kottler collars. On an inner transition, $r_j\asymp a$ and the companion distance is comparable to $d$. Expanding the explicit Kottler fields~\eqref{AX-eq-kottler-radial-fields} in $m/r_j$ gives $h_{K,j}=(1+2m/r_j)\delta_X+\Ocal((m/a)^2+a^2)$ and $N_{K,j}=1-m/r_j+\Ocal((m/a)^2+a^2)$ in every fixed scaled norm. We write $G_I\coloneqq h_I+N_I^2dt_E^2$ for the intermediate product metric. Its difference from $G_{K,j}$ satisfies
\be
\bigl\|G_I-G_{K,j}\bigr\|_{C^{k+2,\alpha}_a} \leq C\bigl((m/a)^2+m/d+d^2+a^2\bigr)\leq C\varepsilon^2.
\ee
The companion and cosmological terms obey the same bound after differentiation, since $a\ll d$. Differentiation of the fixed cutoff preserves this scaled estimate. Smooth dependence of the Ricci formula on the metric and its first two derivatives, relative to the rescaled positive Kottler metric on that annulus, gives $a^2\|S_3\|_{C^{k,\alpha}_a}\leq C\varepsilon^2$. The weighted contribution is therefore
\be
a^{-\gamma}\varepsilon^2=\varepsilon^{2-2\gamma}.
\ee

In the exact intermediate region, we use $r\leq C(\omega+d)$ and the derivative bounds on $m/r_+$ and $m/r_-$. The first-order static equations verified after~\eqref{AX-eq-intermediate} cancel the constant and linear terms in the Einstein equation, including its cosmological source. The remaining terms in the coordinate Ricci formula are products of the metric perturbation with its second derivatives, or products of first derivatives; the cosmological correction is bounded by the same quantity after multiplication by $\omega^2$. Thus we obtain
\be
\omega^2\|S_3(G_I)\|_{C^{k,\alpha}_\omega} \leq C\bigl((m/\omega)^2+d^4+\omega^4\bigr).
\ee
More explicitly, the scaled norms of the metric perturbation and its first two derivatives are bounded by $C(m/\omega+d^2+\omega^2)$. The coefficient of $dt_E^2$ $N_I^2$ adds only a quadratic term to its first-order perturbation. Taylor's formula for the Ricci tensor, with uniformly bounded inverse-metric coefficients, bounds the remaining expression by the square of this quantity. Since $a\leq C\omega\leq Cb$ and $2+\gamma>0$, multiplication by $\omega^{-\gamma}$ bounds this contribution by
\be
C\bigl(m^2a^{-2-\gamma}+d^4b^{-\gamma}+b^{4-\gamma}\bigr) \leq C\bigl(\varepsilon^{2-2\gamma}+\varepsilon^{4-\beta_0\gamma}+\varepsilon^{\beta_0(4-\gamma)}\bigr).
\ee

On the outer transition, $r\asymp b$ and $r_j\asymp b$. The de Sitter factors in~\eqref{AX-eq-vector-reference-models} have expansions $h_D=(1-r^2/2+\Ocal(r^4))\delta_X$ and $N_D=1-r^2/2+\Ocal(r^4)$. Consequently the scaled difference from the intermediate product is bounded by $C(m/b+b^4)$. The cutoff has bounded derivatives in the scaled chart. Since $S_3(G_D)=0$, the smooth Ricci formula now gives
\be
b^{-\gamma}b^2\|S_3(G_\varepsilon)\|_{C^{k,\alpha}_b} \leq C\bigl(mb^{-1-\gamma}+b^{4-\gamma}\bigr) =C\bigl(\varepsilon^{3-\beta_0(1+\gamma)}+\varepsilon^{\beta_0(4-\gamma)}\bigr).
\ee
The tensor vanishes in the exact outer de Sitter and inner Kottler regions, and the inequalities defining $\beta_0$ make each outer exponent at least $2-2\gamma$; the intermediate exponent $4-\beta_0\gamma$ is larger as well. Every exponent obtained above is therefore at least $2-2\gamma$, and the error is bounded by $C\varepsilon^{2-2\gamma}\leq C\varepsilon^\nu$. The lower restriction $\nu>2-\gamma$ ensures that a correction of norm $\Ocal(\varepsilon^\nu)$ is $o(\varepsilon^2)$ on a separation annulus, by~\eqref{AX-eq-correction-comparison}. Replacing $m$ by two comparable masses changes only the constants.
\ese
\end{proof}

Throughout the vector estimates, the fields are spatial, orthogonal to the axial symmetry orbits and invariant under rotation. Equivalently, they are independent of $t_E$ and $\varphi$ and have no component in either of these two Killing directions. They extend smoothly in the invariant local angular lifts at every horizon, axis and pole. Norms below use $G_\varepsilon$ and any fixed positive integration length in $t_E$. Each integral has the same length factor, which does not affect the constants. We write $B_\varepsilon V\coloneqq\tfrac12\mathcal L_VG_\varepsilon$.

\begin{proposition}[A uniform energy estimate for spatial vector fields orthogonal to the axial symmetry orbits]
\label{AX-prop-uniform-vector}
For the references~\eqref{AX-eq-vector-reference-patching}, there is a constant $C$ independent of sufficiently small $\varepsilon$, $D\in I$ and the integration period in $t_E$ such that every smooth spatial vector field orthogonal to the axial symmetry orbits satisfies
\be
\label{AX-eq-uniform-vector-proved}
\|\nabla^{G_\varepsilon}V\|_{L^2(G_\varepsilon)} +\|\omega^{-1}V\|_{L^2(G_\varepsilon)} \leq C\|B_\varepsilon V\|_{L^2(G_\varepsilon)}.
\ee
The same conclusion holds for two comparable positive masses $m_j=\theta_j\varepsilon^3$, with $\theta_j$ in a fixed compact subset of $(0,\infty)$, and centers $p_j$ on the axis satisfying $|p_j|\leq C_0\varepsilon$ and $c_0\varepsilon\leq|p_+-p_-|\leq C_0\varepsilon$, with fixed positive constants $c_0,C_0$. A reflection interchanging the two horizons is not required. The estimate persists for sufficiently small perturbations in the scaled $C^1$ norm with weight zero that are smooth in the same reference angular lifts.
\end{proposition}
\begin{proof}
\bse
\emph{1. Weighted Hardy inequality.} We first establish a weighted inequality for scalar functions directly on a Kottler collar. In its areal coordinate, $F_j(R)=(1-r_{H,j}/R)b_j(R)$, where $b_j(R)\geq\beta_*>0$; condition~\eqref{AX-eq-kottler-collar-positivity} gives $\beta_*=\beta_{\rm col}$. After division by the integration length in $t_E$, the volume measure is $\alpha_{K,j}R^2\,dR\,d\Omega$, and the radial gradient integral has density $\alpha_{K,j}F_jR^2|\del_Rf|^2$. For a regular invariant scalar and $r_{H,j}<L\leq a\sqrt{c_{h,j}}(1+m_j/(2a))^2$, integration of $\del_R((R-r_{H,j})f^2)$ gives
\be
\int_{r_{H,j}}^{L}f^2\,dR \leq2Lf(L)^2+\frac4{\beta_*}\int_{r_{H,j}}^{L}F_jR^2|\del_Rf|^2\,dR.
\ee
The lower boundary term vanishes. Young's inequality bounds $2(R-r_{H,j})|f\del_Rf|$ by $f^2/2+2(R-r_{H,j})^2|\del_Rf|^2$, and $(R-r_{H,j})^2\leq F_jR^2/\beta_*$. The constant $\alpha_{K,j}$ multiplies both radial integrals and is uniformly bounded above and below. Also, $R$ and $r_j$ are uniformly comparable down to the horizon. This proves the required weighted inequality on the inner collar.

On the complement of the two radius-$a$ balls inside a fixed smaller ball, the spatial coefficients and lapse are uniformly comparable to the flat product. Integrate $\operatorname{div}((X-p_j)/r_j^2)=r_j^{-2}$ against a compactly supported $f^2$ and apply Cauchy--Schwarz. This proves the Euclidean Hardy inequality $\int r_j^{-2}f^2\leq4\int|df|^2$. To apply it on the punctured region, we extend $f$ into each radius-$a$ ball from its surrounding annulus. An extension with gradient norm bounded by the annular gradient norm is obtained by subtracting the annular average, extending the difference on the fixed unit annulus by reflection in boundary charts and a partition of unity, and adding the average back. Rescaling gives constants independent of $a$; the Poincar\'e inequality for the mean-zero difference controls the cutoff terms. A fixed outer cutoff adds only an unweighted $L^2$ term. The trace at $r_j=a$ is controlled by the same annular gradient and weighted $L^2$ integrals, by integrating the fundamental theorem of calculus along radii in $[a,2a]$. Combining these estimates with the inequality on the inner region gives
\be
\label{AX-eq-vector-weighted-Hardy}
\|\omega^{-1}f\|_{L^2(G_\varepsilon)}^2 \leq C_H\bigl(\|df\|_{L^2(G_\varepsilon)}^2+\|f\|_{L^2(G_\varepsilon)}^2\bigr).
\ee
All transition annuli have uniformly comparable metrics, so the comparison constants are uniform. On regions where $r$ has a positive lower bound, $\omega^{-1}$ is bounded, so the last term suffices, including near the cosmological horizon. The scalar inequality applied to $(|V|^2+\delta^2)^{1/2}$ passes to the limit as $\delta\downarrow0$. The bound $|df|\leq|\nabla V|$ then gives the vector inequality.
\ese

\bse
\emph{2. Curvature control and the symmetric-gradient identity.} We next control the curvature term in the integrated symmetric-gradient identity. Both the exact Kottler collars and the exact de Sitter region have $\Ric_{G_\varepsilon}=3G_\varepsilon$. On every transition or intermediate chart, the metric differs from the flat product by at most $C\varepsilon$ in the scaled $C^2$ norm at radius comparable to $\omega$. This follows from $m/r_j\leq C\varepsilon$, $r^2\leq C\varepsilon$ and the factors $a^{-k}$ or $b^{-k}$ contributed by the cutoff derivatives of order $k\leq2$. The coordinate Ricci formula therefore yields
\be
\label{AX-eq-vector-reference-Ricci}
|\Ric_{G_\varepsilon}|_{G_\varepsilon} \leq C_0+C_1\varepsilon\omega^{-2}.
\ee
For the outer transition, $\omega\asymp b$; for an inner transition, $\omega\asymp a$. In the remaining intermediate region, derivatives of each potential obey the same estimates with $\omega\leq r_j$ and $\omega\leq C(r+d)$. This verifies the curvature bound throughout the region where the reference is not an exact model.

Integration by parts and commutation of one covariant derivative give
\be
2\|B_\varepsilon V\|_2^2 =\|\nabla V\|_2^2+\|\operatorname{div}V\|_2^2 -\int\Ric_{G_\varepsilon}(V,V)\,dV_{G_\varepsilon}.
\ee
For clarity, the cross term in $2|B V|^2=|\nabla V|^2+\nabla_iV_j\nabla^jV^i$ integrates to $\|\operatorname{div}V\|_2^2-\int\Ric(V,V)$. The horizon and axis tube boundary terms vanish in their regular angular lifts, as in Proposition~\ref{AX-lem-Einstein-recovery}. Substitution of~\eqref{AX-eq-vector-reference-Ricci} and~\eqref{AX-eq-vector-weighted-Hardy}, followed by absorption for small $\varepsilon$, proves
\be
\label{AX-eq-vector-augmented-Korn}
\|\nabla V\|_2^2+\|\omega^{-1}V\|_2^2 \leq C\bigl(\|B_\varepsilon V\|_2^2+\|V\|_2^2\bigr).
\ee
Here and in the next step the integration period in $t_E$ may be fixed to one, since it multiplies every integral.
\ese

\bse
\emph{3. Removal of the lower norm and stability.} We remove the last term by contradiction. If~\eqref{AX-eq-uniform-vector-proved} failed along a sequence $\varepsilon_i\to0$, we could choose fields with the left side of~\eqref{AX-eq-vector-augmented-Korn} equal to one and $\|B_{\varepsilon_i}V_i\|_2\to0$. On each fixed region $r\geq r_0>0$, the metrics are exactly de Sitter for large $i$. Their $H^1$ bounds give a subsequence converging strongly in $L^2$ on that region and weakly in $H^1$. This compactness also holds at the cosmological horizon: a finite set of regular four-dimensional angular charts turns it into the usual compactness on bounded smooth domains. The latter follows by approximation on cubes by their averages, whose $L^2$ error is bounded by the cube size times the $H^1$ norm. A diagonal subsequence gives a spatial vector field orthogonal to the axial symmetry orbits $V_0$ with $B_DV_0=0$ on $0<r\leq2$.

This Killing field vanishes. At a pole of the cosmological horizon, the two circle actions rotate two orthogonal coordinate planes. We set $s^2\coloneqq x_1^2+x_2^2$ and $u^2\coloneqq x_3^2+x_4^2$. Invariance and the absence of components along the circles give smooth scalar functions $A,C$ such that
\be
V_0=A(s^2,u^2)(x_1\del_{x_1}+x_2\del_{x_2}) +C(s^2,u^2)(x_3\del_{x_3}+x_4\del_{x_4}).
\ee
A weak Killing field is smooth: differentiating its equation and commuting derivatives expresses its second covariant derivatives as curvature contractions with the field, in the distributional sense. Starting from $H^1$, this gives $H^2$ regularity in local coordinates; iteration gives smoothness. Its value at the pole is zero. Its derivative there is diagonal on the two planes; the Killing equation makes that derivative skew, hence zero. Differentiating the Killing equation and commuting derivatives expresses $\nabla\nabla V_0$ as a curvature contraction with $V_0$. Along a geodesic, $(V_0,\nabla V_0)$ therefore satisfies a linear homogeneous ordinary differential system. The zero initial data propagate, first in a neighborhood and then throughout the connected region. Thus $V_0=0$.

There is no remaining concentration in the unweighted norm at the shrinking holes. Since $\omega\leq r_j\leq r+Cd$, the normalization gives
\be
\int_{r<r_0}|V_i|^2\,dV_{G_{\varepsilon_i}} \leq(r_0+Cd_i)^2\|\omega^{-1}V_i\|_2^2 \leq(r_0+Cd_i)^2.
\ee
Strong convergence outside this region, followed by $i\to\infty$ and $r_0\downarrow0$, gives $\|V_i\|_2\to0$. This contradicts~\eqref{AX-eq-vector-augmented-Korn}. If a proposed failure sequence has a positive limiting value of $\varepsilon$, the same compactness argument applies on the fixed smooth family; its Killing field vanishes by the identical pole argument. This completes the uniform estimate.

For comparable unequal masses, the inequality on the inner region uses the corresponding $r_{H,j}$, and all cutoff, curvature and concentration estimates have the same bounds. The estimate does not require a reflection interchanging the inner horizons. The position estimates in this argument use only $|p_j|\leq C_0\varepsilon$ and the two-sided bound for the separation. In particular, the concentration estimate uses $r_j\leq r+C_0\varepsilon$, while the cutoff and companion estimates use the separation; no positive lower bound for an individual $|p_j|$ enters the proof. The estimate therefore remains uniform when a center passes through the origin. Finally, if $G'=G_\varepsilon+q$ with $|q|+\omega|\nabla q|\leq\delta$, we have
\be
|B_{G'}V-B_\varepsilon V| \leq C\delta(|\nabla V|+\omega^{-1}|V|).
\ee
Metric comparison and~\eqref{AX-eq-uniform-vector-proved} absorb this difference for sufficiently small $\delta$. The connection difference satisfies the same bound, so the left side remains equivalent to its $G'$ version. Constant rescaling restores any fixed $\ell>0$, with $\omega$ replaced by the corresponding length function.
\ese
\end{proof}

The estimate gives a weak inverse for $A_{G_\varepsilon}=\delta_{G_\varepsilon}B_\varepsilon$ on the completion of regular spatial vector fields orthogonal to the axial symmetry orbits in its left-hand norm. For a bounded linear functional $L$, coercivity makes a minimizing sequence for $\frac12\int|B_\varepsilon V|^2-L(V)$ bounded. Weak lower semicontinuity gives a minimizer, strict positivity gives uniqueness, and the Euler equation yields the weak solution with norm at most $C\|L\|$. For smooth admissible sources, local elliptic regularity in the invariant lifts gives smoothness at fixed positive masses. The weighted projection requires the further estimates of Lemma~\ref{AX-lem-vector-weight-upper}; Proposition~\ref{AX-lem-vector-weight-loss} shows why the energy bound alone is insufficient.

\paragraph{Improving the approximation near the inner horizons}
To compensate for the negative power of $\varepsilon$ in the inverse estimate, we match the constant companion field at each small horizon and adjust the transition radii. At inner radius $\varepsilon^a$, the leading matching errors are of orders $\varepsilon^{6-2a}$ and $\varepsilon^{1+a}$. Equating these orders gives $a=5/3$. At outer radius $\varepsilon^b$, the weighted mass and cosmological truncation errors are of orders $\varepsilon^{3-b(1+\gamma)}$ and $\varepsilon^{b(4-\gamma)}$. Their balance gives $b=3/5$. Each choice maximizes the smaller of its two competing error exponents within these matching estimates. We therefore put
\be
\label{AX-eq-improved-reference-constants}
r_{\rm in}\coloneqq\varepsilon^{5/3},\qquad r_{\rm out}\coloneqq\varepsilon^{3/5},\qquad c_{h,j}\coloneqq1+\frac{2m_{\bar j}}d-\frac{|p_j|^2}{2},\qquad c_{N,j}\coloneqq1-\frac{m_{\bar j}}d-\frac{|p_j|^2}{2}.
\ee
Here $m_{\bar j}$ is the companion mass, equal to $m$ in the equal-mass case. For comparable unequal masses the collar at $p_j$ uses its own mass $m_j$. The improved reference $\widetilde G_\varepsilon\coloneqq\widetilde h_\varepsilon+\widetilde N_\varepsilon^2dt_E^2$ is defined by~\eqref{AX-eq-vector-reference-patching}, with $(a,b)$ replaced by $(r_{\rm in},r_{\rm out})$ and with the Kottler fields determined by these $c_{h,j},c_{N,j}$. The intermediate and outer fields are unchanged. We require $|p_j|\leq Cd$. Figure~\ref{AX-fig-gluing-scales} shows the regions of this reference and derives the main error powers. Table~\ref{AX-tab-estimate-exponents} records how the subsequent inverse and nonlinear estimates turn the approximation order into the order of the final correction.

\begin{lemma}[A smaller Einstein error with the same fields on the flux annuli]
\label{AX-lem-improved-reference-error}
The reference $\widetilde G_\varepsilon$ has three regular nondegenerate horizons and regular axes, and equals $h_I+N_I^2dt_E^2$ on the enlarged annuli used for the axial fluxes. For every fixed $k\geq0$, $0<\alpha<1$ and admissible weight $\gamma$, one has
\be
\label{AX-eq-improved-reference-error}
\|S_3(\widetilde G_\varepsilon)\|_{k,\alpha,\gamma-2;\omega}\leq C\varepsilon^p.
\ee
The constants are uniform for $D$ in its fixed compact interval and for comparable positive masses. The uniform energy estimate for the auxiliary vector field of Proposition~\ref{AX-prop-uniform-vector} also holds for these references in their own regular angular lifts.
\end{lemma}
\begin{proof}
\bse
The constants $c_{h,j},c_{N,j}$ are positive and differ from one by $\Ocal(\varepsilon^2)$. The resulting collars are exact Kottler collars with the same positive cosmological constant. Their horizon radius is $2m_j\sqrt{c_{h,j}}$ and their surface gravity is $c_{N,j}/(4m_j\sqrt{c_{h,j}})$. Condition~\eqref{AX-eq-kottler-collar-positivity} holds with $a=r_{\rm in}$ once $\varepsilon$ is small enough. The regular angular formula in Lemma~\ref{AX-lem-kottler-comparison} proves the horizon and pole regularity directly for this reference. Its local angular coordinate is determined by its own surface gravity.

Since $r_{\rm in}/d\to0$ and $d/r_{\rm out}\to0$, the inner supports are disjoint and avoid the annuli used for the axial fluxes. These annuli also lie in $r<r_{\rm out}$, where the outer cutoff equals one. The lapse and metric are positive where the cutoffs act, by the same convex-combination argument as before. The three horizons and their pole regularity therefore persist.

We give the error estimates on each region. On an inner transition, $r_j\asymp r_{\rm in}$. The companion potential and the cosmological quadratic term have been matched at $p_j$. Their variation on this annulus is bounded in every fixed scaled norm by $C(mr_{\rm in}/d^2+dr_{\rm in}+r_{\rm in}^2)$. Expansion of~\eqref{AX-eq-kottler-radial-fields} gives $h_{K,j}=(c_{h,j}+2m_j/r_j)\delta_X$ and $N_{K,j}=c_{N,j}-m_j/r_j$, up to scaled errors bounded by $C((m_j/r_{\rm in})^2+\varepsilon^2m_j/r_{\rm in}+r_{\rm in}^2)$. The radial anisotropic coefficient is $\Ocal(r_{\rm in}^2)$. Combining these bounds with the companion variation gives
\be
\label{AX-eq-improved-core-matching}
\begin{aligned}
&\|h_I-h_{K,j}\|_{C^{k+2,\alpha}_{r_{\rm in}}}
+\|N_I-N_{K,j}\|_{C^{k+2,\alpha}_{r_{\rm in}}}\\
&\leq C\bigl((m/r_{\rm in})^2+\varepsilon^2m/r_{\rm in}
+mr_{\rm in}/d^2+dr_{\rm in}+r_{\rm in}^2\bigr)
\leq C\varepsilon^{8/3}.
\end{aligned}
\ee
For example, the additional terms proportional to $m/r_j$ are $2(1-c_{h,j})m/r_j$ in the metric and $(c_{N,j}-1)m/r_j$ in the lapse; these give the second term on the right. The Taylor remainders of the exact Kottler factors give the first term; the terms containing $b_j-1$ and $Q_{H,j}$ are included in the final term. Derivatives of the fixed cutoffs, after multiplication by the corresponding powers of $r_{\rm in}$, remain bounded. Smooth dependence of the Ricci expression on the metric, its inverse and its first two derivatives therefore bounds the weighted Einstein error on each transition by $Cr_{\rm in}^{-\gamma}(\varepsilon^{8/3}+r_{\rm in}^2)$. On the exact Kottler collars, the Einstein tensor vanishes.

On the intermediate region, the fields coincide with those in Lemma~\ref{AX-lem-reference-error}. Its pointwise estimate, which follows from cancellation of the linearized static equations, still gives
\be
\label{AX-eq-improved-middle-error}
\|S_3(\widetilde G_\varepsilon)\|_{k,\alpha,\gamma-2;\omega;\mathrm{middle}} \leq C\bigl(m^2r_{\rm in}^{-2-\gamma} +d^4r_{\rm out}^{-\gamma}+r_{\rm out}^{4-\gamma}\bigr).
\ee
Here the extra final subscript restricts the supremum to that region. On the outer transition the intermediate fields differ from de Sitter by $\Ocal(m/r_{\rm out}+r_{\rm out}^4)$ in scaled norms. Since de Sitter has zero Einstein error, this contributes at most $C(mr_{\rm out}^{-1-\gamma}+r_{\rm out}^{4-\gamma})$. The outermost region is exactly de Sitter. Substituting the two radii, all contributions are bounded by
\be
\label{AX-eq-improved-error-exponents}
C\bigl(\varepsilon^{(8-5\gamma)/3}+\varepsilon^{(10-5\gamma)/3} +\varepsilon^{4-3\gamma/5}+\varepsilon^{(12-3\gamma)/5}\bigr) \leq C\varepsilon^p.
\ee
Indeed, $r_{\rm in}^{-\gamma}\varepsilon^{8/3}=m^2r_{\rm in}^{-2-\gamma}=\varepsilon^{(8-5\gamma)/3}$, while both outer-transition terms have exponent $p=(12-3\gamma)/5$. The other exponents are larger for $\gamma<0$. This proves the tensor error bound, including derivatives and H\"older seminorms. Comparable masses change only the constants.

For the vector estimate, we apply the proof of Proposition~\ref{AX-prop-uniform-vector} directly to this reference. On each Kottler collar, the factor $b_j$ is uniformly bounded below and the constants $\alpha_{K,j},c_{h,j},c_{N,j}$ are uniformly bounded above and below. The direct Hardy inequality and the radial measure comparisons therefore have uniform constants. In the remaining regions the gradient extension uses the new radius $r_{\rm in}$; its rescaled constants are independent of that radius, and the two extension annuli remain disjoint. The Ricci bound used in that proof has the same form $|\Ric|\leq C_0+C_1\varepsilon\omega^{-2}$: on the exact Kottler collars the Ricci tensor equals $3\widetilde G_\varepsilon$; on the inner transitions their scaled Ricci bound is $\Ocal(\varepsilon^{8/3})$; in the middle its quadratic error is bounded as above; and on the outer transition $m/r_{\rm out}=\varepsilon^{12/5}$. The cosmological terms have bounded unscaled Ricci norm and are absorbed in $C_0$. The integrated Korn identity and absorption therefore give the same preliminary inequality with an additional $L^2$ term. The outer metric is exactly de Sitter for $r\geq2r_{\rm out}\to0$. The Killing-field limit and the estimate excluding $L^2$ concentration near the centers are consequently identical. This proves the uniform energy estimate without comparing two different angular lifts.
\ese
\end{proof}

The exponent in~\eqref{AX-eq-improved-reference-error} concerns the unprojected Einstein tensor. The contracted Bianchi identity implies that the reference projection leaves this error unchanged, as proved in Proposition~\ref{AX-prop-projected-tensor-bound}. Corollary~\ref{AX-cor-improved-projection-quadratic} controls the nonlinear projection, and Theorem~\ref{AX-thm-physical-force-inverse} gives the inverse used in the construction.

\begin{figure}[!htbp]
\centering
\begin{tikzpicture}[x=1cm,y=1cm,line cap=round,line join=round,font=\small]
\path[use as bounding box] (-2.5,-3.35) rectangle (11.1,2.55);

\fill[blue!9] (0,0) circle[radius=2.2];
\fill[orange!23] (0,0) circle[radius=1.7];
\fill[green!7] (0,0) circle[radius=1.3];
\draw[blue!65!black,line width=.9pt] (0,0) circle[radius=2.2];
\draw[orange!65!black,line width=.5pt] (0,0) circle[radius=1.7];
\draw[orange!65!black,line width=.5pt] (0,0) circle[radius=1.3];

\foreach \a in {-.6,.6}{
  \fill[orange!40] (\a,0) circle[radius=.4];
  \fill[red!12] (\a,0) circle[radius=.25];
  \draw[orange!65!black,line width=.45pt] (\a,0) circle[radius=.4];
  \draw[red!60!black,line width=.45pt] (\a,0) circle[radius=.25];
  \fill[white] (\a,0) circle[radius=.095];
  \draw[red!75!black,line width=.9pt] (\a,0) circle[radius=.095];
  \fill[black] (\a,0) circle[radius=.017];
}

\draw[black!45,densely dashed,-{Stealth[length=1.6mm]}] (-2.3,0)--(2.48,0);
\node[anchor=north west] at (2.35,-.04) {$z$};
\draw[black!45,densely dashed] (-.6,.43)--(-.6,.82) (.6,.43)--(.6,.82);
\draw[{Stealth[length=1.5mm]}-{Stealth[length=1.5mm]}] (-.6,.76)--(.6,.76);
\node[above,inner sep=1.5pt] at (0,.76) {$d$};
\node at (-.6,-.63) {$p_-$};
\node at (.6,-.63) {$p_+$};
\node[above,blue!65!black] at (0,2.22) {$H_c:\ r=2$};
\node[fill=blue!9,inner sep=1.2pt] at (0,1.7) {$r=2r_{\rm out}$};
\node[fill=green!7,inner sep=1.2pt] at (0,1.3) {$r=r_{\rm out}$};

\fill[blue!9] (2.95,1.91) rectangle (3.24,2.16);
\draw[blue!65!black,line width=.45pt] (2.95,1.91) rectangle (3.24,2.16);
\node[anchor=west] at (3.45,2.07) {Exact de Sitter};
\node[anchor=west] at (3.45,1.66) {$(h_D,N_D),\qquad 2r_{\rm out}\leq r\leq2$};

\fill[orange!23] (2.95,.94) rectangle (3.24,1.19);
\draw[orange!65!black,line width=.45pt] (2.95,.94) rectangle (3.24,1.19);
\node[anchor=west] at (3.45,1.10) {Outer transition};
\node[anchor=west] at (3.45,.69) {$r_{\rm out}\leq r\leq2r_{\rm out}$};

\fill[green!7] (2.95,-.03) rectangle (3.24,.22);
\draw[green!45!black,line width=.45pt] (2.95,-.03) rectangle (3.24,.22);
\node[anchor=west] at (3.45,.13) {Intermediate fields $(h_I,N_I)$};
\node[anchor=west] at (3.45,-.28) {$r\leq r_{\rm out},\qquad r_\pm\geq2r_{\rm in}$};

\fill[orange!40] (2.95,-1.00) rectangle (3.24,-.75);
\draw[orange!65!black,line width=.45pt] (2.95,-1.00) rectangle (3.24,-.75);
\node[anchor=west] at (3.45,-.84) {Inner transitions};
\node[anchor=west] at (3.45,-1.25) {$r_{\rm in}\leq r_j\leq2r_{\rm in},\qquad j\in\{+,-\}$};

\fill[red!12] (2.95,-1.97) rectangle (3.24,-1.72);
\draw[red!60!black,line width=.45pt] (2.95,-1.97) rectangle (3.24,-1.72);
\node[anchor=west] at (3.45,-1.81) {Exact Kottler regions};
\node[anchor=west] at (3.45,-2.22) {$(h_{K,j},N_{K,j}),\qquad m_j/2\leq r_j\leq r_{\rm in}$};

\node at (4.3,-3.03) {$m_j\asymp\varepsilon^3\ll r_{\rm in}=\varepsilon^{5/3}\ll d\asymp\varepsilon\ll r_{\rm out}=\varepsilon^{3/5}\ll1$};
\end{tikzpicture}

\par\smallskip
\noindent\begin{minipage}[t]{.49\textwidth}
\centering\small
\textbf{Inner transition: $r_{\rm in}=\varepsilon^{5/3}$}\par\smallskip
\begin{tabular}{@{}ll@{}}
Matching term & Power of $\varepsilon$\\[2pt]
$(m/r_{\rm in})^2$ & $2(3-5/3)=8/3$\\[2pt]
$mr_{\rm in}/d^2$ & $3+5/3-2=8/3$\\[2pt]
$dr_{\rm in}$ & $1+5/3=8/3$\\[2pt]
$\varepsilon^2m/r_{\rm in},\ r_{\rm in}^2$ & $10/3$
\end{tabular}\par\smallskip
$\displaystyle r_{\rm in}^{-\gamma}\varepsilon^{8/3}=\varepsilon^{(8-5\gamma)/3}$
\end{minipage}\hfill
\begin{minipage}[t]{.49\textwidth}
\centering\small
\textbf{Outer transition: $r_{\rm out}=\varepsilon^{3/5}$}\par\smallskip
\begin{tabular}{@{}ll@{}}
Weighted error term & Power of $\varepsilon$\\[2pt]
$mr_{\rm out}^{-1-\gamma}$ & $3-\frac35(1+\gamma)$\\[4pt]
$r_{\rm out}^{4-\gamma}$ & $\frac35(4-\gamma)$
\end{tabular}\par\medskip
Both exponents equal $p=(12-3\gamma)/5$.\par\smallskip
The outer transition determines the global bound\\[2pt]
$\displaystyle\|S_3(\widetilde G_\varepsilon)\|_Y\leq C\varepsilon^p$.
\end{minipage}
\caption{Regions of the Kottler reference metric, before the global nonlinear correction, in reference coordinates with $\ell=1$. The transition radii are $r_{\rm in}$ and $r_{\rm out}$. Each transition annulus extends from its inner radius to twice that radius. The red inner circles are $r_j=m_j/2$. Across the inner annuli, the Kottler fields are joined to $(h_I,N_I)$. Across the outer annulus, the intermediate fields are joined to de Sitter, with cosmological horizon $r=2$. The lower panels derive the powers in~\eqref{AX-eq-improved-core-matching} and the weighted transition errors, using $m\asymp\varepsilon^3$ and $d\asymp\varepsilon$. The upper panel is schematic; distances and sizes are not drawn to scale.}
\label{AX-fig-gluing-scales}
\end{figure}

\begin{lemma}[Uniform Kottler collar bounds and convergence to Schwarzschild]
\label{AX-lem-kottler-comparison}
The Kottler collars~\eqref{AX-eq-kottler-fields} are regular when expressed in the local angular coordinates~\eqref{AX-eq-kottler-collar-identification}. For the local comparison, keep $r_{H,j}$ and $\kappa_j$ fixed and denote by $G_{j,Q}$ the lifted Kottler metric with cosmological coefficient $Q/r_{H,j}^2$. Its zero-coefficient member is written $G_{j,0}$ and corresponds to $Q=0$. The reference uses the member $Q=Q_{H,j}=\Lambda r_{H,j}^2$ for the prescribed $\Lambda$. For this reference member $Q=Q_{H,j}$, on an annulus $r_j\asymp t\geq m_j$ within the inserted collar, including a regular lifted chart when $t\asymp m_j$, one has
\be
\label{AX-eq-kottler-comparison-bounds}
\begin{aligned}
\|G_{j,Q}-G_{j,0}\|_{C^{k+2,\alpha}_t(G_{j,0})}&\leq C\Lambda t^2,\\
\|G_{j,Q}-G_{j,0}\|_{k+2,\alpha,\gamma;\omega;\{r_j\leq2r_{\rm in}\}}
&\leq C\varepsilon^{p_K},\qquad p_K\coloneqq\tfrac53(2-\gamma).
\end{aligned}
\ee
The norm in the first line measures tensor components relative to $G_{j,0}$; the second is restricted to the collar. Thus the $Q=0$ member is a local limiting model, and no second global comparison metric is needed. For $G=\widetilde G_\varepsilon$, the reference projection satisfies
\be
\label{AX-eq-kottler-projected-error}
\Pi_GS_3(G)=S_3(G),\qquad \|\Pi_GS_3(G)\|_Y\leq C\varepsilon^p.
\ee
The constants are uniform on the fixed parameter ranges.
\end{lemma}
\begin{proof}
\bse
We put $x\coloneqq R/r_{H,j}$ and take $Q=Q_{H,j}$. An exact factorization is
\be
\label{AX-eq-kottler-factorization}
F_j(r_{H,j}x)=(1-x^{-1})b_Q(x),\qquad b_Q(x)\coloneqq1-\frac Q3(x^2+x+1).
\ee
In the coordinates~\eqref{AX-eq-kottler-collar-identification}, the lifted metric is
\be
\label{AX-eq-kottler-regular-lift}
\frac{4r_{H,j}^2(1+y_j^2)}{b_Q(1+y_j^2)}dy_j^2 +r_{H,j}^2(1+y_j^2)^2d\Omega^2 +\frac{4r_{H,j}^2y_j^2b_Q(1+y_j^2)}{(1-Q)^2(1+y_j^2)}d\vartheta_j^2.
\ee
At $y_j=0$, the coefficients of $dy_j^2$ and $y_j^2d\vartheta_j^2$ agree. They are smooth in $y_j^2$, and their difference vanishes to second order. The metric is therefore smooth in $(y_j\cos\vartheta_j,y_j\sin\vartheta_j)$, also at the axial poles. Its $Q=0$ member has the same horizon radius and surface gravity. This family supplies the rescaled inner limit.

On an annulus $r_j\asymp t\geq m_j$, $R\asymp t$. The radial coefficient ratio between $G_{j,Q}$ and $G_{j,0}$ is $b_Q^{-1}$, and the coefficient of $dt_E^2$ ratio is $b_Q/(1-Q)^2$. Their differences from one, with any fixed number of scaled derivatives, are bounded by $C\Lambda t^2$, because $Qx^2=\Lambda R^2$. The regular angular formula proves the same estimate near the horizon with $t\asymp m_j$. This gives the first line of~\eqref{AX-eq-kottler-comparison-bounds}. Multiplication by $\omega^{-\gamma}$ and the bound $\omega\leq Cr_{\rm in}$ on the collar give $Cr_{\rm in}^{2-\gamma}=C\varepsilon^{p_K}$. These are local coefficient estimates, independent of a global tensor comparison.

Finally, the contracted Bianchi identity gives $\delta_GS_3(G)=0$. The definition~\eqref{AX-eq-projection} consequently gives $\Pi_GS_3(G)=S_3(G)$. The norm bound is Lemma~\ref{AX-lem-improved-reference-error}.
\ese
\end{proof}


\subsection{Axial fluxes of the approximation and the metric correction}
\label{AX-section-flux-proof}

We compute the leading axial fluxes from the intermediate fields and estimate the effect of the global correction. A closed-sphere identity cancels every linearized flux contribution, including those of spatial corrections that are not conformal. The numerical coefficients in the identities below are fixed by the curvature formulas and the chosen normalizations. In particular, the coefficient $8\pi$ is obtained by integrating the dipole term over a unit sphere. The annulus widths are the freely chosen constants $a_{\rm ann},b_{\rm ann}$ fixed above.

The leading axial flux therefore comes from the quadratic term of the approximation, while the correction contributes through its interaction with that approximation and through its own quadratic term. This is the reason that a correction of size $o(\varepsilon^2)$, rather than $o(\varepsilon^4)$, suffices on the separation scale.

We put $a_\varepsilon\coloneqq\varepsilon\ell$, $t\coloneqq\varepsilon^2$ and $x\coloneqq X/a_\varepsilon$. After absorbing the constant normalization of the auxiliary coordinate, we use a local auxiliary coordinate $\tau$ and define
\be
\label{AX-eq-flux-reference}
\begin{aligned}
&p_\pm^0\coloneqq\pm(D/2)e_z,\qquad U(x)\coloneqq|x-p_+^0|^{-1}+|x-p_-^0|^{-1},\\
&Q(x)\coloneqq|x|^2/2,\qquad G_t\coloneqq(1-t(U+Q))^2d\tau^2+(1+t(2U-Q))\delta.
\end{aligned}
\ee
Here $\delta$ is the Euclidean spatial metric. Around the sphere $\Sigma_D\coloneqq\{|x-p_+^0|=\vartheta D\}$, take the annular neighborhood $\mathcal A_D\coloneqq\{a_{\rm ann}\vartheta D<|x-p_+^0|<b_{\rm ann}\vartheta D\}$. The inequality $2b_{\rm ann}\vartheta<1$ and the compact positive range of $D$ keep these annuli uniformly separated from both centers after rescaling. All derivatives and norms in the following calculation are Euclidean, including the component in the auxiliary coordinate in the flat product.
\begin{proposition}[The leading axial flux and its change under a metric correction]
\label{AX-prop-flux-stability}
For a metric independent of $\tau$ $\widehat g=G_t+k$ on $\mathcal A_D$, with $\rho_k\coloneqq\|k\|_{C^2(\mathcal A_D)}$ sufficiently small, the vector flux satisfies
\be
\label{AX-eq-flux-stability}
\bigl|\int_{\Sigma_D}(S_{3t}(\widehat g))_{ij}n^i\,dA -8\pi t^2B(D)(e_z)_j\bigr| \leq C\bigl(t^3+t\rho_k+\rho_k^2\bigr).
\ee
The constant is uniform for $D\in I$ and sufficiently small $t\geq0$. The tensor and the flux have the conventions of~\eqref{AX-eq-projection} and~\eqref{AX-eq-spatial-flux}. In particular, if a metric $g$ has this form after the rescaling $X=a_\varepsilon x$ and $\rho_k\leq C\varepsilon^{2+\delta_{\rm fl}}$ for a fixed $0<\delta_{\rm fl}\leq2$, one has
\be
\label{AX-eq-physical-flux-proved}
\bigl|\operatorname{Flux}_+(S_\Lambda(g))-8\pi\varepsilon^4B(D)\bigr| \leq C\varepsilon^{4+\delta_{\rm fl}}.
\ee
Thus the axial flux expansion follows from the reference formula and the correction estimate on $\mathcal A_D$.
\end{proposition}
\begin{proof}
\bse
\emph{1. Closed-sphere cancellation.} We first establish the closed-sphere cancellation. A prime denotes a metric variation with the cosmological coefficient held fixed; $S'_0$ denotes this variation at the flat metric and zero cosmological coefficient. Greek indices range over the four flat product coordinates, and $G_0=d\tau^2+\delta$. The first variation of the Christoffel coefficients is $\dot\Gamma^\gamma_{\alpha\beta}=\tfrac12(\del_\alpha k_{\beta\gamma}+\del_\beta k_{\alpha\gamma}-\del_\gamma k_{\alpha\beta})$. Taking its two contractions gives
\be
\label{AX-eq-flat-linear-tensor}
\begin{aligned}
(\Ric'_0[k])_{\alpha\beta}
&=\tfrac12\bigl(\del_\gamma\del_\alpha k_{\beta\gamma}
+\del_\gamma\del_\beta k_{\alpha\gamma}
-\Delta_4 k_{\alpha\beta}
-\del_\alpha\del_\beta\operatorname{tr}_0 k\bigr),\\
\Scal'_0[k]&=\del_\alpha\del_\beta k_{\alpha\beta}
-\Delta_4\operatorname{tr}_0 k,\qquad
S'_0[k]=\Ric'_0[k]-\tfrac12\Scal'_0[k]G_0.
\end{aligned}
\ee
Here $\Delta_4$ and $\operatorname{tr}_0$ are the flat four-dimensional Laplacian and trace. Commuting derivatives in the first line proves $\del_\alpha(\Ric'_0[k])_{\alpha\beta}=\tfrac12\del_\beta\Scal'_0[k]$, so $\del_\alpha(S'_0[k])_{\alpha\beta}=0$. Independence of $\tau$ removes the derivative in $\tau$ from this identity. We multiply $k$ by a spatial cutoff supported in $\mathcal A_D$ and equal to one near $\Sigma_D$, then extend its components by zero through the enclosed ball. The resulting tensor is $C^2$ on that ball and agrees with $k$ near the sphere. For a smooth tensor the divergence theorem gives the identity below. For a $C^2$ tensor, smooth approximation converges in $C^2$ near the sphere and proves the same identity:
\be
\label{AX-eq-flat-zero-flux}
\int_{\Sigma_D}(S'_0[k])_{ij}n^i\,dA=0.
\ee
Only the tensor in the linear calculation is extended. The black-hole metric need not extend through a horizon for this argument.
\ese

\emph{2. The quadratic flux of the reference.} The explicit calculation in Appendix~\ref{AX-appendix-quadratic-flux} gives
$\int_{\Sigma_D}(S_{3t}(G_t))_{ij}n^i\,dA=8\pi t^2B(D)(e_z)_j+\Ocal_I(t^3)$.
The closed-sphere cancellation removes the Hessian part of the quadratic tensor. The remaining dipole term integrates to $8\pi B(D)e_z$.

\bse
\emph{3. Stability under the nonlinear correction.} We finally transfer this calculation to a finite correction. The Einstein tensor is a smooth function of the metric, its inverse, and its first two derivatives on uniformly positive metrics. On $\mathcal A_D$, the background differs from $G_0$ in $C^2$ by $\Ocal_I(t)$. Its differentiated tensor map therefore satisfies
\be
\label{AX-eq-flux-finite-increment}
\begin{aligned}
&S_{3t}(G_t+k)-S_{3t}(G_t)-S'_0[k] =\int_0^1\bigl(S'_{3t,G_t+uk}-S'_{0,G_0}\bigr)[k]\,du,\\
&\bigl\|S_{3t}(G_t+k)-S_{3t}(G_t)-S'_0[k]\bigr\|_{C^0} \leq C(t\rho_k+\rho_k^2).
\end{aligned}
\ee
For sufficiently small $t,\rho_k$ the path remains positive, and uniform bounds on inverse-metric coefficients and their first two metric derivatives prove the last inequality directly from the local Ricci formula; the cosmological variation $3tk$ is included in the first product. Integration, followed by~\eqref{AX-eq-flat-zero-flux}, proves~\eqref{AX-eq-flux-stability}.

To return to the original coordinates, the spatial map is $X=a_\varepsilon x$ and the map in the auxiliary coordinate includes the same length factor after normalization. Denoting this coordinate map by $\Phi$, we have $\widehat g=a_\varepsilon^{-2}\Phi^*g$. Ricci curvature as a covariant tensor is unchanged by constant metric multiplication and $\Lambda a_\varepsilon^2=3t$. Thus we obtain
\be
\label{AX-eq-flux-scaling}
S_{3t}(\widehat g)=\Phi^*S_\Lambda(g),\qquad (S_{3t}(\widehat g))_{ij}=a_\varepsilon^2(S_\Lambda(g))_{ij}\circ\Phi,\qquad dA_X=a_\varepsilon^2dA_x.
\ee
The tensor and area factors cancel in the spatial flux, while the Euclidean normal is unchanged. If $\rho_k\leq C\varepsilon^{2+\delta_{\rm fl}}$, the three errors are bounded by $C(\varepsilon^6+\varepsilon^{4+\delta_{\rm fl}}+\varepsilon^{4+2\delta_{\rm fl}})\leq C_1\varepsilon^{4+\delta_{\rm fl}}$. This proves~\eqref{AX-eq-physical-flux-proved}. In particular, a weighted correction of norm $\Ocal(\varepsilon^\nu)$, with $2-\gamma<\nu\leq2-2\gamma$, gives $\delta_{\rm fl}=\nu+\gamma-2>0$ by~\eqref{AX-eq-correction-comparison}.
\ese
\end{proof}


\subsection{From the reduced Einstein equation to the Einstein equation}
\label{AX-section-Einstein-recovery}

We show that a projected solution satisfies the Einstein equation when the axial obstructions vanish. The contracted Bianchi identity reduces this question to an energy estimate for the auxiliary vector field and a bound for its flux contribution. We use the normalization $B_GV=\tfrac12\mathcal L_VG$. The factors $1/2$ in the identities below follow from this definition. The estimates involve constants depending on the uniform coefficient bounds and on the coercivity constant $C_K$; smallness of the metric correction is chosen relative to these constants.

The reduced Einstein equation still permits a term $B_GV$. Here the reduction consists in applying the projection $\Pi_G$ onto divergence-free tensors and retaining the two obstruction terms; a solution of this equation need not yet satisfy the full Einstein equation. The uniform vector estimate required in the proposition is proved in Section~\ref{AX-section-uniform-vector}. Integrated brackets denote the $L^2$ tensor pairing over one fixed integration period in $t_E$, with the indicated metric. The coefficient regularity and parity conditions are those of Section~\ref{AX-section-boundaries}.
\begin{proposition}[The Bianchi identity and removal of the auxiliary vector term]
\label{AX-lem-Einstein-recovery}
Fix a regular static auxiliary metric $g=G+q_g$, a scalar $\lambda\in\mathbb R$, a regular invariant symmetric tensor $\psi$ and a regular invariant spatial vector field orthogonal to the axial symmetry orbits $V$ satisfying
\be
\label{AX-eq-recovery-decomposition}
S_\Lambda(g)=\lambda\psi+B_GV,\qquad \delta_G\psi=0.
\ee
Assume also that $V$ obeys
\be
\label{AX-eq-recovery-vector-bound}
\|\nabla^GV\|_{L^2(G)}+\|\omega^{-1}V\|_{L^2(G)} \leq C_K\|B_GV\|_{L^2(G)}.
\ee
If $\eta=\|q_g\|_{C^{2,\alpha}_{0;\omega}}$ is sufficiently small in terms of $C_K$, one has
\be
\label{AX-eq-recovery-conclusion}
\|B_GV\|_{L^2(G)}\leq C|\lambda|\eta\|\psi\|_{L^2(G)}.
\ee
In particular, $\lambda=0$ implies $S_\Lambda(g)=0$. The constants are uniform whenever the vector constant and the coefficient comparisons are uniform. Assume, in addition, that the reference metric and its inverse have uniform scaled coefficient bounds on an enlarged interior annulus $\mathcal A_d$ containing the flux sphere, that $N_0$ is uniformly bounded above and below there, and that
\be
\label{AX-eq-recovery-local-test}
\|\psi\|_{L^2(G)}\leq C(P_E/m)^{1/2},\qquad d^2\|\psi\|_{C^{1,\alpha}_d(\mathcal A_d)}\leq C.
\ee
With $\eta_{\rm sep}\coloneqq\|q_g\|_{C^{2,\alpha}_d(\mathcal A_d)}$, one then has
\be
\label{AX-eq-recovery-local-flux}
|\operatorname{Flux}_+(B_GV)| \leq C|\lambda|\Bigl(\eta\sqrt{d/m}+\eta_{\rm sep}\Bigr).
\ee
\end{proposition}
\begin{proof}
\bse
\emph{1. Global energy estimate.} Contracting the second Bianchi identity gives $\delta_gS_\Lambda(g)=0$. We excise small horizon and axis tubes before integrating against $V$. For fixed positive mass, all fields have bounded components in the regular local angular lifts. The tube area contains its vanishing polar radius, while the vector and tensor remain bounded. The boundary integrals tend to zero; the joint pole conditions give the same conclusion at the intersections of tubes. The boundary limits therefore justify the integrated identity $\langle S_\Lambda(g),B_gV\rangle_g=0$. Since $B_gV=B_GV+\tfrac12\mathcal L_Vq_g$, the decomposition gives
\be
\label{AX-eq-bianchi-pairing}
\|B_GV\|_{L^2(g)}^2 =-\bigl\langle B_GV,\tfrac12\mathcal L_Vq_g\bigr\rangle_g -\lambda\langle\psi,B_gV\rangle_g.
\ee
The weighted coefficient bound implies $|q_g|_G+\omega|\nabla^Gq_g|_G\leq C\eta$. The formula for a tensor Lie derivative consequently gives
\be
\label{AX-eq-recovery-Lie-estimate}
\bigl|\tfrac12\mathcal L_Vq_g\bigr|_G \leq C\eta\bigl(|\nabla^GV|_G+\omega^{-1}|V|_G\bigr),\qquad \bigl\|\tfrac12\mathcal L_Vq_g\bigr\|_{L^2(G)} \leq C C_K\eta\|B_GV\|_{L^2(G)}.
\ee
Small $\eta$ also compares the two inverse metrics, volume measures and tensor pairings, with relative error at most $C\eta$. The identity $\delta_G\psi=0$ and the same boundary argument give $\langle\psi,B_GV\rangle_G=0$. We subtract this zero pairing from $\langle\psi,B_gV\rangle_g$ and use the Lie estimate to obtain
\be
\label{AX-eq-recovery-absorption}
\begin{aligned}
&|\langle\psi,B_gV\rangle_g| \leq C(1+C_K)\eta\|\psi\|_{L^2(G)}\|B_GV\|_{L^2(G)},\\
&\|B_GV\|_{L^2(G)}^2 \leq C(1+C_K)\eta\|B_GV\|_{L^2(G)}^2\qquad \hspace{9mm}+C(1+C_K)|\lambda|\eta\|\psi\|_{L^2(G)}\|B_GV\|_{L^2(G)}.
\end{aligned}
\ee
Choose $\eta$ small enough to absorb the first term on the right into the left-hand side. If the remaining norm is positive, divide by it to obtain~\eqref{AX-eq-recovery-conclusion}. If it is zero, the same estimate holds without division. At $\lambda=0$, the decomposition therefore gives $S_\Lambda(g)=0$.
\ese

\bse
\emph{2. Interior vector estimate and sphere flux.} To estimate the sphere flux, we use the vector equation on the interior annulus. Applying $\delta_g$ to the decomposition and using the contracted Bianchi identity and $\delta_G\psi=0$ gives
\be
\label{AX-eq-recovery-local-vector-equation}
\delta_g B_GV=-\lambda(\delta_g-\delta_G)\psi,\qquad \|\lambda(\delta_g-\delta_G)\psi\|_{C^{0,\alpha}_d(\mathcal A_d)} \leq C|\lambda|\eta_{\rm sep}d^{-3}.
\ee
The last estimate follows by writing the two divergences as inverse-metric contractions of a first tensor derivative and connection terms. Their coefficient difference is bounded by $C\eta_{\rm sep}$, and their connection difference by $C\eta_{\rm sep}/d$. The bound for $\psi$ in~\eqref{AX-eq-recovery-local-test} controls both contributions, including the H\"older seminorm.

We shrink the annulus by a fixed proportion and denote the resulting neighborhood of the sphere by $\mathcal A'_d$. For a spatial covector $\xi$, the principal symbol at the reference metric is
\be
\label{AX-eq-recovery-reference-vector-symbol}
\sigma_2(\delta_G B_G)(\xi)V =\tfrac12\bigl(|\xi|_G^2V^{\flat_G}+\xi\,\xi(V)\bigr), \qquad \bigl\langle\sigma_2(\delta_G B_G)(\xi)V,V\bigr\rangle =\tfrac12\bigl(|\xi|_G^2|V|_G^2+|\xi(V)|^2\bigr).
\ee
Here the last pairing is the natural pairing between a one-form and a vector. The uniform bounds for the reference metric and its inverse give a positive lower bound for this symbol in the rescaled spatial coordinates. The coefficient difference between $\delta_gB_G$ and $\delta_GB_G$ is $\Ocal(\eta_{\rm sep})$. For sufficiently small $\eta_{\rm sep}$, rescaling $X=p_++dx$ consequently gives a uniformly strongly elliptic vector system on fixed nested annuli. Its interior estimate from $L^2$ to $C^{1,\alpha}$, followed by restoring the spatial volume and integration-length factors, gives
\be
\label{AX-eq-recovery-local-regularity}
\begin{aligned}
\|V\|_{C^{1,\alpha}_d(\mathcal A'_d)}
&\leq C\Bigl((P_Ed^3)^{-1/2}\|V\|_{L^2(G;\mathcal A_d)}
+d^2\|\delta_gB_GV\|_{C^{0,\alpha}_d(\mathcal A_d)}\Bigr)\\
&\leq C\Bigl(P_E^{-1/2}d^{-1/2}\|B_GV\|_{L^2(G)}
+|\lambda|\eta_{\rm sep}d^{-1}\Bigr).
\end{aligned}
\ee
Here $\omega\asymp d$ on $\mathcal A_d$ and~\eqref{AX-eq-recovery-vector-bound} give $\|V\|_{L^2(G;\mathcal A_d)}\leq Cd\|B_GV\|_{L^2(G)}$. For an invariant field, the volume factor is $P_Ed^3$: the integral in $t_E$ contributes $P_E$, the lapse contributes a uniformly bounded positive factor, and rescaling the spatial variables contributes $d^3$. The interior elliptic estimate uses only this annulus and requires no trace at a horizon.

In the same reference coordinates, $|B_GV|\leq Cd^{-1}\|V\|_{C^{1,\alpha}_d(\mathcal A'_d)}$ on the sphere, whose Euclidean area is of order $d^2$. We therefore obtain
\be
\label{AX-eq-recovery-flux-from-energy}
|\operatorname{Flux}_+(B_GV)| \leq C\Bigl(P_E^{-1/2}d^{1/2}\|B_GV\|_{L^2(G)} +|\lambda|\eta_{\rm sep}\Bigr).
\ee
Substitution of~\eqref{AX-eq-recovery-conclusion} and~\eqref{AX-eq-recovery-local-test} proves~\eqref{AX-eq-recovery-local-flux}.
\ese
\end{proof}
Proposition~\ref{AX-prop-uniform-vector} proves the global vector estimate for an explicit regular reference and sufficiently small scaled $C^1$ perturbations of it. Proposition~\ref{AX-prop-global-force-columns} constructs the obstruction tensors globally. The factor $\sqrt{d/m}$ in~\eqref{AX-eq-recovery-local-flux} is combined with the bounds for the obstruction tensors in the proof of Theorem~\ref{AX-thm-binary-existence}.


\section{\texorpdfstring{Kernels of the limiting linearized operators (linear analysis)}{Kernels of the limiting linearized operators (linear analysis)}}
\label{AX-section-two-masses}


\subsection{Barycenter, separation and the leading compatibility equations}
\label{AX-section-two-mass-geometry}

\paragraph{The barycenter and the separation}
We identify the kernels of the local limiting linearized operators and the two axial obstructions in the global linearized equation. Weighted estimates on the inner, intermediate and outer scales then control the differential tensor operator and its projection. We introduce the barycenter and separation for two unequal positive reference masses. Their leading equilibrium equations reduce to a barycenter condition and a scalar equation involving the total mass. The reflection-symmetric problem has one scalar separation parameter. With two different positive masses, the location of the pair relative to the cosmological center must also be selected. Write $m_+,m_->0$ for the reference masses. Their total mass and mass fraction are $\mathfrak m\coloneqq m_++m_-$ and $\theta\coloneqq m_+/\mathfrak m\in(0,1)$. For centers $p_\pm=z_\pm e_z$ satisfying $z_+>z_-$, define the separation and mass barycenter by
\be
\label{AX-eq-two-mass-coordinates}
d\coloneqq z_+-z_-,\qquad z_{\rm bar}\coloneqq\theta z_++(1-\theta)z_-,\qquad z_+=z_{\rm bar}+(1-\theta)d,\qquad z_-=z_{\rm bar}-\theta d.
\ee
The leading lapse is $N_{\rm lead}(X)\coloneqq1-m_+/|X-p_+|-m_-/|X-p_-|-|X|^2/(2\ell^2)$. After omitting the singular summand associated with each center, the negative axial derivatives of $N_{\rm lead}$ at $p_\pm$ are
\be
\label{AX-eq-two-mass-accelerations}
a_+\coloneqq\frac{z_+}{\ell^2}-\frac{m_-}{d^2},\qquad a_-\coloneqq\frac{z_-}{\ell^2}+\frac{m_+}{d^2},\qquad \theta a_++(1-\theta)a_-=\frac{z_{\rm bar}}{\ell^2},\qquad a_+-a_-=\frac d{\ell^2}-\frac{\mathfrak m}{d^2}.
\ee
The weighted identity follows because the two interaction terms cancel; the difference adds their masses. The weighted mean and difference of $(a_+,a_-)$ determine the pair uniquely for each $0<\theta<1$. The two leading balance equations are thus equivalent to $z_{\rm bar}=0$ and $d^3=\mathfrak m\ell^2$. These are reference-coordinate relations. They do not prescribe a proper distance or prove existence of the nonlinear static spacetime.

We recall the total-mass parameter and define the associated dimensionless coordinates by
\be
\label{AX-eq-total-mass-scale}
\varepsilon_{\mathrm{tot}}=(\mathfrak m/\ell)^{1/3},\qquad x\coloneqq X/(\varepsilon_{\mathrm{tot}}\ell),\qquad D_{\mathrm{tot}}\coloneqq d/(\varepsilon_{\mathrm{tot}}\ell),\qquad C_{\mathrm{tot}}\coloneqq z_{\rm bar}/(\varepsilon_{\mathrm{tot}}\ell).
\ee
The subscript prevents confusion with the equal-mass variables of Section~\ref{AX-section-scales}. When the masses are equal, $\varepsilon_{\mathrm{tot}}=2^{1/3}\varepsilon$. Passing from one normalization to the other therefore leaves all exponents in the estimates unchanged.


\paragraph{The two leading equilibrium equations}
The extension of the scalar separation argument requires two parameters and two corrected obstructions. The leading axial obstructions, divided by $\varepsilon_{\mathrm{tot}}/\ell$, are $C_{\mathrm{tot}}+(1-\theta)(D_{\mathrm{tot}}-D_{\mathrm{tot}}^{-2})$ and $C_{\mathrm{tot}}-\theta(D_{\mathrm{tot}}-D_{\mathrm{tot}}^{-2})$. Their weighted mean and difference define the leading map
\be
\label{AX-eq-two-mass-balance-map}
F_0(C,D)\coloneqq(C,D-D^{-2}),\qquad
v_0\coloneqq(0,1),\qquad
J_0\coloneqq DF_0(v_0)=
\begin{pmatrix}1&0\\0&3\end{pmatrix}.
\ee
In the two-mass construction, replace the notation $(C_{\mathrm{tot}},D_{\mathrm{tot}})$ by $(C,D)$. For arbitrary parameters $(C,D)$ near $(0,1)$, we set $q_+(C,D)\coloneqq(C+(1-\theta)D)e_z$, $q_-(C,D)\coloneqq(C-\theta D)e_z$ and $U_{\theta,C,D}(x)\coloneqq\theta/|x-q_+(C,D)|+(1-\theta)/|x-q_-(C,D)|$. For the corresponding intermediate metric, the dipole calculation~\eqref{AX-eq-flux-dipole} includes the coefficient of the singular term $\theta$ at the positive center and $1-\theta$ at the negative center. Consequently the leading spatial fluxes, divided by $8\pi\varepsilon_{\mathrm{tot}}^4$, are
\be
\label{AX-eq-two-mass-leading-fluxes}
\begin{aligned}
&f_+(C,D)=\theta\bigl(C+(1-\theta)(D-D^{-2})\bigr),\qquad f_-(C,D)=(1-\theta)\bigl(C-\theta(D-D^{-2})\bigr),\\
&f_++f_-=C,\qquad \frac{f_+}{\theta}-\frac{f_-}{1-\theta}=D-D^{-2}.
\end{aligned}
\ee
Indeed, near the positive center $\nabla U_{\theta,C,D}=-\theta n/r^2-(1-\theta)D^{-2}e_z+\Ocal(r)$, and the cosmological gradient is $(C+(1-\theta)D)e_z+rn$. The resulting flux is $8\pi\theta$ times the sum of the regular axial gradients at the positive center; exchanging the centers gives the second formula. The Hessian combination still has zero closed-sphere flux, by the proof of Proposition~\ref{AX-prop-flux-stability}. Thus the two combinations in~\eqref{AX-eq-two-mass-leading-fluxes} are the leading barycenter and separation equations.

\paragraph{Uniform axial flux estimates for unequal masses}
We extend Proposition~\ref{AX-prop-flux-stability} to the preceding family. We fix $0<\theta_0\leq1/2$, $c_{\rm bar}>0$, and a compact positive interval $I_{\rm fl}\coloneqq[D_{\rm lo},D_{\rm hi}]$, and allow $\theta\in[\theta_0,1-\theta_0]$, $|C|\leq c_{\rm bar}$ and $D\in I_{\rm fl}$. For $a\in\{+,-\}$, we define the sphere and its annular neighborhood by
\bse
\be
\label{AX-eq-two-mass-flux-annuli}
\Sigma_a\coloneqq\{|x-q_a|=\vartheta D\},\qquad \mathcal A_a\coloneqq\{a_{\rm ann}\vartheta D<|x-q_a|<b_{\rm ann}\vartheta D\}.
\ee
The normal $n_a$ points away from $q_a$. With $Q=|x|^2/2$ as in~\eqref{AX-eq-flux-reference}, the intermediate metric is
\be
\label{AX-eq-two-mass-flux-reference}
G_{t;\theta,C,D}\coloneqq(1-t(U_{\theta,C,D}+Q))^2d\tau^2 +(1+t(2U_{\theta,C,D}-Q))\delta.
\ee
For a symmetric tensor $k$ independent of $\tau$, we put $\rho_k\coloneqq\max_a\|k\|_{C^2(\mathcal A_a)}$, measuring components and derivatives in the flat product $d\tau^2+\delta$. There exist $t_{\rm fl},\rho_{\rm fl},C_{\rm stab}>0$, depending only on $\theta_0,c_{\rm bar},D_{\rm lo},D_{\rm hi},a_{\rm ann},b_{\rm ann}$ and $\vartheta$, such that, for $0\leq t\leq t_{\rm fl}$ and $\rho_k\leq\rho_{\rm fl}$, the metric $\widehat g\coloneqq G_{t;\theta,C,D}+k$ is positive on both annuli and its vector flux satisfies, for $a\in\{+,-\}$ and $j\in\{1,2,3\}$,
\be
\label{AX-eq-two-mass-uniform-flux-stability}
\Bigl|\int_{\Sigma_a}(S_{3t}(\widehat g))_{ij}n_a^i\,dA -8\pi t^2 f_a(C,D)(e_z)_j\Bigr| \leq C_{\rm stab}(t^3+t\rho_k+\rho_k^2).
\ee
\ese
Here $f_\pm$ are defined in~\eqref{AX-eq-two-mass-leading-fluxes}; the axial component of this vector flux is the functional used in~\eqref{AX-eq-spatial-flux}.

We justify the uniform constants directly. On $\mathcal A_a$, the distance to $q_a$ is at least $a_{\rm ann}\vartheta D_{\rm lo}$, the distance to the other center is at least $(1-b_{\rm ann}\vartheta)D_{\rm lo}$, and $|x|\leq c_{\rm bar}+(1+b_{\rm ann}\vartheta)D_{\rm hi}$. These inequalities uniformly bound $U_{\theta,C,D},Q$ and their first two derivatives, and give $G_{t;\theta,C,D}-(d\tau^2+\delta)=\Ocal(t)$ in $C^2$ on both annuli. Decreasing $t_{\rm fl},\rho_{\rm fl}$ if necessary makes every segment $G_{t;\theta,C,D}+u k$, $0\leq u\leq1$, uniformly positive. The reference expansion~\eqref{AX-eq-quadratic-flux-tensor} has a uniform $\Ocal(t^3)$ remainder, and its quadratic flux is $8\pi t^2 f_a(C,D)e_z$ by~\eqref{AX-eq-two-mass-leading-fluxes}. The closed-sphere identity~\eqref{AX-eq-flat-zero-flux} applies to each sphere without a reflection assumption. Finally, the finite-increment estimate~\eqref{AX-eq-flux-finite-increment} is uniform on these annuli, and the areas of the spheres are uniformly bounded. Integration gives~\eqref{AX-eq-two-mass-uniform-flux-stability}. The scaling identity~\eqref{AX-eq-flux-scaling} transfers it to the physical fluxes without an additional length factor.


\subsection{Obstruction tensors for the flat and de Sitter limiting equations}
\label{AX-section-two-mass-perturbation}

\paragraph{The two tensors obstructing solvability in the flat limit}
We identify the two axial point obstructions in the flat limit and their continuation in the de Sitter exterior. Their different far-field behavior determines the weighted vector estimates needed for the global tensor projection. We use Euclidean spatial vector fields and define $B_0V\coloneqq\tfrac12(\del_iV_j+\del_jV_i)$, $\delta_0T\coloneqq-\del_iT_{ij}$ and $A_0V\coloneqq\delta_0B_0V=-\tfrac12(\Delta V+\nabla\operatorname{div}V)$. The symmetric gradients of spatial vector fields are viewed as static four-tensors by adjoining zero component in the auxiliary coordinates. For a point $q$, $\delta_q^{\rm pt}$ denotes the Dirac measure at $q$. The matrix
\be
\label{AX-eq-flat-force-fundamental}
\Gamma_{ij}(x)\coloneqq\frac{1}{8\pi|x|} \Bigl(3\delta_{ij}+\frac{x_i x_j}{|x|^2}\Bigr)
\ee
will provide the two obstruction tensors and their exact distributional normalization.
\begin{lemma}[Two independent obstructions to the flat linearized Einstein equation]
\label{AX-lem-two-force-columns}
For distinct points $q_\pm$ and constant vectors $F_\pm$, the tensor
\be
\label{AX-eq-flat-force-columns}
\Psi\coloneqq B_0\bigl(\Gamma(\,\cdot-q_+)F_+ +\Gamma(\,\cdot-q_-)F_-\bigr)
\ee
is divergence free off the two points and satisfies
\be
\label{AX-eq-flat-force-source}
\delta_0\Psi=F_+\delta_{q_+}^{\rm pt}+F_-\delta_{q_-}^{\rm pt},\qquad \int_{\Sigma_\pm}\Psi_{ij}n^i\,dA=-(F_\pm)_j,
\ee
where each outward oriented sphere $\Sigma_\pm$ encloses its indicated point and excludes the other. If a static four-tensor $u$, smooth away from those points, satisfies $S'_0[u]=\Psi$ there and $|u|+r_\pm|\nabla u|\leq Cr_\pm^\gamma$ near them for some $\gamma>-1$, then $F_+=F_-=0$. In particular, the two independent axial point-obstruction tensors are linearly independent modulo the range of the flat linearized Einstein operator in this growth class.

For centers whose magnitudes are bounded by a fixed multiple of their separation $d$, the distant expansion is
\be
\label{AX-eq-flat-force-far-field}
\Psi(x)=B_0\bigl(\Gamma(x)(F_++F_-)\bigr) +\Ocal\!\bigl(d(|F_+|+|F_-|)|x|^{-3}\bigr).
\ee
The mode with nonzero total axial flux has a leading term of order $|x|^{-2}$ when $F_++F_-\neq0$, whereas a vanishing total axial flux removes that term.
\end{lemma}
\begin{proof}
\bse
The scalar functions $G_f(x)\coloneqq(4\pi|x|)^{-1}$ and $H_f(x)\coloneqq|x|/(8\pi)$ satisfy $\Delta G_f=-\delta_0^{\rm pt}$ and $\Delta H_f=G_f$ in the sense of distributions. The first identity follows by the outward flux of $\nabla G_f$ through a small sphere; the second follows from $\Delta|x|=2/|x|$ away from zero, with a vanishing inner flux. Direct differentiation gives $\Gamma_{ij}=2G_f\delta_{ij}-\del_i\del_jH_f$. Its divergence and Laplacian consequently satisfy
\be
\label{AX-eq-flat-force-Green-calculation}
\del_i\Gamma_{ij}=\del_jG_f,\qquad \Delta\Gamma_{ij}=2\Delta G_f\,\delta_{ij}-\del_i\del_jG_f,\qquad (A_0\Gamma)_{ij}=-\Delta G_f\,\delta_{ij}=\delta_0^{\rm pt}\delta_{ij}.
\ee
This proves the distributional divergence formula for $\Psi$. It also determines the sign of its flux. More explicitly, with $r=|x|$ and $n=x/r$, symmetrizing the derivative of $\Gamma F$ gives
\be
\label{AX-eq-flat-force-strain-explicit}
\bigl(B_0(\Gamma F)\bigr)_{ij} =\frac{1}{8\pi r^2}\bigl( -n_iF_j-n_jF_i+\delta_{ij}(n\cdot F)-3n_i n_j(n\cdot F)\bigr).
\ee
Contraction with $n^i$ gives $-(F_j+3n_j(n\cdot F))/(8\pi r^2)$. The identities $\int_{\mathbb S^2}1\,dA=4\pi$ and $\int_{\mathbb S^2}n_i n_j\,dA=(4\pi/3)\delta_{ij}$ give flux $-F_j$. The other center contributes zero through a sphere excluding it, by the divergence theorem. This proves~\eqref{AX-eq-flat-force-source}.

For the assertion about the range of the linearized Einstein operator, both $u$ and $\Psi$ are locally integrable near each point. In the second-order equation, integration by parts across a radius-$r$ excision leaves terms bounded by $C(r^{\gamma+1}+r^{\gamma+2})$ against a fixed smooth test tensor. They vanish because $\gamma>-1$. Thus $S'_0[u]=\Psi$ holds on all of $\mathbb R^3$ in the sense of distributions. The constant-coefficient linearized Bianchi identity proved from~\eqref{AX-eq-flat-linear-tensor} applies to distributions and gives
\be
\label{AX-eq-flat-force-independent-columns}
0=\delta_0S'_0[u]=\delta_0\Psi=F_+\delta_{q_+}^{\rm pt}+F_-\delta_{q_-}^{\rm pt}.
\ee
Tests supported near one point at a time show that both coefficients vanish. This proves linear independence modulo the range of $S'_0$, without imposing a coordinate condition on $u$.

Finally, $\Gamma$ is homogeneous of degree minus one and its first derivative of degree minus two. Taylor expansion at $x$, for $|x|$ larger than a fixed multiple of $d$, gives $B_0(\Gamma(x-q)F)=B_0(\Gamma(x)F)+\Ocal(d|F||x|^{-3})$. Addition proves~\eqref{AX-eq-flat-force-far-field}. The leading symmetric gradient is not zero for a nonzero total vector: on the ray parallel to that vector, formula~\eqref{AX-eq-flat-force-strain-explicit} has a nonzero radial component. Homogeneity gives decay of order $|x|^{-2}$ for this nonzero leading term.
\ese
\end{proof}

The barycenter mode in~\eqref{AX-eq-two-mass-leading-fluxes} has point-source coefficients proportional to $\theta e_z$ and $(1-\theta)e_z$, so their sum is nonzero. The separation mode has opposite point-source coefficients proportional to $\theta(1-\theta)e_z$, so their sum is zero. On the reflection-symmetric family, the latter is the permitted mode. Without midpoint reflection, the total axial flux may be nonzero, and the associated tensor has the slower decay $\Ocal(|x|^{-2})$. Its normalization and outer estimate must account for this decay. Thus the flat proof of independence of the obstruction tensors gives two tensors that are independent modulo the range of the flat linearized Einstein operator, while a global inverse estimate must still account for their two different decay rates. 

The difference is also visible in the flat weighted norm. In reference coordinates $X$, we denote by $\Psi_{\rm sum}$ the symmetric gradient in Lemma~\ref{AX-lem-two-force-columns} with point-source coefficients $(\theta e_z,(1-\theta)e_z)$ at $p_\pm$, and by $\Psi_{\rm sep}$ the symmetric gradient with point-source coefficients $(e_z,-e_z)$. We set $\omega_P(X)\coloneqq\min\{|X-p_+|,|X-p_-|\}$ and restrict the calculation to $\mathcal E\coloneqq\{\omega_P\geq m_{\rm cut},\ |X|\leq\ell\}$, with $0<m_{\rm cut}\ll d\ll\ell$. For $-1<\gamma<0$, we define the dimensionless model norm
\be
\label{AX-eq-two-force-model-norm}
\|\Psi\|^{\rm flat}_{\gamma-2}\coloneqq \sup_{\mathcal E}\ell^2(\omega_P/\ell)^{2-\gamma}|\Psi|_\delta.
\ee
For centers bounded by a fixed multiple of $d$, there are positive constants $c_0,K_0$, independent of the three lengths, such that
\be
\label{AX-eq-two-force-different-normalizations}
c_0\leq\|\Psi_{\rm sum}\|^{\rm flat}_{\gamma-2}\leq K_0,\qquad c_0(d/\ell)^{-\gamma}\leq\|\Psi_{\rm sep}\|^{\rm flat}_{\gamma-2} \leq K_0(d/\ell)^{-\gamma}.
\ee
We prove the bounds directly. Within a fixed multiple of $d$, each symmetric gradient is bounded by $K_0\omega_P^{-2}$, by~\eqref{AX-eq-flat-force-strain-explicit}. Multiplication by the weight gives an upper bound $K_0(d/\ell)^{-\gamma}$ there. Beyond that region, the sum of the axial fluxes gives $K_0|X|^{-2}$, while cancellation gives $K_0d|X|^{-3}$ for the separation field. Their weighted maxima are attained, respectively, at a radius comparable to $\ell$ and at a radius comparable to $d$. This proves both upper bounds. On a fixed proportional sphere about $p_+$, the separation flux is $-e_z$. The sphere area is of order $d^2$ and $\omega_P\asymp d$ there, so its largest symmetric gradient component is at least $c_0d^{-2}$; this gives the separation lower bound. For the sum field, we take $X=(\ell/2)e_z$. Its leading radial component is $-1/(2\pi|X|^2)$, and the error in~\eqref{AX-eq-flat-force-far-field} is $\Ocal(d|X|^{-3})$. For sufficiently small $d/\ell$, this component gives the other lower bound. The norm of the separation tensor tends to zero, whereas the norm of the tensor with nonzero total axial flux remains bounded below. Their normalizations in the weighted inverse estimate must therefore differ. Proposition~\ref{AX-prop-global-force-columns} below constructs the corresponding global obstruction tensors and proves their estimates.

\paragraph{A de Sitter tensor with nonzero total axial flux}
The sum of the axial fluxes has a regular continuation to the cosmological horizon on the outer model. We give its construction in the normalization $\ell=1$, so that $\Lambda=3$. In this paragraph we use the areal Cartesian coordinates $x=rn$, with $n\in\mathbb S^2$, and denote the unit round metric on $\mathbb S^2$ by $g_{\mathbb S^2}$. The de Sitter metric is
\be
\label{AX-eq-deSitter-outer-force-metric}
G_D=F\,dt_E^2+F^{-1}dr^2+r^2g_{\mathbb S^2}, \qquad F\coloneqq1-r^2,\qquad 0\leq r\leq1.
\ee
The auxiliary coordinate $t_E$ has period $2\pi$, giving the regular completion at $r=1$. At $r=0$ the auxiliary circle remains nonzero. All distributions below act on static fields per unit length in $t_E$; thus $\delta_0^{\rm pt}$ is a three-dimensional Dirac measure. We identify one-forms and vectors by $G_D$, and write $B_DV\coloneqq\tfrac12\mathcal L_VG_D$, $A_DV\coloneqq(\delta_DB_DV)^\sharp$ and $J_z\coloneqq e_z\delta_0^{\rm pt}$. The matrix $\Gamma$ is the flat fundamental solution in~\eqref{AX-eq-flat-force-fundamental}.

\begin{lemma}[Construction of a de Sitter tensor with prescribed axial flux]
\label{AX-lem-deSitter-outer-force}
There is a unique static spatial vector field orthogonal to the axial symmetry orbits $V_D$, regular in the angular lifts at the cosmological horizon and its poles, satisfying the following conditions:
\begin{itemize}
\item One has $A_DV_D=J_z$ in the sense of distributions, and, near $r=0$, for every integer $b\geq0$,
\be
\label{AX-eq-deSitter-outer-force-vector-asymptotic}
\bigl|D^b(V_D-\Gamma e_z)\bigr|\leq C_b r^{-b}.
\ee
Here $D$ denotes Cartesian differentiation of the spatial vector components.
\item The tensor $\Psi_D\coloneqq B_DV_D$ is regular at the cosmological horizon, is locally integrable at the central circle, and satisfies
\be
\label{AX-eq-deSitter-outer-force-tensor-asymptotic}
(\delta_D\Psi_D)^\sharp=J_z,\quad \bigl|D^b(\Psi_D-B_0(\Gamma e_z))\bigr|\leq C_b r^{-1-b},\quad \lim_{r\downarrow0}\int_{|x|=r}(\Psi_D)_{ij}n^i\,dA_0=-(e_z)_j.
\ee
In the pointwise bound, Cartesian components include the component in the auxiliary coordinate, and the flat symmetric gradient has zero component in the auxiliary coordinates.
\item The tensor $\Psi_D$ is not of the form $DS_3(G_D)[k]$ for a static tensor $k$ regular at the cosmological horizon and satisfying $|k|+r|Dk|\leq Cr^\gamma$ near the central circle for any $\gamma>-1$.
\end{itemize}
A prescribed total axial flux is represented by a scalar multiple of $\Psi_D$. The flux assertion is a limit at the central circle; it asserts no conservation of the Cartesian flux at finite radii in de Sitter.
\end{lemma}
\begin{proof}
\bse
\emph{1. The Kelvin field on the curved model.} We first calculate the error of the flat field $V_K\coloneqq\Gamma e_z$ on this curved model. In areal Cartesian coordinates,
\be
\label{AX-eq-deSitter-Cartesian-force-coefficients}
(h_D)_{ij}=\delta_{ij}+\frac{x_ix_j}{F}, \qquad N_D=\sqrt F, \qquad N_D\,d\mu_{h_D}=dx.
\ee
The determinant of $h_D$ is $1/F$, which proves the measure identity. For a spatial vector $V$ and $S\coloneqq B_{h_D}V$, the component in the auxiliary coordinate of $B_DV$ is $N_DV(N_D)$, and its mixed component in the auxiliary coordinates vanish. The warped connection therefore gives the spatial one-form identity
\be
\label{AX-eq-deSitter-force-warped-divergence}
\delta_DB_DV =\delta_{h_D}S-N_D^{-1}S(\nabla^{h_D}N_D,\cdot) +N_D^{-2}V(N_D)\,dN_D.
\ee
Indeed, the contraction in the auxiliary coordinate in the divergence is obtained from $\Gamma^i_{tt}=-N_D\nabla^iN_D$ and $\Gamma^t_{ti}=N_D^{-1}\del_iN_D$; it gives exactly the last two terms, while the spatial contraction gives $\delta_{h_D}S$.

We compute this field explicitly to fix both the singular order and the normalization. With $u\coloneqq\cos\theta$, its components are $8\pi V_K^r=4u/r$ and $8\pi V_K^u=3(1-u^2)/r^2$. We put $a_0\coloneqq4/F^2-4/(Fr^2)$ and $b_0\coloneqq(3r^2-1)/(Fr)$. Differentiation of the Lie derivative gives the nonzero components
\be
\label{AX-eq-deSitter-force-Kelvin-strain}
\begin{aligned}
8\pi(B_DV_K)_{tt}&=-4u,&
8\pi(B_DV_K)_{rr}&=u a_0,&
8\pi(B_DV_K)_{ru}&=b_0,\\
8\pi(B_DV_K)_{uu}&=\frac{u}{1-u^2},&
8\pi(B_DV_K)_{\varphi\varphi}&=u(1-u^2).
\end{aligned}
\ee
For example, the first component is $4uF'/(2r)=-4u$, and the mixed component is $\tfrac12(-6/r+4/(Fr))=b_0$. In the two angular components the radial derivative contributes $4ug_{\mathbb S^2}$, while the angular vector contributes $-3ug_{\mathbb S^2}$. These calculations also give the radial component by differentiating $F^{-1}$ and $4u/r$. Contraction in~\eqref{AX-eq-deSitter-force-warped-divergence} gives, with primes denoting radial derivatives,
\be
\label{AX-eq-deSitter-force-Kelvin-image-components}
\begin{aligned}
8\pi(A_DV_K)^r
&=-Fu\Bigl[Fa_0'+\Bigl(\frac{2F}{r}-3r\Bigr)a_0
-\frac{4r}{F^2}-\frac{2b_0}{r^2}-\frac2{r^3}\Bigr]
=-\frac{8u}{r},\\
8\pi(A_DV_K)^u
&=-\frac{1-u^2}{r^2}\Bigl[Fb_0'+\Bigl(\frac{2F}{r}-2r\Bigr)b_0+\frac1{r^2}\Bigr]
=-\frac{9(1-u^2)}{r^2}.
\end{aligned}
\ee
The first bracket simplifies to $8/(Fr)$ and the second to $9$, by substituting $F=1-r^2$ and the rational functions $a_0,b_0$. Returning to Cartesian vector components yields the exact formula
\be
\label{AX-eq-deSitter-force-Kelvin-image}
A_DV_K=\frac{(n\cdot e_z)n-9e_z}{8\pi r}, \qquad r>0.
\ee
All coordinate expressions extend across the axis because the initial Cartesian field is smooth there away from the origin.
\ese

\bse
\emph{2. The global correction and uniqueness.} Choose radii $0<\rho_{\rm cut}<R_{\rm cut}<1$ and a smooth radial cutoff $\chi$ that equals one on $r\leq\rho_{\rm cut}$ and zero on $r\geq R_{\rm cut}$. In the sense of distributions, we have
\be
\label{AX-eq-deSitter-force-cutoff-error}
A_D(\chi V_K)=J_z+f, \qquad f\in L^2(G_D), \qquad |D^bf|\leq C_b r^{-1-b}\quad(r<\rho_{\rm cut}).
\ee
The point coefficient is exactly that of $A_0V_K$. In the Cartesian expression for $A_D-A_0$, its second-order coefficients are $\Ocal(r^2)$, its first-order coefficients are $\Ocal(r)$, and its zero-order coefficients are bounded, by~\eqref{AX-eq-deSitter-Cartesian-force-coefficients}. Boundary terms for their action on $V_K=\Ocal(r^{-1})$, $DV_K=\Ocal(r^{-2})$ are therefore $\Ocal(r^2)$ and tend to zero. They produce no additional distribution at the origin. The ordinary function away from the origin is given by~\eqref{AX-eq-deSitter-force-Kelvin-image} where $\chi=1$, and is smooth on the cutoff annulus. Its square is integrable by the measure identity in~\eqref{AX-eq-deSitter-Cartesian-force-coefficients}.

We solve $A_DW=-f$ on the round four-sphere obtained by extending the de Sitter metric across the cosmological horizon: the map $(t_E,x)\mapsto(\sqrt{1-r^2}\cos t_E,\sqrt{1-r^2}\sin t_E,x)$ has induced metric $G_D$. For these spatial vector fields orthogonal to the axial symmetry orbits, the symmetric-gradient identity and compactness give
\be
\label{AX-eq-deSitter-outer-force-energy-inverse}
\|Y\|_{H^1(G_D)}\leq C\|B_DY\|_2, \qquad A_DW=-f, \qquad \|W\|_{H^1(G_D)}\leq C\|f\|_2.
\ee
We justify the removal of the lower norm in this fixed-metric estimate. The integrated identity gives $\|\nabla Y\|_2^2\leq2\|B_DY\|_2^2+3\|Y\|_2^2$. If the first inequality failed, compactness would give a nonzero Killing field orthogonal to the axial symmetry orbits. The pole argument in the proof of Proposition~\ref{AX-prop-uniform-vector} excludes such a field: its value and skew derivative vanish at a point fixed by the two rotations, and Killing prolongation makes it zero everywhere. The coercive form $\int|B_DY|^2$ now gives the unique weak solution $W$ in~\eqref{AX-eq-deSitter-outer-force-energy-inverse} by the same Hilbert-space argument used for the inverse of the operator on vector fields there.

Averaging local tests over the two rotations and the reflections in $t_E$ and $\varphi$ shows that this weak equation holds against arbitrary local vector fields. Near the center, staticity reduces the equation for $W$ to a uniformly elliptic three-dimensional system with smooth coefficients. The $L^2$ source gives local $H^2$ regularity, and the three-dimensional Sobolev embedding makes $W$ bounded. On balls of radius comparable to $r$ away from the center, the rescaled interior Schauder estimate and the bounds on $f$ then give $|D^bW|\leq C_b r^{-b}$ for every $b\geq0$. Near the cosmological horizon and its poles, $f$ is smooth in the regular four-dimensional lifts, so elliptic regularity makes $W$ smooth there. Consequently $V_D\coloneqq\chi V_K+W$ satisfies the first assertion. The connection difference from the flat product is $\Ocal(r)$ near the center. Differentiating the symmetric gradient therefore proves the second-line estimate in~\eqref{AX-eq-deSitter-outer-force-tensor-asymptotic}. This bound makes its sphere-flux error $\Ocal(r)$, and the exact flat flux~\eqref{AX-eq-flat-force-source} proves the flux limit. The distributional divergence follows directly from $A_DV_D=J_z$.

For uniqueness, the difference of two such fields is bounded, has derivative $\Ocal(r^{-1})$, and belongs to $H^1$ near the center. The distributional equation and the energy inequality in~\eqref{AX-eq-deSitter-outer-force-energy-inverse} make this difference zero. In particular, changing $\chi$ does not change $V_D$.
\ese

\emph{3. Exclusion from the Einstein image.} Finally, we suppose that $DS_3(G_D)[k]=\Psi_D$ off the central circle under the bound $|k|+r|Dk|\leq Cr^\gamma$, with $\gamma>-1$. Both tensors are locally integrable. The boundary terms in extending this second-order equation distributionally have largest order $r^{\gamma+1}$ and vanish for $\gamma>-1$. Since $S_3(G_D)=0$, differentiation of its contracted Bianchi identity gives the smooth-coefficient differential identity $\delta_DDS_3(G_D)=0$, also on distributions. It would follow that
\be
\label{AX-eq-deSitter-force-not-Einstein-image}
0=(\delta_DDS_3(G_D)[k])^\sharp =(\delta_D\Psi_D)^\sharp=J_z,
\ee
which is impossible. This proves the last assertion without a coordinate condition on $k$. Scalar multiplication gives any prescribed total axial flux.
\end{proof}

The field $\Psi_D$ describes the limit away from the shrinking horizons for a nonzero total axial flux. Section~\ref{AX-section-global-force-inverse} constructs the global inverse while keeping the two inner Kottler regions separate.

\paragraph{Dependence of the weighted vector inverse bound on the small masses}
The uniform energy estimate of Proposition~\ref{AX-prop-uniform-vector} cannot be replaced by a uniform inverse in every weighted H\"older norm. For the tensor source weight $\gamma-2$, the vector field in $B_GV$ naturally has weight $\beta\coloneqq\gamma-1$. The two-center point-source field gives a concrete obstruction in this weight.

\begin{proposition}[Growth of the weighted vector inverse constant as the masses vanish]
\label{AX-lem-vector-weight-loss}
For the admissible weight $\gamma$ of Table~\ref{AX-tab-admissible-exponents} and the equal-mass references of~\eqref{AX-eq-vector-reference-patching}, there are smooth reflection-invariant spatial vector fields orthogonal to the axial symmetry orbits $V_\varepsilon$, vanishing near every horizon, such that
\be
\label{AX-eq-vector-weight-loss}
\|V_\varepsilon\|_{k+2,\alpha,\beta;\omega}\geq c d^{-\gamma},\qquad \|A_{G_\varepsilon}V_\varepsilon\|_{k,\alpha,\beta-2;\omega}\leq C m^{-\gamma}.
\ee
Consequently any inverse estimate between these two weighted spaces has a constant at least $c(d/m)^{-\gamma}$, which is of order $\varepsilon^{2\gamma}$. This conclusion concerns the inverse of the operator on vector fields in weighted spaces, and is distinct from the tensor operator with two additional scalar unknowns in Theorem~\ref{AX-thm-physical-force-inverse}.
\end{proposition}
\begin{proof}
We use the matrix $\Gamma$ in Lemma~\ref{AX-lem-two-force-columns} and set
\be
\label{AX-eq-weighted-loss-example}
W(X)\coloneqq\bigl(\Gamma(X-p_+)-\Gamma(X-p_-)\bigr)e_z,\qquad V_\varepsilon\coloneqq\chi_f(|X|)\eta(r_+/m)\eta(r_-/m)W(X).
\ee
Here $\eta$ is smooth, zero on $[0,2]$ and one on $[4,\infty)$; the fixed radial cutoff $\chi_f$ is one on $[0,\rho_{\rm cut}]$ and zero on $[R_{\rm cut},\infty)$, using the preceding fixed radii. The field vanishes on neighborhoods of all three horizons and extends smoothly across the axis. The covariance of $\Gamma$ under orthogonal transformations gives the axial and midpoint reflection symmetries. At $X=2d e_z$, all cutoffs equal one for small $\varepsilon$. Since $\Gamma(te_z)e_z=e_z/(2\pi|t|)$, its value there is a fixed nonzero multiple of $d^{-1}e_z$. The metric is uniformly close to the flat product and $\omega\asymp d$ at that point. This proves the first inequality.

The estimate for $A_{G_\varepsilon}V_\varepsilon$ is verified in three regions in Appendix~\ref{AX-appendix-weighted-loss}. Dividing the two estimates gives the required lower bound on any inverse constant.
\end{proof}

The weighted estimates must therefore treat the point-source fields separately or retain their explicit mass dependence when estimating the projection and nonlinear terms. The energy estimate proves injectivity of the vector operator and permits removal of the auxiliary vector term when the obstruction coefficients vanish. It does not by itself give a uniform inverse in the weighted H\"older spaces. This proposition also explains why constructing the inverse of the global tensor operator requires an estimate joining the separate models: the two centers approach each other and produce a mode whose weighted norm is largest at distances comparable to their separation.


\subsection{Global estimates for the tensor equation in Bianchi coordinates}
\label{AX-section-global-differential-tensor}

\paragraph{A global tensor estimate with inner and outer remainder terms}
We prove a global differential tensor estimate by separating the inner collars, the intermediate region and the de Sitter exterior. The axial displacement mode in the kernel of the inner limiting operator is removed by one centering condition at each horizon. We first establish a global estimate which prevents a tensor from concentrating in the regions between the horizons. Its remaining terms are supported in fixed rescalings of the Kottler collars and in the outer region. This is an estimate for the differential tensor operator in a specified coordinate condition, before the nonlocal projection is imposed.

For this purpose, $G$ denotes the Kottler--de Sitter reference of Lemma~\ref{AX-lem-improved-reference-error}, with $\Lambda=3$. For a symmetric covariant tensor $v$, introduce the one-form $\mathfrak b_Gv\coloneqq\delta_Gv+\tfrac12d\operatorname{tr}_Gv$. Our convention for the rough Laplacian is $\nabla^*\nabla=-\operatorname{tr}_G\nabla^2$. The curvature action is $\mathcal R_G(v)(X,Y)\coloneqq\sum_{a,b}v(e_a,e_b)G(R(e_a,X)Y,e_b)$ in an orthonormal frame. We define
\be
\label{AX-eq-gauged-tensor-operator}
\mathcal P_Gv \coloneqq D(\Ric-3g)_G[v]+B_G(\mathfrak b_Gv)^\sharp =\tfrac12\bigl(\nabla^*\nabla v+\Ric_G\circ v+v\circ\Ric_G-2\mathcal R_G(v)\bigr)-3v.
\ee
The products with the Ricci tensor mean contraction of one index with the Ricci endomorphism, and $\sharp$ raises the one-form index using $G$. This identity follows by differentiating the connection covariantly: $\dot\Gamma^a_{bc}=\tfrac12G^{ad}(\nabla_bv_{cd}+\nabla_cv_{bd}-\nabla_dv_{bc})$. In $\dot\Ric_{bc}=\nabla_a\dot\Gamma^a_{bc}-\nabla_c\dot\Gamma^a_{ba}$, moving the two divergence derivatives next to $v$ produces $-B_G(\mathfrak b_Gv)^\sharp$. The commutators on its two covariant indices give $\tfrac12(\Ric_G\circ v+v\circ\Ric_G-2\mathcal R_G(v))$, and the remaining contracted second derivative is $\tfrac12\nabla^*\nabla v$. Differentiation of $-3g$ gives the last term in~\eqref{AX-eq-gauged-tensor-operator}.

\paragraph{Inverting the tensor operator on the de Sitter region}
The outer remainder in the global estimate can be removed. We first give the tensor calculation on the auxiliary de Sitter metric. In areal radius $R$, its metric is $G_D=dR^2/(1-R^2)+R^2d\Omega^2+(1-R^2)dt_E^2$. When $t_E$ has period $2\pi$, the map $(R,n,t_E)\mapsto(Rn,\sqrt{1-R^2}\cos t_E,\sqrt{1-R^2}\sin t_E)$ extends the metric regularly to the unit round sphere $\mathbb S^4\subset\mathbb R^5$. The point $R=0$ in the spatial quotient corresponds to a circle, denoted by $\mathcal Z$, in this completion.

\begin{lemma}[Invertibility on de Sitter and extension across the central circle]
\label{AX-lem-deSitter-tensor-inverse}
On the unit round $\mathbb S^4$, the operator $\mathcal P_{G_D}$ in~\eqref{AX-eq-gauged-tensor-operator} has a bounded inverse on symmetric tensors. For every smooth source $T$, its solution satisfies
\be
\label{AX-eq-deSitter-inverse}
\|v\|_{L^2}\leq\|T\|_{L^2},\qquad \|v\|_{C^{k+2,\alpha}}\leq C_{k,\alpha}\|T\|_{C^{k,\alpha}}, \qquad \mathcal P_{G_D}v=T.
\ee
The inverse preserves the static and axial symmetries. Moreover, if a tensor solves $\mathcal P_{G_D}v=0$ off $\mathcal Z$ and satisfies $|v|+R|\nabla v|\leq CR^\gamma$ near $\mathcal Z$, with $\gamma>-1$, then it extends smoothly across $\mathcal Z$ and is zero.
\end{lemma}
\begin{proof}
\bse
The round curvature gives $\Ric=3G_D$ and $\mathcal R(v)=(\operatorname{tr}v)G_D-v$. Thus the trace and trace-free parts are preserved, and
\be
\label{AX-eq-deSitter-operator-splitting}
\mathcal P_{G_D}v=\tfrac12\nabla^*\nabla v+v-(\operatorname{tr}v)G_D, \qquad \operatorname{tr}(\mathcal P_{G_D}v) =\bigl(\tfrac12\Delta_+-3\bigr)\operatorname{tr}v,
\ee
where $\Delta_+\coloneqq-\Delta_{G_D}$ on scalars. On trace-free tensors the quadratic form is $\tfrac12\|\nabla v\|_2^2+\|v\|_2^2$. This quadratic form is coercive. Minimize half this quadratic form minus the pairing with the trace-free part of $T$. The minimizer is the unique weak solution. Testing the equation against $v$ gives $\|v\|_2\leq\|T\|_2$ on the trace-free subspace.

We give the scalar calculation to check the possible trace obstruction. The restrictions of harmonic homogeneous polynomials of degree $l$ on $\mathbb R^5$ have $\Delta_+$ eigenvalue $l(l+3)$. Indeed, the polar formula for the Euclidean Laplacian is $\del_r^2+4r^{-1}\del_r+r^{-2}\Delta_{\mathbb S^4}$, and substitution of $r^lY$ proves the assertion. Every homogeneous polynomial is a sum of terms $r^{2j}H_{l-2j}$ with harmonic homogeneous $H_{l-2j}$. This follows by subtracting its traces recursively, using
\be
\label{AX-eq-sphere-harmonic-decomposition}
\Delta_{\mathbb R^5}(r^{2j}H_k) =2j(2k+2j+3)r^{2j-2}H_k.
\ee
All coefficients for $j\geq1$ are nonzero, so induction on the degree supplies the decomposition. Polynomial restrictions are dense in $L^2(\mathbb S^4)$. A continuous function on the sphere extends continuously to a surrounding cube: use radial projection, with a cutoff near the origin. The multivariate Bernstein polynomials of this extension converge uniformly. Uniform continuity controls the terms evaluated near the chosen point; the binomial variance bound makes the total weight of the other terms tend to zero. Restricting these polynomials to the sphere proves the density assertion. The decomposition~\eqref{AX-eq-sphere-harmonic-decomposition} therefore supplies a complete orthogonal family of scalar eigenspaces.

For every integer $l\geq0$, the trace eigenvalue satisfies
\be
\label{AX-eq-deSitter-trace-gap}
\bigl|\tfrac12l(l+3)-3\bigr|\geq1.
\ee
At $l=0,1,2$ the respective eigenvalues are $-3,-1,2$; from that point onward they increase. Expand $\operatorname{tr}T$ in these spaces and divide each component by its nonzero eigenvalue. This gives the unique trace solution, whose $L^2$ norm is at most $\|\operatorname{tr}T\|_2$. The resulting series also converges in $H^1$, because $l(l+3)$ divided by the square of the eigenvalue remains bounded. Orthogonality of the trace and trace-free tensor parts proves the first estimate in~\eqref{AX-eq-deSitter-inverse}. Local elliptic regularity gives a smooth solution for smooth data. The interior Schauder estimate on a finite cover of the sphere, together with this $L^2$ estimate and interpolation, gives the H\"older bound in~\eqref{AX-eq-deSitter-inverse}. Every isometry commutes with the operator; applying it to the equation and using the proved zero kernel shows that the inverse preserves its imposed symmetries.

Finally, a radius-$R$ tube about $\mathcal Z$ has boundary measure of order $R^2$. Integrating the equation against a fixed smooth test tensor on the complement leaves boundary terms bounded by $C(R^{\gamma+1}+R^{\gamma+2})$. They tend to zero for $\gamma>-1$, while the tensor is locally integrable. Hence the extended tensor solves the same homogeneous elliptic equation globally in the sense of distributions. Distributional elliptic regularity, obtained locally by the interior estimates for its scalar principal part, makes it smooth. Its trace and trace-free energy calculation above give $v=0$.
\ese
\end{proof}

\paragraph{A global tensor bound with a remainder on the inner collars}
We use the weighted norms of~\eqref{AX-eq-weighted-norm} on tensors, with covariant scaled derivatives and the reference regular angular lifts. For fixed $R>1$ and $r_0>0$, the two remainder terms, localized on the inner collars and in the outer region, are
\be
\label{AX-eq-tensor-model-remainders}
\mathcal C_\varepsilon(v)\coloneqq \max_{j=\pm}\;\sup_{m/2\leq r_j\leq2Rm}m^{-\gamma}|v|_G,\qquad \mathcal O_\varepsilon(v)\coloneqq \sup_{r\geq r_0/2}|v|_G.
\ee
For unequal comparable masses, each inner supremum and its weight use that collar's reference mass $m_j$. Tensor fields are smooth in the reference lifts, including the axes and horizons, and invariant under the two circle actions.

\begin{proposition}[A uniform tensor bound up to terms on the inner collars]
\label{AX-prop-global-tensor-remainders}
For fixed $-1<\gamma<0$, $k\geq0$ and $0<\alpha<1$, there exist $R$, $r_0$ and $C$, independent of sufficiently small $\varepsilon$, such that
\be
\label{AX-eq-global-tensor-remainder-estimate}
\|v\|_{k+2,\alpha,\gamma;\omega} \leq C\Bigl(\|\mathcal P_Gv\|_{k,\alpha,\gamma-2;\omega} +\mathcal C_\varepsilon(v)+\mathcal O_\varepsilon(v)\Bigr).
\ee
The outer remainder can in fact be removed, so that
\be
\label{AX-eq-global-tensor-core-estimate}
\|v\|_{k+2,\alpha,\gamma;\omega} \leq C\Bigl(\|\mathcal P_Gv\|_{k,\alpha,\gamma-2;\omega} +\mathcal C_\varepsilon(v)\Bigr).
\ee
These estimates are uniform in the compact separation interval and for comparable positive masses. The constant in these estimates remains bounded as $\varepsilon\to0$ on the regions between the Kottler collars and the outer region. The inner remainder $\mathcal C_\varepsilon(v)$ must be estimated under the inner orthogonality conditions before an inverse bound follows.
\end{proposition}
\begin{proof}
\bse
\emph{1. Coefficient bounds and a positive barrier.} We first work on the region $\mathcal A_\varepsilon\coloneqq\{r<r_0,\ r_\pm>Rm\}$. For small $\varepsilon$, its two inner boundary spheres are disjoint and lie inside the outer sphere. With $\rho=\min(r_+,r_-)$, the reference formulas and differentiated cutoffs give, for each fixed derivative order $j$,
\be
\label{AX-eq-neck-coefficient-smallness}
\begin{aligned}
\rho^j|D^j(G-G_0)|&\leq C_j(m/\rho+r^2+d^2),\\
\rho^2(|\operatorname{Rm}_G|_G+3)&\leq C(m/\rho+r^2+d^2)
\leq C(R^{-1}+r_0^2+d^2).
\end{aligned}
\ee
Here the flat product coordinates are regular because the lapse is bounded below on this region; the axis is treated in spatial Cartesian coordinates. On the exact Kottler collars, the factorization~\eqref{AX-eq-kottler-factorization} adds terms with scaled bounds $C\Lambda r_j^2$ to the inner limiting coefficients; differentiation gives the same inequalities. On the intermediate region, they follow by differentiating $m/r_++m/r_-$ and $r^2/2$. On either transition, a derivative of a cutoff has inverse transition scale, which is comparable to $\rho$ there. The constant adjustments on the inner regions are $\Ocal(d^2)$, and on the outer transition the mass contribution is $\Ocal(m/\rho)$. These observations prove the first line. The coordinate curvature formula contains second metric derivatives and quadratic first derivatives, so it gives the second line. The additional constant $3$ is controlled because $\rho\leq r+Cd$.

We use the positive scalar weight $w\coloneqq r_+^\gamma+r_-^\gamma$. Its size is comparable to $\rho^\gamma$. All its differentiated components satisfy the corresponding scaled bounds. The Euclidean identity $\Delta r^\gamma=\gamma(\gamma+1)r^{\gamma-2}$ gives
\be
\label{AX-eq-neck-positive-barrier}
-\Delta_0w=-\gamma(\gamma+1)(r_+^{\gamma-2}+r_-^{\gamma-2}) \geq c_\gamma\rho^{\gamma-2}.
\ee
For $-1<\gamma<0$, this constant is positive. For invariant scalars, $\Delta_G=\Delta_h+N^{-1}\langle\nabla N,\nabla\,\cdot\,\rangle_h$. The coefficient estimate therefore bounds $(\Delta_G-\Delta_0)w$ by $C(R^{-1}+r_0^2+d^2)\rho^{\gamma-2}$. Decompose the operator as $\mathcal P_G=\tfrac12\nabla^*\nabla+\mathcal V_G$. The endomorphism in~\eqref{AX-eq-gauged-tensor-operator} is bounded in norm by $C(|\operatorname{Rm}_G|+3)$. Choose $b_\gamma\in(0,c_\gamma)$ first. Next take $R$ large enough and $r_0$ and $\varepsilon$ small enough that the coefficient error is at most $c_\gamma-b_\gamma$ and the following inequality holds:
\be
\label{AX-eq-neck-barrier-potential}
\bigl(-\Delta_G-2|\mathcal V_G|\bigr)w \geq b_\gamma\rho^{\gamma-2}.
\ee
\ese

\bse
\emph{2. Weighted maximum estimate.} Where $v\neq0$, differentiate the norm and use $|d|v||\leq|\nabla v|$ to obtain $-\Delta_G|v|\leq\langle\nabla^*\nabla v,v\rangle/|v|$. The norm of the tensor consequently satisfies
\be
\label{AX-eq-neck-tensor-Kato}
\bigl(-\Delta_G-2|\mathcal V_G|\bigr)|v| \leq2|\mathcal P_Gv|.
\ee
Suppose the maximum of $|v|/w$ on $\overline{\mathcal A_\varepsilon}$ is positive and occurs in the interior. At such a point, $v\neq0$. The gradient of the quotient is zero at this maximum, and its Laplacian is nonpositive. Apply the product rule to $|v|=w(|v|/w)$ and use~\eqref{AX-eq-neck-barrier-potential} together with~\eqref{AX-eq-neck-tensor-Kato}. The maximum is then bounded by $C\|\mathcal P_Gv\|_{0,0,\gamma-2;\omega}$. A boundary maximum is bounded by the inner and outer suprema in~\eqref{AX-eq-tensor-model-remainders}. The zero-maximum case is immediate. Since $\omega\asymp\rho$ on this region and $w\asymp\rho^\gamma$, we obtain
\be
\label{AX-eq-neck-Czero-estimate}
\|v\|_{0,0,\gamma;\omega;\mathcal A_\varepsilon} \leq C\Bigl(\|\mathcal P_Gv\|_{0,0,\gamma-2;\omega} +\mathcal C_\varepsilon(v)+\mathcal O_\varepsilon(v)\Bigr).
\ee
Outside this region, the two localized remainder terms directly control the weighted $C^0$ norm: $\omega\asymp m_j$ on each fixed rescaled neighborhood $m_j/2\leq r_j\leq2Rm_j$, and $\omega$ is bounded below on the fixed outer region. These comparisons have constants depending on the chosen $R,r_0$, which remain fixed.
\ese

\emph{3. Derivative estimates and the outer limit.} We finish with the derivative estimates. The principal part of $\mathcal P_G$ is the scalar Laplacian on each tensor component. After scaling a coordinate ball by $\omega$, its leading coefficients are uniformly elliptic with bounded $C^{k,\alpha}$ norms, and its lower-order tensor coefficients are uniformly bounded. Near an inner horizon, this follows by scaling the exact Kottler collar by its reference mass and using its regular angular lift. The parameter $Q_j=\Lambda r_{H,j}^2$ tends to zero, so the scaled coefficient bounds are uniform. Near the outer horizon it follows from the fixed de Sitter lift. On the remaining balls it follows from the coefficient estimates and the cutoff formulas. There is no boundary in a horizon lift. The componentwise interior Schauder estimate, with bounded first-order coupling terms, gives
\be
\label{AX-eq-global-tensor-Schauder}
\|v\|_{k+2,\alpha,\gamma;\omega} \leq C\Bigl(\|\mathcal P_Gv\|_{k,\alpha,\gamma-2;\omega} +\|v\|_{0,0,\gamma;\omega}\Bigr).
\ee
For the coupling estimate, apply the scalar interior estimate to each component on nested balls. The first-order terms are bounded by $C\|v\|_{k+1,\alpha}$ and are then controlled by the interpolation inequality $\|v\|_{k+1,\alpha}\leq\eta\|v\|_{k+2,\alpha}+C_\eta\|v\|_0$. Choosing $\eta$ small absorbs these terms; the rescaled balls and a cover with bounded overlap give this estimate. These are the interior estimates and interpolation inequalities of~\cite[Chapter 6]{LS-GilbargTrudinger}. Combining it with the $C^0$ bound proves~\eqref{AX-eq-global-tensor-remainder-estimate}.

To remove the outer term, we argue by contradiction. If~\eqref{AX-eq-global-tensor-core-estimate} fails as $\varepsilon\to0$, there are $\varepsilon_i\to0$ and tensors $v_i$ with weighted norm one for which both terms on its right tend to zero. The estimate already proved forces $\mathcal O_{\varepsilon_i}(v_i)\geq c>0$. On every fixed region $r\geq r_*>0$, the reference is exactly de Sitter for large $i$. The local estimates and compactness therefore give, after a diagonal subsequence, a limit $v_0$ solving $\mathcal P_{G_D}v_0=0$ off the central circle. The uniform weighted bounds imply $|v_0|+r|\nabla v_0|\leq Cr^\gamma$ near that circle. Lemma~\ref{AX-lem-deSitter-tensor-inverse} gives $v_0=0$. Convergence is uniform on the compact outer region, including its regular horizon lifts, contradicting the lower bound on $\mathcal O_{\varepsilon_i}$. This proves~\eqref{AX-eq-global-tensor-core-estimate} for all sufficiently small $\varepsilon$.
\end{proof}

\paragraph{Orthogonality conditions on the two inner collars}
The inner remainder can be replaced by finitely many linear functionals expressing orthogonality to the kernel of the inner limiting operator. This gives a uniform global estimate for the differential operator. The decaying kernel of the inner limiting operator is classified below.

We denote by $G_S$ the $Q=0$ endpoint of the unit-mass rescaled Kottler family, completed at its horizon with auxiliary coordinate period $8\pi$, and by $\mathcal P_S\coloneqq\tfrac12(\nabla^*\nabla-2\mathcal R_{G_S})$ its zero-$\Lambda$ operator. In the isotropic coordinates $y$, its horizon is $|y|=1/2$. We use weight $\rho_S\coloneqq1+|y|$, regular angular lifts and the tensor symmetries of a single inner collar. The decaying kernel of $\mathcal P_S$ is
\be
\label{AX-eq-Schwarzschild-model-kernel}
\mathcal K_S\coloneqq \{w:\ \mathcal P_Sw=0,\quad \|w\|_{k+2,\alpha,\gamma;\rho_S}<+\infty,\quad w\text{ is regular at the horizon}\}.
\ee
For a sufficiently large fixed $R$, the compact set $K_S\coloneqq\{1/2\leq|y|\leq2R\}$ contains the horizon. The proof in Appendix~\ref{AX-appendix-model-proof-1} shows that $n_S\coloneqq\dim\mathcal K_S$ is finite and that the $L^2(K_S,G_S)$ pairing is positive definite on this kernel. We choose an orthonormal basis $w_1,\ldots,w_{n_S}$ for that pairing; if the kernel is zero, the following lists of functionals are empty.

On the $j$th exact Kottler collar, the coordinate map
\be
\label{AX-eq-core-rescaling-map}
\Phi_{j,\varepsilon}(y,\tau) \coloneqq\Bigl(p_j+m_jy,\frac{\sqrt{c_{h,j}}m_j}{c_{N,j}}\tau\Bigr), \qquad \Phi_{j,\varepsilon}^*G=c_{h,j}m_j^2\widehat G_{Q_j},\qquad Q_j\coloneqq\Lambda r_{H,j}^2
\ee
is understood in the regular local angular lift. The normalized metric $\widehat G_{Q_j}$ is smooth at the horizon and converges to $G_S$ in every smooth norm on each fixed compact subset as $Q_j\to0$. We put $\widehat v_j\coloneqq(c_{h,j}m_j^2)^{-1}\Phi_{j,\varepsilon}^*v$ and define
\be
\label{AX-eq-complete-core-centering}
\ell_{j,a,\varepsilon}(v)\coloneqq \int_{K_S}\langle m_j^{-\gamma}\widehat v_j,w_a\rangle_{G_S}\,d\mu_{G_S}, \qquad 1\leq a\leq n_S.
\ee
The factor $m_j^{-\gamma}$ is required by the weighted tensor norm. The inner orthogonality (or centering) conditions are $\ell_{j,a,\varepsilon}(v)=0$ for both inner collars and every indicated index. The proof is given in Appendix~\ref{AX-appendix-model-proof-1}.

\begin{proposition}[A uniform tensor estimate under the two orthogonality conditions]
\label{AX-prop-completely-centered-inverse-bound}
The space $\mathcal K_S$ is finite dimensional. With the centering functionals above, the global reference satisfies
\be
\label{AX-eq-completely-centered-global-bound}
\|v\|_{k+2,\alpha,\gamma;\omega} \leq C\Bigl(\|\mathcal P_Gv\|_{k,\alpha,\gamma-2;\omega} +\sum_{j=\pm}\sum_{a=1}^{n_S}|\ell_{j,a,\varepsilon}(v)|\Bigr)
\ee
for all sufficiently small $\varepsilon$, with $C$ independent of $\varepsilon$ and of the comparable mass ratio. On the domain defined by the inner orthogonality conditions, $\mathcal P_G$ is therefore injective, has closed range, and its inverse onto that range is uniformly bounded in these weighted norms.
\end{proposition}

\paragraph{The tensor obtained by displacing the Schwarzschild center}
The inner orthogonality conditions are necessary even on a single Schwarzschild model. In areal coordinates, with $F(R)\coloneqq1-2/R$, we set $\varphi(R,\theta)\coloneqq(R-1)\cos\theta$. The scalar Laplacian on invariant fields gives
\be
\label{AX-eq-Schwarzschild-displacement-harmonic}
\Delta_{G_S}\varphi =\frac1{R^2}\del_R\bigl(R^2F\del_R\varphi\bigr) +\frac1{R^2}\Delta_{\mathbb S^2}\varphi =\frac{(2R-2)-2(R-1)}{R^2}\cos\theta=0.
\ee
The function is smooth at the horizon: $R-2$ is a smooth even function of the normal distance, and $\cos\theta$ is smooth on its sphere. Its gradient $V\coloneqq\nabla\varphi$ approaches the axial translation at infinity. The tensor $w_{\rm disp}\coloneqq\mathcal L_VG_S=2\Hess_{G_S}\varphi$ is regular and has size $\Ocal(R^{-2})$ with its scaled derivatives, because $\varphi=R\cos\theta-\cos\theta$ and the Schwarzschild coefficients have differentiated asymptotic expansions in powers of $R^{-1}$. It is nonzero on the horizon sphere, where its tangential part is $-2\cos\theta\,d\Omega^2$. The Ricci variation along this Lie derivative is $\mathcal L_V\Ric_{G_S}=0$. Moreover, contraction of the covariant derivatives of the Hessian gives
\be
\label{AX-eq-Schwarzschild-displacement-gauge}
\operatorname{tr}w_{\rm disp}=2\Delta\varphi=0, \qquad \mathfrak b_{G_S}w_{\rm disp} =-d\Delta\varphi-2\Ric_{G_S}(\nabla\varphi,\,\cdot\,)=0.
\ee
The gauged variation identity then gives $\mathcal P_Sw_{\rm disp}=0$, and its decay places it in $\mathcal K_S$. The following calculation proves that this mode exhausts the kernel in the prescribed symmetry class. The proof is given in Appendix~\ref{AX-appendix-model-proof-2}.

\begin{proposition}[The Schwarzschild displacement tensor spans the decaying kernel]
\label{AX-prop-Schwarzschild-kernel-classification}
For the metric $G_S=F\,d\tau^2+F^{-1}dR^2+R^2g_{\mathbb S^2}$, where $F\coloneqq1-2/R$, $R\geq2$, $g_{\mathbb S^2}$ is the unit round metric, and $\tau$ has period $8\pi$, one has the following classification. In the class of static, axisymmetric tensors with orthogonal symmetry orbits, regular at the horizon, the kernel in~\eqref{AX-eq-Schwarzschild-model-kernel} is
\be
\label{AX-eq-Schwarzschild-full-kernel}
\mathcal K_S=\operatorname{span}\{w_{\rm disp}\}, \qquad w_{\rm disp}=2\Hess_{G_S}((R-1)\cos\theta).
\ee
This assertion holds for every weight $-1<\gamma<0$. With this auxiliary period, $\|w_{\rm disp}\|_{L^2(G_S)}^2=128\pi^2$.
\end{proposition}

The proof is given in Appendix~\ref{AX-appendix-model-proof-2}. Orthogonality to the Schwarzschild displacement mode therefore imposes one condition on each inner collar. Under midpoint reflection the two conditions are paired, so only one independent condition remains. The estimate~\eqref{AX-eq-completely-centered-global-bound} controls the differential tensor operator on the domain satisfying these orthogonality conditions; the next section transfers it to the reduced Einstein equation.


\subsection{Estimates for the reduced Einstein equation}
\label{AX-section-projected-tensor-estimates}

\paragraph{A bound for the weighted vector inverse in terms of the masses}
We estimate the nonlocal projection defined by the inverse of the operator on vector fields and apply it to the global tensor equation. The resulting linear and quadratic estimates retain the powers of the mass parameter needed for the nonlinear contraction. The energy estimate for the auxiliary vector field also gives a quantitative weighted upper bound. This bound is not sharp, but its power is sufficient for the linearized Bianchi argument. We use the admissible weight $\gamma$, $\beta\coloneqq\gamma-1$, and $m\asymp\varepsilon^3$ in the normalization $\ell=1$. All fields and sources satisfy the regularity, symmetry and orthogonality conditions used to define $A_G$.

\begin{lemma}[A weighted vector inverse bound with explicit mass dependence]
\label{AX-lem-vector-weight-upper}
For the Kottler--de Sitter reference, the inverse supplied by the energy estimate for the auxiliary vector field satisfies
\be
\label{AX-eq-vector-weight-upper}
\|A_G^{-1}f\|_{k+2,\alpha,\beta;\omega} \leq C m^{\beta+1/2}\|f\|_{k,\alpha,\beta-2;\omega} \leq C\varepsilon^{3\gamma-3/2}\|f\|_{k,\alpha,\beta-2;\omega}.
\ee
Consequently, for $k\geq1$, the divergence-free projection obeys
\be
\label{AX-eq-weighted-projection-upper}
\|\Pi_GT\|_{k,\alpha,\gamma-2;\omega} \leq C\varepsilon^{3\gamma-3/2}\|T\|_{k,\alpha,\gamma-2;\omega}.
\ee
These estimates give a sufficient power of the mass parameter for the argument below.
\end{lemma}
\begin{proof}
\bse
We write $V=A_G^{-1}f$, and integrate over an interval in $t_E$ of any fixed length $P_E$. The weak vector equation, the energy estimate and Cauchy--Schwarz give
\be
\label{AX-eq-vector-source-energy}
\|\nabla V\|_2+\|\omega^{-1}V\|_2 \leq C\|\omega f\|_2 \leq C P_E^{1/2}m^{\beta+1/2}\|f\|_{k,\alpha,\beta-2;\omega}.
\ee
For the second inequality, $|\omega f|\leq\omega^{\beta-1}\|f\|$, while the spatial measure per unit length in $t_E$ is $N\,d\mu_h$. The reference coefficients bound this measure above by a fixed multiple of Euclidean spatial measure. On each radial region centered at $p_j$, its integral is therefore bounded by $C\int_{cm}^C r^{2\beta}\,dr\leq C m^{2\beta+1}$, since $2\beta+1<0$. The remaining fixed region contributes a bounded amount. This proves the second inequality in~\eqref{AX-eq-vector-source-energy} uniformly for comparable masses.

On nested coordinate neighborhoods at scale $\omega$, the local elliptic estimate for the invariant vector system gives
\be
\label{AX-eq-vector-local-Ltwo-Schauder}
\begin{aligned}
\|V\|_{C^{k+2,\alpha}_{\omega}}
&\leq C\Bigl((P_E\omega^3)^{-1/2}\|V\|_{L^2(\mathrm{larger\ neighborhood})}
+\omega^2\|f\|_{C^{k,\alpha}_{\omega}}\Bigr)\\
&\leq C\Bigl(P_E^{-1/2}\omega^{-1/2}\|\omega^{-1}V\|_2
+\omega^\beta\|f\|_{k,\alpha,\beta-2;\omega}\Bigr).
\end{aligned}
\ee
Away from horizons the volume factor is the spatial factor $\omega^3$ times $P_E$, with lapse bounded below in the relevant rescaled charts. At an inner horizon, a regular period of $t_E$ has size comparable to $m$, a lifted ball has four-dimensional volume comparable to $m^4$, and $\omega\asymp m$. For invariant fields, converting its $L^2$ integral to an arbitrary integration length in $t_E$ gives exactly the same factor $P_E m^3$. The fixed outer horizon has the corresponding estimate at scale one. Thus the local constants are uniform, also at poles. The estimate follows from interior elliptic regularity and the Schauder estimate on these fixed rescaled charts; the principal symbol of $A_G$ is uniformly positive on vectors.

After multiplication by $\omega^{-\beta}$, the energy bound gives a factor which satisfies
\be
\omega^{-\beta-1/2}m^{\beta+1/2}+1\leq2m^{\beta+1/2}, \qquad -\beta-1/2>0,\quad \omega\leq1,\quad m\leq1.
\ee
Taking the supremum therefore proves~\eqref{AX-eq-vector-weight-upper}. Existence and fixed-parameter regularity follow from the previously proved weak inverse and local regularity; approximation gives the same bound for weighted H\"older data. Finally, $\delta_G$ maps the tensor source space to weight $\gamma-3=\beta-2$ with one fewer derivative, while $B_G$ maps weight $\beta$ to weight $\gamma-2$. Applying the vector estimate at order $k-1$ to the definition of $\Pi_G$ proves~\eqref{AX-eq-weighted-projection-upper}.
\ese
\end{proof}

\paragraph{Estimating the effect of the divergence-free projection}
We now pass from the differential operator to the projected linearization of the Einstein tensor. We denote by $Y_\varepsilon$ and $X_\varepsilon$ the weighted H\"older spaces $C^{k,\alpha}_{\gamma-2;\omega}$ and $C^{k+2,\alpha}_{\gamma;\omega}$, respectively, with $k\geq1$, in the regular static tensor class. These are spaces of covariant metric variations; applying the estimate in a regular coefficient chart of Section~\ref{AX-section-boundaries} requires that chart to identify its domain with this tensor domain. For this statement the domain consists of tensors satisfying both $\mathfrak b_Gv=0$ and the inner orthogonality conditions~\eqref{AX-eq-complete-core-centering}. We define
\be
\label{AX-eq-projected-linear-tensor}
L_\varepsilon v\coloneqq\Pi_GD S_3(G)[v], \qquad \mathcal T_G T\coloneqq T-\tfrac12(\operatorname{tr}_GT)G.
\ee
The map $\mathcal T_G$ replaces a tensor by its trace reversal in dimension four; applying it twice gives back the original tensor.

\begin{proposition}[A uniform estimate for the linearized reduced Einstein equation]
\label{AX-prop-projected-tensor-bound}
On this domain, one has
\be
\label{AX-eq-projected-tensor-bound}
\|v\|_{X_\varepsilon}\leq C\|L_\varepsilon v\|_{Y_\varepsilon}
\ee
for the Kottler--de Sitter references and all sufficiently small $\varepsilon$, with a constant independent of $\varepsilon$. The operator is injective with closed range, and its inverse onto that range has this uniform bound. In addition, one has
\be
\label{AX-eq-initial-error-projection}
\Pi_GS_3(G)=S_3(G),\qquad \|\Pi_GS_3(G)\|_{Y_\varepsilon}\leq C\varepsilon^{p}.
\ee
\end{proposition}
\begin{proof}
\bse
We put $E_G\coloneqq\Ric_G-3G$ and recall the approximation exponent $p$ from~\eqref{AX-eq-exponent-relations}. Since $S_3(G)=\mathcal T_GE_G$, differentiation in dimension four gives
\be
\label{AX-eq-Einstein-Ricci-linear-relation}
D S_3(G)[v] =\mathcal T_GD(\Ric-3g)_G[v]+\mathcal C_G(v), \qquad \mathcal C_G(v)\coloneqq\tfrac12\langle v,E_G\rangle_GG -\tfrac12(\operatorname{tr}_GE_G)v.
\ee
Indeed, the derivative of $\operatorname{tr}_g E_g$ is $\operatorname{tr}_GD E_G[v]-\langle v,E_G\rangle_G$. The reference error bound holds at every fixed derivative order, so the product estimate and $\omega^\gamma\leq C m^\gamma=C\varepsilon^{3\gamma}$ give
\be
\label{AX-eq-reference-times-tensor}
\|\mathcal C_G(v)\|_{Y_\varepsilon} \leq C\varepsilon^{p+3\gamma}\|v\|_{X_\varepsilon}.
\ee

Differentiating the contracted Bianchi identity $\delta_gS_3(g)=0$ at $G$ yields
\be
\label{AX-eq-linearized-Bianchi-error}
\delta_GD S_3(G)[v]=-D\delta_G[v]S_3(G),\qquad \|D\delta_G[v]S_3(G)\|_{k-1,\alpha,\gamma-3;\omega} \leq C\varepsilon^{p+3\gamma}\|v\|_{X_\varepsilon}.
\ee
To see the estimate, the divergence variation is a sum of an inverse-metric variation times $\nabla S_3(G)$ and connection variations times $S_3(G)$. These terms have sizes $|v||\nabla S_3(G)|$ and $|\nabla v||S_3(G)|$. Their weights are $2\gamma-3$; conversion to the indicated weight costs at most $m^\gamma$. Differentiating these products proves the same assertion through order $k-1$.

We set $V\coloneqq A_G^{-1}\delta_GD S_3(G)[v]$. By the definition of the projection, $D S_3(G)[v]=L_\varepsilon v+B_GV$. Lemma~\ref{AX-lem-vector-weight-upper} and~\eqref{AX-eq-linearized-Bianchi-error} give
\be
\label{AX-eq-small-projection-on-linear-image}
\|B_GV\|_{Y_\varepsilon} \leq C\varepsilon^{p+6\gamma-3/2}\|v\|_{X_\varepsilon} =C\varepsilon^{p_{\rm lin}}\|v\|_{X_\varepsilon}.
\ee
On the coordinate domain, $D(\Ric-3g)_G[v]=\mathcal P_Gv$. Trace reversal of~\eqref{AX-eq-Einstein-Ricci-linear-relation}, followed by the global estimate under the inner orthogonality conditions, consequently gives
\be
\label{AX-eq-projected-bound-absorption}
\|v\|_{X_\varepsilon} \leq C\|\mathcal P_Gv\|_{Y_\varepsilon} \leq C\|L_\varepsilon v\|_{Y_\varepsilon} +C\bigl(\varepsilon^{p_{\rm lin}}+\varepsilon^{p+3\gamma}\bigr)\|v\|_{X_\varepsilon}.
\ee
The last term is absorbed for small $\varepsilon$, proving~\eqref{AX-eq-projected-tensor-bound}. Injectivity and closed range follow as in Proposition~\ref{AX-prop-completely-centered-inverse-bound}, because the coordinate and centering conditions define a closed subspace. Finally, the unvaried contracted Bianchi identity gives $\delta_GS_3(G)=0$, so the projection leaves this tensor exactly unchanged. Lemma~\ref{AX-lem-improved-reference-error} proves the last estimate.
\ese
\end{proof}

\paragraph{A quadratic estimate for corrections added to the reference metric}
There is also a direct nonlinear estimate in these tensor norms. For $g=G+v$, we define $\mathcal Q_G(v)\coloneqq S_3(G+v)-S_3(G)-D S_3(G)[v]$. On a ball where $\varepsilon^{3\gamma}\|v\|_{X_\varepsilon}$ and $\varepsilon^{3\gamma}\|w\|_{X_\varepsilon}$ are sufficiently small, the inverse metrics are uniformly comparable and the local Ricci formula gives
\be
\label{AX-eq-unprojected-quadratic-weight}
\|\mathcal Q_G(v)-\mathcal Q_G(w)\|_{Y_\varepsilon} \leq C\varepsilon^{3\gamma} (\|v\|_{X_\varepsilon}+\|w\|_{X_\varepsilon})\|v-w\|_{X_\varepsilon}.
\ee
We justify the weight and the constant. In a chart scaled by $\omega$, Taylor's formula for the metric, inverse metric and its first two derivatives bounds the second variation of $S_3$ by $C\omega^{-2}$ times the two scaled $C^{k+2,\alpha}$ tensor norms. Each tensor factor contributes its weight $\omega^\gamma$. Multiplying by the source weight $\omega^{2-\gamma}$ leaves $\omega^\gamma\leq C\varepsilon^{3\gamma}$. Integration of this second variation along the segment between $v$ and $w$ proves the inequality; the reference charts have uniform constants by the preceding coefficient estimates. Applying~\eqref{AX-eq-weighted-projection-upper} now gives
\be
\label{AX-eq-projected-quadratic-proved}
\|\Pi_G(\mathcal Q_G(v)-\mathcal Q_G(w))\|_{Y_\varepsilon} \leq C\varepsilon^{6\gamma-3/2} (\|v\|_{X_\varepsilon}+\|w\|_{X_\varepsilon})\|v-w\|_{X_\varepsilon}.
\ee
This proves the quadratic power used in the contraction scheme for corrections added to the reference metric in the specified regular coordinate domain. Section~\ref{AX-section-force-surjectivity} constructs the two additional tensors needed for surjectivity and improves this quadratic estimate by using a vector weight above the decay exponent of the point-source potential.


\section{\texorpdfstring{Solvability with two additional scalar unknowns (linear analysis)}{Solvability with two additional scalar unknowns (linear analysis)}}
\label{AX-section-force-surjectivity}


\subsection{Fredholm index and the two inner orthogonality conditions}
\label{AX-section-centered-range}

We determine the range of the linearized operator under the inner orthogonality conditions and construct two divergence-free tensors representing its cokernel, the quotient of the source space by its range. We add these tensors to the equation with two unknown scalar coefficients and invert the resulting \emph{linear operator with two additional scalar unknowns} for use in the nonlinear correction. We compute the Fredholm index of the projected operator at fixed positive masses. Independence of the two centering conditions then determines the cokernel of the restricted operator. We first work at a fixed positive value of $\varepsilon$. Constants in the index argument may depend on this value; the final inverse estimate uses the uniform constant already obtained in Proposition~\ref{AX-prop-projected-tensor-bound}.

We write $\mathcal D_\varepsilon\coloneqq\{v\in C^{k+2,\alpha}_{\gamma;\omega}:\mathfrak b_Gv=0\}$ for the domain before imposing the inner orthogonality conditions, and $\mathcal Y_\varepsilon\coloneqq\{F\in C^{k,\alpha}_{\gamma-2;\omega}:\delta_GF=0\}$ for its divergence-free target. All tensors have the prescribed symmetries and are regular in their local angular lifts. Write $\mathcal D_\varepsilon^0$ for the subspace defined by the inner orthogonality conditions, and $n_S\coloneqq\dim\mathcal K_S$ as in~\eqref{AX-eq-Schwarzschild-model-kernel}.

\begin{proposition}[Index zero and a two-dimensional obstruction to solvability]
\label{AX-prop-centered-index}
For every fixed, sufficiently small $\varepsilon>0$, the operator $L_\varepsilon:\mathcal D_\varepsilon\to\mathcal Y_\varepsilon$, before imposing the inner orthogonality conditions, has the Fredholm property and index zero. In the class with both inner regions without midpoint reflection, the inner orthogonality conditions on the two inner regions are independent and
\be
\label{AX-eq-centered-cokernel-dimension}
\dim\bigl(\mathcal Y_\varepsilon/L_\varepsilon\mathcal D_\varepsilon^0\bigr)=2n_S=2.
\ee
In a reflection-invariant class, the same conclusion holds with the number of independent compatible paired centerings in place of $2n_S$.
\end{proposition}
\begin{proof}
\bse
\emph{1. Trace reversal and the constrained form.} We first establish the relation between the conditions $\mathfrak b_Gv=0$ and $\delta_Gq=0$. Trace reversal is selfadjoint and satisfies $\mathcal T_G^2=I$ in dimension four. Direct contraction of~\eqref{AX-eq-gauged-tensor-operator} gives
\be
\label{AX-eq-trace-reversal-commutation}
\begin{aligned}
&\mathfrak b_G(\mathcal T_Gq)=\delta_Gq,\qquad \operatorname{tr}_G(\mathcal P_Gq)=\bigl(\tfrac12\nabla^*\nabla-3\bigr)\operatorname{tr}_Gq,\\
&\mathcal P_G(fG)=\bigl(\tfrac12\nabla^*\nabla f-3f\bigr)G, \qquad \mathcal P_G\mathcal T_G=\mathcal T_G\mathcal P_G.
\end{aligned}
\ee
The Ricci contractions cancel the trace of the curvature action in the second identity. The same cancellation on a tensor $fG$ gives the third. Thus, for a divergence-free tensor $q$, the variation identity~\eqref{AX-eq-Einstein-Ricci-linear-relation} gives
\be
\label{AX-eq-transverse-index-operator}
L_\varepsilon(\mathcal T_Gq)=\Pi_G\mathcal P_Gq+K_Gq, \qquad K_Gq\coloneqq\Pi_G\mathcal C_G(\mathcal T_Gq).
\ee
The operator $K_G$ has order zero: multiplication by the smooth reference error is followed by the bounded projection $\Pi_G$.

We justify the Sobolev spaces and compact embedding used below without choosing a single period of $t_E$ regular at all horizons. We integrate invariant fields over an interval in $t_E$ of any fixed length $P_E$ with measure $P_EN\,d\mu_h$. In a regular angular lift at a horizon of gravity $\kappa$, this integral is $P_E\kappa/(2\pi)$ times the corresponding smooth four-dimensional integral. The factor is positive and constant on that horizon. The same comparison holds at a pole in the product of the two regular angular planes. For $s=1,2$, we define $H^s$ as the completion of the smooth invariant fields regular in these lifts, using the global measure $P_EN\,d\mu_h$ and the covariant derivatives of orders at most $s$. A finite invariant partition of unity and a finite collection of these charts give equivalent fixed-parameter Sobolev norms. Compactness of $H^1\hookrightarrow L^2$ follows by applying compactness in each smooth bounded chart and extracting a finite diagonal subsequence. All fields are invariant, and averaging local tests over the circle actions and their reflection symmetries gives the same weak equations as unrestricted tests in a regular lift. For smooth fields, integration by parts on a region with the horizon tubes of radius $\delta$ removed has limiting boundary contribution zero: the tube measure per unit length in $t_E$ is $\Ocal(\delta)$, while the fields and their first covariant derivatives are bounded in the regular lift. At a pole we use the product of the two regular rotation planes. Density then extends the identities to their $H^1$ and $H^2$ domains.

We write $H\coloneqq\ker\delta_G\subset L^2$ and $V\coloneqq H\cap H^1$. The inverse of the operator on vector fields shows that $\Pi_G$ is the orthogonal projection onto $H$: $B_G$ is the formal adjoint of $\delta_G$, the decomposition $T=\Pi_GT+B_GA_G^{-1}\delta_GT$ is orthogonal, and the integration by parts has zero limiting contribution at the regular angular axes. The tensor operator has the form $\mathcal P_G=\tfrac12\nabla^*\nabla+\mathcal V_G$, where $\mathcal V_G$ is a symmetric endomorphism with a finite bound for this fixed $\varepsilon$. For $c>\|\mathcal V_G\|_\infty+1$, its constrained form satisfies
\be
\label{AX-eq-transverse-coercive-form}
a_c(q,z)\coloneqq\tfrac12(\nabla q,\nabla z)_{L^2} +(\mathcal V_Gq,z)_{L^2}+c(q,z)_{L^2},\quad a_c(q,q)\geq\tfrac12\|\nabla q\|_{L^2}^2+\|q\|_{L^2}^2, \quad q,z\in V.
\ee
Minimize $a_c(q,q)/2-(F,q)$ on the Hilbert space $V$ for each $F\in H$. Completing the square in the inner product $a_c$ identifies the minimizer with the Riesz representative and proves uniqueness of the weak solution. The solution operator is symmetric and sends bounded subsets of $H$ to bounded subsets of $V$. The embedding proved above makes it compact on $H$.
\ese

\bse
\emph{2. Regularity and the index before the inner orthogonality conditions.} We verify its regularity and domain. The divergence $\delta_G:H^1\to L^2$ is onto in this class of spatial vector fields orthogonal to the axial symmetry orbits, since $B_GA_G^{-1}$ is a right inverse and the fixed-parameter vector regularity maps $L^2$ into $H^2$. The functional $(\mathcal P_G+c)q-F\in H^{-1}$ vanishes on $V$. Composing it with this right inverse gives an $L^2$ functional and hence a one-form annihilating both Killing directions $p\in L^2$ such that
\be
\label{AX-eq-constrained-pressure-system}
(\mathcal P_G+c)q=F+B_Gp^\sharp, \qquad \delta_Gq=0.
\ee
Here we identify vector fields and one-forms by $G$ when applying $A_G$. Covariant commutation gives
\be
\label{AX-eq-divergence-tensor-commutator}
\delta_G\mathcal P_Gq =\tfrac12\nabla^*\nabla(\delta_Gq)+R_1q, \qquad A_Gp^\sharp=(R_1q)^\sharp,
\ee
where $R_1$ is first order with smooth coefficients. Indeed, the third derivatives cancel, and $[\nabla,\nabla^2]q$ is a sum of curvature times $\nabla q$ and a differentiated curvature tensor times $q$. Differentiating the potential in $\mathcal P_G$ has the same form. The second equation follows by taking divergence in~\eqref{AX-eq-constrained-pressure-system}. Since $q\in H^1$, its right side belongs to $L^2$. The inverse of the operator on regular vector fields gives $p\in H^2$, and the tensor equation then gives $q\in H^2$. The local elliptic estimates apply in the smooth lifts; the weak equations extend across the axes by the averaged-test argument. Iteration proves the corresponding higher Sobolev regularity.

The same argument gives the H\"older domain. For smooth data and $k\geq1$, the vector estimate, the tensor Schauder estimate, and interpolation give
\be
\label{AX-eq-constrained-fixed-parameter-Schauder}
\|p\|_{C^{k+1,\alpha}}\leq C_\varepsilon\|q\|_{C^{k,\alpha}},\qquad \|q\|_{C^{k+2,\alpha}}\leq C_\varepsilon \bigl(\|F\|_{C^{k,\alpha}}+\|q\|_{L^2}\bigr).
\ee
For the first inequality we apply the vector estimate at order $k-1$ to the first-order right side. For the second, the tensor estimate first produces the lower norm $\|q\|_{C^{k,\alpha}}$ through $B_Gp^\sharp$; interpolation absorbs it, as well as the intermediate $C^0$ norm, into a small multiple of the left side and an $L^2$ term. Approximation in the regular charts, followed by convergence at a smaller H\"older exponent, proves the same estimate for $C^{k,\alpha}$ data. The uniform bound passes to the full exponent by the defining difference quotients. These are fixed-parameter estimates. Since $\omega$ is positive on the compact region for fixed $\varepsilon$, the weighted and unweighted norms are equivalent here.

Consequently the constrained operator $A\coloneqq\Pi_G\mathcal P_G$, with domain $H\cap H^2$, has a compact selfadjoint inverse after addition of $c$. Its kernel is finite dimensional and its range is closed: both assertions follow by the usual compactness contradiction applied to $I-c(A+c)^{-1}$. Orthogonality and selfadjointness identify the cokernel with the kernel, so its index is zero. We can see directly that the perturbation in~\eqref{AX-eq-transverse-index-operator} preserves this index. With $R\coloneqq(A+c)^{-1}$, we have
\be
\label{AX-eq-index-compact-factorization}
(A+K_G)R=I+(K_G-c)R.
\ee
The operator $(K_G-c)R$ is compact. Approximating it in operator norm by a finite-rank operator with error less than one factors the right side into an invertible operator and an identity plus a finite-rank operator. The latter has equal kernel and cokernel dimensions, as is seen by splitting off the finite-dimensional range of its finite-rank part. Thus its index is zero. The regularity argument applies also to $A+K_G$ and its adjoint, because $K_G$ preserves every Sobolev order. It follows that the same kernel, cokernel, and index occur in the regular H\"older spaces. Trace reversal therefore shows that $L_\varepsilon$ has Fredholm index zero.
\ese

\bse
\emph{3. Independence of the centering conditions.} We next show that the model centerings are independent. Proposition~\ref{AX-prop-Schwarzschild-kernel-classification} proves that every element of $\mathcal K_S$ is divergence-free and trace free, and that $n_S=1$.

Under the rescaling~\eqref{AX-eq-core-rescaling-map}, the functional $q\mapsto\ell_{j,a,\varepsilon}(\mathcal T_Gq)$ has an $L^2(G)$ representative supported on $K_j\coloneqq\Phi_{j,\varepsilon}(K_S)$. After division by its nonzero scaling constant, its pullback equals $w_a+\Ocal(Q_j)$ in every smooth norm on $K_S$. This follows from $\widehat G_{Q_j}=G_S+\Ocal(Q_j)$ on that fixed compact region, including trace reversal and the volume density. The positive definite model Gram matrix implies that, for every linear combination $\sigma_j$ of these representatives,
\be
\|\delta_G\sigma_j\|_{L^2(K_j)} \leq C Q_jm_j^{-1}\|\sigma_j\|_{L^2(K_j)}.
\ee
The divergence in this estimate is taken in the interior of $K_j$; the representatives are not differentiated across their outer boundary.

We suppose that a combination of the centering functionals vanishes on all divergence-free tensors. Its representative $H_0$ is orthogonal to $H$, and the vector energy inverse gives $H_0=B_GU$. Smooth divergence-free tensors are dense in $H$: we approximate in $L^2$ and apply the orthogonal projection, which preserves smoothness. Outside the two sets $K_j$, the field $U$ is Killing. It vanishes on the outer de Sitter region by the symmetry and orthogonality conditions on $U$. The Killing prolongation equation then gives $U=0$ on the connected complement of $K_+\cup K_-$. In particular its outer trace on each $K_j$ is zero.

The scalar Dirichlet Poincar\'e inequality on the fixed rescaled collar, applied to $|U|$, gives $\|U\|_{L^2(K_j)}\leq Cm_j\|\nabla U\|_{L^2(K_j)}$, with a constant uniform for small $Q_j$. Since $\Ric_G=\Lambda G$ on $K_j$, the integrated symmetric-gradient identity gives
\be
2\|B_GU\|_{L^2(K_j)}^2 =\|\nabla U\|_{L^2(K_j)}^2+\|\operatorname{div}_GU\|_{L^2(K_j)}^2 -\Lambda\|U\|_{L^2(K_j)}^2.
\ee
For small $Q_j$ the last term is absorbed, and hence $\|U\|_{L^2(K_j)}\leq Cm_j\|B_GU\|_{L^2(K_j)}$. In the interior $B_GU=\sigma_j$, and integration by parts gives
\be
\label{AX-eq-centering-independence-energy}
\|\sigma_j\|_{L^2(K_j)}^2 =\langle\delta_G\sigma_j,U\rangle_{L^2(K_j)} \leq C Q_j\|\sigma_j\|_{L^2(K_j)}^2.
\ee
The outer trace is zero and the regular horizon contributes no boundary term. Absorption and independence of the finite model basis make every coefficient zero. This proves rank $2n_S$ on the Bianchi domain. The same argument with paired representatives proves the rank in the reflection-invariant class.

Finally, restriction of a Fredholm operator to a domain of codimension $r$ reduces its index by $r$. To check this finite-dimensional assertion, we choose a complement of dimension $r$ to the restricted domain; adjoining its basis vectors recovers the original domain and increases the index by $r$. Equivalently, kernel, range, and quotient fit into the exact sequence obtained by applying the operator to this decomposition. Proposition~\ref{AX-prop-projected-tensor-bound} makes the restricted kernel zero and its range closed. Its cokernel therefore has dimension $r$, which proves~\eqref{AX-eq-centered-cokernel-dimension}.
\ese
\end{proof}


\subsection{Global tensors representing the two obstructions}
\label{AX-section-global-force-columns}

\paragraph{Constructing global divergence-free tensors with prescribed axial fluxes}
We construct two global divergence-free tensors from the displacement fields on the Kottler collars. Their axial fluxes identify the two obstructions, and their outer limit describes the part carried by the de Sitter region. The displacement mode of the inner limit defines regular fields on each exact Kottler collar. Lemma~\ref{AX-lem-deSitter-outer-force} identifies the de Sitter field carrying the total axial flux in the outer limit. We work with comparable reference masses $m_\pm\asymp\varepsilon^3$, separation $d\asymp\varepsilon$, and $\ell=1$. All $L^2$ pairings are divided by one integration length in $t_E$; the invariant volume is $N\,d\mu_h$.

On the unit inner limit, we put $\varphi=(R-1)\cos\theta$, $V_S=\nabla_{G_S}\varphi$ and $w=2\Hess_{G_S}\varphi$. We transport its vector field and define
\be
U_j\coloneqq m_j^{-1}(\Phi_{j,\varepsilon})_*V_S, \qquad w_j\coloneqq2m_jB_GU_j.
\ee
The fields are smooth at the horizon in its regular angular lift. In reference radial coordinates the spatial radial coefficient acquires the factor $b_{Q_j}^{-1}$ and the lapse acquires $\sqrt{b_{Q_j}}/(1-Q_j)$. Their product in the volume density is the constant $(1-Q_j)^{-1}$. Since $\operatorname{div}_{G_S}V_S=0$, this proves
\be
\operatorname{div}_GU_j=0, \qquad\operatorname{tr}_Gw_j=0.
\ee
We denote by $w_j^0$ the transport of $c_{h,j}m_j^2w$ from the local $Q=0$ model. Differentiating the exact collar coefficients gives, for every fixed $b\geq0$,
\be
\begin{aligned}
|U_j-e_z|&\leq Cm_j/r_j,&
|\nabla^b(A_GU_j)|&\leq C_b\Lambda r_j^{-b},\\
|\nabla^b(w_j-w_j^0)|&\leq C_b\Lambda m_jr_j^{1-b},&
|\nabla^bw_j^0|&\leq C_bm_j^2r_j^{-2-b}.
\end{aligned}
\ee
These estimates hold for $r_j\geq m_j$ and in the regular scaled lift for $r_j\asymp m_j$. For the vector estimate, the coefficient difference from $Q=0$ and its scaled derivatives are $\Ocal(\Lambda r_j^2)$; we apply it to the limiting solution $A_{G_S}V_S=0$. The estimate for $w_j-w_j^0$ follows from $w_j=m_j\mathcal L_{U_j}G$.

We choose a smooth decreasing radial function $\chi$ which is one on $[0,1]$ and zero on $[2,\infty)$. For large $L$, possibly depending on $\varepsilon$, we require $4Lm_j$ to remain in each exact Kottler collar, $Lm_j/d\to0$, and $\Lambda m_j^2L^3\to0$. We define the localized tensors and their vector corrections by
\be
\label{AX-eq-global-force-tensor-construction}
\begin{aligned}
q_j&\coloneqq m_j^{-2}\chi(r_j/(Lm_j))w_j
=\frac2{m_j}\chi(r_j/(Lm_j))B_GU_j,\\
W_j&\coloneqq A_G^{-1}\delta_Gq_j,
\qquad \psi_j\coloneqq\Pi_Gq_j=q_j-B_GW_j.
\end{aligned}
\ee
The tensors $q_j$ are trace free and $\delta_G\psi_j=0$. Replacing $L$ by $R_j=m_j^{-\zeta}$ on the $j$th collar, with $0<\zeta<4/9$, satisfies these scale conditions.

\begin{proposition}[Two global tensors with an invertible matrix of axial fluxes]
\label{AX-prop-global-force-columns}
For $L$ sufficiently large and $\varepsilon$ sufficiently small under the preceding scale conditions, the two-by-two matrix $\mathsf F_{ij}\coloneqq\operatorname{Flux}_i(\psi_j)$ on the two separation spheres is invertible, with uniformly bounded inverse. The tensors in~\eqref{AX-eq-global-force-tensor-construction} satisfy
\be
\label{AX-eq-global-force-tensor-bounds}
\|\psi_j\|_2\leq Cm_j^{-1/2}, \qquad d^2\|\psi_j\|_{C^{1,\alpha}_d(\mathcal U_i)} \leq C\sqrt{d/m_j},
\ee
where $\mathcal U_i$ is a fixed enlarged neighborhood of either flux sphere. Bounded linear combinations therefore give divergence-free tensors $\widehat\psi_j$ with $\operatorname{Flux}_i(\widehat\psi_j)=\delta_{ij}$. In the equal-mass reflection-symmetric class, their difference supplies the single tensor associated with separation.
\end{proposition}
\begin{proof}
\bse
\emph{1. Projection and the Gram matrix.} The inverse of the operator on vector fields and integration by parts show that $\Pi_G$ is the orthogonal $L^2$ projection onto the divergence-free tensors in the prescribed invariant class. More explicitly, $\langle T,B_GV\rangle=\langle\delta_GT,V\rangle$, and regularity in the angular lifts makes the boundary contributions at all horizons vanish. The range of $B_G$ is closed by the vector estimate. Thus $\|\psi_j\|_2\leq\|q_j\|_2\leq Cm_j^{-1/2}$.

We next estimate $B_GW_j$. The exact product formula is
\be
\delta_Gq_j =-\frac2{m_j}(B_GU_j)(\nabla\chi_j,\,\cdot\,) +\frac2{m_j}\chi_jA_GU_j, \qquad \chi_j\coloneqq\chi(r_j/(Lm_j)).
\ee
The cutoff source has the previous size $Cr_j^{-3}$; the second term is bounded by $C\Lambda/m_j$. We put $\rho_j=Lm_j$. The vector estimate, tested against $W_j$, gives
\be
\label{AX-eq-force-projection-tail}
\|\nabla W_j\|_2+\|\omega^{-1}W_j\|_2 \leq C\bigl(\rho_j^{-1/2}+\Lambda m_j^{-1}\rho_j^{5/2}\bigr) =C m_j^{-1/2}L^{-1/2}(1+\Lambda m_j^2L^3) \leq C m_j^{-1/2}L^{-1/2}.
\ee
For the term proportional to $\Lambda$, $\omega\asymp r_j$ away from the fixed scaled horizon region and the volume per unit length in $t_E$ is comparable to $r_j^2dr_jd\Omega$; its squared norm is bounded by $C\Lambda^2m_j^{-2}\rho_j^5$. The fixed scaled horizon region obeys the same bound.

The supports of $q_+$ and $q_-$ are disjoint. The orthogonality of $\Pi_G$, together with~\eqref{AX-eq-force-projection-tail}, gives
\be
\label{AX-eq-force-column-Gram}
\langle\psi_j,q_j\rangle=\|\psi_j\|_2^2 =\|q_j\|_2^2-\|B_GW_j\|_2^2,\qquad |\langle\psi_i,q_j\rangle|\leq C(m_im_j)^{-1/2}L^{-1},\qquad i\ne j.
\ee
For the off-diagonal bound we use the sharper identity $\langle\Pi_Gq_i,\Pi_Gq_j\rangle=-\langle B_GW_i,B_GW_j\rangle$. The constants $c_{h,j}$ and $c_{N,j}$ tend to one, and using the local $Q=0$ coefficient formulas gives
\be
\label{AX-eq-force-column-diagonal-limit}
m_j\|q_j\|_2^2 =c_{h,j}^{3/2}c_{N,j}\,a_L +\Ocal(\Lambda m_j^2L^{5/2}+\Lambda^2m_j^4L^5), \qquad a_L\coloneqq\int \chi(|y|/L)^2|w|_{G_S}^2\,d\nu_S \longrightarrow a_\infty>0.
\ee
Here $d\nu_S$ is Schwarzschild volume per unit length in the model coordinate $\tau$. The error follows from $m_j^{-2}(w_j-w_j^0)=\Ocal(\Lambda r_j/m_j)$ and the constant volume-density ratio. The integral defining $a_L$ converges because $w=\Ocal(|y|^{-2})$. Its value can also be computed. In isotropic Cartesian coordinates, with $r=|y|$ and $n=y/r$, we have $\varphi=y^3+\Ocal(r^{-1})$, and differentiation of the Schwarzschild factors gives
\be
\label{AX-eq-displacement-force-normalization}
w_{ij}=\frac{2}{r^2}(\delta_{i3}n_j+\delta_{j3}n_i-\delta_{ij}n_3)+\Ocal(r^{-3}), \qquad \int_{\{r=s\}}w_{ij}n^i\delta^j_3\,dA_0=8\pi+\Ocal(s^{-1}).
\ee
Since $w=2B_{G_S}V_S$ and $\delta_{G_S}w=0$, integration by parts gives $a_\infty=2\lim_{s\to\infty}F_{V_S}(w;s)=16\pi$. The model lapse and area factors tend to their Euclidean values, and $V_S=e_z+\Ocal(s^{-1})$, so these replacements do not change the limit. In particular, the diagonal entries multiplied by $m_j$ are bounded below independently of sufficiently large $L$ and sufficiently small $\varepsilon$.
\ese

\bse
\emph{2. Detection by the axial fluxes.} We next connect this positive pairing to an axial flux. For a divergence-free tensor $T$, we define
\be
\label{AX-eq-force-weighted-flux}
F_{j,U}(T;s)\coloneqq \int_{\{r_j=s\}}N\,T(n_h,U_j)\,dA_h.
\ee
Using $q_j=(2/m_j)\chi(r_j/(Lm_j))B_GU_j$ and testing $\delta_GT=0$ against the regular compactly supported vector $\chi U_j$ gives the exact identity
\be
\label{AX-eq-force-pairing-flux-identity}
\frac{m_j}{2}\langle T,q_j\rangle =-\int_0^\infty\frac{d}{ds}\chi(s/(Lm_j))\, F_{j,U}(T;s)\,ds.
\ee
The derivative of the cutoff has integral $-1$, so the right-hand side is an average of the indicated flux over its cutoff annulus. Positivity of $\chi$ and monotonicity can be imposed in choosing the cutoff.

The flux varies little between that annulus and the sphere $r_j=4Lm_j$. The divergence identity gives the difference as the volume integral of $T:B_GU_j$. On this part of the exact Kottler collar, $|B_GU_j|\leq C(m_j/r_j^2+\Lambda r_j)$. Hence, for $T=\psi_i$,
\be
\label{AX-eq-force-core-transport}
\begin{aligned}
 |F_{j,U}(\psi_i;s_2)-F_{j,U}(\psi_i;s_1)|
 &\leq \|\psi_i\|_2\|B_GU_j\|_{L^2(\{Lm_j\leq r_j\leq4Lm_j\})}\\
 &\leq C\bigl(L^{-1/2}+\Lambda m_j^2L^{5/2}\bigr)
 \leq C L^{-1/2}.
\end{aligned}
\ee
The last inequality uses $\Lambda m_j^2L^3\to0$. Comparable masses are used in the mass estimates. At $r_j=4Lm_j$, the difference $U_j-e_z$ is $\Ocal(L^{-1})$. The source of the equation for $W_i$ vanishes on a neighborhood of that sphere. Scaled interior estimates for $A_GW_i=0$, together with~\eqref{AX-eq-force-projection-tail} or its weaker bound without $L^{-1/2}$, give $|\psi_i|\leq C m_i^{-1/2}(Lm_j)^{-3/2}$. Multiplying by the sphere area and by $|U_j-e_z|$ bounds the resulting flux change by $CL^{-1/2}$.

We transport the flux against $e_z$ to the sphere $r_j=\vartheta d$ used in the axial flux calculation. The reference coefficient estimate on this annulus implies
\be
\label{AX-eq-force-neck-transport}
\|B_Ge_z\|_{L^2(\{4Lm_j\leq r_j\leq\vartheta d\})}^2 \leq C(m_j/L+d^5).
\ee
Indeed, its pointwise bound is $C(m_j/r_j^2+d^2/r_j)$, and the measure is comparable to Euclidean spatial volume. Squaring and integrating gives this inequality. The divergence identity and the first bound in~\eqref{AX-eq-global-force-tensor-bounds} therefore show that this transport changes the flux by at most $C(L^{-1/2}+d^{5/2}m_i^{-1/2})=C(L^{-1/2}+\varepsilon)$. Local vector estimates at scale $d$ give the second bound in~\eqref{AX-eq-global-force-tensor-bounds}. Since $G-G_0=\Ocal(\varepsilon^2)$ on these annuli, replacing $N\,dA_h$, the unit normal, and the metric contractions by their Euclidean counterparts changes the flux by at most $C\varepsilon^2\sqrt{d/m_i}=C\varepsilon$.

Combining these estimates with~\eqref{AX-eq-force-pairing-flux-identity} gives
\be
\label{AX-eq-global-force-matrix}
\operatorname{Flux}_j(\psi_i) =\frac{m_j}{2}\langle\psi_i,q_j\rangle +\Ocal(L^{-1/2}+\varepsilon).
\ee
Its diagonal entries tend to $a_\infty/2=8\pi$ as first $L$ becomes large and then $\varepsilon$ becomes small; its off-diagonal entries tend to zero. These bounds prove invertibility and uniform boundedness of the inverse. They apply also to $L=m^{-\zeta}$, or to comparable mass-dependent choices $R_j=m_j^{-\zeta}$, because the supports of $q_j$ remain inside the exact Kottler collars and have radii $o(d)$. Multiplication by the inverse matrix gives the normalized obstruction tensors. In the equal-mass class with the reflection interchanging the inner horizons, that reflection exchanges the two obstruction tensors with a change of sign because it reverses $e_z$; their difference is invariant and has opposite axial fluxes on the two spheres.
\ese
\end{proof}

The local bound in~\eqref{AX-eq-global-force-tensor-bounds} contributes $C|\lambda|\eta_{\rm sep}\sqrt{d/m}$ to the estimate for the auxiliary vector term, while the energy contribution is $C|\lambda|\eta\sqrt{d/m}$. Since $\eta_{\rm sep}\leq\eta$, their sum is bounded by $C|\lambda|\eta\sqrt{d/m}$. These bounds suffice for the normalization of the axial fluxes and the proof of the Einstein equation. Theorem~\ref{AX-thm-physical-force-inverse} below shows that the two obstruction tensors span a complement of the range obtained after imposing the two inner orthogonality conditions.

\paragraph{The de Sitter limit of the two global obstruction tensors}
The global obstruction tensors have the same limiting field on each fixed outer region. Their sum and difference therefore have different limits, although their two fluxes around the inner horizons remain independent. We prove this directly on the Kottler--de Sitter references. We use the growing cutoffs $R_j\coloneqq m_j^{-\zeta}$ and their radii in the reference coordinates $\rho_j\coloneqq m_jR_j=m_j^{1-\zeta}$ in~\eqref{AX-eq-global-force-tensor-construction}. The two masses satisfy $m_j\asymp\varepsilon^3$, their separation satisfies $d\asymp\varepsilon$, and $|p_j|\leq Cd$. The outer field $V_D$ and its symmetric gradient $\Psi_D$ are those of Lemma~\ref{AX-lem-deSitter-outer-force}. We identify its areal coordinates with the isotropic reference coordinates by $x=X/(1+|X|^2/4)$, leaving the auxiliary coordinate $t_E$ fixed. This map has identity differential at the center.

\begin{lemma}[Convergence of the global obstruction tensors to the de Sitter field]
\label{AX-lem-global-columns-outer-limit}
For every compact subset $K$ of the completed de Sitter model away from its central circle, every integer $k\geq0$ and every $0<\alpha<1$, the global fields in~\eqref{AX-eq-global-force-tensor-construction} satisfy
\be
\label{AX-eq-global-columns-outer-limit}
\|W_j-8\pi V_D\|_{C^{k+1,\alpha}(K)} +\|\psi_j+8\pi\Psi_D\|_{C^{k,\alpha}(K)} \leq C_{K,k,\alpha}\varepsilon^{3\zeta/2},\qquad j\in\{+,-\}.
\ee
The estimate holds for sufficiently small $\varepsilon$ depending on $K$, uniformly on the prescribed compact ranges of mass fractions and geometric parameters. The norms at the cosmological horizon and its poles use the regular de Sitter lifts. In particular, for arbitrary real coefficients $\lambda_j$, one has
\be
\label{AX-eq-global-columns-outer-sum-separation}
\bigl\|\lambda_+\psi_++\lambda_-\psi_- +8\pi(\lambda_++\lambda_-)\Psi_D\bigr\|_{C^{k,\alpha}(K)} \leq C_{K,k,\alpha}\varepsilon^{3\zeta/2}(|\lambda_+|+|\lambda_-|).
\ee
Thus the sum of the axial fluxes has the de Sitter limit in~\eqref{AX-eq-global-columns-outer-sum-separation}, while a vanishing total axial flux cancels this outer term. The assertion concerns the obstruction tensors, rather than the nonlinear Einstein metrics.
\end{lemma}
\begin{proof}
\bse
\emph{1. Transported tests and their equation error.} All pairings are per unit length in $t_E$. We first transport a smooth spatial test vector field orthogonal to the axial symmetry orbits $\phi$ from the de Sitter model to the reference with its two small horizons. We assume that $A_D\phi=0$ on a fixed central coordinate neighborhood $U$. We choose a sufficiently large fixed $L_0$, and smooth radial cutoffs $\eta_j$ equal to one for $r_j\leq L_0m_j$ and zero for $r_j\geq2L_0m_j$. These sets are disjoint and lie inside the exact Kottler collars. Axial invariance gives $\phi(p_j)=a_j e_z$, where $a_j\coloneqq\phi^z(p_j)$. We set
\be
\label{AX-eq-outer-matching-test-transport}
\phi_\varepsilon\coloneqq\phi+\sum_{j=\pm}\eta_j(a_jU_j-\phi), \qquad \mathcal R_{\varepsilon,\phi}\coloneqq A_G\phi_\varepsilon-A_D\phi.
\ee
Here $U_j$ is the regular transported displacement field, with $2B_GU_j=m_j^{-1}w_j$ and $|\nabla^b(A_GU_j)|\leq C_b\Lambda r_j^{-b}$ on the exact Kottler collar. Consequently $\phi_\varepsilon$ is regular in every reference angular lift, including at the small horizons and their poles. It equals $\phi$ on the cutoff annulus $\rho_j\leq r_j\leq2\rho_j$, because $\rho_j/m_j\to\infty$. The error $\mathcal R_{\varepsilon,\phi}$ is bounded by $C\Lambda\|\phi\|_{C^2(U)}$ on $r_j\leq L_0m_j$ and vanishes outside $|X|<2r_{\rm out}$, where $r_{\rm out}=\varepsilon^{3/5}$.

We prove its required norm bound. On the two cutoff annuli, rescaling by $m_j$ gives $|\mathcal R_{\varepsilon,\phi}|\leq Cm_j^{-2}\|\phi\|_{C^2(U)}$. Away from these annuli and the smaller Kottler collars, the differentiated coefficient estimate~\eqref{AX-eq-neck-coefficient-smallness} gives
\be
\label{AX-eq-outer-matching-test-error-pointwise}
|\mathcal R_{\varepsilon,\phi}| \leq C\Bigl(1+\sum_{j=\pm}\frac{m_j}{r_j^3}\Bigr)\|\phi\|_{C^2(U)}.
\ee
We justify the power in this estimate. The operator $A_G$ has coefficients depending on the metric and its first two derivatives, and differentiates its vector argument at most twice. If $\rho=\min(r_+,r_-)$, the coefficient difference from the de Sitter operator is therefore bounded by $C\rho^{-2}(m/\rho+r^2+d^2)$ when applied to a field of bounded $C^2$ norm. Here $m\asymp m_j$ and $r=|X|$. When $\rho\leq d$, we have $r\leq C'd$ and $d^2\rho\leq Cm$, since $m\asymp d^3$. The entire expression is then bounded by $Cm/\rho^3$. When $\rho\geq d$, the inequality $r\leq\rho+Cd$ gives $r^2+d^2\leq C\rho^2$, so the remaining cosmological contribution is bounded. This proves~\eqref{AX-eq-outer-matching-test-error-pointwise}, including both reference transitions. The lapse is bounded below on this region, and the measure is uniformly comparable to spatial Euclidean measure.

We write $\mathfrak m=m_++m_-$. On a cutoff annulus the weighted $L^2$ bound is $Cm_j^{1/2}\|\phi\|_{C^2(U)}$. On its complement, $\omega\leq r_j$ and $\omega\leq r+Cd$ give
\be
\label{AX-eq-outer-matching-test-dual-bound}
\|\omega\mathcal R_{\varepsilon,\phi}\|_2 \leq C\Bigl(r_{\rm out}^{5/2}+\sum_{j=\pm}\sqrt{m_j}\Bigr)\|\phi\|_{C^2(U)} \leq C\sqrt{\mathfrak m}\,\|\phi\|_{C^2(U)}.
\ee
Indeed, the square of the constant-term contribution is bounded by $C\int_0^{2r_{\rm out}}(r+d)^2r^2\,dr\leq Cr_{\rm out}^5$, while each singular contribution is bounded by $Cm_j^2\int_{L_0m_j}^{C}r_j^{-2}\,dr_j\leq Cm_j$. On $r_j\leq L_0m_j$, the additional weighted $L^2$ contribution is at most $C\Lambda m_j^{5/2}\|\phi\|_{C^2(U)}$, which has smaller order than the right-hand side of~\eqref{AX-eq-outer-matching-test-dual-bound}. Finally $r_{\rm out}^{5/2}=\varepsilon^{3/2}\asymp\sqrt{\mathfrak m}$.
\ese

\bse
\emph{2. Normalization of the source.} We next normalize the action of the source on this test. On a sphere of radius $s$ in the exact Kottler collar, the finite-radius formula of the inner limit gives
\be
\label{AX-eq-outer-matching-exact-core-flux}
m_j^{-2}\int_{r_j=s}N_Gw_j(n_h,e_z)\,dA_h =8\pi c_{h,j}^{3/2}c_{N,j} \Bigl(1-\frac{2m_j}{3s}+\frac{m_j^2}{12s^2}\Bigr) +\Ocal(\Lambda s^3/m_j).
\ee
Here $n_h$ points toward increasing $r_j$, and $c_{h,j},c_{N,j}$ are the exact constants on the inner regions. We derive its principal term on the unit $Q=0$ model. In isotropic radius $s$, we put $v(s)\coloneqq2\log(1+1/(2s))$ and $f(s)\coloneqq s+1/(4s)$, so that the harmonic function is $f(s)\cos\theta$ and the spatial metric is $e^{2v}\delta$. The difference of its connection from the flat connection, evaluated on $(n,e_z)$, is $v'e_z$. Hence the sphere average of $\Hess_h(f\cos\theta)(n,e_z)$ is
\be
\label{AX-eq-outer-matching-core-flux-calculation}
\frac13\Bigl[f''+2\Bigl(\frac{f'}s-\frac f{s^2}\Bigr) -v'\Bigl(f'+\frac{2f}s\Bigr)\Bigr] =\frac{6s-1}{3s^2(2s+1)}.
\ee
We used the averages of $\cos^2\theta$ and $1-\cos^2\theta$, which are $1/3$ and $2/3$. Multiplication by $2$ for $w=2\Hess\varphi$, by the sphere area, and by $N_Se^v=1-1/(4s^2)$ gives $8\pi(1-2/(3s)+1/(12s^2))$. The spatial scaling of the Schwarzschild metric contributes $c_{h,j}^{3/2}$, its lapse scaling contributes $c_{N,j}$, and replacing $s$ by $s/m_j$ gives the principal term in~\eqref{AX-eq-outer-matching-exact-core-flux}. The Kottler coefficient change contributes $\Ocal(\Lambda s^2)$ to the flux of $m_j^{-2}w_j^0$. The difference $m_j^{-2}(w_j-w_j^0)=\Ocal(\Lambda s/m_j)$ contributes $\Ocal(\Lambda s^3/m_j)$; since $s\geq m_j$, this bounds both errors.

The divergence of $q_j$ contains its cutoff source and its bulk source. Coarea and the equality $\phi_\varepsilon=\phi$ on the cutoff annulus give
\be
\label{AX-eq-outer-matching-source-pairing}
\langle\delta_Gq_j,\phi_\varepsilon\rangle =-m_j^{-2}\int_{\rho_j}^{2\rho_j}\frac{\chi'(s/\rho_j)}{\rho_j} \int_{r_j=s}N_Gw_j(n_h,\phi)\,dA_h\,ds +\frac2{m_j}\int\chi_j\langle A_GU_j,\phi_\varepsilon\rangle\,d\nu_G.
\ee
The term proportional to $\Lambda$ has absolute value at most $C\Lambda\rho_j^3m_j^{-1}\|\phi\|_{C^2(U)}$, since $\phi_\varepsilon$ is uniformly bounded. On the cutoff annulus $s\asymp\rho_j$, the replacement of $\phi$ by $\phi(p_j)$ has error at most $Cs\|\phi\|_{C^1(U)}$. The flux estimate, $-\int_{\rho_j}^{2\rho_j}\chi'(s/\rho_j)\,ds/\rho_j=1$, $c_{h,j},c_{N,j}=1+\Ocal(\varepsilon^2)$ and $|p_j|\leq Cd$ therefore prove
\be
\label{AX-eq-outer-matching-source-Dirac-limit}
\bigl|\langle\delta_Gq_j,\phi_\varepsilon\rangle-8\pi\phi^z(0)\bigr| \leq C(d+\varepsilon^2+R_j^{-1}+\rho_j+\Lambda\rho_j^3/m_j)\|\phi\|_{C^2(U)}.
\ee
\ese

\bse
\emph{3. Duality and elliptic regularity.} We now take an arbitrary smooth spatial vector source orthogonal to the axial symmetry orbits $F_{\mathrm o}$ supported in a fixed region away from the central circle, and solve $A_D\phi=F_{\mathrm o}$ by the fixed de Sitter energy inverse of~\eqref{AX-eq-deSitter-outer-force-energy-inverse}. Fixed-metric elliptic estimates give $\|\phi\|_{C^2(U)}\leq C\|F_{\mathrm o}\|_{H^4(G_D)}$. The support of $F_{\mathrm o}$ lies where $G=G_D$ for sufficiently small $\varepsilon$. Selfadjointness, $A_GW_j=\delta_Gq_j$, and~\eqref{AX-eq-outer-matching-test-transport} then give
\be
\label{AX-eq-outer-matching-duality}
\langle W_j,F_{\mathrm o}\rangle_{G_D} =\langle\delta_Gq_j,\phi_\varepsilon\rangle -\langle W_j,\mathcal R_{\varepsilon,\phi}\rangle_G.
\ee
All horizon boundary terms vanish because both transported fields are regular in their reference lifts. The energy estimate~\eqref{AX-eq-force-projection-tail} bounds $\|\omega^{-1}W_j\|_2$ by $C\rho_j^{-1/2}$. Combining the last identity with~\eqref{AX-eq-outer-matching-test-dual-bound} and~\eqref{AX-eq-outer-matching-source-Dirac-limit} yields
\be
\label{AX-eq-outer-matching-distribution-rate}
\begin{aligned}
\bigl|\langle W_j,F_{\mathrm o}\rangle_{G_D}-8\pi\phi^z(0)\bigr|
&\leq C\Bigl(\sqrt{\frac{\mathfrak m}{\rho_j}}
+d+\varepsilon^2+R_j^{-1}+\rho_j+\Lambda\rho_j^3/m_j\Bigr)\|F_{\mathrm o}\|_{H^4(G_D)}\\
&\leq C\varepsilon^{3\zeta/2}\|F_{\mathrm o}\|_{H^4(G_D)}.
\end{aligned}
\ee
The powers are $\sqrt{\mathfrak m/\rho_j}=\Ocal(\varepsilon^{3\zeta/2})$, $R_j^{-1}=\Ocal(\varepsilon^{3\zeta})$, $\rho_j=\Ocal(\varepsilon^{3(1-\zeta)})$ and $d=\Ocal(\varepsilon)$. The additional power is $\Lambda\rho_j^3/m_j=\Ocal(\varepsilon^{6-9\zeta})=o(\varepsilon^{3\zeta/2})$, since $\zeta<4/9<4/7$. The distributional equation for $V_D$ gives $\langle V_D,F_{\mathrm o}\rangle=\phi^z(0)$, with the same normalization of $t_E$. Thus~\eqref{AX-eq-outer-matching-distribution-rate} is a local $H^{-4}$ bound for $W_j-8\pi V_D$. We enlarge the fixed test neighborhoods to be invariant under the two rotations and their reflections. Averaging test fields under these isometries then extends this bound from tests by spatial vector fields orthogonal to the axial symmetry orbits to arbitrary local tests.

On each fixed outer region, $A_D(W_j-8\pi V_D)=0$ for all sufficiently small $\varepsilon$. Interior elliptic regularity gives bounds of order $\varepsilon^{3\zeta/2}$ for every derivative on fixed smaller outer regions. We recall the estimate used in this passage: for nested fixed neighborhoods $K_0\Subset K_1$ and any real $s$, localization and the Fourier estimate for the frozen elliptic symbol give
\be
\label{AX-eq-outer-matching-elliptic-upgrade}
\|Y\|_{H^{s+2}(K_0)} \leq C_s\Bigl(\|A_DY\|_{H^s(K_1)}+\|Y\|_{H^{s+1}(K_1)}\Bigr).
\ee
Freezing the smooth coefficients on sufficiently small charts absorbs their principal oscillation; derivatives of the cutoffs and coefficients have lower order and give the last term. This argument applies also to negative $s$, since multiplication by a fixed smooth coefficient is bounded on the relevant Sobolev spaces. Starting with $s=-5$ and using a nested finite sequence of neighborhoods raises the $H^{-4}$ bound one derivative at a time. Sobolev embedding then gives every required H\"older bound. At a horizon or pole these are interior estimates in the regular lift. Finally $q_j=0$ and $G=G_D$ on these outer regions, so $\psi_j=-B_DW_j$ there. This proves~\eqref{AX-eq-global-columns-outer-limit}. The triangle inequality proves~\eqref{AX-eq-global-columns-outer-sum-separation}.
\ese
\end{proof}

The same outer limit does not make the two obstruction tensors dependent on the full spatial region. Their separate fluxes around the inner horizons remain independent by Proposition~\ref{AX-prop-global-force-columns}, and the Gram matrix~\eqref{AX-eq-force-column-Gram} remains nondegenerate. The de Sitter limit depends only on the sum of the two axial fluxes. Their difference remains detectable by the two separate flux integrals at the scale of the separation and in the Kottler collars. These statements require neither equal masses nor cancellation of the sum before the equation with two obstruction terms is solved.


\subsection{Solving the linear equation with two additional scalar unknowns}
\label{AX-section-global-force-inverse}

\paragraph{A vector inverse estimate at weights greater than minus one}
We prove surjectivity by adjoining the two obstruction tensors to the operator restricted by the inner orthogonality conditions. A weighted vector estimate controls the part introduced by the projection outside the cutoff supports, and the inverse bounds record the mass dependence needed in the nonlinear estimates. The exponent $-1$ is the decay rate of the flat point-source potential. Choosing a fixed weight just above this exponent gives an inverse estimate whose constant grows as an arbitrarily small negative power of the mass. This estimate applies to the vector equation; its consequence for the obstruction tensors retains the resulting power of the reference mass.

\begin{lemma}[A uniform inverse estimate for the operator on vector fields at weights greater than minus one]
\label{AX-lem-vector-above-force-weight}
For $0<\sigma_{\rm v}<1/2$ and $\beta\coloneqq-1+\sigma_{\rm v}$, the inverse of the operator on regular spatial vector fields orthogonal to the axial symmetry orbits on the Kottler--de Sitter references satisfies
\be
\label{AX-eq-vector-above-force-weight}
\|V\|_{k+2,\alpha,\beta;\omega} \leq C_{\sigma_{\rm v}}\|A_GV\|_{k,\alpha,\beta-2;\omega}.
\ee
The constant is independent of sufficiently small $\varepsilon$ and comparable positive masses. The regularity and symmetry conditions are those used in the uniform energy estimate for the auxiliary vector field.
\end{lemma}
\begin{proof}
\bse
The principal symbol of $A_G$ is $\tfrac12(|\xi|^2\operatorname{Id}+\xi\otimes\xi)$ in orthonormal coordinates, so it is uniformly positive. The coefficients on the scaled regular charts have uniform bounds. Away from a horizon we can unwrap the auxiliary coordinate $t_E$ when the chart radius exceeds a local period of $t_E$ at an inner horizon; invariance makes the estimates independent of that choice. Near a horizon we use the regular angular coordinates associated with that horizon. Thus these local estimates require no single global period of $t_E$. The interior estimates on these charts give
\be
\label{AX-eq-vector-uniform-local-estimate}
\|V\|_{k+2,\alpha,\beta;\omega} \leq C\Bigl(\|A_GV\|_{k,\alpha,\beta-2;\omega} +\|V\|_{0,0,\beta;\omega}\Bigr).
\ee
We justify the system estimate by its frozen principal operator. In a static interior chart its three-dimensional inverse is the Kelvin matrix computed in Section~\ref{AX-section-two-mass-perturbation}. In a regular four-dimensional lift, the fundamental solution matrix of $-\tfrac12(\Delta I+\nabla\operatorname{div})$ is
\be
\label{AX-eq-four-dimensional-vector-kernel}
\Gamma^{(4)}_{ij}(x)=\frac{3\delta_{ij}+2n_in_j}{8\pi^2|x|^2},\qquad n\coloneqq x/|x|.
\ee
Indeed, the inverse Fourier symbol is $2|\xi|^{-2}I-|\xi|^{-4}\xi\otimes\xi$. Equivalently, $\Gamma^{(4)}_{ij}=2\Phi\delta_{ij}+\del_i\del_j\Psi$, where $\Phi=1/(4\pi^2|x|^2)$ and $\Psi=-\log|x|/(8\pi^2)$ satisfy $-\Delta\Phi=\delta_0^{\rm pt}$ and $-\Delta\Psi=\Phi$, with the Dirac measure at the origin. Differentiating $\log|x|$ gives the fundamental solution matrix in~\eqref{AX-eq-four-dimensional-vector-kernel}. Its second derivatives have homogeneity $-4$ and zero annular mean; this cancellation follows by the divergence theorem and the equal homogeneity of the two boundary integrals. In the convolution for a second derivative, subtraction of the source value at the evaluation point leaves the integrable radial bound $r^{\alpha-1}$. To compare two evaluation points a distance $h$ apart, integration on $r\leq2h$ gives $Ch^\alpha$, while the derivative bound on this matrix on $r>2h$ gives $Ch\int_{2h}^1r^{\alpha-2}\,dr\leq Ch^\alpha$. These estimates prove the local H\"older bound for constant coefficients; the same argument works for the three-dimensional Kelvin matrix. Freezing coefficients on sufficiently small scaled charts and absorbing their oscillation gives the variable-coefficient estimate. Interpolation controls the lower-order terms. Regular invariant vectors are estimated on the complete local angular lift, so the horizons require no boundary estimate here.

We suppose that the desired bound fails along $\varepsilon_i\to0$. Normalizing the weighted norm to one and using~\eqref{AX-eq-vector-uniform-local-estimate}, we can choose points where the weighted $C^0$ norm is bounded below, while the source norms tend to zero. Local compactness on scaled charts gives a nonzero homogeneous limit. There are three possible locations for this limit.
\begin{itemize}
\item If the scale is comparable to one reference mass, the limit is a regular solution of $A_{G_S}V=0$ on the complete unit Schwarzschild model, with $|V|+r|\nabla V|\leq C(1+r)^\beta$. Integration by parts on a radius-$R$ exhaustion gives
\be
\label{AX-eq-Schwarzschild-vector-energy-decay}
\int_{\{r<R\}}|B_{G_S}V|^2 =\int_{\{r=R\}}(B_{G_S}V)(n,V), \qquad \bigl|\int_{\{r=R\}}(B_{G_S}V)(n,V)\bigr| \leq C R^{2\beta+1}\longrightarrow0.
\ee
The measure is taken per unit length in the model coordinate $\tau$. The fixed surface of the auxiliary circle action contributes no boundary term. Thus $B_{G_S}V=0$. Ricci flatness and the Killing equation give $\Delta|V|^2=2|\nabla V|^2\geq0$. Decay at infinity allows us to apply the maximum principle, which gives $V=0$. This contradicts its normalization.
\item If the scale tends to zero but is large compared with the reference masses, the limit is the flat static vector operator on $\mathbb R^3$ with either one or two points removed. A second center may escape to infinity or coincide with the first; both cases give a single excluded point. Near every finite excluded point, the estimate is $|V|+r|DV|\leq Cr^\beta$, and $\beta>-1$ makes all boundary terms in the distributional equation tend to zero, at largest order $r^{\beta+1}$. Hence $A_0V=0$ on all of $\mathbb R^3$. The limiting weighted bound gives decay at infinity. Taking the divergence yields $\Delta\operatorname{div}V=0$, whose bounded decaying solution is zero. The vector equation then gives $\Delta V=0$ componentwise, so the mean-value property and decay give $V=0$.
\item If the scale stays bounded below, the limit is the regular de Sitter vector equation off its central circle. Its codimension is three, and the same boundary calculation gives removability because $\beta>-1$. The resulting smooth global field satisfies $A_{G_D}V=0$. Integration by parts gives $B_{G_D}V=0$, and the argument excluding Killing fields orthogonal to the axial symmetry orbits in Lemma~\ref{AX-lem-deSitter-outer-force} gives $V=0$.
\end{itemize}
These cases exhaust the ratios of the local scale to the reference masses and the separation. In the intermediate case, the scaled cutoff coefficients tend to the flat coefficients, since $m/\omega\to0$ and $d^2+\omega^2\to0$; a cutoff transition produces no additional limit. In the outer case the metric is exactly de Sitter on every fixed region away from the shrinking Kottler collars. The nonzero normalized limit is therefore impossible in every case, which proves the estimate. The previously established weak inverse and fixed-parameter regularity give existence for regular data, so the estimate is an inverse bound rather than only a homogeneous uniqueness statement.
\ese
\end{proof}

For the localized tensors $q_j$ in~\eqref{AX-eq-global-force-tensor-construction}, the cutoff parameter $L$ may be fixed or may tend to infinity, provided the scale conditions in the construction hold. Its cutoff source has size $m_j^{-3}L^{-3}$, and the bulk divergence is bounded by $C\Lambda/m_j$. Consequently, we have
\be
\label{AX-eq-force-source-near-critical-weight}
\begin{aligned}
 \|\delta_Gq_j\|_{k,\alpha,\beta-2;\omega}
 &\leq C\bigl((m_jL)^{-\sigma_{\rm v}}+\Lambda m_j^{-1}(m_jL)^{3-\sigma_{\rm v}}\bigr)\\
 &\leq C(m_jL)^{-\sigma_{\rm v}},\\
 \|q_j\|_{k+1,\alpha,\gamma-2;\omega}
 &\leq C(m_jL)^{-\gamma}.
\end{aligned}
\ee
The same powers hold for scaled derivatives and H\"older seminorms, by the differentiated collar bounds and the fixed cutoff profile. The bulk source is absorbed because $\Lambda m_j^2L^3\to0$. Since $B_G$ lowers the weight by one, the preceding lemma gives a bound for $B_GW_j$ in weight $\beta-1=-2+\sigma_{\rm v}$. On the region $\omega\leq C$, conversion to weight $\gamma-2$ costs only the factor $\omega^{\sigma_{\rm v}-\gamma}\leq C$. We therefore obtain
\be
\label{AX-eq-force-column-small-weight-loss}
\|\psi_j\|_{k,\alpha,\gamma-2;\omega} \leq C\big((m_jL)^{-\gamma}+(m_jL)^{-\sigma_{\rm v}}\big) \leq C m_j^{-\sigma_{\rm v}}.
\ee
Thus an obstruction tensor with prescribed axial flux has norm at most $C\varepsilon^{-3\sigma_{\rm v}}$. Its contribution to the inverse of the linear operator with two additional scalar unknowns also depends on the centering functional and on the distinction between the total flux and the separation obstruction.

\begin{corollary}[Bounds for the divergence-free projection and the quadratic error]
\label{AX-cor-improved-projection-quadratic}
For $k\geq1$, the weighted projection satisfies
\be
\label{AX-eq-improved-weighted-projection-upper}
\|\Pi_GT\|_{k,\alpha,\gamma-2;\omega} \leq C_{\sigma_{\rm v}} m^{\gamma-\sigma_{\rm v}} \|T\|_{k,\alpha,\gamma-2;\omega}.
\ee
Its norm is therefore $\Ocal(\varepsilon^{-s_{\rm inv}})$. For the affine regular tensor neighborhood and the remainder $\mathcal Q_G$ defined above, one has
\be
\label{AX-eq-improved-projected-quadratic}
\|\Pi_G(\mathcal Q_G(v)-\mathcal Q_G(w))\|_{Y_\varepsilon} \leq C\varepsilon^{-s_{\rm quad}} (\|v\|_{X_\varepsilon}+\|w\|_{X_\varepsilon}) \|v-w\|_{X_\varepsilon}.
\ee
\end{corollary}
\begin{proof}
\bse
Since $\delta_G$ lowers the source weight by one, conversion from weight $\gamma-3$ to $\beta-2=\sigma_{\rm v}-3$ costs at most $Cm^{\gamma-\sigma_{\rm v}}$. After applying the uniform inverse of the operator on vector fields, $B_G$ maps to weight $\sigma_{\rm v}-2$, which embeds in weight $\gamma-2$ with a uniform constant because $\sigma_{\rm v}>\gamma$ and $\omega$ is bounded above. Explicitly,
\be
\label{AX-eq-projection-weight-conversion}
\begin{aligned}
 \|\delta_GT\|_{k-1,\alpha,\sigma_{\rm v}-3;\omega}
 &\leq C m^{\gamma-\sigma_{\rm v}}\|T\|_{k,\alpha,\gamma-2;\omega},\\
 \|B_GA_G^{-1}\delta_GT\|_{k,\alpha,\gamma-2;\omega}
 &\leq C_{\sigma_{\rm v}}\|\delta_GT\|_{k-1,\alpha,\sigma_{\rm v}-3;\omega}.
\end{aligned}
\ee
The formula $\Pi_G=I-B_GA_G^{-1}\delta_G$ proves the first assertion; its identity term is absorbed because $m^{\gamma-\sigma_{\rm v}}\geq1$ for small $m$. The source weight conversion uses the lower bound $\omega\geq cm$, and the output conversion uses the upper bound $\omega\leq C$, so their constants are uniform.

The unprojected quadratic estimate already proved in the affine tensor chart has loss $m^\gamma\asymp\varepsilon^{3\gamma}$. Applying the first assertion gives the total factor
\be
\label{AX-eq-quadratic-weight-conversion}
m^{\gamma-\sigma_{\rm v}}m^\gamma=m^{2\gamma-\sigma_{\rm v}}\asymp\varepsilon^{-s_{\rm quad}}.
\ee
This proves the second assertion. Applying the same projection estimate to the differentiated Bianchi bound improves the difference between the linearized Einstein tensor and its projection to $C\varepsilon^{p+3\gamma-s_{\rm inv}}\|v\|_{X_\varepsilon}=C\varepsilon^{p_{\rm cor}+3\gamma}\|v\|_{X_\varepsilon}$.
\ese
\end{proof}

\paragraph{Solvability after adding the two obstruction tensors}
The regular displacement fields on the Kottler collars determine the two divergence-free tensors adjoined to obtain a surjective operator. We use the tensors $q_j$ and $\psi_j=\Pi_Gq_j$ from~\eqref{AX-eq-global-force-tensor-construction}, with the growing cutoff radii $R_j=m_j^{-\zeta}$ of Lemma~\ref{AX-lem-global-columns-outer-limit} in place of a fixed $L$. The supports of $q_j$ are contained in the reference-coordinate balls $r_j\leq2m_jR_j=2m_j^{1-\zeta}$, which lie inside the exact Kottler collars for sufficiently small $\varepsilon$. The parameters $\gamma,\sigma_{\rm v},\zeta$ satisfy Table~\ref{AX-tab-admissible-exponents}. The two masses are comparable, and we write $m\asymp m_j\asymp\varepsilon^3$ in estimates which are unchanged by a fixed comparability factor. Pairings in this argument are divided by the auxiliary coordinate period.

\begin{theorem}[Existence, uniqueness and bounds for the linear equation with two additional unknowns]
\label{AX-thm-physical-force-inverse}
Fix $0<f_*<8\pi$. On the tensor domain defined by the Bianchi coordinate condition and the centering conditions on both Kottler collars, the operator
\be
\label{AX-eq-global-force-augmented-map}
\mathbb L_\varepsilon(v,\lambda) \coloneqq L_\varepsilon v+\lambda_+\psi_++\lambda_-\psi_-
\ee
is bijective onto $\mathcal Y_\varepsilon$ for every sufficiently small $\varepsilon$. Its solution for a source $f$ satisfies
\be
\label{AX-eq-global-force-inverse-loss}
\|v\|_{X_\varepsilon}\leq C\varepsilon^{-s_{\rm inv}}\|f\|_{Y_\varepsilon}, \qquad |\lambda|\leq C\varepsilon^{3\gamma}\|f\|_{Y_\varepsilon}.
\ee
The constants are uniform for comparable positive masses and for the fixed compact range of separations. For the same growing cutoffs, the matrix of axial fluxes of Proposition~\ref{AX-prop-global-force-columns} satisfies
\be
\label{AX-eq-physical-force-matrix-rate}
\|\mathsf F-8\pi I\|_{\mathbb R^2\to\mathbb R^2}\leq C\varepsilon^{3\zeta/2}, \qquad |\mathsf F\lambda|\geq f_*|\lambda|\quad(\lambda\in\mathbb R^2),
\ee
for sufficiently small $\varepsilon$. In the equal-mass reflection-invariant class, the corresponding operator with one obstruction coefficient is bijective with the same bounds. Replacing $\psi_+$ and $\psi_-$ by the linear combinations $\widehat\psi_j$ satisfying $\operatorname{Flux}_i(\widehat\psi_j)=\delta_{ij}$ preserves these inverse estimates.
\end{theorem}
\begin{proof}
\bse
\emph{1. A priori estimate and the Gram matrix.} We first estimate any solution of~\eqref{AX-eq-global-force-augmented-map}. The obstruction tensor estimate~\eqref{AX-eq-force-column-small-weight-loss} and the tensor bound under the inner orthogonality conditions give
\be
\label{AX-eq-force-inverse-first-bound}
\|v\|_{X_\varepsilon} \leq C\bigl(\|f\|_{Y_\varepsilon}+m^{-\sigma_{\rm v}}|\lambda|\bigr).
\ee
We next obtain a stronger estimate for the two coefficients by testing against the cut-off displacement tensors themselves. The trace-free property of $q_j$ gives $\langle\mathcal T_G\psi_i,q_j\rangle=\langle\psi_i,q_j\rangle$. Since $\Pi_G$ is orthogonal, the coefficient matrix is the positive Gram matrix
\be
\label{AX-eq-force-inverse-Gram}
\mathsf H_{ji}\coloneqq\langle\psi_i,q_j\rangle =\langle\Pi_Gq_i,\Pi_Gq_j\rangle.
\ee
The disjoint supports of $q_+$ and $q_-$, the inner rescaling and the projection-tail estimate~\eqref{AX-eq-force-projection-tail} imply
\be
\label{AX-eq-force-inverse-Gram-bounds}
\mathsf H_{jj}\geq c m_j^{-1}, \qquad |\mathsf H_{ji}|\leq C(m_im_j)^{-1/2}(R_iR_j)^{-1/2}\quad(i\ne j).
\ee
For the first inequality, $m_j\|q_j\|_2^2$ tends to the strictly positive $L^2$ norm of the displacement tensor per unit length in the model coordinate $\tau$, while $m_j\|B_GW_j\|_2^2\leq C/R_j$. For the second, orthogonality gives $\mathsf H_{ji}=-\langle B_GW_i,B_GW_j\rangle$. Thus the matrix obtained by multiplying row $j$ by $m_j$ has a uniformly bounded inverse. We use comparability of the masses when expressing this statement with a single $m$.
\ese

\bse
\emph{2. Pairings with the localized displacement tensors.} The linearized Bianchi estimates in~\eqref{AX-eq-small-projection-on-linear-image} and~\eqref{AX-eq-reference-times-tensor} rewrite the equation as
\be
\label{AX-eq-force-inverse-differential-equation}
\mathcal P_Gv=\mathcal T_Gf-\sum_i\lambda_i\mathcal T_G\psi_i+E_v, \qquad \|E_v\|_{Y_\varepsilon}\leq C\varepsilon^{p_{\rm lin}}\|v\|_{X_\varepsilon}.
\ee
The projection acts as the identity on the prescribed obstruction tensors in this equation because they are divergence-free. Pairing with $q_j$ therefore gives the same positive matrix~\eqref{AX-eq-force-inverse-Gram}.

We estimate the other pairings explicitly. The tensor $q_j$ has size at most $Cr_j^{-2}$, with the corresponding scaled derivatives, outside a fixed multiple of $m_j$, and size at most $Cm_j^{-2}$ in the remaining fixed scaled collar. Since $\gamma<1$, integration of $r_j^{\gamma-2}$ from $m_j$ to $2m_jR_j$ gives
\be
\label{AX-eq-force-inverse-source-pairing}
|\langle f,q_j\rangle| \leq C m_j^{\gamma-1}\|f\|_{Y_\varepsilon}.
\ee
The metric measure per unit length in $t_E$ is comparable to $r_j^2\,dr_j\,d\Omega$ away from the fixed surface of the auxiliary circle action, and the scaled integral over a compact rescaled Kottler collar has the same power. The estimate also applies to $E_v$ using~\eqref{AX-eq-force-inverse-differential-equation}.

For the differential pairing, formal selfadjointness and the Einstein equation on the exact Kottler collar give
\be
\label{AX-eq-force-inverse-kernel-test}
\begin{gathered}
\langle\mathcal P_Gv,q_j\rangle=\langle v,\mathcal P_Gq_j\rangle,\qquad
\chi_j\coloneqq\chi(r_j/(m_jR_j)),\\
\mathcal P_Gq_j=m_j^{-2}[\mathcal P_G,\chi_j]w_j
+\frac2{m_j}\chi_jB_G(A_GU_j)^\sharp.
\end{gathered}
\ee
There is no outer boundary term because $q_j$ is compactly supported in the exact Kottler collar, and no horizon term in the regular lift. The commutator is supported where $r_j\asymp m_jR_j$ and has size at most $Cr_j^{-4}$. Indeed, its two terms contain $m_j^{-2}\nabla\chi_j\nabla w_j$ and $m_j^{-2}(\nabla^2\chi_j)w_j$, and $|w_j|\leq Cm_j^2r_j^{-2}$. The tensor variation has size at most $C r_j^\gamma\|v\|_X$ there. The annular volume is of order $r_j^3$, so the commutator pairing is bounded by $C(m_jR_j)^{\gamma-1}\|v\|_X$. The term in the region where the cutoff is constant follows from covariance: $D(\Ric-\Lambda g)_G[w_j]=0$, while $\operatorname{tr}_Gw_j=0$ gives $\mathcal P_Gw_j=2m_jB_G(A_GU_j)^\sharp$. Thus it has size at most $C\Lambda/(m_jr_j)$. Its pairing is bounded by $C\Lambda m_j^{-1}(m_jR_j)^{\gamma+2}\|v\|_X$, because $\gamma>-1$. We have proved
\be
\label{AX-eq-force-inverse-differential-pairing}
|\langle\mathcal P_Gv,q_j\rangle| \leq C\bigl(m_j^{\gamma-1}R_j^{\gamma-1} +\Lambda m_j^{\gamma+1}R_j^{\gamma+2}\bigr)\|v\|_{X_\varepsilon}.
\ee
\ese

\bse
\emph{3. Coefficient bounds and absorption.} We invert the Gram matrix. Writing $a\coloneqq m^{-\gamma}|\lambda|$ and $R\coloneqq m^{-\zeta}$, the preceding estimates give
\be
\label{AX-eq-force-inverse-coefficient-bound}
a\leq C\Bigl(\|f\|_{Y_\varepsilon} +\bigl(R^{\gamma-1}+\Lambda m^2R^{\gamma+2}+\varepsilon^{p_{\rm lin}}\bigr) \|v\|_{X_\varepsilon}\Bigr).
\ee
This estimate is the point where the Schwarzschild displacement mode controls the obstruction coefficients. Substitution into~\eqref{AX-eq-force-inverse-first-bound} leaves an absorbable multiple of $\|v\|_X$. In fact, $R=m^{-\zeta}$ gives
\be
\label{AX-eq-force-inverse-absorption-powers}
\begin{aligned}
&m^{\gamma-\sigma_{\rm v}}\asymp\varepsilon^{-s_{\rm inv}},\qquad m^{\gamma-\sigma_{\rm v}}R^{\gamma-1}=m^{\gamma-\sigma_{\rm v}+\zeta(1-\gamma)},\\
&m^{\gamma-\sigma_{\rm v}}\Lambda m^2R^{\gamma+2}=\Lambda m^{\gamma-\sigma_{\rm v}+2-\zeta(2+\gamma)},\qquad m^{\gamma-\sigma_{\rm v}}\varepsilon^{p_{\rm lin}}\asymp\varepsilon^{p_{\rm lin}-s_{\rm inv}}.
\end{aligned}
\ee
The three coefficients multiplying the correction norm after substitution tend to zero. For the new cosmological term, the admissible inequalities imply $\gamma-\sigma_{\rm v}+2-\zeta(2+\gamma)>44/45-56\gamma/45>0$. After absorption, the estimate reads $\|v\|_X\leq C \varepsilon^{-s_{\rm inv}}\|f\|_Y$, as required for the first bound in~\eqref{AX-eq-global-force-inverse-loss}. Substitution back into~\eqref{AX-eq-force-inverse-coefficient-bound} gives $a\leq C\|f\|_Y$. Therefore $|\lambda|\leq C m^\gamma\|f\|_Y\leq C\varepsilon^{3\gamma}\|f\|_Y$, proving the second bound.
\ese

\bse
\emph{4. Surjectivity and normalization of the flux matrix.} It remains to establish surjectivity. Proposition~\ref{AX-prop-Schwarzschild-kernel-classification} gives $n_S=1$, and Proposition~\ref{AX-prop-centered-index} gives index $-2$ for the tensor operator under the inner orthogonality conditions. Extend $L_\varepsilon$ to $\mathcal D_\varepsilon^0\times\mathbb R^2$ by $(v,\lambda)\mapsto L_\varepsilon v$. This adds $\mathbb R^2$ to the kernel and leaves the range unchanged, so the index goes from $-2$ to zero. The added map $(v,\lambda)\mapsto\lambda_+\psi_++\lambda_-\psi_-$ has finite rank. The compact-factorization argument of Proposition~\ref{AX-prop-centered-index} therefore leaves the index equal to zero. The estimate already proved makes the kernel of the linear operator with two additional scalar unknowns zero. Its index is zero, so its cokernel is zero, proving surjectivity. The solution is therefore unique, and the preceding estimates give its bounded inverse. Restricting to the reflection class gives one independent paired centering and one paired obstruction tensor; the same argument applies with index $-1$ before augmentation.

We also verify the quantitative normalization of the axial fluxes for these growing cutoffs. The displacement components in~\eqref{AX-eq-Schwarzschild-displacement-components} give a tail integral bounded by $C/R_j$ per unit length in the model coordinate $\tau$. Thus the integral $a_{R_j}$ in~\eqref{AX-eq-force-column-diagonal-limit} satisfies $|a_{R_j}-16\pi|\leq C/R_j$. The additional terms in~\eqref{AX-eq-force-column-diagonal-limit} satisfy $m_j^2R_j^{5/2}=o(R_j^{-1})$ and $m_j^4R_j^5=o(R_j^{-1})$, since $\zeta<4/9<4/7$ and $\zeta<2/3$. The projection-tail estimate and $c_{h,j},c_{N,j}=1+\Ocal(\varepsilon^2)$ therefore give
\be
\label{AX-eq-growing-cutoff-Gram-normalization}
m_j\mathsf H_{jj}=16\pi+\Ocal(R_j^{-1}+\varepsilon^2), \qquad |m_j\mathsf H_{ji}| \leq C(R_iR_j)^{-1/2}\quad(i\ne j).
\ee
Here the masses remain comparable, as required in the off-diagonal estimate. The exact pairing-to-flux identity and its transport estimates, proved in Proposition~\ref{AX-prop-global-force-columns}, give
\be
\label{AX-eq-growing-cutoff-force-normalization}
\mathsf F_{ji} =\frac{m_j}{2}\mathsf H_{ji} +\Ocal(R_j^{-1/2}+\varepsilon).
\ee
Since $R_j=m_j^{-\zeta}\asymp\varepsilon^{-3\zeta}$, the diagonal error is $\Ocal(\varepsilon^{3\zeta}+\varepsilon^2)$ and the transport error is $\Ocal(\varepsilon^{3\zeta/2}+\varepsilon)$. The restriction $0<\zeta<4/9$ makes both $\Ocal(\varepsilon^{3\zeta/2})$. This proves~\eqref{AX-eq-physical-force-matrix-rate}. For sufficiently small $\varepsilon$, the matrix error is at most $8\pi-f_*$, giving $|\mathsf F\lambda|\geq f_*|\lambda|$ and $\|\mathsf F^{-1}\|\leq f_*^{-1}$.
\ese
\end{proof}

\begin{corollary}[The inverse used to correct the Kottler reference metric]
\label{AX-cor-kottler-inverse}
Fix the improved Kottler--de Sitter reference $G=\widetilde G_\varepsilon$. Use the coordinate condition $\mathfrak b_Gw=0$, the two centerings~\eqref{AX-eq-complete-core-centering}, the target $\ker\delta_G$ and the obstruction tensors $\psi_j$. The operator
\be
\label{AX-eq-kottler-linearization}
L_K\coloneqq L_\varepsilon=\Pi_GDS_3(G),\qquad (w,\lambda)\longmapsto L_Kw-\sum_j\lambda_j\psi_j
\ee
is an isomorphism onto that target. Denote its inverse components by $\mathcal R_G^w,\mathcal R_G^\lambda$. They satisfy
\be
\label{AX-eq-kottler-inverse}
\|\mathcal R_G^wF\|_X\leq C\varepsilon^{-s_{\rm inv}}\|F\|_Y, \qquad |\mathcal R_G^\lambda F|\leq C\varepsilon^{3\gamma}\|F\|_Y.
\ee
\end{corollary}
\begin{proof}
Theorem~\ref{AX-thm-physical-force-inverse} applies directly to $G$. Changing the sign of the coefficient component gives the inverse in~\eqref{AX-eq-kottler-inverse}. In particular, we have
\be
\label{AX-eq-kottler-operator-perturbation}
L_K\mathcal R_G^wF-\sum_j(\mathcal R_G^\lambda F)_j\psi_j=F.
\ee
All operators, regular angular lifts and coordinate conditions in this identity refer to the same global metric $G$.
\end{proof}

\paragraph{Horizon angular coordinates and the space of regular tensors}
The tensor domain in Theorem~\ref{AX-thm-physical-force-inverse} uses the fixed regular angular lifts of the reference. Varying a surface gravity therefore requires a corresponding change of these identifications.

This restriction has a direct tensor test. In spatial Gaussian distance $s$, a horizon with surface gravity $\kappa$ has leading normal metric $ds^2+\kappa^2s^2dt_E^2$. In the fixed reference coordinates $\theta_{E,0}\coloneqq\kappa_0t_E$, $x=s\cos\theta_{E,0}$, $y=s\sin\theta_{E,0}$, a variation of surface gravity with the Gaussian radial coefficient fixed gives
\be
\label{AX-eq-horizon-angle-tensor-test}
\left.\frac{d}{da}\right|_{a=0} \Bigl(ds^2+\frac{(\kappa_0+a\dot\kappa)^2}{\kappa_0^2}s^2d\theta_{E,0}^2\Bigr) =\frac{2\dot\kappa}{\kappa_0} \frac{(-y\,dx+x\,dy)^2}{x^2+y^2}.
\ee
We used $s^2d\theta_{E,0}^2=(-y\,dx+x\,dy)^2/(x^2+y^2)$. If $\dot\kappa\neq0$, the coefficient of $dy^2$ has different limits along $y=0$ and $x=0$, so this variation is not continuous at the origin. Thus smoothness in this fixed tensor lift prescribes the leading angular coefficient. A parametrization allowing independent horizon gravities must instead separate these variations from the smooth tensor unknown, or identify parameter-dependent regular lifts before differentiating.


\section{Nonlinear existence of the static solutions}
\label{AX-section-comparable-mass-existence}


\subsection{Selecting the barycenter and the separation}
\label{AX-section-two-parameter-selection}

We solve the reduced nonlinear Einstein equation and select the barycenter and separation so that both obstruction coefficients vanish. The construction applies to every sufficiently small pair of comparable positive reference masses, including unequal masses. We select the barycenter and separation from a continuous map whose leading term is explicit. Opposite signs on the faces of a parameter square give an interior zero without a differentiated estimate.

We recall the leading map $F_0(C,D)=(C,D-D^{-2})$ from~\eqref{AX-eq-two-mass-balance-map}. We use the maximum norm on $\mathbb R^2$ and its induced matrix norm, and define the parameter square $\mathcal Q\coloneqq\{(C,D):|C|\leq r_0,\ |D-1|\leq r_0\}$, where $r_0\in(0,1)$ is fixed. We choose $q_0\in(0,1)$ and require $C_h r_0<q_0$, with $C_h\coloneqq2(1-r_0)^{-4}$.
\begin{lemma}[Existence of a nearby zero of the perturbed compatibility equations]
\label{AX-prop-two-mass-balance}
Fix $0<\theta_0\leq\theta\leq1-\theta_0$. Assume that the continuous maps $F_{\varepsilon_{\mathrm{tot}},\theta}$ on $\mathcal Q$ satisfy the following estimate with fixed constants $A>0$ and $\tau>0$:
\be
\label{AX-eq-two-mass-error}
F_{\varepsilon_{\mathrm{tot}},\theta}=F_0+E_{\varepsilon_{\mathrm{tot}},\theta},\qquad \sup_{\mathcal Q}|E_{\varepsilon_{\mathrm{tot}},\theta}|\leq A\varepsilon_{\mathrm{tot}}^{\tau}.
\ee
For all sufficiently small $\varepsilon_{\mathrm{tot}}>0$, the following conclusions hold:
\begin{itemize}
\item There is a zero $(C_{\varepsilon_{\mathrm{tot}},\theta},D_{\varepsilon_{\mathrm{tot}},\theta})$ of the map in the interior of $\mathcal Q$. Every zero in $\mathcal Q$ satisfies
\be
\label{AX-eq-two-mass-root-bound}
|C_{\varepsilon_{\mathrm{tot}},\theta}|+|D_{\varepsilon_{\mathrm{tot}},\theta}-1| \leq2A\varepsilon_{\mathrm{tot}}^{\tau},\qquad z_{\rm bar}=\Ocal(\varepsilon_{\mathrm{tot}}^{1+\tau}\ell),\qquad d=(\mathfrak m\ell^2)^{1/3}\bigl(1+\Ocal(\varepsilon_{\mathrm{tot}}^{\tau})\bigr).
\ee
The constants are uniform in the prescribed mass interval.
\item If the maps are $C^1$ and also satisfy $\sup_{\mathcal Q}\|DE_{\varepsilon_{\mathrm{tot}},\theta}\|\leq A\varepsilon_{\mathrm{tot}}^{\tau}$, this zero is unique in $\mathcal Q$. Joint continuity of the maps in the parameters gives a continuous selected zero.
\end{itemize}
\end{lemma}
Figure~\ref{AX-fig-parameter-square} uses the abbreviation $F=(F_1,F_2)$ for $F_{\varepsilon_{\mathrm{tot}},\theta}$. The arrows show the strict sign inequalities used in the proof.

\begin{figure}[!htbp]
\centering
\begin{tikzpicture}[
  x=2.2cm,y=2.2cm,font=\small,
  line cap=round,line join=round,
  >={Stealth[length=2.0mm,width=1.5mm]}
]
\path[use as bounding box] (-2.45,-1.55) rectangle (2.45,1.55);
\fill[blue!3] (-1,-1) rectangle (1,1);
\draw[black!75,line width=.7pt] (-1,-1) rectangle (1,1);
\node[anchor=north west,inner sep=1pt] at (-.86,.88) {$\mathcal Q$};

\draw[->,black!50,line width=.4pt] (-1.70,0) -- (1.70,0)
  node[anchor=west,black,inner sep=3pt] {$C$};
\draw[->,black!50,line width=.4pt] (0,-1.35) -- (0,1.37)
  node[anchor=south,black,inner sep=2pt] {$D$};
\draw[black!65,line width=.4pt] (-1,-.035)--(-1,.035) (1,-.035)--(1,.035)
  (-.035,-1)--(.035,-1) (-.035,1)--(.035,1);
\node[below,fill=white,inner sep=1.5pt] at (-1,-.045) {$-r_0$};
\node[below,fill=white,inner sep=1.5pt] at (1,-.045) {$r_0$};
\node[left,fill=white,inner sep=1.5pt] at (-.045,1) {$1+r_0$};
\node[left,fill=white,inner sep=1.5pt] at (-.045,-1) {$1-r_0$};

\draw[->,blue!65!black,line width=.8pt] (-1,.55)--(-1.33,.55);
\node[left,blue!65!black,inner sep=3pt] at (-1.36,.55) {$F_1<0$};
\draw[->,blue!65!black,line width=.8pt] (1,.55)--(1.33,.55);
\node[right,blue!65!black,inner sep=3pt] at (1.36,.55) {$F_1>0$};
\draw[->,blue!65!black,line width=.8pt] (.60,1)--(.60,1.28);
\node[above,blue!65!black,inner sep=3pt] at (.60,1.29) {$F_2>0$};
\draw[->,blue!65!black,line width=.8pt] (.60,-1)--(.60,-1.28);
\node[below,blue!65!black,inner sep=3pt] at (.60,-1.29) {$F_2<0$};

\draw[black,line width=.7pt,fill=white] (0,0) circle[radius=2.25pt];
\node[anchor=north east,align=right,inner sep=1pt] at (-.10,-.10)
  {$(0,1)$\\$F_0=0$};
\fill[red!70!black] (.25,-.25) circle[radius=2.25pt];
\draw[red!70!black,line width=.4pt] (.28,-.27)--(.97,-.446);
\draw[red!70!black,line width=.4pt] (1.03,-.462)--(1.18,-.50);
\node[anchor=west,align=left,inner sep=2pt,fill=white,text=red!70!black]
  at (1.22,-.52) {a zero of $F$\\(schematic)};
\end{tikzpicture}
\caption{Selection of the barycenter and separation. The arrow on a face of $\mathcal Q$ gives the sign of the corresponding component of $F$. These strict signs ensure the existence of an interior zero. The open point represents the zero $(0,1)$ of $F_0$. The solid point indicates a possible location of a zero of $F$.}
\label{AX-fig-parameter-square}
\end{figure}
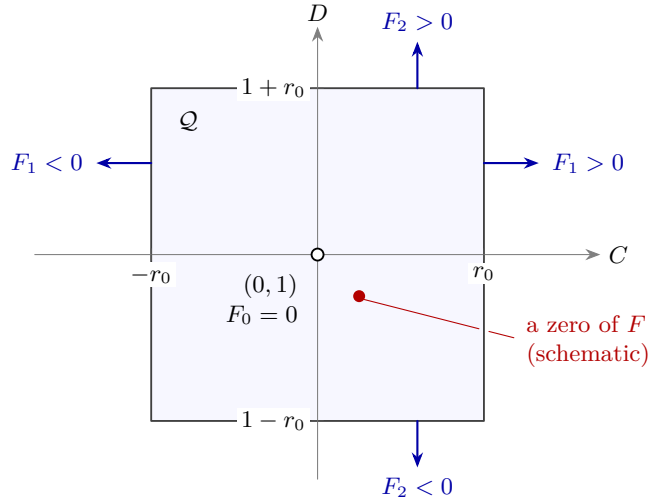

\begin{proof}
\bse
We set $e_\varepsilon\coloneqq A\varepsilon_{\mathrm{tot}}^{\tau}$ and use the fixed half-width $r_0$ of $\mathcal Q$. The second component of $F_0$ is $h_b(D)\coloneqq D-D^{-2}$, with $h_b(1)=0$ and $h_b'(D)=1+2D^{-3}\geq1$. On the two faces $C=\pm r_0$, the first component of $F_{\varepsilon_{\mathrm{tot}},\theta}$ has the corresponding strict sign when $e_\varepsilon<r_0$. On the faces $D=1\pm r_0$, its second component also has the corresponding strict sign, because $|h_b(1\pm r_0)|\geq r_0$. These face signs give a zero by the following planar argument, whose proof also covers continuous maps.

We argue by contradiction and assume that no zero exists. Compactness then gives $\min_{\mathcal Q}|F_{\varepsilon_{\mathrm{tot}},\theta}|>0$. We extend the map to a neighborhood of $\mathcal Q$ by projecting each coordinate to its closed interval, and we smooth this extension by convolution. A sufficiently close smooth approximation $\widetilde F$ is nonzero throughout $\mathcal Q$ and retains the strict face signs. On $\mathbb R^2\setminus\{0\}$, we define the closed one-form
\be
\label{AX-eq-two-mass-winding-form}
\omega_{\rm ang}\coloneqq\frac{y_1\,dy_2-y_2\,dy_1}{y_1^2+y_2^2},\qquad d\omega_{\rm ang}=0.
\ee
The derivative identity follows by differentiating the numerator and denominator. Stokes' theorem gives $\int_{\del\mathcal Q}\widetilde F^*\omega_{\rm ang}=0$, because $\widetilde F$ is nonzero on all of $\mathcal Q$. On the boundary, the interpolation $(1-s)F_0+s\widetilde F$ is nonzero for every $0\leq s\leq1$: on each face its indicated component retains the same strict sign. Applying Stokes' theorem to this interpolation on the boundary cylinder gives
\be
\label{AX-eq-two-mass-winding-contradiction}
0=\int_{\del\mathcal Q}\widetilde F^*\omega_{\rm ang} =\int_{\del\mathcal Q}F_0^*\omega_{\rm ang}=2\pi.
\ee
For the last equality, $F_0$ maps the oriented boundary to that of the rectangle $[-r_0,r_0]\times[h_b(1-r_0),h_b(1+r_0)]$. This rectangle surrounds the origin, and the polar angle increases once through $2\pi$ along its positively oriented boundary; equivalently, integration of the one-form on its four segments gives that value. The contradiction proves existence. The strict face signs exclude boundary zeros.

At any zero, the error estimate gives $|C|\leq e_\varepsilon$ and $|h_b(D)|\leq e_\varepsilon$. The lower derivative bound on $h_b$ gives $|D-1|\leq e_\varepsilon$, proving the bound for the selected parameters in~\eqref{AX-eq-two-mass-root-bound}. Multiplication by $\varepsilon_{\mathrm{tot}}\ell$ and the identity $\varepsilon_{\mathrm{tot}}\ell=(\mathfrak m\ell^2)^{1/3}$ give the estimates for $z_{\rm bar}$ and $d$.

Under the additional derivative hypothesis, we consider $\mathcal T(v)\coloneqq v-J_0^{-1}F_{\varepsilon_{\mathrm{tot}},\theta}(v)$, with $v_0=(0,1)$ and $J_0=\operatorname{diag}(1,3)$ as in~\eqref{AX-eq-two-mass-balance-map}. Its fixed points are the zeros. For $|D-1|\leq r\leq r_0$, the mean value theorem gives
\be
\label{AX-eq-two-mass-contraction-derivative}
\begin{aligned}
\|I-J_0^{-1}DF_0(C,D)\|
&=\frac23|D^{-3}-1|
\leq C_h|D-1|\leq C_hr,\\
|\mathcal T(v_0)-v_0|&\leq e_\varepsilon,\qquad
\sup_{|v-v_0|\leq r}\|D\mathcal T(v)\|\leq C_hr+e_\varepsilon.
\end{aligned}
\ee
For $e_\varepsilon\leq\min\{(1-q_0)r_0,q_0-C_hr_0\}$, the radius-$R_\varepsilon$ square centered at $v_0$, with $R_\varepsilon\coloneqq e_\varepsilon/(1-q_0)$, lies in $\mathcal Q$. The derivative bound on the full square is at most $q_0<1$, and
\be
\label{AX-eq-two-mass-contraction-ball}
|\mathcal T(v)-v_0|\leq e_\varepsilon+q_0|v-v_0|\leq R_\varepsilon.
\ee
Iteration converges in this square. The contraction bound on $\mathcal Q$ gives uniqueness there. If the parameters vary continuously, the same bound estimates the difference between their fixed points by $1/(1-q_0)$ times the difference of the corresponding maps, evaluated at either fixed point. Joint continuity on the compact square therefore gives continuity of the selected zero. All constants are uniform in $\theta$ because the normalization~\eqref{AX-eq-two-mass-leading-fluxes} is uniformly invertible on the prescribed positive mass interval.
\ese
\end{proof}


\subsection{Construction of the exact static solutions}
\label{AX-section-unconditional-existence}

We correct the Kottler--de Sitter reference with the inverse of the linear operator with two additional scalar unknowns and then select the two geometric parameters. Vanishing axial fluxes eliminate the obstruction coefficients, and the vector energy identity gives the Einstein equation. We use the Kottler reference of Lemma~\ref{AX-lem-kottler-comparison}, the global obstruction tensors of Proposition~\ref{AX-prop-global-force-columns}, and the inverse of the linear operator with two additional scalar unknowns of Corollary~\ref{AX-cor-kottler-inverse}. We write the corrected auxiliary metric as the sum of the reference metric and a symmetric tensor. Smoothness of this tensor in the Cartesian coordinates associated with the reference angular coordinates preserves the three reference surface gravities; the reference masses and geometric parameters remain free parameters of the construction.

\begin{theorem}[Existence of axisymmetric static solutions with two comparable small masses]
\label{AX-thm-binary-existence}
Fix $\Lambda>0$ and $0<\theta_0\leq1/2$, and set $\ell=\sqrt{3/\Lambda}$. There exists $\varepsilon_0>0$ such that, whenever $m_+,m_->0$ satisfy
\be
\label{AX-eq-binary-existence-parameters}
\mathfrak m= m_++m_-,\qquad \varepsilon_{\mathrm{tot}}=(\mathfrak m/\ell)^{1/3}<\varepsilon_0, \qquad \theta= m_+/\mathfrak m\in[\theta_0,1-\theta_0],
\ee
One obtains a smooth axisymmetric static solution on a compact spatial region diffeomorphic to a closed ball with two disjoint open balls removed. Its lapse is positive in the interior and vanishes exactly on its three boundary components. All three horizons are nondegenerate, the axis is regular, and each horizon admits a regular local extension. The conclusions are as follows:
\begin{itemize}
\item In the reference coordinates, the two inner centers have separation $d$ and mass barycenter $z_{\rm bar}$ satisfying
\be
\label{AX-eq-binary-existence-balance}
d=(\mathfrak m\ell^2)^{1/3}\bigl(1+\Ocal(\varepsilon_{\mathrm{tot}}^{p_{\rm bal}})\bigr), \qquad z_{\rm bar}=\Ocal(\ell\varepsilon_{\mathrm{tot}}^{1+p_{\rm bal}}).
\ee
\item Relative to the selected Kottler--de Sitter reference $G$, the auxiliary metric has the form $g=G+w$. The norms, coordinate condition, projection and correction sizes are all defined using $G$. One has
\be
\label{AX-eq-binary-existence-correction}
\|w\|_{X_{\varepsilon_{\mathrm{tot}}}}\leq C\varepsilon_{\mathrm{tot}}^{p_{\rm cor}}, \qquad \eta\leq C\varepsilon_{\mathrm{tot}}^{p_{\rm cor}+3\gamma}, \qquad \eta_{\rm sep}\leq C\varepsilon_{\mathrm{tot}}^{p_{\rm cor}+\gamma}.
\ee
\item If $A_+,A_-$ denote the inner horizon areas, then $A_j=16\pi m_j^2(1+\Ocal(\varepsilon_{\mathrm{tot}}^2))$. The parameters $m_j$ specify the reference scales; these area estimates identify their geometric meaning. For $m_+=m_-$, the construction can also impose the reflection interchanging the two inner horizons.
\end{itemize}
The exponents are defined in~\eqref{AX-eq-exponent-relations}. The constants may depend on the fixed admissible triple and are uniform on the prescribed interval of mass fractions. A constant rescaling of the lapse can impose $N_{\max}=1$.
\end{theorem}

The smallness threshold is chosen after the admissible exponents, the compact mass-fraction range and the fixed geometric margins. The estimates below involve a finite number of strictly positive error and absorption exponents, so all required inequalities hold for sufficiently small $\varepsilon_{\mathrm{tot}}$. Section~\ref{AX-section-fixed-reference-uniqueness} distinguishes uniqueness of the projected correction from the still open uniqueness of the selected barycenter and separation.

\begin{proof}
\bse
\emph{1. Regular tensor domain and projected estimates.} We work first with $\ell=1$ and abbreviate $\varepsilon_{\mathrm{tot}}$ to $\varepsilon$. The geometric parameters $(C,D)$ range over the fixed square $\mathcal Q$ of Lemma~\ref{AX-prop-two-mass-balance}, and the centers are $p_\pm=\varepsilon q_\pm(C,D)$. We write $G=\widetilde G_\varepsilon$ for the improved Kottler--de Sitter reference. The reference has exact Kottler fields on the two inner collars and exact de Sitter fields near the outer horizon. Its horizon parameters are given by~\eqref{AX-eq-kottler-horizon-parameters}--\eqref{AX-eq-kottler-fields}; the matching constants $c_{h,j},c_{N,j}$ are $1+\Ocal(\varepsilon^2)$ uniformly on the parameter square. The improved inner and outer transition radii are $r_{\rm in}=\varepsilon^{5/3}$ and $r_{\rm out}=\varepsilon^{3/5}$. Since $r_{\rm in}/d\to0$ and $d/r_{\rm out}\to0$, the inner cutoffs vanish on both enlarged separation annuli and the outer cutoff equals one there. The reference is therefore exactly the original intermediate metric on these annuli, by Lemma~\ref{AX-lem-improved-reference-error}; its axial flux calculation is unchanged.

We specify the nonlinear domain. Its elements $w$ satisfy $\mathfrak b_Gw=0$ and the centering conditions~\eqref{AX-eq-complete-core-centering} on the two inner Kottler collars, with the prescribed static and axial symmetries. We use the affine auxiliary metric $g\coloneqq G+w$. Its relative change obeys $\eta\leq C\varepsilon^{3\gamma}\|w\|_X$ by~\eqref{AX-eq-correction-comparison}. If it is small, this metric is positive definite in every reference lift. Near a horizon, the auxiliary circle action is the rotation of the two-dimensional plane normal to $H_j$ in the reference Riemannian extension, with angle $\theta_{E,j}=\kappa_jt_E$. A smooth invariant tensor preserves this action and its fixed set. More explicitly, regularity and rotation invariance give
\be
\label{AX-eq-affine-horizon-gravity}
g_{\theta_{E,j}\theta_{E,j}}=s^2a(s^2,x),\qquad h_{ss}=a(s^2,x)+\Ocal(s^2),\qquad h_{sa}=\Ocal(s),\qquad N=\kappa_js\sqrt{a(s^2,x)},
\ee
where $a(0,x)>0$. The factor $\kappa_j$ in the lapse comes from $d\theta_{E,j}=\kappa_jdt_E$. Since $h^{ss}=a(0,x)^{-1}$ at $s=0$, the last formula gives $|dN|_h=\kappa_j$ on that horizon. Thus the affine domain preserves its positive surface gravity as well as its fixed set. The same argument in the axial plane preserves the regular axis. At a pole, simultaneous invariance under these two rotations gives smoothness in the product of their Cartesian coordinates. Away from the fixed sets, small relative change of the coefficient of $dt_E^2$ preserves its positivity. Consequently the affine domain preserves exactly the three horizons and all the required horizon, axis and pole regularity conditions.

We solve the reduced Einstein equation with two undetermined obstruction coefficients. The projection, tensor domain, obstruction tensors, linearization and nonlinear remainder are all defined by $G$. We write
\be
\label{AX-eq-binary-kottler-remainder}
E_\varepsilon\coloneqq\Pi_GS_3(G)=S_3(G),\qquad \mathcal Q_G(w)\coloneqq S_3(G+w)-S_3(G)-DS_3(G)[w].
\ee
The inverse of Corollary~\ref{AX-cor-kottler-inverse} solves $L_\varepsilon w-\sum_j\lambda_j\psi_j=F$. Its components and Lemma~\ref{AX-lem-improved-reference-error} give
\be
\label{AX-eq-binary-composed-estimates}
\|\mathcal R_G^wF\|_X\leq C\varepsilon^{-s_{\rm inv}}\|F\|_Y, \qquad |\mathcal R_G^\lambda F|\leq C\varepsilon^{3\gamma}\|F\|_Y, \qquad \|E_\varepsilon\|_Y\leq C\varepsilon^{p}.
\ee
The rescaled coefficients of $G$ have the uniform bounds proved above. Taylor's formula gives the unprojected quadratic estimate~\eqref{AX-eq-unprojected-quadratic-weight}. Applying the fixed projection $\Pi_G$ yields
\be
\label{AX-eq-binary-improved-quadratic}
\|\Pi_G(\mathcal Q_G(w)-\mathcal Q_G(z))\|_Y \leq C\varepsilon^{-s_{\rm quad}}(\|w\|_X+\|z\|_X)\|w-z\|_X.
\ee
Here the factor $\varepsilon^{-s_{\rm quad}}$ is the product of the factor $\varepsilon^{3\gamma}$ in the unprojected quadratic estimate and the factor $\varepsilon^{-s_{\rm inv}}$ in the projection bound, as in Corollary~\ref{AX-cor-improved-projection-quadratic}.
\ese

\bse
\emph{2. Contraction and the corrected family.} The metric component is a fixed point of
\be
\label{AX-eq-binary-affine-contraction}
\mathcal A_G(w)\coloneqq-\mathcal R_G^w\bigl(E_\varepsilon+\Pi_G\mathcal Q_G(w)\bigr), \qquad r_\varepsilon\coloneqq K\varepsilon^{p_{\rm cor}}.
\ee
Here $C_0$ bounds $\|\mathcal R_G^wE_\varepsilon\|_X/\varepsilon^{p_{\rm cor}}$ and $p_{\rm cor}=p-s_{\rm inv}$. We let $C_1$ bound the composition of the inverse of the tensor operator with the projected quadratic remainder after factoring out $\varepsilon^{-s_{\rm inv}-s_{\rm quad}}$. We fix $q\in(0,1)$ and $K>0$ satisfying $C_0/K+q/2<1$. On the radius-$r_\varepsilon$ ball, we obtain
\be
\label{AX-eq-binary-contraction-margin}
\begin{aligned}
\|\mathcal A_G(w)\|_X&\leq C_0\varepsilon^{p_{\rm cor}}+C_1\varepsilon^{-s_{\rm inv}-s_{\rm quad}}r_\varepsilon^2,\\
\|\mathcal A_G(w)-\mathcal A_G(z)\|_X
&\leq2C_1K\varepsilon^{p-2s_{\rm inv}-s_{\rm quad}}\|w-z\|_X.
\end{aligned}
\ee
The inequality $p>2s_{\rm inv}+s_{\rm quad}$ permits $\varepsilon$ to be chosen so that $2C_1K\varepsilon^{p-2s_{\rm inv}-s_{\rm quad}}\leq q$. The norm bound divided by $r_\varepsilon$ is then at most $C_0/K+q/2<1$. Thus the map preserves the ball and has contraction factor at most $q$. Its iterates from zero converge geometrically to the unique fixed point in this ball. We then put $\lambda=-\mathcal R_G^\lambda(E_\varepsilon+\Pi_G\mathcal Q_G(w))$. The two components of the same inverse give
\be
\label{AX-eq-binary-projected-family}
\Pi_GS_3(G+w)=\sum_{j=1}^2\lambda_j\psi_j, \qquad \|w\|_X\leq C\varepsilon^{p_{\rm cor}}, \qquad |\lambda|\leq C\varepsilon^{p+3\gamma}.
\ee
The projected quadratic term has size $\Ocal(\varepsilon^{2p_{\rm cor}-s_{\rm quad}})=o(\varepsilon^p)$, because $p>2s_{\rm inv}+s_{\rm quad}$. The coefficient bound follows by applying the coefficient component of the inverse of the linear operator with two additional scalar unknowns, whose norm is $\Ocal(\varepsilon^{3\gamma})$. Since $g-G=w$, the comparisons~\eqref{AX-eq-correction-comparison} give $\eta\leq C\varepsilon^{p_{\rm cor}+3\gamma}$ and $\eta_{\rm sep}\leq C\varepsilon^{p_{\rm cor}+\gamma}$ directly. In particular, $p_{\rm cor}+3\gamma>2$ places this ball in the space of positive metrics with the prescribed Cartesian regularity.
\ese

\bse
\emph{3. Continuity in the geometric parameters.} We justify continuity of the corrected family in $(C,D)$ for each fixed $\varepsilon>0$. We fix a base parameter $p_*=(C_*,D_*)$ and denote a nearby parameter by $p$. Smooth axis-preserving spatial maps $\Psi_p$ identify the coordinates near the inner horizons and equal the identity near the outer horizon. In the base spatial region, we choose disjoint smooth invariant cutoffs $\chi_j$, equal to one near the inner horizons and zero near the outer horizon. If $\kappa_j(p)$ is the corresponding reference gravity, we define
\be
\label{AX-eq-binary-continuous-angular-identifications}
a_p\coloneqq\exp\Bigl(\sum_{j=\pm}\chi_j\log\frac{\kappa_j(p)}{\kappa_j(p_*)}\Bigr), \qquad J_p(q_{\rm sp}+q_{tt}dt_E^2)\coloneqq q_{\rm sp}+a_p^2q_{tt}dt_E^2, \qquad I_p\coloneqq J_p^{-1}\Psi_p^*.
\ee
Here $q_{\rm sp}$ is the restriction of the tensor to the spatial directions. Near each inner horizon $a_p$ is constant. Dividing the coefficient of $dt_E^2$ by its square changes the angular coordinate $\kappa_j(p)t_E$ to $\kappa_j(p_*)t_E$ and preserves the matching of the radial and angular normal coefficients. Thus $I_p$ identifies the regular tensor spaces, including their pole conditions; it is the identity near the outer horizon. Its derivatives in space and in $p$ have finite bounds on a fixed positive-mass neighborhood, since its nonconstant factors are supported in the interior. The conjugated differential operators include derivatives of $a_p$ on that support and consequently depend continuously on $p$. The spaces of spatial vector fields orthogonal to the axial symmetry orbits and their dual one-forms are identified by spatial pullback: their normal coefficients have the same odd or even regular expansions independently of the angular coordinate. The energy inverse $A_G^{-1}$ is continuous in these fixed spaces by the inverse identity. Identifying vectors and one-forms by $G$, we define the projection
\be
\label{AX-eq-binary-domain-projections}
Q_G\coloneqq I-\mathcal T_G B_GA_G^{-1}\mathfrak b_G, \qquad \mathfrak b_G\mathcal T_GB_GA_G^{-1}=I.
\ee
The projection $Q_G$ therefore depends continuously on the parameters and has range $\ker\mathfrak b_G$. The identity uses $\mathfrak b_G=\delta_G\mathcal T_G$ and $\mathcal T_G^2=I$. By Proposition~\ref{AX-prop-centered-index}, the centering functionals are independent on this range. At a fixed parameter value, we choose two tensors in that range whose centering matrix is the identity. Applying $Q_G$ to their transported values gives a continuous family whose centering matrix remains invertible nearby. Inverting this finite matrix gives a continuous right inverse $Z_G$ for the centering functionals on $\ker\mathfrak b_G$. Hence $Q_G-Z_G\ell_GQ_G$ is a continuous projection onto the space satisfying both $\mathfrak b_Gw=0$ and the two centering conditions. The target is similarly identified by the continuous projection $\Pi_G$. Two projections $P_0,P$ that are sufficiently close in norm have isomorphic ranges. Indeed, $PP_0+(I-P)(I-P_0)$ is close to the identity and hence invertible by the geometric series; it intertwines the projections. These constructions identify both the domain and target with fixed spaces locally. The Kottler coefficients depend smoothly on $(C,D)$ in these identified charts. The linear operator with two additional scalar unknowns is continuous in those identifications, so its inverse is continuous by $T^{-1}-S^{-1}=T^{-1}(S-T)S^{-1}$. Comparing the two uniformly contracting fixed-point maps now proves continuity of $w$, $\lambda$, and the fluxes of the corrected metrics.
\ese

\bse
\emph{4. Selection of a zero of the two fluxes.} We apply the axial flux calculation to the corrected metric. We write $F_+,F_-$ for the two axial spatial fluxes of $S_3(g)$, with the orientations already fixed in~\eqref{AX-eq-spatial-flux}. On the separation annuli, the difference from the original intermediate metric has size $\Ocal(\varepsilon^{p_{\rm cor}+\gamma})$. We choose $c_{\rm bar}\geq r_0$ and $[1-r_0,1+r_0]\subset I_{\rm fl}$, so that the entire parameter square lies in the range of~\eqref{AX-eq-two-mass-uniform-flux-stability}. With $t=\varepsilon^2$, that estimate and the scaling identity~\eqref{AX-eq-flux-scaling} give
\be
\label{AX-eq-binary-corrected-two-fluxes}
\frac{F_\pm}{8\pi\varepsilon^4} =f_\pm(C,D)+\Ocal(\varepsilon^{p_{\rm bal}}).
\ee
Dividing $t^3+t\rho+\rho^2$ by $t^2$ gives $\Ocal(\varepsilon^2+\rho/\varepsilon^2+(\rho/\varepsilon^2)^2)$. Here $\rho=\Ocal(\varepsilon^{p_{\rm cor}+\gamma})$ and $p_{\rm bal}=p_{\rm cor}+\gamma-2>0$. Also $p_{\rm bal}<2$, since $p=(12-3\gamma)/5$, $\gamma>-1/4$ and $s_{\rm inv}>0$. This proves the error estimate. The leading functions $f_\pm$ are those of~\eqref{AX-eq-two-mass-leading-fluxes}, uniformly for the comparable masses. The continuous map
\be
\label{AX-eq-binary-actual-flux-map}
\mathfrak F_\varepsilon(C,D)\coloneqq \frac1{8\pi\varepsilon^4} \Bigl(F_++F_-,\frac{F_+}{\theta}-\frac{F_-}{1-\theta}\Bigr) =(C,D-D^{-2})+\Ocal(\varepsilon^{p_{\rm bal}})
\ee
has a zero in the interior of $\mathcal Q$ by Lemma~\ref{AX-prop-two-mass-balance}. That lemma also gives $|C|+|D-1|\leq C\varepsilon^{p_{\rm bal}}$ at every such zero. The defining relations $z_{\rm bar}=\varepsilon C$, $d=\varepsilon D$ give the balance estimates in the normalization $\ell=1$.
\ese

\bse
\emph{5. Energy control of the auxiliary vector field.} We apply the energy argument of Proposition~\ref{AX-lem-Einstein-recovery} to both obstructions. The reduced Einstein equation gives
\be
\label{AX-eq-binary-recovery-decomposition}
S_3(g)=T_\lambda+B_GV, \qquad T_\lambda\coloneqq\sum_{j=1}^2\lambda_j\psi_j, \qquad V\coloneqq A_G^{-1}\delta_GS_3(g).
\ee
We integrate over a fixed interval in $t_E$ of length $P_E>0$. The obstruction tensors are divergence-free with respect to $G$, so $\delta_GT_\lambda=0$ and
\be
\label{AX-eq-binary-selected-source-energy}
\langle T_\lambda,B_GV\rangle_{L^2(G)}=0, \qquad \|T_\lambda\|_{L^2(G)} \leq C|\lambda|(P_E/m)^{1/2}.
\ee
The tensor bounds of Proposition~\ref{AX-prop-global-force-columns} give the second estimate. The first follows by the integration by parts in Proposition~\ref{AX-lem-Einstein-recovery}, with vanishing tube contributions in the regular lifts.

The Bianchi pairing~\eqref{AX-eq-bianchi-pairing} uses the change of symmetric gradient
\be
\label{AX-eq-binary-gauge-metric-change}
B_gV=B_GV+\tfrac12\mathcal L_Vw, \qquad \|\tfrac12\mathcal L_Vw\|_{L^2(G)} \leq C\eta\Bigl(\|\nabla^GV\|_{L^2(G)} +\|\omega^{-1}V\|_{L^2(G)}\Bigr).
\ee
The bound is~\eqref{AX-eq-recovery-Lie-estimate}. The uniform vector estimate holds for the improved reference in its own angular lifts, and the inverse metrics and volume measures differ by $\Ocal(\eta)$. Thus~\eqref{AX-eq-recovery-absorption}, with $E\coloneqq\|B_GV\|_{L^2(G)}$, gives
\be
\label{AX-eq-binary-gauge-energy-absorption}
E^2\leq C\eta E^2+C\eta E\|T_\lambda\|_{L^2(G)}, \qquad E\leq C\eta\|T_\lambda\|_{L^2(G)} \leq C|\lambda|\eta(P_E/m)^{1/2}.
\ee
Since $\eta=\Ocal(\varepsilon^{p_{\rm cor}+3\gamma})$, absorption and division by $E$ when positive give the second inequality; it also holds for $E=0$.
\ese

\bse
\emph{6. Local flux estimates and vanishing of the coefficients.} The sphere flux is estimated on the enlarged separation annuli, where $\omega\asymp d$, the reference lapse is bounded below, and the rescaled coefficients have uniform bounds. The obstruction tensors satisfy
\be
\label{AX-eq-binary-selected-source-annulus}
d^2\|T_\lambda\|_{C^{1,\alpha}_d} \leq C|\lambda|\sqrt{d/m}.
\ee
Applying $\delta_g$ to~\eqref{AX-eq-binary-recovery-decomposition} and using $\delta_GT_\lambda=0$ gives the exact vector equation and its source estimate
\be
\label{AX-eq-binary-gauge-annular-equation}
\delta_gB_GV=-(\delta_g-\delta_G)T_\lambda, \qquad d^3\|\delta_gB_GV\|_{C^{0,\alpha}_d} \leq C|\lambda|\eta_{\rm sep}\sqrt{d/m}.
\ee
The divergence comparison in~\eqref{AX-eq-recovery-local-vector-equation} now carries the factor $\sqrt{d/m}$ from~\eqref{AX-eq-binary-selected-source-annulus}. The interior estimate~\eqref{AX-eq-recovery-local-regularity}, on a smaller annulus containing either sphere, gives
\be
\label{AX-eq-binary-gauge-annular-estimate}
\|V\|_{C^{1,\alpha}_d} \leq C\Bigl(P_E^{-1/2}d^{-1/2}E +|\lambda|\eta_{\rm sep}d^{-1}\sqrt{d/m}\Bigr).
\ee
The rescaling and volume factors are those of~\eqref{AX-eq-recovery-local-regularity}; the additional $\sqrt{d/m}$ affects only the source term. Both annuli lie in the positive-lapse interior.

The sphere estimate~\eqref{AX-eq-recovery-flux-from-energy} and~\eqref{AX-eq-binary-gauge-energy-absorption} therefore give, for $a\in\{+,-\}$,
\be
\label{AX-eq-binary-recovery-flux-loss}
|\operatorname{Flux}_a(B_GV)| \leq C|\lambda|(\eta+\eta_{\rm sep})\sqrt{d/m} \leq C\varepsilon^{p_{\rm cor}+3\gamma-1}|\lambda|.
\ee
Here the integration-length factors cancel, $d/m\asymp\varepsilon^{-2}$, and $\eta_{\rm sep}=\Ocal(\varepsilon^{p_{\rm cor}+\gamma})\leq C\eta$. The exponent $p_{\rm cor}+3\gamma-1$ is positive.

We set $\mathbf F\coloneqq(F_+,F_-)$ and use the matrix of axial fluxes $\mathsf F$ of Theorem~\ref{AX-thm-physical-force-inverse}, with its lower bound $f_*$. We fix $e_*\in(0,f_*)$. Taking the two fluxes in~\eqref{AX-eq-binary-recovery-decomposition} gives
\be
\label{AX-eq-binary-force-coefficient-detection}
\mathbf F=\mathsf F\lambda+ \bigl(\operatorname{Flux}_+(B_GV),\operatorname{Flux}_-(B_GV)\bigr), \qquad (f_*-C\varepsilon^{p_{\rm cor}+3\gamma-1})|\lambda|\leq|\mathbf F|.
\ee
For sufficiently small $\varepsilon$, we have $C\varepsilon^{p_{\rm cor}+3\gamma-1}\leq e_*$, and therefore $(f_*-e_*)|\lambda|\leq|\mathbf F|$. The selected zero of~\eqref{AX-eq-binary-actual-flux-map} has $\mathbf F=0$, since its two defining linear combinations form an invertible map when $0<\theta<1$. Both coefficients therefore vanish. The energy estimate gives $B_GV=0$, and~\eqref{AX-eq-binary-recovery-decomposition} gives $S_3(g)=0$. Conversely, $\lambda=0$ gives the same conclusion by that energy estimate and hence makes both axial fluxes zero. The zero of the axial flux map and the vanishing of the two coefficients are therefore equivalent throughout this small corrected family.
\ese

\emph{7. The static equations and regularity.} Taking the trace of $S_3(g)=0$ gives $\Scal_g=12$, and hence $\Ric_g=3g$. Because $g=h+N^2dt_E^2$, the warped Ricci formulas give the full static equations for $(h,N)$. The affine-domain argument already established the three nondegenerate horizons and the regular axis. Smoothness follows at fixed positive $\varepsilon$ in every regular lift: the equation $\Ric(G+w)-3(G+w)+B_G(\mathfrak b_Gw)^\sharp=0$ is a quasilinear elliptic system for $w$ when its relative size is small. Its principal symbol is a small perturbation of that of one half of the connection Laplacian on symmetric two-tensors at $G$. Differentiating in a smaller chart and applying the local Schauder estimate upgrades regularity one derivative at a time. The smooth reference and fixed chart coefficients permit iteration; the circle symmetries preserve the reconstructed spatial metric and odd lapse at a horizon. The local Riemannian extensions across the horizons therefore satisfy Proposition~\ref{AX-lem-local-horizon-extension}, which proves the local analytic horizon extensions. On the static interior, the Riemannian metric $g=h+N^2dt_E^2$ is Einstein, hence analytic in harmonic coordinates by DeTurck--Kazdan~\cite[Theorem~5.2, p.~258]{LS-DeTurckKazdan}. Its Killing field and orthogonal static sections are analytic, giving analytic spatial data in adapted coordinates. This analyticity holds across the transition radii of the reference as well: the corrected solution satisfies the Einstein equation throughout these regions.

The $j$th Kottler reference horizon has area $4\pi r_{H,j}^2=16\pi c_{h,j}m_j^2$. The relative induced-metric correction is $\Ocal(\eta)$, and $c_{h,j}=1+\Ocal(\varepsilon^2)$, whereas $\eta=\Ocal(\varepsilon^{p_{\rm cor}+3\gamma})$. The area expansion in the statement follows. For equal masses we restrict every operator and correction to the reflection-invariant class; the same contraction argument and proof of the Einstein equation then use the single obstruction tensor permitted by reflection symmetry and the scalar separation equation. Reflection imposes $z_{\rm bar}=0$, and the opposite signs at the two ends of the separation interval supply the required separation zero. Finally, multiplying spatial lengths by $\ell$ restores $\Lambda=3/\ell^2$, the masses and the balance relations in~\eqref{AX-eq-binary-existence-balance}. Multiplication of $N$ by a positive constant preserves the static equations and supplies the indicated maximum normalization.
\end{proof}


\subsection{Prescribed Kottler reference masses}
\label{AX-section-prescribed-Kottler-masses}

We select the two auxiliary masses together with the barycenter and separation to prescribe the Kottler parameters. This selection requires only joint continuity of the corrected family. We recall the conversion $\mu_j=m_j\sqrt{c_{h,j}}(1-Q_{H,j}/3)$ between a reference scale and its Kottler parameter. The matching estimates give the uniform conversion bound needed in the proof below.

\begin{proposition}[Existence with prescribed Kottler reference masses]
\label{AX-prop-prescribed-Kottler-masses}
Fix $0<\widehat\theta_0\leq1/2$, $0<\rho<1$, and $\delta_\mu>0$ satisfying $(1+\rho)\delta_\mu<\rho$. Fix target masses $\widehat\mu_\pm>0$, and define
\be
\label{AX-eq-prescribed-Kottler-targets}
\widehat M\coloneqq\widehat\mu_++\widehat\mu_-,\qquad \widehat\varepsilon\coloneqq(\widehat M/\ell)^{1/3},\qquad \widehat\theta\coloneqq\widehat\mu_+/\widehat M \in[\widehat\theta_0,1-\widehat\theta_0].
\ee
Set $\theta_0\coloneqq(1-\rho)\widehat\theta_0/(1+\rho)$ and $\varepsilon_{\max}\coloneqq(1+\rho)^{1/3}\widehat\varepsilon$. After decreasing the threshold $\varepsilon_0(\theta_0)$ of Theorem~\ref{AX-thm-binary-existence}, there is a constant $C_\mu>0$, depending only on this interval of mass fractions and the fixed construction choices, for which the following assertion holds. Assume
\be
\label{AX-eq-prescribed-Kottler-sufficient-bounds}
\varepsilon_{\max}<\varepsilon_0(\theta_0),\qquad C_\mu\varepsilon_{\max}^2\leq\delta_\mu.
\ee
This gives a smooth axisymmetric static solution for which the reference collars have the prescribed Kottler mass parameters $\widehat\mu_+$ and $\widehat\mu_-$. These parameters refer to the collars before the global nonlinear correction.
\end{proposition}
\begin{proof}
\bse
We work with $\ell=1$ and vary the four parameters in
\be
\label{AX-eq-prescribed-Kottler-parameter-box}
\mathcal B\coloneqq\prod_{j=+,-} [(1-\rho)\widehat\mu_j,(1+\rho)\widehat\mu_j]\times\mathcal Q.
\ee
For each pair, $\theta\in[\theta_0,1-\theta_0]$ and $\varepsilon\leq\varepsilon_{\max}$ hold. The corrected projected metrics therefore exist throughout this box. The dependence is jointly continuous. The horizon and spatial transports from Theorem~\ref{AX-thm-binary-existence} apply to $p=(m_+,m_-,C,D)$; the positive-mass reference coefficients and centering functionals depend smoothly on these parameters. The inverse identity gives continuity of the inverses of the linear operators with two additional scalar unknowns, and the fixed-point iterates converge uniformly because their contraction constants are bounded by $q<1$.

We write the normalized fluxes as $(C,h_b(D))+E(p)$, where $h_b(D)=D-D^{-2}$ and $\|E\|\leq\alpha_{\rm eq}<r_0$, and put $\phi_j(p)\coloneqq\sqrt{c_{h,j}}(1-Q_{H,j}/3)$. The matching formulas~\eqref{AX-eq-improved-reference-constants} give $|c_{h,j}-1|\leq C_h\varepsilon^2$ uniformly on this box. Decreasing $\varepsilon_0$ ensures $\varepsilon\leq1$ and $1/2\leq c_{h,j}\leq3/2$. The identity $Q_{H,j}=12c_{h,j}m_j^2$ then gives $Q_{H,j}\leq18\varepsilon^6$. Thus we may take $C_\mu\coloneqq C_h/(1+1/\sqrt2)+6\sqrt{3/2}$, since
\be
\label{AX-eq-prescribed-Kottler-conversion-bound}
|\phi_j-1|\leq\frac{|c_{h,j}-1|}{1+\sqrt{c_{h,j}}} +\frac{\sqrt{c_{h,j}}}{3}Q_{H,j} \leq C_\mu\varepsilon^2\leq\delta_\mu.
\ee
We consider the continuous map
\be
\label{AX-eq-prescribed-Kottler-four-component-map}
\Phi(p)\coloneqq(m_+\phi_+-\widehat\mu_+,\, m_-\phi_--\widehat\mu_-,\,C+E_1,\,h_b(D)+E_2).
\ee
On the two faces corresponding to $m_j$, its $j$th component has opposite strict signs with magnitude at least $[\rho-(1+\rho)\delta_\mu]\widehat\mu_j>0$. Its last two components have the corresponding strict signs on $\mathcal Q$.

We suppose that $\Phi$ has no zero and approximate its coordinate-clipped extension by a smooth nonvanishing map $\widetilde\Phi$ retaining the face signs. The three-form
\be
\label{AX-eq-prescribed-Kottler-angular-three-form}
\omega_{\rm ang}\coloneqq|y|^{-4}\sum_{j=1}^4(-1)^{j-1}y_j\, dy_1\wedge\cdots\wedge\widehat{dy_j}\wedge\cdots\wedge dy_4
\ee
is closed on $\mathbb R^4\setminus\{0\}$ because $\operatorname{div}(y/|y|^4)=0$. Thus Stokes' theorem gives a zero boundary integral for $\widetilde\Phi^*\omega_{\rm ang}$. Its boundary interpolation to
\be
\label{AX-eq-prescribed-Kottler-leading-four-component-map}
\Phi_0(p)\coloneqq(m_+-\widehat\mu_+,m_--\widehat\mu_-,C,h_b(D))
\ee
never vanishes, by the face signs. The boundary integral is consequently also the integral for $\Phi_0$. The image of the oriented boundary under $\Phi_0$ is the boundary of a rectangular box containing the origin; its integral equals the integral on a centered sphere, namely $2\pi^2$, a contradiction. Hence $\Phi$ has an interior zero. Its first two equations prescribe the two Kottler masses and its last two equations make the fluxes vanish. The Bianchi and integral estimates in the proof of Theorem~\ref{AX-thm-binary-existence} give the Einstein equation at this zero. We restore $\ell$ by scaling.
\ese
\end{proof}
For fixed $\rho,\delta_\mu$ and $\widehat\theta_0$, the sufficient target-mass condition is
\be
\label{AX-eq-prescribed-Kottler-dimensional-threshold}
\sqrt\Lambda\,(\widehat\mu_++\widehat\mu_-) <\frac{\sqrt3}{1+\rho} \min\{\varepsilon_0(\theta_0)^3, (\delta_\mu/C_\mu)^{3/2}\}.
\ee
The assertion prescribes reference parameters and does not assert uniqueness of the selected equilibrium or exhaustiveness of the constructed family.


\section{Critical points and geometry of the level set}
\label{AX-section-critical-geometry}

\subsection{The maximal circle and nearby levels}
\label{AX-section-maximal-circle}

We study the solutions constructed in Theorem~\ref{AX-thm-binary-existence}. We determine the maximum set and the axial critical points, describe their local analytic geometry, and derive the topology of the lapse levels. \begin{proposition}[A circle of maxima and nearby torus levels]
\label{AX-prop-lapse-maximum-circle}
Take a solution from Theorem~\ref{AX-thm-binary-existence} and retain its reference lapse normalization for now, before imposing $N_{\max}=1$. Set $\ell=1$, $\varepsilon=(m_++m_-)^{1/3}$ and $\theta=m_+/(m_++m_-)$. After decreasing the smallness threshold, the maximum set of $N$ is one smooth embedded circle in $M$, whose image is one interior point of the orbit domain $\Omega$. The Hessian of $N$ is negative definite in the two directions transverse to this circle. Every level $\{N=c\}$ with $c<N_{\max}$ sufficiently close to $N_{\max}$ is a smooth two-torus. In the reference coordinates $x=X/\varepsilon$, with axial coordinate $z$ and cylindrical radius $\varrho$, this maximum has
\be
\label{AX-eq-lapse-maximum-location-value}
\begin{aligned}
z_{\max}&=\frac{1-2\theta}{2}+\Ocal(\varepsilon^{p_{\rm bal}}),\qquad
\varrho_{\max}=\frac{\sqrt3}{2}+\Ocal(\varepsilon^{p_{\rm bal}}),\\
N_{\max}&=1-\varepsilon^2\Bigl(\frac32-\frac{\theta(1-\theta)}2\Bigr)
+\Ocal(\varepsilon^{2+p_{\rm bal}}).
\end{aligned}
\ee
All constants are uniform for $\theta\in[\theta_0,1-\theta_0]$. A constant rescaling of the lapse preserves its maximum set.
\end{proposition}
\begin{proof}
\emph{1. The minimum of the leading potential.} We write $q_+^0\coloneqq(1-\theta)e_z$, $q_-^0\coloneqq-\theta e_z$ and $r_j^0\coloneqq|x-q_j^0|$. The leading potential at barycenter zero and separation one is $f_\theta\coloneqq\theta/r_+^0+(1-\theta)/r_-^0+|x|^2/2$. We set $f_{\theta,*}\coloneqq3/2-\theta(1-\theta)/2$. The weighted identity $\theta(r_+^0)^2+(1-\theta)(r_-^0)^2=|x|^2+\theta(1-\theta)$ and the function $\phi(r)\coloneqq r^{-1}+r^2/2$ give
\bse
\be
\label{AX-eq-leading-maximum-factorization}
\phi(r)-\frac32=\frac{(r-1)^2(r+2)}{2r},\qquad f_\theta-f_{\theta,*} =\theta\frac{(r_+^0-1)^2(r_+^0+2)}{2r_+^0} +(1-\theta)\frac{(r_-^0-1)^2(r_-^0+2)}{2r_-^0}.
\ee
The minimum is therefore attained exactly where $r_+^0=r_-^0=1$. In coordinates on the orbit domain this is the single point $p_\theta\coloneqq((1-2\theta)/2,\sqrt3/2)$, whose rotation is a circle. Differentiating $f_\theta$ at this point gives
\be
\label{AX-eq-leading-maximum-hessian}
H_\theta\coloneqq D^2_{(z,\varrho)}f_\theta(p_\theta)
=\begin{pmatrix}
3/4&(3\sqrt3/4)(1-2\theta)\\
(3\sqrt3/4)(1-2\theta)&9/4
\end{pmatrix},\qquad
\det H_\theta=\frac{27}{4}\theta(1-\theta).
\ee
Since its trace is $3$, its least eigenvalue is at least $(9/4)\theta_0(1-\theta_0)>0$.
\ese

\emph{2. Localization of all global maxima.} We denote the reference lapse by $N_{\rm ref}$ and use the static Killing field $T=\del_{t_E}$. The relative tensor estimate gives $|w(T,T)|\leq C\eta N_{\rm ref}^2$. Since $N^2=N_{\rm ref}^2+w(T,T)$, we obtain
\bse
\be
\label{AX-eq-lapse-maximum-Czero}
|N-N_{\rm ref}|\leq C\eta N_{\rm ref},\qquad \|N-N_{\rm ref}\|_{C^0(M)}=\Ocal(\varepsilon^{p_{\rm cor}+3\gamma})=o(\varepsilon^2).
\ee
The reference lapse is uniformly bounded by its explicit formulas, and $p_{\rm cor}+3\gamma>2$. On an inner collar, $c_{N,j}\leq1$, $Q_{H,j}=\Ocal(\varepsilon^6)$ and $b_{Q_{H,j}}\leq1$ imply $N_{K,j}\leq(1+CQ_{H,j})(1-m_j/(r_j+m_j/2))$. The same inner deficit bound holds for the intermediate lapse on the transition. On the outer transition, both the intermediate lapse and the de Sitter lapse have a deficit bounded below by a constant times $r_{\rm out}^2$. Convex interpolation consequently gives
\be
\label{AX-eq-lapse-maximum-exclusion}
\begin{aligned}
1-N_{\rm ref}&\geq c\varepsilon^{4/3}
&&\text{if }\min_j r_j\leq2r_{\rm in},\qquad r_{\rm in}=\varepsilon^{5/3},\\
1-N_{\rm ref}&\geq c\varepsilon^{6/5}
&&\text{if }|X|\geq r_{\rm out},\qquad r_{\rm out}=\varepsilon^{3/5}.
\end{aligned}
\ee
These regions cannot contain a maximum, since on the reference orbit over $p_\theta$ the lapse is $1-\varepsilon^2f_{\theta,*}+o(\varepsilon^2)$. In the remaining region, $N_{\rm ref}(\varepsilon x)=1-\varepsilon^2f_{\theta,C,D}(x)$, where $f_{\theta,C,D}\coloneqq\theta/|x-q_+(C,D)|+(1-\theta)/|x-q_-(C,D)|+|x|^2/2$. We choose $L_*>\sup_\theta f_{\theta,*}$, then $R_*>0$ and $r_*>0$ with $R_*^2/2>L_*$ and $\theta_0/(2r_*)>L_*$. The estimates $C=\Ocal(\varepsilon^{p_{\rm bal}})$ and $D-1=\Ocal(\varepsilon^{p_{\rm bal}})$ imply that, for small enough $\varepsilon$, all maxima belong to the compact rescaled region $K_\theta\coloneqq\{|x|\leq R_*,\ \min_j|x-q_j^0|\geq r_*\}$. Indeed, outside this region the intermediate potential exceeds $L_*$, while~\eqref{AX-eq-lapse-maximum-Czero} is $o(\varepsilon^2)$.
\ese

\emph{3. Strict convexity in the orbit domain.} On $K_\theta$, the reference is exactly intermediate for sufficiently small $\varepsilon$, and $\omega\asymp\varepsilon$. The $X$ estimate, with $k\geq3$, transports to the reference Cartesian coordinates and gives a scaled $C^2$ lapse error $\Ocal(\varepsilon^{p_{\rm cor}+\gamma})$. The reference lapse is bounded below there, so taking the square root of its coefficient of $dt_E^2$ preserves that estimate. Since $p_{\rm cor}+\gamma=2+p_{\rm bal}$, and the centers differ from $q_j^0$ by $\Ocal(\varepsilon^{p_{\rm bal}})$, we obtain
\be
\label{AX-eq-lapse-maximum-Ctwo}
\Bigl\|\frac{1-N(\varepsilon x)}{\varepsilon^2}-f_\theta(x)\Bigr\|_{C^2(K_\theta)} \leq C\varepsilon^{p_{\rm bal}}.
\ee
We choose a small disk in the orbit domain about $p_\theta$, uniformly in $\theta$, on which the Hessian of $f_\theta$ is positive definite. The factorization~\eqref{AX-eq-leading-maximum-factorization}, compactness and~\eqref{AX-eq-lapse-maximum-Ctwo} place every maximum in this disk. The normalized deficit in~\eqref{AX-eq-lapse-maximum-Ctwo} is strictly convex there, so it has exactly one critical point. This point is interior and has $\varrho>0$; its orbit is one smooth embedded circle, and the Hessian of the lapse is negative definite transversely. The uniform Hessian lower bound and the $C^1$ part of~\eqref{AX-eq-lapse-maximum-Ctwo} give distance $\Ocal(\varepsilon^{p_{\rm bal}})$ from $p_\theta$. Evaluation at this point proves~\eqref{AX-eq-lapse-maximum-location-value}. The maximum outside the corresponding orbit neighborhood is strictly smaller, by compactness and uniqueness of the maximal circle. All levels sufficiently close to $N_{\max}$ therefore lie in that neighborhood. Strict convexity makes each nearby level in the orbit domain a smooth simple closed curve: every ray from the minimum meets that level once, with nonzero radial derivative. Over this disk the circle action is free, and $\varrho>0$. In the reference cylindrical coordinates $(z,\varrho,\varphi)$, the inverse image is the product of the disk with $S^1$. The level curve therefore lifts to $S^1\times S^1$.
\end{proof}


\subsection{\texorpdfstring{Local geometry at the axis and the critical sets}{Local geometry at the axis and the critical sets}}
\label{AX-section-local-critical-geometry}

\paragraph{The circle action and regular coordinates}
The spatial metric, the lapse and its level sets have different regularity properties. The construction gives analytic spatial data at every interior point, including the fixed axis and the critical set of $N$. A critical level can nevertheless be singular as a subset of this smooth manifold. We describe first the circle action at the axis, then three critical points supplied by the construction, and finally the normal geometry of the maximal circle. The local geometric identities retain a general $\Lambda>0$. In the asymptotic estimates, we take $\ell=1$, $\varepsilon\coloneqq(m_++m_-)^{1/3}$ and $\theta\coloneqq m_+/(m_++m_-)$, with $\theta\in[\theta_0,1-\theta_0]$. We keep the reference normalization of the lapse except where $N_{\max}=1$ is stated explicitly. Rescaling the lapse by a positive constant leaves its critical sets and the topology of its levels unchanged. Choose an interior point $p$ on the fixed axis. The analytic exponential map of $h$ identifies a neighborhood of $p$ with one of the origin in $T_pM$. Since rotations are isometries fixing $p$, the exponential map intertwines the circle action and its differential. In an orthonormal frame, this differential acts by rotations on a two-plane and fixes the perpendicular line. The circle action therefore has the following expression in analytic geodesic normal coordinates $(x,y,s)$:
\be
\label{AX-eq-local-axis-action}
(x+iy,s)\longmapsto(e^{i\varphi}(x+iy),s),\qquad N(x,y,s)=F(x^2+y^2,s),
\ee
where $F$ is analytic. Rotation invariance forces every homogeneous term in the convergent Taylor series of $N$ to be a polynomial in $x^2+y^2$ and $s$. The equations of the axis are $x=y=0$. The cylindrical coordinates $r\coloneqq\sqrt{x^2+y^2}$ and $\varphi$ are singular on the axis. In Cartesian coordinates, both the metric and the lapse are analytic there. The set $r=0$ is a boundary of the orbit domain, although its points lie in the interior of the spatial manifold. On the axis, $dN=F_s(0,s)\,ds$. Symmetry makes the transverse derivatives zero; the point is critical precisely when the derivative along the axis also vanishes. Under the condition $F_s(0,s)\ne0$, the implicit function theorem gives the nearby level in the form $s=\chi(x^2+y^2,c)$, where $\chi$ is analytic. The level crosses the axis as a smooth disk. At the intersection, its tangent plane is perpendicular to the axis. Its profile in the orbit domain meets the axis orthogonally; apparent endpoints of that profile are regular spatial points.

\paragraph{The Hessian at an axial critical point and its conical level}
Suppose $dN(p)=0$, and set $N_p\coloneqq N(p)>0$. By rotation invariance, there are real coefficients $a,b,\alpha,\beta$ for which, on writing $r_3\coloneqq(x^2+y^2+s^2)^{1/2}$, the Taylor expansion is
\be
\label{AX-eq-local-axis-taylor}
N=N_p+\tfrac12\bigl(a(x^2+y^2)+bs^2\bigr) +\alpha s^3+\beta s(x^2+y^2)+\Ocal(r_3^4),\qquad 2a+b=-\Lambda N_p.
\ee
In the orthonormal normal coordinates, the lapse equation $\Delta_hN(p)=-\Lambda N_p$ gives the last identity. The tensor equation also gives $a=N_p(\Ric_h(e_x,e_x)-\Lambda)$ and $b=N_p(\Ric_h(e_s,e_s)-\Lambda)$. These are geometric eigenvalues of the Hessian, rather than coordinate singularities.

When $a>0$ and $b<0$, $p$ is a nondegenerate saddle with two positive eigenvalues. There is an exact analytic description of its critical level. Apply the implicit function theorem to $F(t,s)=N_p$, solving for $t$. This gives $t=\psi(s)$ because $F_t(0,0)=a/2>0$. At the critical point, $\psi(0)=\psi'(0)=0$. Two differentiations give $\psi''(0)=-2b/a$. Taylor's formula now reads $\psi(s)=s^2g(s)$, where $g$ is analytic and $g(0)=-b/a>0$. The analytic change $\zeta\coloneqq s\sqrt{g(s)}$ is invertible and yields
\be
\label{AX-eq-local-axis-exact-cone}
\{N=N_p\}=\{x^2+y^2=\zeta^2\}.
\ee
This is a cone in analytic coordinates, not a claim that the spatial metric is Euclidean or conically singular. The metric remains positive definite at its vertex. The tangent cone in the original orthonormal coordinates has slope $r/|s|=\sqrt{-b/a}$. Equation~\eqref{AX-eq-local-axis-taylor} gives $-b/a=2+\Lambda N_p/a>2$. Thus the static equation constrains its opening, without making the critical level smooth. Since $N_p>0$, the Lorentzian static metric is also regular at this point. In a sufficiently small neighborhood, a level just below $N_p$ has two local sheets along the axis; a level just above it has a connected annular neck. To see this, solve $F(t,s)=c$ for $t=\psi(s,c)$. The function $\psi$ has positive second $s$-derivative near $(0,N_p)$. The implicit function theorem gives its unique local minimum at $s=s_0(c)$. Put $\psi_0(c)\coloneqq\psi(s_0(c),c)$. Taylor factorization gives a positive analytic coefficient $g(s,c)$ such that
\be
\label{AX-eq-local-axis-level-family}
\psi(s,c)=\psi_0(c)+(s-s_0(c))^2g(s,c),\qquad g>0,\quad \psi_0(N_p)=0,\quad \psi_0'(N_p)=2/a>0.
\ee
The change $\zeta\coloneqq(s-s_0(c))\sqrt{g(s,c)}$ gives the exact level equation $x^2+y^2=\zeta^2+\psi_0(c)$. Thus a negative $\psi_0(c)$ gives two sheets and a positive $\psi_0(c)$ gives the annular neck. The local topology changes at the saddle. The number of global connected components also depends on how these sheets connect outside the neighborhood.

\begin{proposition}[Existence and local geometry of three axial saddle points]
\label{AX-prop-three-axial-saddles}
For the solutions of Theorem~\ref{AX-thm-binary-existence}, and sufficiently small total mass with comparable mass ratio, there are at least three interior critical points of $N$ on the axis: one between the two inner horizons and one on each outer axis interval. Each has two positive transverse Hessian eigenvalues and one negative axial eigenvalue, and its critical level has the local form~\eqref{AX-eq-local-axis-exact-cone}. Their rescaled positions converge to the three axial critical points of the leading potential $f_\theta$ used in Proposition~\ref{AX-prop-lapse-maximum-circle}. Their exhaustion and, in the reflection-symmetric family, the ordering and level topologies are established below.
\end{proposition}
\begin{proof}
\emph{1. Critical points of the leading potential.} Put $z_+\coloneqq1-\theta$, $z_-\coloneqq-\theta$, $\theta_+\coloneqq\theta$ and $\theta_-\coloneqq1-\theta$. Away from the two centers, the following formulas hold on the axis:
\be
\label{AX-eq-axis-leading-derivatives}
\begin{aligned}
&S_\theta(z)\coloneqq\sum_{j=\pm}\frac{\theta_j}{|z-z_j|^3},\qquad \del_z f_\theta(z,0)=z-\sum_{j=\pm}\frac{\theta_j(z-z_j)}{|z-z_j|^3},\\
&\del_{zz}f_\theta(z,0)=1+2S_\theta(z),\qquad \del_{xx}f_\theta(z,0)=\del_{yy}f_\theta(z,0)=1-S_\theta(z).
\end{aligned}
\ee
The axial derivative is strictly increasing on each of the three intervals separated by $z_-,z_+$. Its limits at the endpoints of each interval have opposite signs, including the limits at infinity on the outer intervals. It consequently has exactly one zero on each interval.

At the middle zero, both distances to the centers are less than one, so $S_\theta>\theta_++\theta_-=1$. At an outer zero $z$, the critical equation gives
\be
\label{AX-eq-axis-leading-transverse-sign}
z(S_\theta(z)-1) =\sum_{j=\pm}\frac{\theta_jz_j}{|z-z_j|^3} =\theta(1-\theta)\bigl(|z-z_+|^{-3}-|z-z_-|^{-3}\bigr).
\ee
On the right interval the right-hand side and $z$ are positive; on the left interval both are negative. Thus $S_\theta>1$ also at the outer zeros. The Hessian of $f_\theta$ has one positive axial eigenvalue and two negative transverse eigenvalues at all three zeros. The zeros vary continuously with $\theta$ by strict positivity of the axial second derivative. For $\theta$ in the fixed compact mass-fraction interval, they remain in bounded sets away from the centers, and their Hessian eigenvalues remain uniformly separated from zero.

\emph{2. Persistence for the exact lapse.} Choose small disjoint Cartesian neighborhoods of these three zeros in rescaled coordinates. They lie in the intermediate reference region for sufficiently small $\varepsilon$. The same local estimate used in~\eqref{AX-eq-lapse-maximum-Ctwo}, applied on their union, gives
\be
\label{AX-eq-axis-critical-Ctwo}
\left\|\frac{1-N(\varepsilon x)}{\varepsilon^2}-f_\theta(x)\right\|_{C^2} \leq C\varepsilon^{p_{\rm bal}}.
\ee
Here the weight satisfies $\omega\asymp\varepsilon$, the reference lapse is bounded below, and the selected centers differ from their leading positions by $\Ocal(\varepsilon^{p_{\rm bal}})$, exactly as in the proof of~\eqref{AX-eq-lapse-maximum-Ctwo}. On each small axial interval, the axial derivative of the normalized deficit retains opposite endpoint signs and a strictly positive derivative. It has exactly one zero there. Rotation invariance makes the transverse derivatives vanish on the axis, so these zeros are genuine critical points in three dimensions. Their displacements are $\Ocal(\varepsilon^{p_{\rm bal}})$ by the uniform axial derivative bound. The $C^2$ estimate preserves the Hessian signatures, reversing those of $f_\theta$ when passing from the deficit to $N$. At a critical point the coordinate Hessian equals the covariant Hessian; its signature is independent of coordinates and of the positive definite metric used to represent it. Formula~\eqref{AX-eq-local-axis-exact-cone} now applies at each point.
\end{proof}

In the equal-mass family with reflection symmetry, the middle saddle lies in the fixed plane of the reflection that interchanges the two horizons. Choose the exponential coordinates so that the reflection maps $(x,y,s)$ to $(x,y,-s)$. At this saddle, the reflection forces the cubic coefficients $\alpha,\beta$ in~\eqref{AX-eq-local-axis-taylor} to vanish. This additional parity is a consequence of the equal-mass reflection; axial symmetry alone permits both cubic terms.

\paragraph{Critical values of the constructed solutions and level transitions}
Proposition~\ref{AX-prop-three-axial-saddles} concerns the exact solutions. It establishes three local saddle transitions, without identifying their critical values with those of $N_{\rm fig}$. Proposition~\ref{AX-prop-complete-critical-set}, below, excludes further critical points. The global progression of level topologies requires the ordering of all relevant critical values and the attachment of the local sheets. For the equal-mass reflection-symmetric solutions, the two outer saddle points obtained above are interchanged by reflection and have equal critical values. For unequal masses, the existence argument does not impose equality of these values. This distinguishes the simultaneous transitions in the symmetric auxiliary picture from a general solution with unequal masses.

\begin{proposition}[The three axial saddles and the maximal circle form the entire critical set]
\label{AX-prop-complete-critical-set}
After decreasing the total-mass threshold, the critical set of every constructed solution consists precisely of the three axial saddles of Proposition~\ref{AX-prop-three-axial-saddles} and the maximal circle of Proposition~\ref{AX-prop-lapse-maximum-circle}.
\end{proposition}
\begin{proof}
\emph{1. A derivative estimate for the lapse correction.} Work in the reference normalization $\ell=1$, and denote the improved reference lapse by $N_{\rm ref}$. The tensor estimate in~\eqref{AX-eq-binary-projected-family}, with the regular angular conditions, implies
\be
\label{AX-eq-lapse-global-Cone}
|\nabla_X(N-N_{\rm ref})| \leq C\varepsilon^{p_{\rm cor}}\omega^{\gamma-1}.
\ee
Here $\nabla_X$ denotes derivatives in reference spatial Cartesian coordinates. To justify this estimate up to an inner horizon, put $T\coloneqq\del_{t_E}$ and $e\coloneqq w(T,T)/N_{\rm ref}^2$, with $w$ the correcting tensor. In a regular polar lift, its angular component divided by the squared polar radius is an even regular function. The scaled tensor bounds, including two derivatives, therefore give $|e|+\omega|\nabla_X e|\leq C\varepsilon^{p_{\rm cor}}\omega^\gamma$. The same assertion holds in the outer horizon lift and at the two-plane poles. Since $N=N_{\rm ref}\sqrt{1+e}$ on the static side, differentiation proves~\eqref{AX-eq-lapse-global-Cone}, using the explicit reference bounds $N_{\rm ref}\leq C$ and $|\nabla_XN_{\rm ref}|\leq C/\omega$. In the remaining interior the argument is ordinary differentiation of the coefficient of $dt_E^2$. All comparisons have uniform constants in the fixed regular lifts.

\emph{2. Inner exclusion regions.} Fix a small number $r_*>0$. On the part of the intermediate region where $r_j\leq r_*\varepsilon$, differentiation in the outward radial direction from $p_j$ gives
\be
\label{AX-eq-inner-lapse-radial-sign}
\del_{r_j}N_I\geq \frac{m_j}{r_j^2}-C\varepsilon \geq c\frac{m_j}{r_j^2}.
\ee
The companion distance is comparable to $\varepsilon$, and $|X|\leq C\varepsilon$ there. Since $m_j\geq\theta_0\varepsilon^3$, the last inequality holds on choosing $r_*$ sufficiently small, independently of $\varepsilon$.

The exact Kottler collar $m_j/2\leq r_j\leq r_{\rm in}$ satisfies the same lower bound for $N_{K,j}$. In~\eqref{AX-eq-kottler-radial-fields}, differentiation of $(1-u_j)/(1+u_j)$ gives $m_j/(r_j+m_j/2)^2$. The remaining factor differs from its positive constant leading value by $\Ocal(r_j^2)$, with derivative $\Ocal(r_j)$; these estimates include $r_j=m_j/2$ by the explicit polynomial expression for $b_j$. Its contribution is smaller than $m_j/r_j^2$, since $r_{\rm in}^3/m_j=\Ocal(\varepsilon^2)$. On the transition $r_j\asymp r_{\rm in}$, the scaled matching bound~\eqref{AX-eq-improved-core-matching} divided by $m_j/r_{\rm in}$ tends to zero: its five terms give respectively $\Ocal(\varepsilon^{4/3})$, $\Ocal(\varepsilon^2)$, $\Ocal(\varepsilon^{4/3})$, $\Ocal(\varepsilon^{4/3})$ and $\Ocal(\varepsilon^2)$. Differentiation of the cutoff costs one factor $r_{\rm in}^{-1}$, so it preserves the lower bound~\eqref{AX-eq-inner-lapse-radial-sign} for the patched reference, after reducing $\varepsilon$.

In all these inner regions $\omega\asymp r_j$. Relative to $m_j/r_j^2$, the correction in~\eqref{AX-eq-lapse-global-Cone} is bounded by
\be
\label{AX-eq-inner-lapse-correction-ratio}
C\varepsilon^{p_{\rm cor}-3}r_j^{\gamma+1} \leq C\varepsilon^{p_{\rm cor}+\gamma-2} =C\varepsilon^{p_{\rm bal}}\longrightarrow0.
\ee
We used $\gamma+1>0$ and $r_j\leq r_*\varepsilon$. Consequently $\del_{r_j}N>0$ throughout both inner exclusion regions, including their horizon limits, and they contain no critical points.

\emph{3. Outer exclusion region.} Choose $R_*>0$ large. In the intermediate region with $r\geq R_*\varepsilon$, the two center distances are comparable to $r$, and
\be
\label{AX-eq-outer-lapse-radial-sign}
\del_rN_I\leq-r+C\varepsilon^3/r^2\leq-cr.
\ee
In the de Sitter region the exact formula~\eqref{AX-eq-vector-reference-models} gives $\del_rN_D=-r/(1+r^2/4)^2\leq-cr$ for $r\leq2$. On the outer transition $r\asymp r_{\rm out}$, the scaled difference between the intermediate lapse and $N_D$ is $\Ocal(\varepsilon^3/r+r^4)$. Its differentiated cutoff error, divided by $r$, is $\Ocal(\varepsilon^3/r^3+r^2)=\Ocal(\varepsilon^{6/5})$. Thus~\eqref{AX-eq-outer-lapse-radial-sign} holds for the patched reference everywhere with $r\geq R_*\varepsilon$. Since $\omega\asymp r$, the relative correction to this bound is at most
\be
\label{AX-eq-outer-lapse-correction-ratio}
C\varepsilon^{p_{\rm cor}}r^{\gamma-2} \leq C\varepsilon^{p_{\rm cor}+\gamma-2} =C\varepsilon^{p_{\rm bal}}\longrightarrow0.
\ee
The outer exclusion region therefore has $\del_rN<0$ and contains no critical point.

\emph{4. The compact intermediate region.} Every critical point now lies in a fixed compact rescaled region away from both leading centers. On that region the normalized deficit converges to $f_\theta$ in $C^2$, as in~\eqref{AX-eq-axis-critical-Ctwo}. The full critical set of $f_\theta$ consists of its three axial zeros and its minimizing circle. Here $r_j^0\coloneqq|x-q_j^0|$ denotes distance to the leading center, in rescaled coordinates. Indeed, at a point with cylindrical radius $\varrho>0$, its radial derivative is
\be
\label{AX-eq-leading-off-axis-critical}
\del_\varrho f_\theta=\varrho\left(1-\sum_{j=\pm}\frac{\theta_j}{(r_j^0)^3}\right).
\ee
Together with $\del_zf_\theta=0$, this gives $\sum_j\theta_j/(r_j^0)^3=1$ and $\sum_j\theta_jz_j/(r_j^0)^3=0$. Since $\theta_+z_+=-\theta_-z_-=\theta(1-\theta)>0$, the latter identity gives $r_+^0=r_-^0$; the former then gives $r_+^0=r_-^0=1$. This is precisely the minimizing circle.

Outside small neighborhoods of these four critical orbits, compactness gives a positive lower bound for $|\nabla f_\theta|$, uniform in the fixed mass-fraction interval. The $C^1$ convergence excludes critical points there. Near each axial saddle, the two transverse second derivatives of the deficit remain negative. Invariance and integration along a transverse radial segment then give $\del_\varrho[(1-N)/\varepsilon^2]<0$ for $\varrho>0$ in this neighborhood. Its critical points must therefore lie on the axis, where strict positivity of the axial second derivative gives the unique saddle already found. Near the minimizing circle, strict convexity of the deficit in the disk in the orbit domain gives the unique critical orbit already established in Proposition~\ref{AX-prop-lapse-maximum-circle}. These observations exhaust the compact region and prove the assertion.
\end{proof}

\begin{corollary}[Ordering of the saddle values for equal masses]
\label{AX-cor-equal-mass-critical-order}
In the equal-mass reflection-symmetric family, let $c_{\rm mid}^{\rm ex}$ and $c_{\rm out}^{\rm ex}$ be the values of the exact lapse at the middle saddle and either of the two outer saddles supplied by Proposition~\ref{AX-prop-three-axial-saddles}. For sufficiently small mass,
\be
\label{AX-eq-equal-mass-critical-order}
c_{\rm mid}^{\rm ex}<c_{\rm out}^{\rm ex}<N_{\max}.
\ee
The ordering is preserved by normalization to $N_{\max}=1$.
\end{corollary}
\begin{proof}
When $\theta=1/2$, the middle zero of the axial derivative of $f_{1/2}$ occurs at $z=0$, where $f_{1/2}(0,0)=2$. Let $z_o>1/2$ be its right outer zero, and put $f_o\coloneqq f_{1/2}(z_o,0)$. Since the axial second derivative is positive, $z_o$ is the minimum on the right outer interval. Evaluation at $z=1$ gives
\be
\label{AX-eq-equal-mass-critical-gap}
f_o\leq f_{1/2}(1,0)=\frac{11}{6}<2.
\ee
The minimum of $f_{1/2}$ over the full spatial region is $11/8$, attained precisely on the circle in Proposition~\ref{AX-prop-lapse-maximum-circle}. In particular $f_o>11/8$, since $z_o$ lies on the axis. The local $C^2$ estimates and the saddle displacement estimate give, in the reference normalization,
\be
\label{AX-eq-equal-mass-critical-values}
\begin{aligned}
&c_{\rm mid}^{\rm ex}=1-2\varepsilon^2+\Ocal(\varepsilon^{2+p_{\rm bal}}),\qquad c_{\rm out}^{\rm ex}=1-f_o\varepsilon^2+\Ocal(\varepsilon^{2+p_{\rm bal}}),\\
&N_{\max}=1-\tfrac{11}{8}\varepsilon^2+\Ocal(\varepsilon^{2+p_{\rm bal}}).
\end{aligned}
\ee
The strict gaps between the three leading coefficients prove~\eqref{AX-eq-equal-mass-critical-order}. Multiplication by a positive normalization constant preserves both inequalities.
\end{proof}

Thus the central local saddle transition occurs before the two simultaneous outer transitions as the level value increases in this family. This ordering concerns the exact equal-mass solutions. For unequal comparable masses there is no pair of outer saddles constrained by reflection to have one common value; Proposition~\ref{AX-prop-general-level-topology} uses their three values in increasing order. Proposition~\ref{AX-prop-equal-mass-level-topology}, below, determines the complete level topologies for the equal-mass solutions.

\begin{proposition}[Topology of every lapse level for equal masses]
\label{AX-prop-equal-mass-level-topology}
For the reflection-symmetric equal-mass solutions, normalize $N_{\max}=1$ and write $c_{\rm mid}^{\rm ex}$ and $c_{\rm out}^{\rm ex}$ for their two distinct axial critical values, with $0<c_{\rm mid}^{\rm ex}<c_{\rm out}^{\rm ex}<1$. Then the regular levels have the topologies
\be
\label{AX-eq-equal-mass-level-topologies}
\{N=c\}\simeq
\begin{cases}
S^2\sqcup S^2\sqcup S^2,&0<c<c_{\rm mid}^{\rm ex},\\
S^2\sqcup S^2,&c_{\rm mid}^{\rm ex}<c<c_{\rm out}^{\rm ex},\\
S^1\times S^1,&c_{\rm out}^{\rm ex}<c<1.
\end{cases}
\ee
The middle critical level is the disjoint union of the outer sphere and two spheres meeting at the middle saddle. The outer critical level is connected and has two conical points; topologically it is obtained from two disjoint spheres by identifying two distinct pairs of points, one pair at each outer saddle. The maximal level is the regular circle $\mathcal C$.
\end{proposition}
\begin{proof}
By Proposition~\ref{AX-prop-complete-critical-set}, the only critical values are $c_{\rm mid}^{\rm ex},c_{\rm out}^{\rm ex}$ and $1$. The levels between consecutive critical values are identified by the flow of $\nabla N/|\nabla N|^2$. On a closed band of regular values this vector field is smooth, and the band is compact. Its flow thus exists for the finite time needed to cross the band and increases $N$ at unit speed. The small positive levels consist of the three horizon spheres.

At the middle saddle, the two lower local sheets in~\eqref{AX-eq-local-axis-exact-cone} lie on opposite sides of the middle axis interval. The axial derivative has no other zero between this saddle and either inner horizon. Invariance makes the gradient tangent to the axis there, so these two branches continue to the two distinct inner horizon collars. As $c$ increases through $c_{\rm mid}^{\rm ex}$, their two spherical level components are joined by the annular neck described above. Removing one disk from each sphere and joining their boundary circles by this neck gives a sphere. The outer level sphere is unchanged. This proves the middle line of~\eqref{AX-eq-equal-mass-level-topologies} and the description at $c=c_{\rm mid}^{\rm ex}$.

Each outer saddle connects the component issuing from an inner horizon to the component issuing from the outer horizon, by the same axial-branch argument on its outer axis interval. Thus, just below $c_{\rm out}^{\rm ex}$, the two sheets at either saddle belong to the two different spheres just identified. At the common critical value these spheres meet at two distinct points. The two disjoint local surgeries replace those meetings by two necks. Removing two disks from each sphere leaves two cylinders, and joining their two boundary pairs by the necks gives a connected orientable closed surface of genus one. This is a torus, consistently with the maximal-circle neighborhood in Proposition~\ref{AX-prop-lapse-maximum-circle}. The regular-band flow extends this topology throughout $c_{\rm out}^{\rm ex}<c<1$.
\end{proof}

\begin{proposition}[Topology of regular lapse levels for general comparable masses]
\label{AX-prop-general-level-topology}
After normalizing $N_{\max}=1$, order the three axial critical values, with multiplicity, as $c_{(1)}\leq c_{(2)}\leq c_{(3)}$. For every nonempty interval indicated below, the regular-level topology is
\be
\label{AX-eq-general-level-topologies}
\{N=c\}\simeq
\begin{cases}
S^2\sqcup S^2\sqcup S^2,&0<c<c_{(1)},\\
S^2\sqcup S^2,&c_{(1)}<c<c_{(2)},\\
S^2,&c_{(2)}<c<c_{(3)},\\
S^1\times S^1,&c_{(3)}<c<1.
\end{cases}
\ee
An interval with equal endpoints is empty. Each distinct axial critical value gives a level with one conical singularity for each saddle having that value. At $c=0$ the level is the three horizons, and at $c=1$ it is the maximal circle.
\end{proposition}
\begin{proof}
The regular-band flow and local saddle description used in the preceding proof apply without reflection symmetry. The three axis intervals label the saddle attachments by the pairs of horizons $(H_-,H_+)$, $(H_-,H_c)$ and $(H_+,H_c)$. Indeed, each descending axial branch continues to the corresponding horizon without encountering any other critical point on that interval. Starting from disjoint horizon collars, passing one saddle attaches a single neck between the current level sheets associated with its two labels. The attachment data therefore form the three edges of a triangle whose vertices are the horizon collars.

The first edge joins two different spherical components and leaves the third sphere separate. Any second edge joins the remaining component to the first joined pair, giving a single sphere. The third edge adds a neck between two disks on that sphere, giving a torus. Equivalently, a connected collection of $v$ initial spherical components joined by $e$ such necks has genus $e-v+1$: each neck removes two disks and glues their boundary circles by a cylinder, reducing Euler characteristic by two. This gives genus zero for either connected tree and genus one for the completed triangle. Saddles at the same critical value have disjoint neighborhoods, so their neck attachments take place simultaneously. The empty intervals in~\eqref{AX-eq-general-level-topologies} indicate which intermediate regular stages are absent. There are no further transitions, by Proposition~\ref{AX-prop-complete-critical-set}. At a critical value, the local singularity is described by~\eqref{AX-eq-local-axis-exact-cone}.
\end{proof}

\paragraph{Expansion of the lapse in directions normal to the maximal circle}
Write $\mathcal C\coloneqq\{N=N_{\max}\}$ for the maximum circle from Proposition~\ref{AX-prop-lapse-maximum-circle}. Choose an orthonormal eigenframe in the normal plane at one point of $\mathcal C$, and transport it around the circle by the rotations. Use the normal exponential map to introduce local coordinates $(t,u,v)$, where $t$ parametrizes the circle and $\rho_\perp\coloneqq\sqrt{u^2+v^2}$. Differentiate $dN=0$ along $\mathcal C$. This gives $\Hess_hN(T,\cdot)=0$ for each vector $T$ tangent to the circle. The transverse Hessian is negative definite. There are positive constants $a_{\mathcal C},b_{\mathcal C}>0$ and a homogeneous cubic polynomial $P_3$ in $u,v$ for which the expansion reads
\be
\label{AX-eq-maximum-normal-expansion}
\begin{aligned}
N&=N_{\max}-\tfrac12(a_{\mathcal C}u^2+b_{\mathcal C}v^2)
+P_3(u,v)+\Ocal(\rho_\perp^4),\\
a_{\mathcal C}+b_{\mathcal C}&=\Lambda N_{\max},\qquad
w\coloneqq|\nabla^hN|_h^2
=a_{\mathcal C}^2u^2+b_{\mathcal C}^2v^2+\Ocal(\rho_\perp^3).
\end{aligned}
\ee
The orbit of any off-axis critical point is a circle consisting of critical points, since invariance carries $dN=0$ along the orbit. Such an orbit cannot be an isolated critical point of the three-dimensional manifold. Its image in $\Omega$ may nevertheless be isolated. In the transported frame the coefficients do not vary with $t$, since rotations preserve $h$ and $N$. Rotations of the normal $(u,v)$-plane need not preserve them: the axial action moves points along $\mathcal C$ and does not act as rotations of that plane at a fixed point. In this frame the cubic term is constant along the circle, but it can vary with the normal direction.

Write $d_N\coloneqq N_{\max}-N$. The expansion~\eqref{AX-eq-maximum-normal-expansion} then yields the following two directional limits:
\be
\label{AX-eq-maximum-gradient-directional-limits}
\lim_{u\to0}\frac{w(u,0)}{d_N(u,0)}=2a_{\mathcal C},\qquad \lim_{v\to0}\frac{w(0,v)}{d_N(0,v)}=2b_{\mathcal C}.
\ee
If $a_{\mathcal C}\ne b_{\mathcal C}$, no coefficient $A$ can satisfy $w=A d_N+o(d_N)$ uniformly over the normal directions. This anisotropy is present in the construction. At equal masses the matrix~\eqref{AX-eq-leading-maximum-hessian} has eigenvalues $3/4$ and $9/4$. The spatial metric in the intermediate region converges to the Euclidean metric on this scale, and~\eqref{AX-eq-lapse-maximum-Ctwo} identifies these with the limits of $a_{\mathcal C},b_{\mathcal C}$ in the reference lapse normalization with $\ell=1$. They remain distinct for sufficiently small mass.

\paragraph{An exact analytic normal form at the maximal circle}
There are also analytic normal coordinates $(U,V)$ for which $N=N_{\max}-U^2-V^2$ exactly. We give the factorization to distinguish this statement about the function from the orthonormal expansion~\eqref{AX-eq-maximum-normal-expansion}. Put $y\coloneqq(u,v)^T$. Taylor's formula, using the vanishing value and first derivative of $d_N$, gives
\be
\label{AX-eq-circle-analytic-quadratic-factor}
d_N(y)=y^TA_\perp(y)y,\qquad A_\perp(y)\coloneqq\int_0^1(1-\tau)D^2d_N(\tau y)\,d\tau.
\ee
The symmetric matrix $A_\perp$ is analytic and positive definite near zero, with $A_\perp(0)=\tfrac12\operatorname{diag}(a_{\mathcal C},b_{\mathcal C})$. Write $A_{ij}$ for its entries in this paragraph. Completing the square gives analytic functions
\be
\label{AX-eq-circle-exact-normal-form}
U\coloneqq\sqrt{A_{11}}\,u+\frac{A_{12}}{\sqrt{A_{11}}}\,v,\qquad V\coloneqq\sqrt{A_{22}-A_{12}^2/A_{11}}\,v,\qquad N=N_{\max}-U^2-V^2.
\ee
The Jacobian of $(U,V)$ in $(u,v)$ is invertible at zero, so these are regular analytic coordinates, together with the orbit variable $t$. The metric remains analytic and positive definite in them, with its normal coefficients at the circle equal to $2/a_{\mathcal C}$ and $2/b_{\mathcal C}$. Thus the exact quadratic expression for the function retains the directional dependence of $w$ through the metric coefficients. It does not imply spherical symmetry or constancy of $w$ on a lapse level.

\paragraph{Variation of the cubic coefficient along a surface of maxima}
For a smooth surface of maxima with negative nonzero normal Hessian, a signed normal variable $z$ can be chosen with $N=N_{\max}-z^2$. This follows by factoring the normal quadratic zero and taking the positive square root of its analytic coefficient. Set $M_z\coloneqq|\nabla z|_h^2$ and choose the normal $\nu\coloneqq\nabla z/\sqrt{M_z}$ to the $z$-levels, with mean curvature $H\coloneqq\operatorname{div}_h\nu$. In coordinates whose tangential variables are transported along $\nabla z/M_z$, the identity $\Delta_hz=H\sqrt{M_z}+\tfrac12\del_zM_z$ and the lapse equation give
\be
\label{AX-eq-maximal-surface-trace-expansion}
2M_z+2zH\sqrt{M_z}+z\del_zM_z =\Lambda N_{\max}-\Lambda z^2.
\ee
Writing $\Sigma\coloneqq\{z=0\}$ and $H_\Sigma\coloneqq H|_\Sigma$, evaluation and differentiation at zero give
\be
\label{AX-eq-maximal-surface-cubic}
\begin{aligned}
M_z|_\Sigma&=\tfrac12\Lambda N_{\max},\qquad
\del_zM_z|_\Sigma=-\tfrac23H_\Sigma\sqrt{\Lambda N_{\max}/2},\\
w&=2\Lambda N_{\max}z^2
-\tfrac83H_\Sigma\sqrt{\Lambda N_{\max}/2}\,z^3+\Ocal(z^4).
\end{aligned}
\ee
The cubic coefficient can vary with the point of $\Sigma$. Reversing the chosen normal and the signed variable reverses its sign. This is the hypersurface expansion relevant to comparison with the spherical maximum of Kottler. At the maximal circle, $\sqrt{N_{\max}-N}$ is instead a radial quantity in two normal dimensions and is not differentiable at the circle. A single regular signed variable $z$ with $N=N_{\max}-z^2$ would force the Hessian there to have rank at most one, contrary to its rank two. Thus the surface expansion cannot be applied to the circle. The constant coefficients along the axial orbit and their dependence on the normal direction are compatible with one another.

\paragraph{Regular horizons and poles, and conical singularities of the levels}
On a nondegenerate horizon $H_j$, $N=0$ and $dN\ne0$, so its points are regular points of the lapse as a boundary defining function. In the spatial double, Proposition~\ref{AX-lem-local-horizon-extension} gives the convergent odd expansion of $N$ and even expansion of $h$ in signed normal distance. At an axis--horizon intersection, the quotient has a corner. The spatial geometry has a smooth boundary with a regular rotation pole; the joint Cartesian conditions specified in Section~\ref{AX-section-boundaries} give this compatibility. In the auxiliary Riemannian extension there are two normal polar planes, one for the auxiliary circle and one for the axial circle.

A conical defect of the metric would instead give a transverse circumference-to-radius ratio $2\pi\beta$ with $\beta\ne1$, as in~\eqref{AX-eq-angle}. The solutions constructed here satisfy the full Cartesian regularity conditions and $\beta=1$ on every open axis arc. Their singular critical levels~\eqref{AX-eq-local-axis-exact-cone} describe singular subsets of an analytic Riemannian manifold, rather than defects of that manifold.

\section{Uniqueness from fixed parameters and boundary data}
\label{AX-section-uniqueness}


\subsection{Boundary conditions and uniqueness at fixed reference parameters}
\label{AX-section-fixed-reference-uniqueness}

We prove uniqueness for a fixed reference in the contraction ball and then for prescribed intrinsic geometry on one horizon. The second result uses analytic continuation and simple connectivity; it applies to existing regular solutions without a small-mass assumption. We specify which boundary quantities are fixed in the construction. The contraction then gives uniqueness within the same coordinate condition, centerings and correction ball. The construction imposes horizon and axis regularity in the reference angular coordinates. On each horizon, the spatial data have the parity~\eqref{AX-eq-horizon-data}, with the positive surface gravity fixed by the reference and the chosen normalization of $N$. The induced horizon metric $\gamma_j(0,\cdot)$ is an unknown. Along each axis arc, the circular orbit radius satisfies $R=r_\perp+\Ocal(r_\perp^3)$, with the Cartesian parity conditions of Section~\ref{AX-section-boundaries}; the joint conditions also hold at each horizon pole. The intrinsic horizon metrics are therefore determined as part of the solution.

\begin{lemma}[Uniqueness in the correction ball at a fixed reference metric]
\label{AX-prop-fixed-reference-local-uniqueness}
Fix the parameters $\Lambda,m_+,m_-,C,D$, the improved reference $G=h_0+N_0^2dt_E^2$ and its local angular lifts. Use the coordinate condition $\mathfrak b_Gw=0$ and the centerings~\eqref{AX-eq-complete-core-centering}. Assume the smallness conditions of Theorem~\ref{AX-thm-binary-existence}, and let $\mathcal B_G$ denote its closed contraction ball in the prescribed regular static and axial tensor space. There is at most one solution $g=G+w=h+N^2dt_E^2$, with $w\in\mathcal B_G$ and $N>0$ in the interior, of the static Einstein equations in the fixed auxiliary coordinate $t_E$.
\end{lemma}
\begin{proof}
We normalize $\ell=1$ and let $\widehat g=G+\widehat w$ be such a solution. Since $S_3(\widehat g)=0$, the nonlinear expansion~\eqref{AX-eq-binary-kottler-remainder} and the definition~\eqref{AX-eq-kottler-linearization} give
\be
\label{AX-eq-uniqueness-fixed-reference-equation}
L_\varepsilon\widehat w=-E_\varepsilon-\Pi_G\mathcal Q_G(\widehat w).
\ee
Applying the inverse of the linear operator with two additional scalar unknowns to the solution pair $(\widehat w,0)$ yields
\be
\label{AX-eq-uniqueness-fixed-reference-map}
\widehat w=-\mathcal R_G^w\bigl(E_\varepsilon+\Pi_G\mathcal Q_G(\widehat w)\bigr) =\mathcal A_G(\widehat w).
\ee
The map~\eqref{AX-eq-binary-affine-contraction} has at most one fixed point in $\mathcal B_G$. The positive square root of its coefficient of $dt_E^2$ determines $N$. Restoring $\ell$ proves the assertion.
\end{proof}

When the two axial fluxes vanish, this fixed point is the solution of Theorem~\ref{AX-thm-binary-existence}. This uniqueness statement compares tensors in the same affine space. A local classification of all nearby solutions would additionally require a theorem placing them in the chosen coordinate condition and centerings, and uniqueness of the zero selecting $(C,D)$. Continuity and the face signs of the reduced equations prove existence of a zero, but do not prove its uniqueness.


\subsection{Global uniqueness from the metric of one horizon}
\label{AX-section-horizon-uniqueness}

The full intrinsic metric of a horizon contains enough data to determine an existing regular solution. We first determine its analytic germ by the static equations and then extend the resulting local isometry across complete spatial doubles. Here, regularity at a nondegenerate horizon means the Cartesian regularity of the auxiliary angular metric in Proposition~\ref{AX-lem-local-horizon-extension}. That proposition gives an even analytic spatial metric and an odd analytic lapse in Gaussian normal coordinates. Interior analyticity follows from the Einstein metric $h+N^2dt_E^2$, as in the proof of Theorem~\ref{AX-thm-binary-existence}. These properties hold for all the solutions constructed above.

\begin{theorem}[Uniqueness from one intrinsic horizon metric]
\label{AX-thm-horizon-global-uniqueness}
Fix $\Lambda>0$ and two compact connected three-dimensional static data sets $(M_i,h_i,N_i)$ satisfying~\eqref{LS-eq-static}, with $N_i>0$ in $M_i^\circ$ and $N_i=0$ on $\del M_i$. Assume that their interiors are simply connected and that every boundary component is a regular nondegenerate horizon in the preceding sense. Choose connected horizon components $H_i$ and an isometry $\phi:(H_1,h_1|_{TH_1})\to(H_2,h_2|_{TH_2})$. Denote their positive surface gravities by $\kappa_i$. Then $\phi$ extends uniquely to an isometry $\Phi:M_1\to M_2$ satisfying
\be
\label{AX-eq-horizon-global-isometry}
\Phi^*h_2=h_1,\qquad \Phi^*(N_2/\kappa_2)=N_1/\kappa_1.
\ee
\end{theorem}
\begin{proof}
\emph{1. Determination of the analytic germ.} We first suppose that the two surface gravities are equal to $\kappa$. We identify the horizons by $\phi$ and use inward spatial normal distance $s$ to write $h=ds^2+\gamma(s)$. We set $k\coloneqq\tfrac12\del_s\gamma$ and $H\coloneqq\operatorname{tr}_\gamma k$. The tangential static equation and the scalar lapse equation give
\bse
\be
\label{AX-eq-horizon-metric-recursion}
\del_s(Nk)=N\bigl(\Ric_\gamma-\Lambda\gamma-Hk+2k\circ_\gamma k\bigr)-\Hess_\gamma N,
\ee
\be
\label{AX-eq-horizon-lapse-recursion}
\del_s^2N=-H\del_sN-\Delta_\gamma N-\Lambda N.
\ee
\ese
Here $(k\circ_\gamma k)_{ab}\coloneqq k_{ac}\gamma^{cd}k_{db}$. Indeed, the Gaussian identities are $\Ric_h|_{TH}=\Ric_\gamma-\del_sk-Hk+2k\circ_\gamma k$ and $\Hess_hN|_{TH}=\Hess_\gamma N+k\del_sN$. Substitution in~\eqref{LS-eq-static} gives the first recursion equation.

The parity and analyticity give convergent expansions
\be
\label{AX-eq-horizon-analytic-series}
\gamma(s)=\gamma_0+\sum_{j\geq1}\gamma_{2j}s^{2j},\qquad N(s)=\kappa s+\sum_{j\geq1}n_{2j+1}s^{2j+1}.
\ee
In~\eqref{AX-eq-horizon-metric-recursion}, the coefficient of $s^{2j-1}$ in $\del_s(\kappa s k)$ is $2\kappa j^2\gamma_{2j}$. All other terms at this order depend only on $\gamma_0,\ldots,\gamma_{2j-2}$ and $\kappa,n_3,\ldots,n_{2j-1}$. In particular, tangential differentiation preserves the power of $s$, and the factors $k=\Ocal(s)$ and $N=\Ocal(s)$ exclude $\gamma_{2j}$ from the right side at this order. Thus this equation determines $\gamma_{2j}$. The coefficient of $s^{2j-1}$ in~\eqref{AX-eq-horizon-lapse-recursion} then determines $n_{2j+1}$, with nonzero multiplier $2j(2j+1)$. For example, the first coefficients are
\be
\label{AX-eq-horizon-first-jets}
\gamma_2=\tfrac12(\Ric_{\gamma_0}-\Lambda\gamma_0),\qquad n_3=-\kappa K_{\gamma_0}/6,
\ee
where $K_{\gamma_0}$ is the Gaussian curvature of the horizon metric. Induction determines all coefficients, so analyticity gives equality on a collar. This proves uniqueness of the germ between two existing solutions; no existence for arbitrary $\gamma_0$ is used. If the surface gravities differ, we apply this argument to $N_i/\kappa_i$ instead.

\emph{2. Extension across the spatial doubles.} We double each $M_i$ along all its boundary components. The even spatial expansions and interior analyticity give an analytic metric $\widehat h_i$ on the closed double $DM_i$. Their analytic structures agree on overlaps, since the identity between two analytic presentations of the same metric is analytic in geodesic normal coordinates. We extend $N_i/\kappa_i$ with opposite signs on the two copies; its odd expansions give an analytic function $\widehat N_i$ on $DM_i$. Compactness makes $\widehat h_i$ complete. This construction doubles the spatial metric and permits different surface gravities on different horizons.

We lift the local isometry of collars to connected open sets in the universal covers of the doubles. The covers are connected, simply connected, complete analytic Riemannian manifolds. The continuation theorem for local isometries of complete simply connected analytic Riemannian manifolds~\cite{AX-KobayashiNomizu} 
extends this germ to a global isometry between them. The lifted functions $\widehat N_i$ agree under this isometry on the initial collar, hence everywhere by the identity theorem for analytic functions. Consequently, the isometry maps the positive component meeting that collar onto the corresponding positive component in the other cover.

Each such component projects as a connected covering onto $M_i^\circ$. Simple connectivity makes that restriction an isometric diffeomorphism. It follows that the two static interiors are isometric and their normalized lapses agree. Their metric completions are the original compact manifolds with boundary, so this isometry extends to $M_i$. Near the boundary, the extension agrees with the analytic continuation on the doubles and is therefore smooth. It extends the prescribed $\phi$. Finally, a boundary isometry and the inward unit normal determine the isometry on a collar; two extensions agree there and hence everywhere on a connected Riemannian manifold. This proves uniqueness of the extension.
\end{proof}

\paragraph{Determination of a solution by its horizon meridian profile}
We fix the topology of a closed ball with two disjoint open balls removed and mark one horizon component. We consider the regular static data at fixed $\Lambda$, modulo isometries preserving the marked component and multiplication of $N$ by a positive constant. Theorem~\ref{AX-thm-horizon-global-uniqueness} says that the map assigning to such data the isometry class of the induced metric on the marked horizon is injective. It thus parametrizes this class by the horizon metrics that actually occur. It does not characterize that image or assert existence for arbitrary horizon metrics. Prescribing only two masses, horizon areas or surface gravities specifies fewer data; uniqueness from these quantities has not been proved here.

For an axial horizon $H\simeq\mathbb S^2$, the intrinsic data are one profile and the length of the orbit interval. We use arc length $u$ on the orbit interval $H/S^1$ and an angle $\varphi$ of period $2\pi$. We denote the interval length by $L_H>0$ and the circular orbit radius by $a_H(u)$. The metric and pole conditions are
\bse
\be
\label{AX-eq-intrinsic-horizon-profile}
\gamma_0=du^2+a_H(u)^2d\varphi^2,\qquad 0\leq u\leq L_H,\qquad a_H(u)>0\quad(0<u<L_H),
\ee
\be
\label{AX-eq-intrinsic-horizon-pole-data}
a_H(0)=a_H(L_H)=0,\qquad a_H'(0)=1,\qquad a_H'(L_H)=-1.
\ee
\ese
To express full regularity, we use $v\coloneqq L_H-u$ near the second pole and analytic functions $A_\pm$ with $A_\pm(0)=1$. The profile satisfies $a_H(u)=uA_-(u^2)$ and $a_H(L_H-v)=vA_+(v^2)$ near the respective poles. The orbit interval is contractible, so a change of angular coordinate removes any mixed term in the intrinsic metric. Its primitive is even analytic at each regular pole, and the change extends there. Thus no additional angular parameter is needed. Interchanging the poles replaces $u$ by $L_H-u$. Among profiles realized by static solutions, $(L_H,a_H)$ therefore determines the solution up to the equivalences in the theorem. These pole conditions alone are not an existence criterion for a static filling.

For instance, an exactly round inner horizon of area $4\pi r_H^2$ with $0<\Lambda r_H^2<1$ would determine the compact Kottler region with that black-hole horizon, up to lapse rescaling. This region has two boundary components, whereas the constructed regions have three. The inner horizons of the constructed solutions therefore cannot remain exactly round. The exact Kottler geometry belongs to the reference collars, and the global nonlinear correction changes their induced horizon metrics.


\paragraph*{Acknowledgments.}

The first author (PLF) was partially supported by the research project ANR-23-CE40-0010-02: Einstein-PPF, entitled {\it ``Einstein constraints: past, present, and future''} funded by the Agence Nationale de la Recherche (ANR), as well as by the MSCA Staff Exchange Project 101131233: Einstein-Waves, entitled {\it ``Einstein gravity and nonlinear waves: physical models, numerical simulations, and data analysis'',} funded by the European Research Council. All of the figures in this paper were generated with the help of ChatGPT 6.

\appendix


\section{Proof of the local horizon extension}
\label{AX-appendix-horizon}

We prove Proposition~\ref{AX-lem-local-horizon-extension}. The argument gives constancy of surface gravity, the even and odd analytic expansions of the spatial metric and lapse, and an analytic local Lorentzian extension.

\begin{proof}
\bse
The undivided static equation gives $\Hess_hN=0$ at $H$. For vectors $X,Y$ tangent to $H$ and the inward unit normal $\nu$, we therefore have
\be
\label{AX-eq-local-horizon-principal-identities}
X(|\nabla^hN|_h^2)=2\Hess_hN(X,\nabla^hN)=0, \qquad \kappa\,h(\nabla_X\nu,Y)=\Hess_hN(X,Y)=0.
\ee
Nondegeneracy and connectedness prove the two assertions about $\kappa$ and the second fundamental form. The warped Ricci formulas give $\Ric_{G_H}=\Lambda G_H$ off $H$. Since $G_H$ extends as a $C^{2,\alpha}$ metric across $H$, its Ricci tensor is continuous, so the equation holds also on the fixed surface. Theorem~5.2 of DeTurck--Kazdan~\cite[Theorem~5.2, p.~258]{LS-DeTurckKazdan} makes this Einstein metric analytic in harmonic and geodesic normal coordinates.

The rotation action is analytic in these coordinates: about a fixed point its action is conjugate by the analytic exponential map to its linear differential. Its fixed surface and an analytic normal frame are consequently analytic. The normal exponential map from that frame gives analytic tubular coordinates. The reflection $\theta\mapsto-\theta$ is an isometry, since the static metric has no mixed components involving $d\theta$. Its fixed hypersurface contains a spatial section through $H$, and the normal exponential curves in that section are parametrized by spatial proper distance $s$. Rotation by $\pi$ sends $s$ to $-s$ while preserving the horizon coordinates. The tangential spatial coefficients are therefore even analytic functions of $s$. The squared length of the rotation field is $s^2$ times a positive even analytic function with value one at $s=0$. Its positive analytic square root gives $N=\kappa s b(s^2,x)$ on the spatial side, with $b(0,x)=1$. The Gauss lemma gives $h=ds^2+\gamma(s^2,x)$. These arguments hold in ordinary coordinates on $H$, including at an axial pole.

We set $X=s\cosh(\kappa t)$, $T=s\sinh(\kappa t)$ and $\sigma\coloneqq X^2-T^2$. Since the analytic function $b(\sigma,x)^2-1$ vanishes at $\sigma=0$, it can be written as $\sigma\beta(\sigma,x)$ with $\beta$ analytic. Direct substitution yields
\be
\label{AX-eq-local-Lorentzian-extension}
\mathbf g=dX^2-dT^2+\gamma(\sigma,x) -\beta(\sigma,x)(X\,dT-T\,dX)^2.
\ee
This is analytic and nondegenerate near $X=T=0$, and restricts to the static metric on $X>|T|$. The construction applies separately to each horizon, with the angle normalized by its own surface gravity.
\ese
\end{proof}


\section{Weighted estimates and the kernel of the Schwarzschild linearized operator}
\label{AX-appendix-Schwarzschild}


\subsection{Weighted estimates under the inner orthogonality conditions}
\label{AX-appendix-model-proof-1}

We prove the estimates and classify the kernel of the linearized operator in the zero-parameter limit of the rescaled Kottler collars. These results justify the centering conditions that remove the two inner displacement modes from the global tensor estimate. We prove Proposition~\ref{AX-prop-completely-centered-inverse-bound} by establishing a weighted estimate on the inner limiting metric and finite dimensionality of its decaying kernel. A rescaling contradiction then gives the global estimate under the centering conditions.

\begin{proof}
\bse
We first prove the fixed-model estimate
\be
\label{AX-eq-Schwarzschild-core-estimate}
\|w\|_{k+2,\alpha,\gamma;\rho_S} \leq C\Bigl(\|\mathcal P_Sw\|_{k,\alpha,\gamma-2;\rho_S} +\sup_{K_S}|w|\Bigr).
\ee
Outside a sufficiently large fixed radius, the Schwarzschild metric differs from the flat product in scaled derivatives by $\Ocal(|y|^{-1})$. Its curvature term is $\Ocal(|y|^{-3})$. The barrier proof of Proposition~\ref{AX-prop-global-tensor-remainders}, with the single radial weight $|y|^\gamma$, therefore applies on exterior annuli. To justify passage to infinity for a tensor merely bounded by that weight, we choose $\gamma<\gamma'<0$ and use $|y|^\gamma+\delta|y|^{\gamma'}$, with $\delta>0$. Both powers have a positive barrier inequality beyond a fixed radius, and the ratio of $|w|$ to their sum tends to zero at infinity. The maximum argument gives a bound with a constant independent of $\delta$. Letting $\delta$ decrease to zero gives the weighted $C^0$ estimate. The same local Schauder estimates, on the end and on the regular compact horizon region, prove~\eqref{AX-eq-Schwarzschild-core-estimate}.

For a bounded sequence in $\mathcal K_S$, these local estimates give a subsequence converging uniformly on $K_S$. The pairwise differences have zero image. Applying~\eqref{AX-eq-Schwarzschild-core-estimate} to them shows that the subsequence is Cauchy in the full weighted norm. The closed unit ball of $\mathcal K_S$ is therefore compact, which makes the space finite dimensional. In an infinite-dimensional normed space, one can successively choose unit vectors whose distance from the span of the preceding vectors is at least a fixed $c_{\rm sep}\in(0,1)$. Each such span is finite dimensional and closed; take a vector outside it, approximate its distance to the span, and rescale to obtain the required separation. Such a sequence has no convergent subsequence, which proves the implication used here. If the $L^2(K_S)$ norm of a kernel element is zero, continuity makes it zero on $K_S$, and~\eqref{AX-eq-Schwarzschild-core-estimate} makes it zero everywhere. This proves positive definiteness of the compact pairing and justifies the chosen basis.

We suppose that~\eqref{AX-eq-completely-centered-global-bound} fails along a sequence $\varepsilon_i\to0$. After normalization, we have $\|v_i\|_{k+2,\alpha,\gamma;\omega}=1$, with both the source norm and every centering functional tending to zero. The estimate~\eqref{AX-eq-global-tensor-core-estimate} forces a positive lower bound for $\mathcal C_{\varepsilon_i}(v_i)$. After passing to a subsequence, this lower bound holds on one fixed inner collar $j$. Its rescaled Kottler metric $\widehat G_{Q_j}$ converges to $G_S$ on every fixed compact region, since $Q_j\to0$ and $r_{\rm in}/m_j\to\infty$. The rescaled tensors $z_i\coloneqq m_j^{-\gamma}\widehat v_{i,j}$ satisfy
\be
\label{AX-eq-core-rescaled-linear-equation}
\mathcal P_{Q_j} z_i =m_j^{-\gamma}\Phi_{j,\varepsilon_i}^*(\mathcal P_Gv_i), \qquad \mathcal P_{Q_j}=\tfrac12\nabla_{\widehat G_{Q_j}}^*\nabla_{\widehat G_{Q_j}}-\mathcal R_{\widehat G_{Q_j}}.
\ee
The right side has $\widehat G_{Q_j}$ norm $c_{h,j}m_j^{2-\gamma}|\mathcal P_Gv_i|_G$, so it tends to zero in the scaled local source norm. The coefficients of $\mathcal P_{Q_j}$ converge smoothly to those of $\mathcal P_S$ on fixed compact sets. Uniform local estimates give a diagonal subsequence converging to a regular tensor $z$ on the complete Schwarzschild model. The global weighted bound passes to this limit and gives $|D^bz|\leq C\rho_S^{\gamma-b}$, with the corresponding H\"older bounds. Hence $z\in\mathcal K_S$.

The positive lower bound on the chosen inner collar passes to the compact limit, so $z\neq0$. On the other hand, the centering conditions pass to the limit and give $\int_{K_S}\langle z,w_a\rangle\,d\mu=0$. Since $z$ belongs to the span of the $w_a$, these equalities give $z=0$, a contradiction. This proves the uniform estimate. On the domain satisfying the orthogonality conditions its right side contains only the source norm, which gives injectivity and the inverse bound on the range. If a sequence of sources in that range converges, the estimate makes the corresponding tensors Cauchy in the closed Banach space defined by those conditions; its limit maps to the limiting source. The range is therefore closed.
\ese
\end{proof}


\subsection{\texorpdfstring{Classification of the decaying kernel of the Schwarzschild limiting operator}{Classification of the decaying kernel of the Schwarzschild limiting operator}}
\label{AX-appendix-model-proof-2}

We prove Proposition~\ref{AX-prop-Schwarzschild-kernel-classification} by classifying the regular decaying static axisymmetric tensors in the kernel of the inner limiting operator. Spherical harmonic decomposition reduces the calculation to radial equations and leaves only the axial displacement mode.

\begin{proof}
\bse
\emph{1. Trace and Bianchi conditions.} We first show that the gauge and trace equations follow from membership in the kernel. On a Ricci-flat metric, contraction of the operator and the linearized contracted Bianchi identity give
\be
\label{AX-eq-Schwarzschild-automatic-gauge}
\operatorname{tr}(\mathcal P_Sw)=\tfrac12\nabla^*\nabla\operatorname{tr}w, \qquad \mathfrak b_{G_S}(\mathcal P_Sw)=\tfrac12\nabla^*\nabla(\mathfrak b_{G_S}w).
\ee
For the second identity, the gauged variation formula gives $\mathcal P_Sw=\Ric'[w]+B(\mathfrak b w)^\sharp$; differentiation of $\mathfrak b_g\Ric_g=0$ at a Ricci-flat metric gives $\mathfrak b\Ric'[w]=0$, while commuting two covariant derivatives gives $\mathfrak b B\alpha^\sharp=\tfrac12\nabla^*\nabla\alpha$. The trace tends to zero at infinity, so its harmonic equation and the maximum principle give $\operatorname{tr}w=0$. The one-form $\mathfrak b w$ tends to zero as $\Ocal(R^{\gamma-1})$. The equation gives $\Delta|\mathfrak b w|^2=2|\nabla\mathfrak b w|^2$, and the maximum principle then implies $\mathfrak b w=0$. The horizon is an interior set in the regular lift and introduces no boundary in these arguments. Consequently, we have
\be
\label{AX-eq-Schwarzschild-kernel-TT}
\operatorname{tr}w=0,\qquad \delta w=0,\qquad \Ric'[w]=0.
\ee
\ese

\bse
\emph{2. Harmonic decomposition and a regular gauge.} We write $P_l$ for the Legendre polynomial of degree $l$, normalized by $P_l(1)=1$, and decompose $w$ into the axisymmetric spherical harmonics $Y_l(\theta)\coloneqq P_l(\cos\theta)$, with $L_l\coloneqq l(l+1)$ and $\Delta_{g_{\mathbb S^2}} Y_l=-L_lY_l$. The angular derivatives of these harmonics give the corresponding one-form and trace-free tensor harmonics. The reflection $\varphi\mapsto-\varphi$ excludes the odd tensor harmonics. Spherical projection commutes with the background operator; thus each projected tensor separately satisfies~\eqref{AX-eq-Schwarzschild-kernel-TT} and the prescribed decay. We can work at each fixed $l$ without requiring an estimate uniform in $l$. Smooth convergence of the harmonic expansion on compact sets then recovers the original tensor.

We recall that $G_S=F\,d\tau^2+F^{-1}dR^2+R^2g_{\mathbb S^2}$ with $F=1-2/R$. For $l\geq2$, a projected tensor has components
\be
\label{AX-eq-Schwarzschild-polar-components}
\begin{aligned}
w_{\tau\tau}&=F A Y_l,& w_{RR}&=F^{-1}B Y_l,& w_{Ra}&=j D_aY_l,\\
w_{ab}&=R^2\bigl(KY_l(g_{\mathbb S^2})_{ab}+G\mathbb Y_{ab}\bigr),&&
\mathbb Y_{ab}\coloneqq D_aD_bY_l+\tfrac12L_lY_l(g_{\mathbb S^2})_{ab}.
\end{aligned}
\ee
Here $a,b$ are angular indices, $D$ is the connection of $g_{\mathbb S^2}$, and the five scalar coefficients depend on $R$. A spatial vector orthogonal to the axial symmetry orbits $X=xY_l\del_R+zD^aY_l\del_a$ changes $G$, $j$ and $K$ by $2z$, $F^{-1}x+R^2z'$ and $2x/R-L_lz$, respectively. Therefore subtracting its Lie derivative with
\be
\label{AX-eq-Schwarzschild-regular-RW-gauge}
z=\tfrac12G,\qquad x=F\bigl(j-\tfrac12R^2G'\bigr)
\ee
sets $G=j=0$. This gauge is regular at the horizon. Indeed, with normal distance $s$, the quantities $R-2$, $A$, $B$, $K$ and $G$ are smooth functions of $s^2$, while regularity gives $w_{sa}=s$ times a smooth function of $s^2$. Since $dR/ds$ is $s$ times a nonzero smooth function, $j$ is also smooth in $R-2$. Hence $x=\Ocal(R-2)$, and $x\del_R$ is a smooth radial vector in the normal plane. At infinity, $j=\Ocal(R^{\gamma+1})$, $G=\Ocal(R^\gamma)$ and $G'=\Ocal(R^{\gamma-1})$, so $X=\Ocal(R^{\gamma+1})$ and $\mathcal L_XG_S=\Ocal(R^\gamma)$, with differentiated versions of these bounds. The resulting tensor retains its decay and solves the linearized Ricci equation.
\ese

\bse
\emph{3. Components of the linearized Ricci tensor.} We compute the components with arbitrary $A,B,K$ before specializing to $l\geq2$, since the same calculation will be used for the dipole. We define
\be
\label{AX-eq-Schwarzschild-general-polar-radial-coefficients}
\begin{aligned}
Q_l&\coloneqq R(R-2)A''+(2R-1)A'+2K'-B'-L_lA,\\
T_l&\coloneqq B+\frac{R-2}{2}(B'-A')
+\Bigl(\frac{L_l}{2}-1\Bigr)K
-(2R-3)K'-\frac{R(R-2)}2K''.
\end{aligned}
\ee
The $\tau\tau$, mixed and spherical components of the Ricci variation are
\be
\label{AX-eq-Schwarzschild-general-polar-Ricci-components}
\begin{aligned}
\Ric'[w]_{\tau\tau}&=-\frac{F}{2R^2}Q_lY_l,\\
\Ric'[w]_{RR}&=\Bigl(-\frac{A''}{2}-\frac{3F'}{4F}A'
+\bigl(\frac1R+\frac{F'}{4F}\bigr)B'
-K''\\
&\qquad-\bigl(\frac2R+\frac{F'}{2F}\bigr)K'
+\frac{L_l}{2FR^2}B\Bigr)Y_l,\\
\Ric'[w]_{Ra}&=\Bigl(-\frac{A'+K'}2+\frac{A+B}{2R}
-\frac{F'}{4F}(A-B)\Bigr)D_aY_l,\\
\Ric'[w]_{ab}&=T_lY_l(g_{\mathbb S^2})_{ab}
-\frac{A+B}{2}D_aD_bY_l.
\end{aligned}
\ee
For the mixed component, the spatial Codazzi equation gives $\Ric_h'[w]_{Ra}=(-K'/2+B/(2R))D_aY_l$, whereas variation of $N^{-1}\Hess_hN$ gives $(A'/2+F'(A-B)/(4F)-A/(2R))D_aY_l$. Their difference gives the mixed component. For the component in $\tau$, we use $\Ric_{\tau\tau}=-N\Delta_hN$. After removing the fixed spherical density, the variation of its radial lapse flux is $(K+(A-B)/2+R^2FA'/2)Y_l$, and its angular divergence is $-L_lAY_l/2$. Multiplication by the unperturbed factor $-F/R^2$ gives the first line.

For the angular component, the radial contribution to $\Ric_{ab}$ is the coefficient
\be
\label{AX-eq-Schwarzschild-angular-radial-Ricci}
1-fa'^2-af a''-aa'\Bigl(\frac{f'}2+f\frac{N'}N\Bigr)
\ee
times $(g_{\mathbb S^2})_{ab}$. Differentiate at $t=0$ after substituting $f=F(1-tBY_l)$, $a=R(1+tKY_l/2)$ and $N=\sqrt F(1+tAY_l/2)$. The factor $Y_l$ is held fixed during radial differentiation. The coefficient obtained is $[B+(R-2)(B'-A')/2-K-(2R-3)K'-R(R-2)K''/2]Y_l$. The intrinsic angular Ricci variation of $R^2(1+tKY_l)g_{\mathbb S^2}$ contributes $L_lKY_l(g_{\mathbb S^2})_{ab}/2$. In two dimensions, the conformal variation $KY_lg_{\mathbb S^2}$ contributes $-\Delta_{g_{\mathbb S^2}}(KY_l)g_{\mathbb S^2}/2$. The remaining angular Hessian terms are $-(A+B)D_aD_bY_l/2$. These follow by varying the normal coefficient $F^{-1}(1+tBY_l)$ in the spatial connection and the angular Hessian of $\sqrt F(1+tAY_l/2)$ in the lapse term. Thus they contribute $-BD_aD_bY_l/2$ and $-AD_aD_bY_l/2$, respectively, proving the last line of~\eqref{AX-eq-Schwarzschild-general-polar-Ricci-components}.

We also give the radial calculation, including the angular term. For the same functions $f,a,N$, the terms in $\Ric_{RR}$ containing only radial derivatives are
\be
\label{AX-eq-Schwarzschild-normal-radial-Ricci}
-\frac{N''}{N}-\frac{f'}{2f}\frac{N'}N -2\frac{a''}a-\frac{f'}f\frac{a'}a.
\ee
Indeed, the radial connection coefficients are $\Gamma^R_{RR}=-f'/(2f)$, $\Gamma^\tau_{R\tau}=N'/N$ and $\Gamma^a_{Rb}=a'\delta^a_b/a$. Substitution into the coordinate Ricci formula cancels the derivative and square of $\Gamma^R_{RR}$ and gives this expression. At first order its four terms contribute, respectively,
\be
\label{AX-eq-Schwarzschild-normal-radial-variation}
-\tfrac12A''-\tfrac{F'}{2F}A', -\tfrac{F'}{4F}A'+\tfrac{F'}{4F}B' -K''-\tfrac2R K', -\tfrac{F'}{2F}K'+\tfrac1R B',
\ee
each multiplied by $Y_l$. The remaining first-order angular contribution is $-\Delta_{g_{\mathbb S^2}}(BY_l)/(2FR^2)$: the varying connection is $\dot\Gamma^a_{RR}=-(2FR^2)^{-1}D^a(BY_l)$, and its angular divergence gives this term. All other terms containing an angular derivative in this component are quadratic in the variation. Adding these contributions proves the radial component of~\eqref{AX-eq-Schwarzschild-general-polar-Ricci-components}. Thus the calculation accounts for all potentially nonzero components in this gauge.
\ese

\bse
\emph{4. The components with $l\geq2$.} We specialize these equations to $l\geq2$. The trace-free angular component of $\Ric'[w]$ is $-\tfrac12(A+B)\mathbb Y_{ab}$. Thus $A=-B$ for $l\geq2$; we put $B\coloneqq H$. The mixed component and the lapse equation give
\be
\label{AX-eq-Schwarzschild-static-Legendre-system}
K'=H'+\frac{F'}F H, \qquad \bigl(R(R-2)H'\bigr)'-\bigl(L_l+\frac4{R(R-2)}\bigr)H=0.
\ee
For the spatial metric $F^{-1}(1+HY_l)dR^2+R^2(1+KY_l)g_{\mathbb S^2}$, the principal curvature of an $R$-sphere is $\sqrt F/R$ times $1+(RK'/2-H/2)Y_l$ to first order. Its Codazzi equation gives $\Ric_h{}_{Ra}=(-K'/2+H/(2R))D_aY_l$. With $N=\sqrt F(1-HY_l/2)$, direct differentiation gives $N^{-1}\Hess_hN{}_{Ra}=(-H'/2-F'H/(2F)+H/(2R))D_aY_l$, proving the first equation. In $\Delta_hN=0$, the perturbed radial flux $\sqrt{\det h}\,h^{RR}N_R$, after suppressing the fixed spherical density, is $1+(K-H-R^2FH'/2)Y_l$; the angular divergence contributes $L_lH Y_l/2$. Hence $(R^2FH')'-L_lH=2(K'-H')$. Substitution of the first equation and $2F'/F=4/(R(R-2))$ proves the second. Finally, the trace-free angular Ricci formula follows either by inserting~\eqref{AX-eq-Schwarzschild-polar-components} with $G=j=0$ into the connection variation or by collecting its angular Hessian terms: the two base coefficients contribute $-\tfrac12D_aD_b((A+B)Y_l)$, whereas a conformal variation of the two-dimensional sphere contributes only a multiple of $(g_{\mathbb S^2})_{ab}$.

Smoothness at the horizon gives $A(2)=B(2)$ with a fixed period of $\tau$, since these are the relative coefficients of the angular and radial directions of the normal plane. Thus $H(2)=0$ and $H=\Ocal(R-2)$. This horizon condition and the decay at infinity give a positive weighted identity throughout the full interval $-1<\gamma<0$. We multiply the radial equation by $H/R$ and integrate twice by parts on $[R_0,R_1]\subset(2,\infty)$. The integrations by parts give the identity
\be
\label{AX-eq-Schwarzschild-polar-weighted-positive-energy}
\bigl[(R-2)HH'+\bigl(\tfrac12-\tfrac1R\bigr)H^2\bigr]_{R_0}^{R_1} =\int_{R_0}^{R_1}\Bigl[(R-2)(H')^2 +\Bigl(\frac{L_l}{R}+\frac1{R^2}+\frac4{R^2(R-2)}\Bigr)H^2\Bigr]dR.
\ee
For the integration by parts, we set $p_{\mathrm{rad}}\coloneqq R(R-2)$ and $b\coloneqq1/R$. The potential after integration is $b(L_l+4/p_{\mathrm{rad}})-(p_{\mathrm{rad}}b')'/2$, with $(p_{\mathrm{rad}}b')'=-2/R^2$, and the boundary term is $bp_{\mathrm{rad}}HH'-p_{\mathrm{rad}}b'H^2/2$. These give exactly this expression. At the horizon, $H=\Ocal(R-2)$ and $H'=\Ocal(1)$ make that boundary term $\Ocal((R-2)^2)$. At infinity, the weighted derivative estimate gives $H=\Ocal(R^\gamma)$ and $H'=\Ocal(R^{\gamma-1})$. The boundary term is thus $\Ocal(R^{2\gamma})$, which tends to zero for every admissible $\gamma$. Both ends give convergent integrals: the terms are $\Ocal(R-2)$ or smaller near the horizon and $\Ocal(R^{2\gamma-1})$ or smaller at infinity. Pass to the limits $R_0\downarrow2$ and $R_1\to\infty$. The integral of the nonnegative integrand is zero, and consequently $H=0$. The factor $1/R$ is needed for this energy argument when $-1/2\leq\gamma<0$, since the unweighted boundary estimate would be $\Ocal(R^{2\gamma+1})$.

The explicit regular solution also identifies the growing mode excluded by this identity. We put $u\coloneqq R-1$ and use the polynomial $P_l$ characterized by
\be
\label{AX-eq-Schwarzschild-Legendre-polynomial}
P_l(1+t)=\sum_{q=0}^l\frac{(l+q)!}{2^q(q!)^2(l-q)!}\,t^q, \qquad \bigl((u^2-1)P_l'\bigr)'=L_lP_l.
\ee
The differential equation follows by substituting these coefficients and comparing successive powers. Differentiating it twice shows that $H_l(u)\coloneqq(u^2-1)P_l''(u)$ satisfies the second equation in~\eqref{AX-eq-Schwarzschild-static-Legendre-system}. It is positive for $u>1$, vanishes to first order at $u=1$, and grows like a positive multiple of $u^l$. The Wronskian $(u^2-1)(H'H_l-HH_l')$ is constant. Smoothness and first-order vanishing make that constant zero at $u=1$, so $H$ is a multiple of $H_l$. Its decay excludes a nonzero multiple. The first equation in~\eqref{AX-eq-Schwarzschild-static-Legendre-system} now makes $K$ constant, and decay gives $K=0$. Thus every component with $l\geq2$ is a Lie derivative along a smooth vector $X=\Ocal(R^{\gamma+1})$.
\ese

\bse
\emph{5. The dipole and spherical components.} For $l=1$, the trace-free harmonic $\mathbb Y$ is zero. We remove $K$ and $j$ by the vector determined by
\be
\label{AX-eq-Schwarzschild-dipole-gauge}
z(R)\coloneqq\frac1{R-2}\int_2^R\Bigl(\frac{(s-2)j(s)}{s^2}-\frac{K(s)}2\Bigr)ds, \qquad x(R)\coloneqq R\bigl(K(R)/2+z(R)\bigr).
\ee
The integral makes $z$ smooth at $R=2$, with $z(2)=-K(2)/2$, and therefore $x=\Ocal(R-2)$. The vector is smooth in the horizon lift. Because $\gamma>-1$, the integral and the weighted derivative bounds give $z=\Ocal(R^\gamma)$ and $x=\Ocal(R^{\gamma+1})$ at infinity. The formulas for its Lie derivative show that the resulting tensor has $K=j=0$. The remaining relative coefficients $A,B$ satisfy
\be
\label{AX-eq-Schwarzschild-dipole-equations}
\begin{aligned}
R(R-2)A'+(3-R)A-(R-1)B&=0,\\
(R-2)(B'-A')+A+3B&=0,\\
\bigl(R(R-2)A'\bigr)'+A'-B'-2A&=0.
\end{aligned}
\ee
The first equation is the mixed static equation, computed as above with arbitrary $A,B$ and $K=0$. The last follows from the radial lapse flux $1+(A-B+R^2FA')Y_1/2$ and its angular divergence $-A Y_1$. For the middle equation, the radial part of the angular Ricci variation is $[B+RF(B'-A')/2]Y_1g_{\mathbb S^2}$; its angular Hessian part is $-(A+B)D^2Y_1/2=(A+B)Y_1g_{\mathbb S^2}/2$. Their sum gives this equation. Differentiate the first equation and subtract the result from the third. This gives $(R-2)(A'+B')-A+B=0$. Combined with the middle equation, this yields $(R-2)B'=-2B$ and $(R-2)A'=A+B$. Substitution into the first gives $B=-3A$, whence $A=C(R-2)^{-2}$. Horizon smoothness forces $C=0$. Thus the dipole component is also a Lie derivative along a smooth vector satisfying $|X|=\Ocal(R^{\gamma+1})$.

We deal separately with $l=0$. We write its angular component as $R^2Kg_{\mathbb S^2}$ and put $c\coloneqq K(2)$. The tensor $cG_S$ solves the linearized Ricci equation. Subtracting this tensor and then the Lie derivative of $X\coloneqq R(K-c)\del_R/2$ removes the angular component. This vector is smooth at the horizon because $K-c=\Ocal(R-2)$. The remaining spherical tensor has relative coefficients $A,B$. The spherical vacuum equations, written for $a(R)d\tau^2+f(R)^{-1}dR^2+R^2g_{\mathbb S^2}$, are $(a/f)'=0$ and $(Rf)'=1$. Their linearizations give $(A+B)'=0$ and $(RFB)'=0$. Regularity at $R=2$ gives $B=0$; equality of the two normal coefficients there then gives $A=0$. Consequently the original spherical tensor is $cG_S+\mathcal L_XG_S$. Since $X=\Ocal(R)$ and $F'=\Ocal(R^{-2})$, the relative coefficient of $d\tau^2$ of $\mathcal L_XG_S$ tends to zero. The decay of the original tensor therefore gives $c=0$. Its trace equation reduces to $(R^2X^R)'=0$. For a smooth radial vector, $X^R(2)=0$. Hence $X=0$, and the spherical component is zero. The two spherical vacuum identities can also be checked directly: the difference of the normalized radial and $\tau\tau$ Ricci equations gives $(\log(a/f))'=0$, and the angular equation then gives $1-f-Rf'=0$.
\ese

\bse
\emph{6. Classification of the harmonic gauge fields.} It remains to identify the vector fields generating these Lie derivatives for $l\geq1$. For each such component we have $w=\mathcal L_XG_S$, with $X=xY_l\del_R+zD^aY_l\del_a$ smooth at the horizon and $|X|=\Ocal(R^{\gamma+1})$. The trace and Bianchi equations in~\eqref{AX-eq-Schwarzschild-kernel-TT} give $\operatorname{div}X=0$ and $\nabla^*\nabla X^\flat=0$. Since $G_S$ is Ricci flat, the latter is also the Hodge Laplace equation. We use the following notation:
\be
\label{AX-eq-Schwarzschild-harmonic-gauge-vector}
\alpha\coloneqq X^\flat=F^{-1}xY_l\,dR+R^2z\,dY_l, \qquad q\coloneqq(R^2z)'-F^{-1}x, \qquad d\alpha=q\,dR\wedge dY_l.
\ee
We have $\delta\alpha=0$ and $\delta d\alpha=0$. The radial component of the last equation is $-L_lqY_l/R^2=0$: this follows from the divergence formula for a two-form and $\sqrt{\det G_S}=R^2\sqrt{\det g_{\mathbb S^2}}$. Therefore $q=0$ and $\alpha=d(pY_l)$ with $p\coloneqq R^2z$. Its divergence equation becomes
\be
\label{AX-eq-Schwarzschild-harmonic-gauge-Legendre}
\bigl(R(R-2)p'\bigr)'-L_lp=0, \qquad p=\Ocal(R^{\gamma+2}).
\ee
The scalar $p$ and its derivative are bounded at $R=2$ by smoothness of the angular vector coefficient. The Wronskian with $P_l(R-1)$ is zero there; hence $p$ is proportional to this polynomial. The growth bound, together with $\gamma+2<2$, rules out all $l\geq2$. For $l=1$ it gives $X=c\nabla((R-1)\cos\theta)$. This proves $\mathcal K_S=\operatorname{span}\{w_{\rm disp}\}$.
\ese

\bse
\emph{7. Normalization of the displacement tensor.} The explicit Hessian gives the norm as follows. In an orthonormal frame ordered by the $\tau$, radial and two spherical directions, the nonzero components are
\be
\label{AX-eq-Schwarzschild-displacement-components}
(w_{00},w_{11},w_{22},w_{33})=\frac{2\cos\theta}{R^2}(1,1,-1,-1), \qquad w_{12}=-\frac{2\sqrt F\sin\theta}{R^2}.
\ee
These follow from $\Gamma^R_{\tau\tau}=-F/R^2$, $\Gamma^R_{RR}=-1/(FR^2)$, $\Gamma^R_{\theta\theta}=-FR$, and $\Gamma^\theta_{R\theta}=1/R$. Thus $|w_{\rm disp}|^2=R^{-4}(16\cos^2\theta+8F\sin^2\theta)$. Integration over the sphere and the auxiliary circle gives
\be
\label{AX-eq-Schwarzschild-displacement-Ltwo}
\|w_{\rm disp}\|_{L^2(G_S)}^2 =8\pi\,\frac{64\pi}{3}\int_2^\infty\frac{1+F}{R^2}\,dR =128\pi^2.
\ee
\ese
\end{proof}


\section{Quadratic flux calculation for the reference metric}
\label{AX-appendix-quadratic-flux}

We give the coefficient calculation used in Proposition~\ref{AX-prop-flux-stability}. The fields $U,Q,G_t$, the sphere $\Sigma_D$ and the annulus $\mathcal A_D$ are those of Section~\ref{AX-section-flux-proof}.

\bse
\emph{1. The quadratic tensor of the reference.} We compute $S_{3t}(G_t)$. For a metric $e^{2v}d\tau^2+e^{2u}\delta$, the spatial Christoffel coefficients are $\Gamma^k_{ij}=\delta^k_j u_i+\delta^k_i u_j-\delta_{ij}u_k$. Their contraction and the Hessian of $e^v$ give
\be
\label{AX-eq-conformal-flux-components}
\begin{aligned}
(\Ric_h)_{ij}&=-u_{ij}+u_i u_j-(\Delta u+|\nabla u|^2)\delta_{ij},\\
e^{-v}(\Hess_h e^v)_{ij}
&=v_{ij}+v_i v_j-u_i v_j-u_j v_i+(\nabla u\cdot\nabla v)\delta_{ij},\\
\Scal_g&=e^{-2u}\bigl(-4\Delta u-2|\nabla u|^2
-2\Delta v-2|\nabla v|^2-2\nabla u\cdot\nabla v\bigr).
\end{aligned}
\ee
For completeness, the derivative terms in the Ricci contraction are $\del_k\Gamma^k_{ij}-\del_j\Gamma^k_{ik}=-u_{ij}-\delta_{ij}\Delta u$, and the quadratic terms are $u_i u_j-\delta_{ij}|\nabla u|^2$. The Hessian formula follows by subtracting $\Gamma^k_{ij}\del_k e^v$ from $\del_i\del_j e^v$. The scalar formula is the trace of $\Ric_g=\Ric_h-e^{-v}\Hess_h e^v$ in the spatial directions, together with the component in the auxiliary coordinate $e^{-2v}(\Ric_g)_{\tau\tau}=-e^{-v}\Delta_h e^v$. In particular, for an arbitrary cosmological coefficient $\Lambda_{\rm loc}$, substituting these expressions in the Einstein tensor gives the exact spatial identity
\be
\label{AX-eq-exact-conformal-flux}
(S_{\Lambda_{\rm loc}})_{ij} =-\del_i\del_j(u+v)+u_i u_j-v_i v_j+u_i v_j+v_i u_j +\bigl(\Delta(u+v)+|\nabla v|^2+\Lambda_{\rm loc} e^{2u}\bigr)\delta_{ij}.
\ee

For $G_t$, Taylor expansion of $\tfrac12\log(1+t(2U-Q))$ and $\log(1-t(U+Q))$, with their first two derivatives on $\mathcal A_D$, yields
\be
\label{AX-eq-flux-logarithms}
\begin{aligned}
u&=t(U-Q/2)-t^2(U-Q/2)^2+\Ocal_I(t^3),\\
v&=-t(U+Q)-\tfrac12t^2(U+Q)^2+\Ocal_I(t^3),\\
u+v&=-\tfrac32tQ+t^2H_2+\Ocal_I(t^3),\qquad
H_2\coloneqq-\tfrac32U^2-\tfrac34Q^2.
\end{aligned}
\ee
Since $\Hess Q=\delta$, $\Delta Q=3$ and $\Lambda_{\rm loc}=3t$, the terms linear in $t$ in~\eqref{AX-eq-exact-conformal-flux} cancel. Expanding each quadratic gradient term and $\Lambda_{\rm loc}e^{2u}=3t+6t^2(U-Q/2)+\Ocal_I(t^3)$ gives
\be
\label{AX-eq-quadratic-flux-tensor}
\begin{aligned}
&(S_{3t}(G_t))_{ij} =t^2\bigl((\Delta H_2)\delta_{ij}-\del_i\del_jH_2+K_{ij}\bigr)+\Ocal_I(t^3),\\
&K_{ij}\coloneqq-2U_iU_j-2(U_iQ_j+Q_iU_j)+\tfrac14Q_iQ_j\qquad \hspace{13mm}+\bigl(|\nabla U+\nabla Q|^2+6U-3Q\bigr)\delta_{ij}.
\end{aligned}
\ee
The Hessian combination has zero flux: its divergence vanishes by commutation of derivatives, and a scalar cutoff extension through the ball gives the divergence theorem as in~\eqref{AX-eq-flat-zero-flux}. This argument is needed because $H_2$ itself is singular at the enclosed center.
\ese

\bse
\emph{2. Evaluation of the axial coefficient.} To compute the remaining flux, we differentiate $K$ away from the centers. Using $\Delta U=0$, $Q_{ij}=\delta_{ij}$ and $\Delta Q=3$, the three tensor terms contribute respectively $-2U_iU_{ij}$, $-8U_j-2Q_iU_{ij}$ and $Q_j$ to $\del_iK_{ij}$. The derivative of the scalar term is $2U_iU_{ij}+2Q_iU_{ij}+8U_j-Q_j$. Their sum vanishes. Consequently, the flux of $K$ through a sphere about $p_+^0$ is independent of its radius below $D$.

Writing $x=p_+^0+rn$, the smooth companion and cosmological gradients have the expansions
\be
\label{AX-eq-flux-centre-expansion}
\nabla U=-n/r^2+A+\Ocal_I(r),\qquad \nabla Q=P+rn,\qquad A\coloneqq-D^{-2}e_z,\qquad P\coloneqq(D/2)e_z.
\ee
The term quadratic in the singular gradient $-n/r^2$ is $r^{-4}(-2n_i n_j+\delta_{ij})$ and has zero vector flux by oddness. The terms of order $r^{-2}$ with a nonzero limiting flux are
\be
\label{AX-eq-flux-dipole}
r^{-2}\bigl(2n_i C_j+2C_i n_j-2(n\cdot C)\delta_{ij}\bigr), \qquad C\coloneqq A+P.
\ee
All remaining terms give $\Ocal_I(r)$ or smaller integrated contributions. Contraction of this tensor with $n^i$ gives $2C_j/r^2$, whose integral is $8\pi C_j$. Conservation then yields $\int_{\Sigma_D}K_{ij}n^i\,dA=8\pi B(D)(e_z)_j$. On the fixed proportional sphere, the uniform Taylor remainder in~\eqref{AX-eq-quadratic-flux-tensor} gives an additional $\Ocal_I(t^3)$.
\ese

\section{Estimates for the vector operator near the horizons} 
\label{AX-appendix-weighted-loss}

We verify the second estimate in Proposition~\ref{AX-lem-vector-weight-loss} for the field defined in~\eqref{AX-eq-weighted-loss-example}, using the notation and cutoffs of its proof.
\bse
We estimate the image in three regions. All derivative bounds below follow from the homogeneity of $\Gamma$ and the scaled bounds on the reference; they hold in $C^{k,\alpha}$ as well. On a cutoff annulus $r_j\asymp m$, the field has scaled norm at most $Cm^{-1}$ and the reference coefficients have bounded scaled norms. The scaled image on this cutoff annulus therefore satisfies
\be
\omega^{2-\beta}\|A_{G_\varepsilon}V_\varepsilon\|_{C^{k,\alpha}_\omega} \leq C m^{-1-\beta}=C m^{-\gamma}.
\ee
Inside the smaller cutoff balls the field is zero. In the region $4m\leq r_j$ and $|X|\leq C_*d$, for a fixed sufficiently large $C_*$, we have $A_0W=0$. If $r_*\coloneqq\min(r_+,r_-)$, then $|D^jW|\leq C_jr_*^{-1-j}$. The scaled metric difference from the flat product, including its first two derivatives, is bounded by $C(m/r_*+d^2)$. This follows directly from the Kottler radial fields, the intermediate fields and the inner cutoff formulas. The coordinate expression for $A_G$ therefore gives
\be
\omega^{2-\beta}|A_{G_\varepsilon}W| \leq C\bigl(mr_*^{-2-\beta}+d^2r_*^{-1-\beta}\bigr) \leq C\bigl(m^{-\gamma}+d^{2-\gamma}\bigr).
\ee
Here $\omega\asymp r_*$, $-2-\beta=-1-\gamma$ and $-1-\beta=-\gamma$.

For $C_*d\leq r\leq R_{\rm cut}$, cancellation between the two poles gives $|D^jW|\leq C_jd r^{-2-j}$. The reference differs from the flat product by at most $C(m/r+r^2)$ in the scaled $C^{k+2,\alpha}_r$ norm; the outer transition obeys this bound as well. Away from the fixed cutoff annulus we again use $A_0W=0$ and obtain
\be
r^{2-\beta}|A_{G_\varepsilon}W| \leq C\bigl(md r^{-3-\beta}+d r^{-\beta}\bigr) \leq C\bigl(md^{-1-\gamma}+d\bigr).
\ee
The derivatives of $\chi_f$ contribute at most $Cd$ on its fixed annulus. Since $m=\varepsilon^3$ and $d=\varepsilon D$, all these terms are bounded by $Cm^{-\gamma}$ for small $\varepsilon$. This proves the second inequality.
\ese

\end{document}